\documentclass[english]{article}

\usepackage[T1]{fontenc}
\usepackage[a4paper]{geometry}
\usepackage{color}
\usepackage{babel}
\usepackage{mathtools}
\mathtoolsset{showonlyrefs}

\usepackage[unicode=true,pdfusetitle,
 bookmarks=true,bookmarksnumbered=false,bookmarksopen=false,
 breaklinks=true,pdfborder={0 0 0},pdfborderstyle={},backref=false,colorlinks=true]
 {hyperref}

\usepackage{graphicx,psfrag,epsf}
\usepackage{enumitem}
\usepackage{longtable}

\usepackage[font=small,skip=0pt]{caption} 
\usepackage[numbers]{natbib}
\usepackage[utf8]{inputenc}

\usepackage{flowchart}
\usepackage{thmtools}
\usepackage{thm-restate}
\usepackage{tikz}

\usepackage{makeshape}
\usetikzlibrary{arrows.meta}
\usepackage{wasysym}
\usepackage{algorithm}
\usepackage{comment}
\usepackage[noend]{algpseudocode}
\usepackage{amsmath, amssymb, amsfonts, amsthm, nccmath}

\pgfdeclarelayer{nodelayer}
\pgfdeclarelayer{edgelayer}
\pgfsetlayers{nodelayer,edgelayer} 

\numberwithin{equation}{section}
\theoremstyle{plain}
\newtheorem{definition}{Definition}

\newtheorem{remark}{Remark}
\newtheorem{corollary}{Corollary}

\newtheorem{lemma}{Lemma}

\newtheorem{assumption}{Assumption}
\newtheorem{proposition}{Proposition}
\newtheorem{theorem}{Theorem}
\newtheorem*{theorem*}{Theorem}
\newcommand{\+}{\mathbf}
\newcommand{\bs}{\boldsymbol}
\newcommand{\Ps}{\mathbb{P}_n^{\bs \theta^*}}

\newcommand{\asls}{ \underset{a.s.}{\overset{\bs\theta^*}{\longrightarrow}} } 
\newcommand{\E}{\mathbb{E}}
\renewcommand{\circ}{ \odot}



\title{Consistent and fast inference in compartmental models of epidemics
using Poisson Approximate Likelihoods}

\author{Michael Whitehouse \\ School of Mathematics, University of Bristol\\
Nick Whiteley\\ School of Mathematics, University of Bristol\\
Lorenzo Rimella\\Department of Mathematics and Statistics, Lancaster University
}

\begin{document}

\maketitle
\begin{abstract}
Addressing the challenge of scaling-up epidemiological inference
to complex and heterogeneous models, we introduce Poisson Approximate
Likelihood  (PAL) methods. In contrast to the popular ODE approach to compartmental modelling, in which a large population limit is used to motivate a deterministic model, PALs are derived from approximate filtering equations for finite-population, stochastic compartmental models, and the large population limit drives consistency of maximum PAL estimators. Our theoretical results appear to be the first likelihood-based parameter estimation consistency results which apply to a broad class of partially observed stochastic compartmental models and address the large population limit. PALs are simple to implement, involving only elementary arithmetic operations and no tuning parameters, and fast to evaluate, requiring no simulation from the model and having computational cost independent of population size. Through examples we demonstrate how PALs can be used to: fit an age-structured model of influenza, taking advantage of automatic differentiation in Stan; compare over-dispersion mechanisms in a model of rotavirus by embedding PALs within sequential Monte Carlo; and evaluate the role of unit-specific parameters in a meta-population model of measles.
\end{abstract}
\setlength{\abovedisplayskip}{6pt}
\setlength{\belowdisplayskip}{6pt}
\section{Introduction}
Compartmental modelling is one of the most widespread methods for
quantifying the dynamics of infectious diseases in populations, rooted in the works of McKendrick and Kermack in the 1920's \citep{m1925applications,kermack1927contribution},
\cite{bartlett1949some,bartlett1966introduction}
and  \cite{kendall1956deterministic}, see \citep{isham2005stochastic}
for an  overview. In this modelling paradigm individuals in a population transition between a
collection of discrete compartments, usually representing disease
states, where the rates of transition may depend on the current state
of the population as a whole as well as possibly unknown parameters. This provides
an interpretable, mechanistic framework in which to infer epidemic
characteristics such as reproduction numbers, forecast disease dynamics
and explore the possible impacts of public health interventions. Compartmental models are also popular in ecology and biochemistry, for example, \citep{komorowski2009bayesian,fearnhead2014inference},  but that is beyond the scope of the present work.

The earliest formulated compartmental
models of epidemics consist of a small number of compartments, just
three in the standard Susceptible-Infected-Recovered (SIR) model. Modern
compartmental models often feature many more compartments,
each corresponding to some combination of disease state and other
variates. By increasing the number of compartments, the modeller can specify a more precise representation
of complex diseases and populations, such as multi-strain dynamics \citep{worden2017products},
subpopulations associated with, e.g., households or age-groups \citep{andrade2020evaluation},
and spatial information \citep{xia2004measles}. Modelling such features is considered a key challenge by
epidemiologists \citep{ball2015seven,funk2015nine,riley2015five,wikramaratna2015five}. 


However, the computational cost of fitting compartmental
models to data, in general, grows with the number of compartments and also, in the
cases of some methods, with the population size.  Exact likelihood-based inference is intractable in general and approximate inference typically either involves deleterious model simplifications or involves highly sophisticated algorithms which incur a substantial computational cost. Thus scaling-up inference to complex models is an important and open challenge -- this is the motivation for the present work.


Compartmental models come in various forms, some stochastic,
some deterministic; some in continuous time, some in discrete time;
some modelling finite populations, some motivated by large population
asymptotics. Deterministic, ODE-based compartmental models are very popular in practice and often motivated by the fact they can be obtained from finite-population stochastic models in the large population limit. As a very simple example, consider the continuous-time, stochastic version of the SEIR model, with fixed population size $n$ and  numbers of susceptible, exposed, infective and removed individuals
denoted $X_{t}^{(n)}\coloneqq[S_{t}^{(n)}\,E_t^{(n)}\,I_{t}^{(n)}\,R_{t}^{(n)}]^{\top}$.  Each susceptible individual becomes exposed at instantaneous
rate $\beta n^{-1}I_{t}^{(n)}$, each exposed individual becomes infective at rate $\rho$, each infective individual is
``removed'' at rate $\gamma$, and $(X_{t}^{(n)})_{t\geq0}$ is a jump-Markov
process.  General results concerning the convergence
of jump-Markov processes to the solutions of ODE's \citep{kurtz1970solutions,kurtz1971limit}
can be applied to show that, if $n^{-1}X_{0}^{(n)}\to x_{0}$ in probability, then for
any $T>0$ and $\delta>0$,
\begin{equation}
\lim_{n\to\infty}\mathbb{P}\left(\sup_{0\leq t\leq T}\|n^{-1}X_{t}^{(n)}-x_{t}\|>\delta\right)=0,\label{eq:kurtz_lln}
\end{equation}
where $(x_{t})_{t\geq0}$ , $x_{t}\equiv[s_{t}\;e_t\;i_{t}\;r_{t}]^{\top}$,
solves:
\begin{equation}
\dfrac{\mathrm{d}s_{t}}{\mathrm{d}t}=-\beta s_{t}i_{t},\qquad\dfrac{\mathrm{d}e_{t}}{\mathrm{d}t}=\beta s_{t}i_{t}-\rho e_{t},\qquad\dfrac{\mathrm{d}i_{t}}{\mathrm{d}t}=\rho e_{t}-\gamma i_{t},\qquad\dfrac{\mathrm{d}r_{t}}{\mathrm{d}t}=\gamma i_{t}.\label{eq:ode}
\end{equation}
It follows from $S_{0}^{(n)}+E_0^{(n)}+I_{0}^{(n)}+R_{0}^{(n)}=n$ together
with $n^{-1}X_{0}^{(n)}\to x_{0}$ and \eqref{eq:ode}, that $s_{t}+e_t+i_{t}+r_{t}=1$
for all $t\geq0$. In order to use this ODE to model a population
of size $n$, $x_{t}$ is scaled back up by a factor of $n$,
$x_{t}^{(n)}\equiv[s_{t}^{(n)}\;e_t^{(n)}\;i_{t}^{(n)}\;r_{t}^{(n)}]^{\top}\coloneqq n[s_{t}\;e_t\;i_t\;r_{t}]^{\top}$,
which satisfies the form of SEIR ODE usually encountered in practice:
\begin{equation}
\dfrac{\mathrm{d}s_{t}^{(n)}}{\mathrm{d}t}=-\beta s_{t}^{(n)}\frac{i_{t}^{(n)}}{n},\quad\dfrac{\mathrm{d}e^{(n)}_{t}}{\mathrm{d}t}=\beta s_{t}^{(n)}\frac{i_{t}^{(n)}}{n}-\rho e^{(n)}_{t}, \quad\dfrac{\mathrm{d}i_{t}^{(n)}}{\mathrm{d}t}=\rho e^{(n)}_{t}-\gamma i_{t}^{(n)},\quad\dfrac{\mathrm{d}r_{t}^{(n)}}{\mathrm{d}t}=\gamma i_{t}^{(n)}.\label{eq:ODE_rescaled}
\end{equation}

To relate $(X_{t}^{(n)})_{t\geq0}$ or $(x_{t}^{(n)})_{t\geq0}$ to
data, for example, error-prone measurements of the number of newly infective
individuals in given time periods, one usually postulates a probabilistic
observation model, and evaluation of the likelihood function for the
parameters $(\beta,\rho,\gamma)$ then involves marginalizing out $(X_{t}^{(n)})_{t\geq0}$
in the case of the finite population stochastic model, which is intractable, or numerical
approximation to $(x_{t}^{(n)})_{t\geq0}$ in the case of the ODE.

Note here that the only way that $x_{t}^{(n)}$ depends on $n$ is
through the scaling factor $x_{t}^{(n)}=n x_{t}$. This, along with
the lack of stochasticity, illustrates the simplicity but inflexibility
of the ODE approach to compartmental modelling. Indeed it has been recognized that ODE models
cannot capture important epidemiological phenomena such as fade-out,
extinction, lack of synchrony, or deviations from stable behaviour
\citep[Sec. 8]{roberts2015nine} and, somewhat more obviously, may under-represent uncertainty
\citep{king2015avoidable}.


To summarise the above, consider the following conceptual workflow:
\begin{enumerate}[label=ODE \arabic*. , wide=0.5em,  leftmargin=*, itemsep=-3pt,topsep=2pt]
\item specify a finite population, stochastic, continuous-time
compartmental model $(X_{t}^{(n)})_{t\geq0}$;
\item scale $X_{t}^{(n)}$ by $n^{-1}$ and take the large population
limit $n\to\infty$ to obtain $(x_{t})_{t\geq0}$;
\item re-scale $(x_{t})_{t\geq0}$ by $n$ to obtain $(x_{t}^{(n)})_{t\geq0},$
on the appropriate scale for a population of size $n$;
\item numerically approximate $(x_{t}^{(n)})_{t\geq0}$ and combine
with an observation model to evaluate the likelihood function.
\end{enumerate}
Of course in practice, someone can use the ODE
model \eqref{eq:ODE_rescaled} without knowing anything about steps
ODE 1.-3. We write out these steps in order to emphasize how the ODE
approach differs to the PAL methods proposed in the present work, where crucially the limit $n\to\infty$ is taken later in the conceptual workflow: 
\begin{enumerate}[label=PAL \arabic*. , wide=0.5em,  leftmargin=*, itemsep=-3pt,topsep=2pt]
\item specify a finite population, stochastic, discrete-time
compartmental model;
\item combine this model with an observation model to obtain discrete-time
filtering equations;
\item recursively approximate the filtering equations using Poisson
distributions, thus defining the PAL;
\item take the large population limit, $n\to\infty$, to establish
the consistency of the parameter estimator obtained by maximizing the PAL.
\end{enumerate}

The Latent Compartmental Model we work with is introduced in section \ref{sec:models}. It allows the probabilities of individuals transitioning between compartments to depend on the state of the population as a whole in a quite general way, as well as allowing for immigration and emigration, constant or random and dynamic population size. Due to the general form of this compartmental model, we can treat classical disease states, such as SEIR, as well as discrete covariates or subpopulations such as spatial locations or age-groups, in a single framework. Also in section \ref{sec:models}, we introduce two types of observation models: one for prevalence data, allowing for under-reporting, mis-reporting and spurious measurements; and one for incidence data, including incidence data which are aggregated in time.  In section \ref{sec:filtering} we introduce recursive Poisson approximations of filtering equations which lead us to PALs. The algorithms used to compute PALs involve only elementary linear algebra operations. Evaluating PALs up to a constant independent of parameter values as is sufficient for optimization, or evaluating ratios of PALs as arises in MCMC algorithms, has computational cost independent of population size.  In section \ref{sec:consistency} we state our main theoretical results concerning consistency of maximum PAL estimators in the large-population regime, and outline the main steps in the proof. As part of the proof we obtain novel results about asymptotically accurate filtering. Section \ref{sec:over-dispersion} discusses how over-dispersion can be handled by numerically integrating out latent variables using sequential Monte Carlo. In section \ref{sec:examples} we demonstrate various aspects of our methodology  and connections to our theory in the context of simulated and real data sets. Opportunities for future research are described in section   \ref{sec:future}.

\section{Connections to the literature}
\subsubsection*{Poisson process approximations}
Recursive approximation of filtering distributions using Poisson processes underlies the so-called
Probability Hypothesis Density (PHD) filter of \cite{mahler2003multitarget},
subsequently re-derived and generalized by \cite{singh2009filters,caron2011conditional}.
A specific but epidemiologically uninteresting (as we shall explain in section \ref{sec:lat_comp_model}) case of one model we consider in section  \ref{sec:case_I} coincides with a discrete-state version of the model considered in these works and the corresponding special case of our algorithm \ref{alg:x} would coincide with a discrete-state version of the PHD filter. The incidence data model we define in sections \ref{sec:obs_model_Z} and \ref{sec:obs_model_Z_agg} is however different, and particularly important for epidemiological data. Parameter estimation using the PHD filter in spatial multi-target models was suggested by \cite{singh2011} but without any rigorous justification and the authors are not aware of any theoretical results concerning parameter estimation consistency using the PHD filter. Approximate filtering for a limited class of epidemic models using multinomial rather than Poisson approximations was proposed by \cite{whiteley2021inference}, but without any consistency theory.
\subsubsection*{Inference algorithms for stochastic compartmental models}\label{sec:inference_methods}
Evaluating the likelihood
function for finite-population, stochastic compartmental models involves marginalizing out over the
set of all possible configurations of the population amongst the
compartments. The cost of this summation explodes with the number
of compartments and the population size. This has prompted the development of a variety
of simulation-based inference methods:  Data Augmentation
MCMC \citep{gibson1998estimating,o1999bayesian,o2010introduction,lekone2006statistical,fintzi2021linear,nguyen2021stochastic},
Approximate Bayesian Computation (ABC) \citep{toni2009approximate,mckinley2009inference}
and Sequential Monte Carlo (SMC) \citep{andrieu2010particle,ionides2011iterated,dukic2012tracking,koepke2016predictive,ju2021sequential}.
If one can simulate from the model, then in principle, one can apply the ABC and SMC methods. However, in practice there are usually algorithmic parameters to tune  and the computational cost of the simulation usually scales up with both the number of compartments and  population size, making these techniques very
computationally intensive in general. 



A functional central limit theorem associated with \eqref{eq:kurtz_lln}  due to \cite{kurtz1971limit}
gives rise to an SDE known as the Linear Noise Approximation (LNA), see e.g., \citep{fearnhead2014inference,komorowski2009bayesian}. Evaluating the Gaussian transition density of the LNA involves solving an ODE for its mean vector and covariance matrix. If combined with a linear-Gaussian observation model, the cost of a marginal likelihood evaluation scales with the  third power of the number of compartments in general.  Other varieties of SDE-based approximations
to finite-population stochastic compartmental models have been proposed
\citep{allen2017primer}, but their transition probabilities are usually not available in closed form and generally
costly simulation-based methods are relied upon to fit these models
to data \citep{roberts2001inference,cauchemez2008likelihood}.


\subsubsection*{Parameter estimation consistency results for compartmental models}
 The literature on consistency of parameter estimation in the large population limit is focused on specific instances of compartmental
models for which inferential calculations can be made in closed form, such as the continuous-time SIR model in which all infection and removal times are observed \citep{becker1993martingale},
\citep[Ch. 9]{andersson2012stochastic},  or  only the initial and final states
of the population are observed \citep[Ch. 10]{andersson2012stochastic}, estimating  the Malthusian parameter in an SEIR
model \citep{lindenstrand2013estimation}, or $R_{0}$ in an SIR model
 \citep{britton2010stochastic}. There appears
to be a lack of consistency results for likelihood-based estimators
for more general classes of compartmental models. Many stochastic compartmental models of epidemics are transient, in the sense that with probability one the entire population eventually ends up in one compartment and stays there, such as the R compartment in SIR. For this reason it seems that the asymptotic regime of a finite, fixed population size and increasingly long time horizon is not a fruitful regime in which to study consistency of parameter estimators for many epidemic models. One  unusual case is the Susceptible-Infective-Susceptible model, see \citep{gourieroux2021temporally} for an analysis in the regime where the time horizon tends to infinity, though for any finite population this epidemic will eventually go extinct.
\section{Models}\label{sec:models}
\subsection{Notation}
The set of natural numbers, including $0$, is denoted $\mathbb{N}_0$. The set of non-negative real numbers is denoted $\mathbb{R}_{\geq 0}$. For an integer $m\geq 1$, $[m]\coloneqq\{1,\ldots,m\}$. Matrices and vectors are denoted by bold upper-case and bold lower-case letters, respectively, e.g., $\+{A}$ and $\+{b}$, with non-bold upper-case and lower case used for their respective elements $A^{(i,j)}$, $b^{(i)}$. All vectors are column vectors unless stated otherwise. We use $\+1_m$ to denote the vector of $m$ $1$'s and $\+0_m$ to denote the vector of $m$ $0$'s. The indicator function is denoted $\mathbb{I}[\cdot]$. The element-wise product of matrices and vectors are denoted $\+A \circ \+B$ and $\+a \circ \+b$ respectively, the element-wise division of matrices and vectors are denoted $\+A \oslash \+B$ and $\+a \oslash \+b$ respectively, the outer product of vectors is denoted $\+a \otimes \+ b$. The logarithm $\log \+A$, factorial $\+A!$, and exponential $\exp(\+A)$ are taken element-wise. For $\+ x \in \mathbb{N}_0^m$ we define $\bs \eta(\+x) = [x^{(1)}/\+1_m^\top \+x \,\cdots\, x^{(m)}/\+1_m^\top \+x]^\top$ if $\+1_m^\top\+x>0$, i.e. $\bs \eta(\+x)$ normalizes $\+x$ to yield a probability vector;  and $\bs \eta(\+x) = \bs 0_m$ if $\+1_m^\top\+x=0$.


For $\+x\in \mathbb{N}_0^m$ and $\bs \lambda \in \mathbb{R}_{\geq 0}^m$ we write $\+x \sim \mathrm{Pois}(\bs \lambda)$ to denote that
the elements of $\+x$ are independent and element $x^{(i)}$ is Poisson distributed with parameter $\lambda^{(i)}$. We shall say that such a
 random vector $\+x$ has a ``vector-Poisson distribution''. For a probability vector $\bs
 \pi$  we write $\mathrm{Mult}(n,\bs\pi)$ for the associated multinomial distribution.
Similarly, for a random matrix $\+X \in \mathbb{N}_0^{m \times l}$ and a matrix $\bs \Lambda \in \mathbb{R}_{\geq 0}^{m \times l}$, we write $\+X\sim \mathrm{Pois}(\bs \Lambda) $ when the elements of $\+X$ are independent with $X^{(i,j)}$ being Poisson distributed with parameter $\Lambda^{(i,j)}$.
We call $\bs \lambda$ (resp. $\bs\Lambda$) the intensity vector (resp. matrix). For a length-$m$ vector $\+b$ with nonnegative elements, we call $\mathrm{supp}(\+b)\coloneqq \{i\in[m]:b^{(i)}>0\}$ the support of $\+b$. By convention, we take a sum over an empty set to be equal to $0$, i.e. a sum of $0$ terms. We write $\bs e_i$ for the vector of zeros except for a $1$ in the $i$th entry.
\subsection{Latent Compartmental Model} \label{sec:lat_comp_model}
The model we consider is defined by: $m$, the number of compartments; $n$ the expected initial population size; $\mathbb{P}_{0,n}$ an initial distribution on $\mathbb{N}_0^m$ such that $\E_{{\+x}_0\sim \mathbb{P}_{0,n}}[\+1_m^\top{\+x}_0] = n$,  e.g., $ \mathrm{Pois}(\bs \lambda_0)$ for some $\bs\lambda_0\in\mathbb{R}_{\geq0}^m$ such that  $\+1_m^\top \bs \lambda_0 =n$, or $\mathrm{Mult}(n,\bs\pi_0)$ for some length-$m$ probability vector $\bs\pi_0$; a sequence, $\{\bs \alpha_{t}\}_{t \geq 1}$ with $\bs \alpha_{t} \in \mathbb{R}_{\geq 0}^m$ for all $t \geq 1$, of immigration intensity vectors; a sequence, $\{\bs \delta_t\}_{t \geq 0}$ with $\bs \delta_t \in [0,1]^m$ for all $t \geq 0$; and for each $t \geq 0$  a mapping from length-$m$ probability vectors to size-$m \times m$ row-stochastic matrices, $\boldsymbol{\eta} \mapsto \mathbf{K}_{t, \boldsymbol{\eta}}$. 

The population at time $t \in \mathbb{N}_0$ is a set of a random number $n_t$ of random variables $\{\xi_{t}^{(1)}, \ldots, \xi_{t}^{(n_t)}\}$, each valued in $[m]$.  The counts of individuals in each of the $m$ compartments at time $t$ are
collected in $\mathbf{x}_{t}=[x_{t}^{(1)} \cdots x_{t}^{(m)}]^\top$, where $x_{t}^{(i)}=\sum_{j=1}^{n_t} \mathbb{I}[\xi_{t}^{(j)}=i].$ The population is initialised as a draw ${\+x}_0 \sim \mathbb{P}_{0,n}$. The members of the population are exchangeable, labelled by, e.g., a uniformly random assignment of indices $\{\xi_{0}^{(1)}, \ldots, \xi_{0}^{(n_0)}\}$ subject to $x_{0}^{(j)}:=\sum_{i=1}^{n_0} \mathbb{I}[\xi_{0}^{(i)}=j].$ For $t\geq 1$, given $\{\xi_{t-1}^{(1)}, \ldots, \xi_{t-1}^{(n_{t-1})}\}$, we obtain $n_t$ and  $\{\xi_{t}^{(1)}, \ldots, \xi_{t}^{(n_{t})}\}$ as follows.  For $i = 1,\dots n_{t-1}$, with probability $1-\delta_{t}^{(\xi_{t-1}^{(i)})}$  the individual $\xi_{t-1}^{(i)}$ emigrates from $[m]$  to a state $0\notin[m]$ from which it does not return.  The counts of remaining individuals are collected in the vector $\bar{\mathbf{x}}_{t-1}$,  where $\bar{x}_{t-1}^{(j)} :=\sum_{i=1}^{n_{t-1}} \mathbb{I}[\xi_{t-1}^{(i)}=j]\mathbb{I}[\phi^{(i)}_t=1]$ and $\phi_t^{(i)} \sim \mathrm{Bernoulli}(\delta_{t}^{(\xi_{t-1}^{(i)})})$. 

For each $i$ such that $\mathbb{I}[\phi^{(i)}_t=1]=1$, i.e. a remaining individual,  $\xi_{t}^{(i)}$ is then drawn from the $\xi_{t-1}^{(i)}$'th row of  $\mathbf{K}_{t, \boldsymbol{\eta}(\bar{\+ x}_{t-1})}$ and the resulting counts of individuals in the compartments $[m]$ are denoted  $\tilde{\+ x}_t$ where  $\tilde{ x}^{(j)}_t \coloneqq \sum_{i=1}^{n_{t-1}}\mathbb{I}[\phi_t^{(i)} = 1 ]\mathbb{I}[\xi_t^{(i)} = j ]$, if $\bar{\+x}_{t-1} = \bs 0_m$ then $\tilde{\+x}_t = \bs 0_m$.   Let $\mathbf{Z}_{t}$ be the $m \times m$
matrix with elements $Z_{t}^{(i, j)}:=\sum_{k=1}^{n_{t-1}} \mathbb{I}[\xi_{t-1}^{(k)}=i, \xi_{t}^{(k)}=j],$ which counts the individuals transitioning from
compartment $i$ at $t-1$ to compartment $j$ at time $t$.  New individuals then immigrate into the compartments $[m]$ according to a vector-Poisson distribution  $\hat{\+ x}_t \sim \mathrm{Pois}(\bs \alpha_t)$ and the resulting combined counts of individuals are $\+x_t \coloneqq \tilde{\+x}_t + \hat{\+x}_t$ with $n_t \coloneqq \+1_m^\top(\tilde{\+x}_t + \hat{\+x}_t)$. The population $\{\xi_{t}^{(1)}, \ldots, \xi_{t}^{(n_t)}\}$ is then obtained by uniformly random assignment of indices subject to $x_{t}^{(j)}:=\sum_{i=1}^{n_t} \mathbb{I}[\xi_{t}^{(i)}=j]$. Note that under this model, the processes $(\+x_t)_{t\geq0}$ and $(\+Z_t)_{t\geq1}$ are Markov chains, although we shall not need explicit expressions for their transition probabilities.

If the matrix $\mathbf{K}_{t, \boldsymbol{\eta}}$ were to have no dependence on $\boldsymbol{\eta}$, then the Latent Compartmental Model is a discrete-state version of the dynamic spatial Poisson-process model underlying the PHD filter \citep{mahler2003multitarget,singh2009filters,caron2011conditional}. However, for epidemiological modelling it is critical that  $\mathbf{K}_{t, \boldsymbol{\eta}}$ does depend on $\boldsymbol{\eta}$; for example in the case of SEIR as we shall now state, it is this dependence which models the mechanism of infection amongst the population. 
\subsubsection*{SEIR example}
As a very simple example of the Latent Compartmental Model consider the SEIR model:\begin{equation*}
S_{t+1} = S_{t} - B_{t},\quad E_{t+1} = E_{t} + B_{t} - C_{t},\quad I_{t+1} = I_{t} + C_{t} - D_{t}, \qquad R_{t+1} = R_{t} + D_{t}.
\end{equation*}
With conditionally independent, binomially distributed random variables:
\begin{equation*}
B_t \sim \mathrm{Bin}(S_t, 1- e^{-h\beta\frac{I_t}{n_t}}), \quad C_t \sim \mathrm{Bin}(E_t, 1- e^{-h\rho}),\quad D_t \sim \mathrm{Bin}(I_t, 1- e^{-h\gamma}),
\end{equation*}
where $h>0$ is a time-step size.
With no immigration or emigration, this model is cast as an instance of the model from section \ref{sec:lat_comp_model} by taking $m=4$, identifying $\+{x}_t\equiv[S_t\;E_t\;I_t\;R_t]^\top$ and:
\begin{equation}\label{eq:SEIRexample}
\mathbf{K}_{t, \boldsymbol{\eta}}=\left[\begin{array}{cccc}
e^{-h \beta \eta^{(3)}} & 1-e^{-h \beta \eta^{(3)}} & 0 &0\\
0 &  e^{-h \rho} & 1-e^{-h \rho} &0\\
0 &  0&e^{-h \gamma} & 1-e^{-h \gamma}\\
0 & 0  &0 &1
\end{array}\right].
\end{equation}


 

\subsection{Observation Models}
\subsubsection{Prevalence data}\label{sec:obs_model_x}
Epidemiological \emph{prevalence data} pertain to the overall levels of susceptibility, exposure and infectivity in the population. In the context of the Latent Compartmental Model, such data are related to the counts of individuals in each compartment at given points in time, i.e., $(\+{x}_t)_{t\geq1}$.  The observation at time $t\geq1$ is an $m$-length vector $\+y_t$ distributed as follows. With a vector $\+q_t \in [0,1]^m$, for each $j\in[m]$ each individual in compartment $j$ is independently detected with probability $q_t^{(j)}$, and the counts of detected individuals are collected  in a vector $\bar{\+y}_t$, i.e.,
\begin{equation}\label{eq:obsx}
\bar{y}_t^{(i)} \sim \mathrm{Bin}(x_t^{(i)}, q_t^{(i)}),\quad i \in[m].
\end{equation}
With $\+G_t$ a row-stochastic matrix of size $m\times m$, each individual detected in compartment $j$ is independently reported in compartment $k$ with probability $G_t^{(j,k)}$. The counts of these reported individuals are collected in an $m$-length vector $\tilde{\+y}_t$. The off-diagonal elements of the matrix $\+G_t$ can be interpreted as the probabilities of mis-reporting between compartments.
Then the observation $\+y_t$ is given by:
\begin{equation*}
\+y_t = \tilde{\+y}_t + \hat{\+y}_t,
\end{equation*}
where independently $\hat{\+y}_t \sim \mathrm{Pois}(\bs \kappa_t)$ for $\bs \kappa_t \in \mathbb{R}^m_{\geq 0}$, which can be interpreted as additive error counts. In epidemiological data usually only individuals associated with some subset of compartments are detected, and only at certain times. If individuals in say compartment $i$ are not observed at time $t$, then for inference we will set $y_t^{(i)}=0$ and $q_t^{(i)}=0$.

More detailed interpretation of this observation model, in terms of e.g. epidemiological testing of the population, probability of false positives, etc., will be specific to the context in which the Latent Compartmental Model is applied. We provide discussion of this point illustrated by example in section \ref{sec:interpretation} of the supplementary material.
\subsubsection{Incidence data}\label{sec:obs_model_Z}
Epidemiological measurements often involve data related to the number of newly infective or recovered individuals over given time periods -- known as \emph{incidence} data. In order to model such data, generalized to allow for transitions from any compartment to any compartment, we consider  an observation at time $t\geq1$ which is an $m\times m$ matrix $\+Y_t$. The elements of $\+Y_t$  are conditionally independent given $\+Z_t$, and with a matrix $\+Q_t\in[0,1]^{m\times m}$,   
\begin{equation}\label{obs_mod_Y}
Y_t^{(i,j)} \sim \mathrm{Bin}(Z_t^{(i,j)}, Q_t^{(i,j)}),\quad (i,j)\in[m]\times[m].
\end{equation}
Similarly to the case of prevalence data,  if  $Y_t^{(i,j)}$ are missing,  then  for inference we set $Y_t^{(i,j)}=0$ and $Q_t^{(i,j)}=0$. One could extend this model to incorporate mis-reporting and/or additive error counts in a similar manner to in section \ref{sec:obs_model_x}, but for simplicity of presentation we do not do so.

In the context of the SEIR model, for example, the variable $Y_t^{(2,3)}$ models the number of individuals which are newly infective at time $t$, i.e. the count of the number of individuals which have transitioned $E\to I$ from time $t-1$ to $t$, subject to random under-reporting parameterized by $Q_t^{(i,j)}$. 
\subsubsection{Aggregated incidence data}\label{sec:obs_model_Z_agg}
In some situations it is desirable to model observations as in section \ref{sec:obs_model_Z}, but with transitions of individuals between compartments occurring on a finer time-scale than observations. For example, consider the SEIR model and suppose each discrete time step corresponds to one week. Then the model in \eqref{eq:SEIRexample} assigns zero probability to a transition $S\to I$ in one week: in order to transition between $S\to I$, an individual must transit $S\to E$ and then $E\to I$, but at least two discrete time steps are needed for that to occur with positive probability. Similarly, transitions $E\to R$ in one week happen with zero probability. To model incidence data as in section \ref{sec:obs_model_Z} but allowing for these sort of multi-step transitions between observation times, we introduce a sequence of increasing integer observation times $(\tau_r)_{r\geq 1}\subset\mathbb{N}_0$ where $\tau_0\coloneqq0$. We then define $\bar{\+Y}_r\coloneqq \sum_{t=\tau_{r-1}+1}^{\tau_r} \+Y_t$, where $(\+Y_t)_{t\geq1}$ are distributed as per section \ref{sec:obs_model_Z}. This model coincides with the model from that section in the case that $\tau_k=k$, we present these two models separately in order to help present a step-by-step explanation in section \ref{sec:filtering} of the corresponding filtering recursions.

In the context of the SEIR model,  $\bar{Y}_r^{(2,3)}$  models the total number of individuals which have become infective between times $\tau_{r-1}$ and $\tau_r$,  subject to random under-reporting. If $\tau_r -\tau_{r-1}\geq2$, this allows for two-step transitions of the form $S\to E\to I$ or $E\to I\to R$ to occur with positive probability between observations times. 
\section{Filtering recursions and Poisson Approximate Likelihoods}\label{sec:filtering}
Our next objective is to state and explain the filtering recursions which are used to compute PALs. In section \ref{sec:case_I} we give the filtering recursion and PAL for the Latent Compartmental Model combined with the prevalence data model from section  \ref{sec:obs_model_x}, we refer to this combination as case (I). In section  \ref{sec:case_II} we give filtering recursions for a simplified case of the Latent Compartmental Model in which $n_0=n$ $a.s.$ for $n \in \mathbb{N}$, $\bs\delta_t=\+1_m$, and $\bs\alpha_{t}=\+0_m$ for all $t$, i.e. no emigration or immigration, combined with the incidence data model from sections \ref{sec:obs_model_Z} and \ref{sec:obs_model_Z_agg}. We refer to this as case (II). We discuss the filtering recursions in case (II) with $\bs\delta_t=\+1_m$ and $\bs\alpha_{t}=\+0_m$ only for ease of exposition. By expanding on the derivations we give in the following sections, the reader could obtain without great difficulty the filtering recursions for case (II) in the full generality of the Latent Compartmental Model and in section \ref{sec:measles} we consider an example involving immigration, emigration and incidence data as an illustration. 

Below we state a collection of lemmas which formalize the derivations of the steps in filtering recursions. The proofs, given in section  \ref{sec:proofs_filtering} of the supplementary materials, rely on moment generating functions and some techniques from the theory of Poisson processes \citep{kingman1992poisson}.
\subsection{Case (I)}\label{sec:case_I}
In this case, the observations $(\+y_t)_{t\geq1}$ follow the model from section \ref{sec:obs_model_x}. The pair of processes $(\+x_t)_{t\geq0}$ and $(\+y_t)_{t\geq1}$ constitutes a hidden Markov model: $(\+x_t)_{t\geq0}$ is a Markov chain, and $(\+y_{t})_{t\geq1}$ are conditionally independent given $(\+x_t)_{t\geq0}$ with the conditional distribution of $\+y_{t}$ given $(\+x_t)_{t\geq0}$ depending only on $\+x_t$. Therefore the filtering distributions $p(\mathbf{x}_t |\+y_{1:t})$, obey a two-step recursion, with steps canonically referred to as ``prediction'' and ``update'':
\begin{equation*} \label{eq:predup}
p(\+x_{t-1}|\+y_{1:t-1}) \overset{\mathrm{prediction}}{\longrightarrow} p(\+x_{t}| \+y_{1:t-1})  \overset{\mathrm{update}}{\longrightarrow} p(\+x_{t}|\+y_{1:t}),
\end{equation*}
where, for $t\geq1$,
\begin{align}
p(\+x_{t}| \+y_{1:t-1})  &= \sum_{\+x_{t-1}\in\mathbb{N}_0^m} p(\+x_t|\+x_{t-1}) p(\+x_{t-1} | \+y_{1:t-1}),\label{eq:prediction_x}\\
p(\+x_{t}| \+y_{1:t}) &=\frac{p(\+y_t | \+x_t)p(\+x_{t}| \+y_{1:t-1})}{p(\+y_t | \+y_{1:t-1})},\label{eq:update_x}\\
p(\+y_t | \+y_{1:t-1})&=\sum_{\+x_t\in\mathbb{N}_0^m} p(\+y_t | \+x_t)p(\+x_{t}| \+y_{1:t-1}),\label{eq:cond_like_x}
\end{align}
and here and below, by convention, conditioning on $\+y_{1:0}$ is understood to mean no conditioning, $p(\cdot|\+y_{1:0})\coloneqq p(\cdot)$. The marginal likelihood of the observations $\+y_1,\ldots,\+y_t$ can be written:
\begin{equation}
p(\+y_{1:t})=\prod_{s=1}^t p(\+y_s | \+y_{1:s-1}).\label{eq:marg_like_y}
\end{equation}
The general idea of the PAL is to obtain vector-Poisson distribution approximation to each of the terms $p(\+y_1)$ and $p(\+y_t | \+y_{1:t-1})$, $t\geq1$, computed via vector-Poisson approximations to each of the filtering distributions $p(\+x_{t} | \+y_{1:t-1})$ and $p(\+x_{t}| \+y_{1:t})$. 

\subsubsection*{Approximating the prediction step}
For time step $t=0$ we take a vector-Poisson approximation $\mathrm{Pois}(\bs \lambda_0)$ to the initial distribution $\mathbb{P}_{0,n}$ by setting $\bs \lambda_0\coloneqq \mathbb{E}_{\+x_0\sim\mathbb{P}_{0,n}}[\+x_0]$ and $\bar{\bs \lambda}_0\coloneqq \bs \lambda_0$.  For $t\geq 1$, suppose we have obtained $\bar{\bs\lambda}_{t-1}$ and so defined a vector-Poisson approximation $\mathrm{Pois}( \bar{\bs\lambda}_{t-1})$ to $p(\+x_{t-1} | \+y_{1:t-1})$. In order to derive a vector-Poisson approximation to $p(\+x_{t} | \+y_{1:t-1})$, we need to consider the operation \eqref{eq:prediction_x} in more detail, in accordance with the definition of the Latent Compartmental Model. We shall not need an explicit formula for the transition probabilities $p(\+x_t|\+x_{t-1})$, but rather work with the intermediate quantities  $\bar{\+x}_{t-1}, \tilde{\+x}_t, \hat{\+x}_t$ introduced in section \ref{sec:lat_comp_model}.

For $\bar{\+x} \in \mathbb{R}^m$ and a length-$m$ probability vector $\bs\eta$, let $M_t(\bar{\+x}, \bs\eta,\cdot)$ be the probability mass function of $(\+1_m^\top \+Z)^\top$ where the $i$th row of $\+Z \in \mathbb{N}_0^{m\times m}$ has distribution $\mathrm{Mult}(\bar{x}^{(i)}, \+K_{t, \bs \eta}^{(i, \cdot)})$. Then we have:
\begin{align}
p(\tilde{\+x}_t | \+y_{1:t-1}) &= \sum_{\bar{\+x}_{t-1}\in \mathbb{N}_0^m}p(\bar{\+x}_{t-1} | \+y_{1:t-1})p(\tilde{\+x}_t | \bar{\+x}_{t-1}) \\
&= \sum_{\bar{\+x}_{t-1}\in \mathbb{N}_0^m}p(\bar{\+x}_{t-1} | \+y_{1:t-1})M_t(\bar{\+x}_{t-1}, \bs \eta(\bar{\+x}_{t-1}), \tilde{\+x}_t), \label{eq:approx}
\end{align}
where $\bar{\+x}_{t-1}$ is related to $\+x_{t-1}$ by $\bar{x}_{t-1}^{(i)}\sim \mathrm{Bin}(x^{(i)}_{t-1}, \delta_t^{(i)})$. The summation in \eqref{eq:approx} is too expensive to compute in general. To define an approximation which circumvents this issue, in \eqref{eq:approx}   we replace $p(\+x_{t-1} | \+y_{1:t-1})$ by its approximation $\mathrm{Pois}(\bar{\bs \lambda}_{t-1}) $, and  replace $\bs \eta(\bar{\+x}_{t-1})$ by $\bs \eta(\E[\bar{\+x}_{t-1}])$ where this expectation is under $\bar{\+x}_{t-1} \sim \mathrm{Pois}(\bar{\bs \lambda}_{t-1}\circ \bs \delta_t)$. Lemma \ref{lem:xpred} explains the rationale for making the  vector-Poisson approximation 
$$p(\tilde{\+x}_t | \+y_{1:t-1})\approx\mathrm{Pois}\left((\bar{\bs \lambda}_{t-1} \circ \bs \delta_t)^\top \+ K_{t, \bs \eta(\bar{\bs \lambda}_{t-1}\circ \bs \delta_t)} \right).$$
\begin{lemma}\label{lem:xpred}
Suppose that ${\+x \sim \mathrm{Pois}(\bs \lambda)}$ for ${\bs \lambda \in \mathbb{R}^m_{\geq 0}}$ and  ${\bar x^{(i)} \sim \mathrm{Bin}(x^{(i)}, \delta^{(i)})}$ for $\bs \delta \in [0,1]^m$. Then ${\bar{\+ x} \sim \mathrm{Pois}(\bs \lambda \circ \bs \delta).}$ Furthermore, if $\mu(\cdot)$ is the probability mass function associated with $\mathrm{Pois}(\bs \lambda \circ \bs \delta)$ and $\E_\mu \left[ \cdot\right]$ is the expected value under $\mu$, then ${\sum_{\bar{\+x}\in \mathbb{N}_{0}^m} \mu(\bar{\+x})M_t(\bar{\+x}, \bs \eta ( \E_\mu \left[  \bar{\+ x} \right]), \cdot )}$ is the probability mass function associated with
$\mathrm{Pois}\left((\bs \lambda \circ \bs \delta)^\top \+ K_{t, \bs \eta(\bs \lambda \circ \bs \delta)} \right)$.
\end{lemma}
\noindent The proof is given in section \ref{sec:proofs_filtering} of the supplementary material.  As per the definition of the Latent Compartmental Model, $\+x_t$ is obtained by summing $\tilde{\+x}_t$ with $\hat{\+x}_t$ where $\hat{\+x}_t\sim\mathrm{Pois}(\bs\alpha_t)$. Since the sum of independent Poisson random variables is also Poisson with intensity given by the sum of the intensities, we then take the approximation $$p(\+x_{t} | \+y_{1:t-1})\approx\mathrm{Pois}(\bs\lambda_t),\quad \text{with}\quad\bs\lambda_t\coloneqq (\bar{\bs \lambda}_{t-1} \circ \bs \delta_t)^\top \+ K_{t, \bs \eta(\bar{\bs \lambda}_{t-1}\circ \bs \delta_t)} +\bs\alpha_t.$$
\subsubsection*{Approximating the update step}
In order to obtain a vector-Poisson approximation to $p(\+x_{t} | \+y_{1:t})$ we substitute $\mathrm{Pois}(\bs \lambda_{t})$ in place of $p(\+x_{t} | \+y_{1:t-1})$ in \eqref{eq:update_x}, which can be viewed as an application of Bayes' rule, and we shall define $\bar{\bs\lambda}_t$ to be the mean vector of the resulting distribution.  Lemma \ref{lem:xupdate} can be applied to calculate $\bar{\bs\lambda}_t$ in accordance with this recipe, leading us to:
\begin{equation*}
p(\+x_{t} | \+y_{1:t})\approx\mathrm{Pois}(\bar{\bs\lambda}_t),\quad\bar{\bs\lambda}_t \coloneqq  [\+ 1_m - \+q_t +  (\{\+ y_t^\top \oslash[(\+ q_t \circ \bs \lambda_t)^\top \+ G_t + \bs \kappa_t^\top ] \}[(\+ 1_m \otimes \+ q_t)\circ \+ G_t^\top] )^\top ]\circ \bs \lambda_t,
\end{equation*} 
Lemma \ref{lem:xupdate}  also tells us how to  obtain a vector-Poisson approximation to $p(\+y_{t} | \+y_{1:t-1})$.

\begin{lemma}\label{lem:xupdate}
Suppose that $\+x \sim \mathrm{Pois}(\bs \lambda)$ for given $\bs \lambda \in \mathbb{R}_{\geq 0}^m$ and let $\bar{\+ y}$ be a vector with conditionally independent elements distributed $\bar{y}^{(i)} \sim \mathrm{Bin}(x^{(i)},q^{(i)})$ for given $\+q \in [0,1]^m$. For $\+G$ a row-stochastic $m\times m$ matrix and $\+M$ an $m\times m$ matrix with rows distributed $\+ M^{(i, \cdot)} \sim \mathrm{Mult}(\bar{y}^{(i)}, \+ G^{(i, \cdot)})$,  let $\tilde{\+y} \coloneqq \sum_{i=1}^m \+ M^{(i, \cdot)}$ and  $\+ y \coloneqq  \tilde{\+y}+ \hat{\+y}$ where  $\hat{\+ y} \sim \mathrm{Pois}(\bs \kappa)$ for a given $\bs \kappa \in \mathbb{R}_{\geq 0}^m$. Then:
\begin{equation} \label{eq:condexp}
\E\left[\+x |\+y \right] = [\+ 1_m - \+q +  (\{\+ y^\top\oslash[(\+ q \circ \bs \lambda)^\top \+ G + \bs \kappa^\top ] \}[(\+ 1_m \otimes \+ q)\circ \+ G^\top] )^\top ]\circ \bs \lambda.
\end{equation}
and $\+ y\sim \mathrm{Pois}([(\bs \lambda \circ \+ q)^\top \+G]^\top + \bs \kappa)$, i.e., 
\begin{equation*}
\log p(\+ y) =  -[(\bs \lambda_t\circ \+ q)^\top \+G + \bs \kappa^\top]\+1_m + \+y^\top \log([(\bs \lambda\circ \+ q)^\top \+G]^\top + \bs \kappa) - \+1_m^\top\log(\+y!),
\end{equation*}
with the convention $0\log 0 \coloneqq 0$.
\end{lemma}
\noindent The proof is given in section \ref{sec:proofs_filtering} of the supplementary material.
\subsubsection*{Computing the PAL}
Gathering together the approximations discussed above we arrive at the following algorithm.
\begin{algorithm}[H]
\caption{Filtering for case (I)}\label{alg:x}
\begin{algorithmic}[1]
  \Statex {\bf initialize:} $\bar{\bs \lambda}_{0} \leftarrow {\bs \lambda}_0$
  \State {\bf for}  $t \geq 1$:
  \State\quad  $ \bs \lambda_{t} \leftarrow [(\bar{\bs {\lambda}}_{t-1}\circ \bs{\delta}_t)^\top\+{K}_{t, \bs{\eta}( \bar{\bs \lambda}_{t-1}\circ \bs{\delta}_t)}]^\top + \bs{\alpha}_t $
   \State \quad$\bar{ \bs \lambda}_{t} \leftarrow [\+ 1_m - \+q_t +  (\{\+ y_t^\top \oslash[(\+ q_t \circ \bs \lambda_t)^\top \+ G_t + \bs \kappa_t^\top ] \}[(\+ 1_m \otimes \+ q_t)\circ \+ G_t^\top] )^\top ]\circ \bs \lambda_t$
   \State \quad $\bs\mu_t\leftarrow [(\bs \lambda_t\circ \+ q_t)^\top \+G_t]^\top + \bs \kappa_t$
  \State \quad$\ell(\+y_t | \+y_{1:t-1}) \leftarrow -\bs\mu_t^\top\+1_m + \+y_t^\top \log(\bs\mu_t) - \+1_m^\top\log(\+y_t!)$
  \State {\bf end for}
\end{algorithmic}
\end{algorithm}
\noindent If, at line $3$ of algorithm \ref{alg:x}, we encounter $0/0$ in performing the element-wise division operation we set the vector element in question to $0$, which is in accordance with $p(\+x_{t} | \+y_{1:t})\approx\mathrm{Pois}(\bar{\bs\lambda}_t)$. At line 5 of algorithm \ref{alg:x} we apply the convention $0\log 0 \coloneqq 0$, in accordance with $p(\+y_{t} | \+y_{1:t-1})\approx\mathrm{Pois}(\bs\mu_t)$.

Mimicking \eqref{eq:marg_like_y}, the log PAL associated with algorithm \ref{alg:x} is:
\begin{equation}
\log p(\+y_{1:t})\approx \sum_{s=1}^t \ell(\+y_s | \+y_{1:s-1}),\label{eq:pal_case_I}
\end{equation}
It is important to note that the term $\+1_m^\top \log(\+y_t!)$ in $\ell(\+y_t | \+y_{1:t-1})$ calculated in algorithm \ref{alg:x} has no dependence on the ingredients of the model, i.e., $\+K_{t,\bs \eta }$, $\bs\kappa_t$, etc. and so in practice if one is computing the PAL in order to maximize it with respect to parameters of the model, or evaluate PAL ratios for different parameter values, the term $\+1_m^\top\log(\+y_t!)$ never needs to be computed.
\subsection{Case (II)} \label{sec:case_II}
In this case we consider the Latent Compartmental Model with $n_0=n$ with probability $1$, $\bs\delta_t=\+1_m$ and $\bs\alpha_{t}=\+0_m$ for all $t$, i.e. no emigration or immigration, and with the observations $(\bar{\+Y}_r)_{r\geq1}$ following the model from section \ref{sec:obs_model_Z_agg}. For ease of exposition we start with the special case that $(\tau_r)_{r\geq1}=\mathbb{N}$, in which case $(\bar{\+Y}_r)_{r\geq1}\equiv(\+Y_t)_{t\geq1}$ and the model from section \ref{sec:obs_model_Z_agg} reduces to that from section \ref{sec:obs_model_Z}.

To derive the filtering recursions we follow a similar programme to case (I), starting from the fact that the pair of processes $(\+Z_t)_{t\geq1}$ and $(\+Y_t)_{t\geq1}$ constitutes a hidden Markov model, and approximating the following prediction and update operations:
\begin{equation*}
p(\+Z_{t-1}| {\+Y}_{1:t-1}) \overset{\mathrm{prediction}}{\longrightarrow}  p(\+Z_{t}| {\+Y}_{1:t-1}) \overset{\mathrm{update}}{\longrightarrow} p(\+Z_{t}| {\+Y}_{1:t}).
\end{equation*}
\subsubsection*{Approximating the prediction step when $(\tau_r)_{r\geq1}=\mathbb{N}$}
For $\+Z\in\mathbb{N}_0^{m\times m}$ and a length-$m$ probability vector $\bs \eta$, let $\bar{M}_t(\+Z, \bs \eta, \cdot)$ be the probability mass function of a random $m \times m$ matrix, say $\tilde{\+ Z}$, such that $\+1_m^\top \+ Z = (\tilde{\+ Z} \+1_m )^\top$ with probability $1$ and such that given the row sums $\tilde{\+Z}\+1_m = \+ x$, the rows of $\tilde{\+Z}$ are conditionally independent with the conditional distribution of the $i^{th}$ row being $\mathrm{Mult}(x^{(i)}, \+K_{t,\bs \eta}^{(i,\cdot)})$. By construction $\bar{M}_t(\+Z_{t-1}, \bs \eta(\+1_m^\top \+Z_{t-1}), \+Z_t)$ is equal to $p(\+Z_{t}|\+Z_{t-1})$ for case (II), hence
\begin{align}
p(\+Z_t | \+Y_{1:t-1}) &= \sum_{\+Z_{t-1}\in \mathbb{N}_0^{m\times m}}p(\+Z_{t-1} | \+Y_{1:t-1})p(\+Z_t | \+Z_{t-1}) \\
&= \sum_{\+Z_{t-1}\in \mathbb{N}_0^{m\times m}}p(\+Z_{t-1} | \+Y_{1:t-1})M_t(\+Z_{t-1}, \bs \eta(\+1_m^\top \+Z_{t-1}), \+Z_t). \label{eq:approx_Z}
\end{align}
Assuming we have already computed $\bar{\bs\Lambda}_{t-1}$ such that $p(\+Z_{t-1} | \+Y_{1:t-1}) \approx \mathrm{Pois}(\bar{\bs\Lambda}_{t-1})$, we substitute this approximation in to \eqref{eq:approx_Z} and replace $\bs \eta(\+1_m^\top \+Z_{t-1})$ by $\bs \eta(\E[\+1_m^\top \+Z_{t-1}])$ where this expectation is under $\+Z_{t-1} \sim \mathrm{Pois}(\bar{\bs\Lambda}_{t-1})$. Lemma \ref{lem:Zpred} explains the rationale for then making the approximation:
\begin{equation*}
p(\+Z_{t} | \+Y_{1:t-1}) \approx \mathrm{Pois}(\bs\Lambda_{t}),\quad
\+\Lambda_{t}  \coloneqq (\bar{\bs\lambda}_{t-1} \otimes \+1_m)\circ \+K_{t,\bs \eta(\bar{\bs\lambda}_{t-1})}, \quad\bar{ \bs\lambda}_{t-1}^\top \coloneqq \+ 1_m^\top \bar{\bs  \Lambda}_{t-1}.
\end{equation*}
\begin{lemma}\label{lem:Zpred}
If for a given $m \times m$ matrix $\+\Lambda$, $\bar \mu$ is the probability mass function associated with $\mathrm{Pois}(\+ \Lambda)$ and $\E_{\bar \mu}[ \+1_m^\top\+Z]$ is the expected value of $\+1^\top_m \+Z$ where $\+Z \sim \bar \mu$, then $\sum_{\+Z\in \mathbb{N}_{0}^{m\times m}}\bar \mu(\+Z) \bar M_t(\+Z,\bs \eta(\E_{\bar \mu}[\+1_m^\top\+Z]),\cdot)$ is the probability mass function associated with $\mathrm{Pois}((\bs \lambda \otimes \+1_m)\circ \+K_{t,\bs \eta(\bs \lambda)})$, where $\bs \lambda^\top \coloneqq \+ 1_m^\top \bs \Lambda$.
\end{lemma}
\noindent The proof is given in section \ref{sec:proofs_filtering} of the supplementary material.
\subsubsection*{Approximating the update step when $(\tau_r)_{r\geq1}=\mathbb{N}$}
We now apply Bayes' rule to $\mathrm{Pois}(\bs\Lambda_{t})$ and shall define  $\bar{\bs\Lambda}_t$ to be the mean vector of the resulting distribution. Lemma \ref{lem:Zupdate} shows how to do this, leading to:
$$
p(\+Z_{t} | \+Y_{1:t}) \approx  \mathrm{Pois}(\bar{\bs\Lambda}_{t}),\quad  \bar{\bs\Lambda}_{t}\coloneqq \+Y_t+  \+ \Lambda_t \circ(\+1_m \otimes \+1_m - \+Q_t),
$$
\begin{lemma} \label{lem:Zupdate}
Suppose that $\+ Z \sim \mathrm{Pois}(\+ \Lambda)$ for some  $\+ \Lambda\in\mathbb{R}_{\geq0}^{m\times m}$, and that for some $\+Q\in\mathbb{R}_{\geq0}^{m\times m}$, given $\+Z$, $\+Y$ is a matrix with conditionally independent entries distributed: $y^{(i,j)} \sim \mathrm{Bin}(Z^{(i,j)},q^{(i,j)})$, then the conditional distribution of $\+Z$ given $\+Y$ is that of $\+ Y + \+Z^*$ where:
\begin{equation*}
\+Z^* \sim \mathrm{Pois}\left( \+ \Lambda \circ(\+1_m \otimes \+1_m - \+Q) \right), 
\end{equation*}
i.e.,
\begin{equation*}
\mathbb{E}[\+Z|\+Y] = \+Y+  \+ \Lambda \circ(\+1_m \otimes \+1_m - \+Q),
\end{equation*}
and $\+Y\sim\mathrm{Pois}(\+\Lambda \circ  \+ Q)$, i.e, 
\begin{equation*}
\log p(\+Y) =  \+1_m^\top (\+ \Lambda \circ \+ Q)\+1_m +\+1_m^\top [\+ Y \circ \log(\+\Lambda\circ \+ Q)]\+1_m - \+1_m^\top \log(\+ Y!)\+1_m,
\end{equation*}
with the convention $0\log 0 \coloneqq 0$.
\end{lemma}
\noindent The proof is given in section \ref{sec:proofs_filtering} of the supplementary material.
\subsubsection*{Computing the PAL when $(\tau_r)_{r\geq1}=\mathbb{N}$}
Combining the above prediction and update approximations we arrive at algorithm \ref{alg:Z}.
\begin{algorithm}[H]
\caption{Filtering for case (II) when $(\tau_r)_{r\geq1}=\mathbb{N}$}\label{alg:Z}
\begin{algorithmic}[1]
  \Statex {\bf initialize:} $\bar{\bs \lambda}_{0} \leftarrow \bs \lambda_0$
  \State {\bf for}  $t \geq 1$:
  \State \quad $\+ \Lambda_t \leftarrow (\bar{\bs \lambda}_t \otimes \+1_m)\circ \+K_{t,\bs \eta(\bar{\bs \lambda}_t)} $
  \State \quad $\bar{\+ \Lambda}_t \leftarrow \+ Y_t+ (\+1_m \otimes \+1_m - \+ Q_t  )\circ \bs \Lambda_{t}$
  \State \quad $\mathcal{L}(\+Y_t|\+Y_{1:t-1}) \leftarrow - \+1_m^\top(\+ \Lambda_t \circ \+ Q_t)\+1_m +\+1_m^\top [\+ Y_t \circ \log(\+\Lambda_t\circ \+ Q_t)]\+1_m - \+1_m^\top \log(\+ Y_t!)\+1_m $
  \State \quad$\bar{\bs \lambda}_t \leftarrow (\+ 1_m^\top \bar{\+ \Lambda}_t)^\top$
  \State {\bf end for}
\end{algorithmic}
\end{algorithm}
\noindent In algorithm \ref{alg:Z} we adopt the same convention $0\log0\coloneqq 0$ as in algorithm \ref{alg:x}. The log PAL associated with algorithm \ref{alg:Z} is:
\begin{equation*}
\log p(\+Y_{1:t})\approx \sum_{s=1}^t \mathcal{L}(\+Y_s | \+Y_{1:s-1}).
\end{equation*}

We now consider general $(\tau_r)_{r\geq1}$. The filtering recursion is:
\begin{equation}\label{eq:case_II_recursion}
\begin{aligned}
p(\+Z_{\tau_{r-1}}|\bar{\+Y}_{1:r-1}) &\overset{\text{prediction}}{\longrightarrow} p(\+Z_{\tau_{r-1}+1}| \bar{\+Y}_{1:r-1}) \overset{\text{prediction}}{\longrightarrow} \dots   \\
&\overset{\text{prediction}}{\longrightarrow} p(\+Z_{\tau_r}|\bar{\+Y}_{1:r-1}) \overset{\text{update}}{\longrightarrow} p(\+Z_{\tau_r}|\bar{\+Y}_{1:r}).
\end{aligned}
\end{equation}
\subsubsection*{Approximating the prediction and update steps for general $(\tau_r)_{r\geq1}$}
Assuming that we are given $\+\Lambda_{\tau_{r-1}}$ such that $p(\+Z_{\tau_{r-1}}|\bar{\+Y}_{1:r-1})\approx \mathrm{Pois}(\+\Lambda_{\tau_{r-1}})$, each of the prediction steps in 
\eqref{eq:case_II_recursion} is approximated by applying lemma \ref{lem:Zpred}, leading to lines 2-6 of algorithm \ref{alg:Ztagg}. To approximate the update step, applying lemma \ref{lem:Zaggup} leads to lines  7-10 of algorithm \ref{alg:Ztagg}.
\begin{lemma}\label{lem:Zaggup}
For $\bs \lambda_0 \in \mathbb{R}_{\geq0}^m$ and $\tau \in \mathbb{N}$,  define:
\begin{equation*}
\bs \Lambda_t \coloneqq (\bs \lambda_{t -1}\otimes \+1_m)\circ \+K_{t, \bs \eta(\bs \lambda_{t -1})},\quad
\bs \lambda_t \coloneqq (\+1^\top_m \bs \Lambda_t)^\top,\quad t = 1, \dots, \tau,
\end{equation*}
and let $(\+ Z_t)_{t=1}^\tau$ be independent with $\+ Z_t \sim \mathrm{Pois}(\bs \Lambda_t)$. Suppose that given $\+Z_t$, $\+Y_t$ is a matrix with conditionally independent entries distributed  $Y_t^{(i,j)}\sim\mathrm{Bin}(Z_t^{(i,j)}, Q^{(i,j)})$, and let $\bar{\+ Y}\coloneqq \sum_{s=1}^\tau \+Y_s$. Then:
\begin{equation*}
\E\left[ \+ Z_{\tau} | \bar{ \+ Y}\right] = (\+1_m \otimes \+1_m - \+Q_\tau)\circ \bs \Lambda_{t} + \bar{\+Y} \circ \bs \Lambda_\tau\circ \+ Q_\tau  \oslash \left( \sum_{t=1}^\tau \bs \Lambda_t\circ \+ Q_t \right)  ,
\end{equation*}
and $\bar{\+Y}\sim\mathrm{Pois}(\sum_{t =1}^{\tau} \bs \Lambda_t \circ \+ Q_t)$, i.e.,
$$
\log p(\bar{\+Y}) =  \+1_m^\top\+M\+1_m +\+1_m^\top (\bar{\+ Y} \circ \log\+M)\+1_m - \+1_m^\top \log(\bar{\+ Y}!)\+1_m,
$$
where $\+M\coloneqq \sum_{t =1}^{\tau} \bs \Lambda_t \circ \+ Q_t$ and by convention $0\log0\coloneqq 0$.
\end{lemma}
\noindent The proof is given in section \ref{sec:proofs_filtering} of the supplementary material.
\subsubsection*{Computing the PAL for general $(\tau_r)_{r\geq1}$}
\begin{algorithm}[H]
\caption{Filtering for case (II) with general $(\tau_r)_{r\geq1}$}\label{alg:Ztagg}
\begin{algorithmic}[1]
  \Statex {\bf initialize:} $\bar{\bs \lambda}_{0} \leftarrow \bs \lambda_0$.
  \State {\bf for}  $r  \geq 1$:
  \State \quad  {\bf for}  $  t = \tau_{r-1}+1, \dots,  \tau_r - 1$:
  \State \quad \quad$\bs \Lambda_{t} \leftarrow  (\bar{ \bs \lambda}_{t-1} \otimes \+1_m)\circ \+K_{t,\bs \eta\left(\bar{ \bs \lambda}_{t-1}\right)}$
  \State \quad \quad$\bar{\bs \lambda}_{t} \leftarrow  (\+1_m^\top \bs \Lambda_{t})^\top$
  \State \quad {\bf end for}

  \State  \quad $\bs \Lambda_{\tau_r}\leftarrow  (\bs \lambda_{\tau_r -1} \otimes \+1_m)\circ \+K_{\tau_r,\bs \eta(\bs \lambda_{\tau_r -1})}$
  \State \quad $\+M_r \leftarrow   \sum_{t = \tau_{r-1}+1}^{\tau_{r}} \bs \Lambda_{t}\circ \+ Q_t$
  \State \quad $\bar{\bs \Lambda}_{\tau_r}\leftarrow (\+1_m \otimes \+1_m - \+Q_{\tau_r})\circ \bs \Lambda_{\tau_r}+ \bar{\+Y}_{{r}} \circ \bs \Lambda_{\tau_r} \circ \+ Q_{\tau_r}  \oslash \+M_r $
  \State\quad $\mathcal{L}(\bar{\+ Y}_r |\bar{\+ Y}_{1:r-1} )\leftarrow  - \+1_m^\top\+M_r\+1_m +\+1_m^\top (\bar{\+ Y}_r \circ \log\+M_r)\+1_m - \+1_m^\top \log(\bar{\+ Y}_r!)\+1_m $
  \State \quad $\bar{\bs \lambda}_{ \tau_{r}} \leftarrow (\+ 1_m^\top \bar{ \+ \Lambda}_{ \tau_{r}})^\top$
  \State {\bf end for}
\end{algorithmic}
\end{algorithm}
In algorithm \ref{alg:Ztagg} we adopt the same conventions concerning $0/0$ and $0\log0\coloneqq 0$ as in algorithm \ref{alg:x}. The log PAL associated with algorithm \ref{alg:Ztagg} is:
\begin{equation}
\log p(\bar{\+Y}_{1:r})\approx \sum_{s=1}^r \mathcal{L}(\bar{\+Y}_s | \bar{\+Y}_{1:s-1}),\label{eq:pal_case_II}
\end{equation}
where, as per line of 9 of  algorithm \ref{alg:Ztagg}, each term $\mathcal{L}(\bar{\+Y}_r | \bar{\+Y}_{1:r-1})$ is the log probability mass function of $\mathrm{Pois}(\+M_r)$ evaluated at $\bar{\+Y}_r$.

\section{Consistency of maximum PAL estimators}\label{sec:consistency}
Whilst the results in section \ref{sec:filtering} explain how the steps in algorithms \ref{alg:x}-\ref{alg:Ztagg} and the associated PALs are motivated by recursive vector-Poisson approximations, so far nothing we have stated quantifies the quality of these approximations, nor the PALs. In this section we present consistency results for parameter estimators defined by maximising PALs. Section \ref{sec:interpretation} of the supplementary material contains a simulation-based example to empirically illustrate our theoretical results.

\subsection{Notation and definitions for the consistency results}
We now introduce explicit notation for dependence of various quantities on a parameter vector $\bs\theta$; we allow   $\mathbb{P}_{0,n},\+ K_{t ,\bs \eta}, \+q_t, \+Q_t, \+G_t, \bs \delta_t$ to depend on $\bs\theta$, and reflect this throughout section \ref{sec:consistency} with notation  $\mathbb{P}_{0,n}^{\bs\theta},\+ K_{t ,\bs \eta}(\bs \theta),\+q_t(\bs \theta), \+Q_t(\bs \theta),\+G_t(\bs \theta),\bs \delta_t(\bs \theta)$. We  allow $\bs \kappa_{t}$ and $\bs \alpha_{t}$ to depend on $\bs \theta$, as well as the expected initial  population size $n$, with notation  $\bs \kappa_{t,n}(\bs \theta)$ and $\bs \alpha_{t,n}(\bs \theta)$.  We also need to make explicit the dependence on $n$ and $\bs\theta$ of the quantities computed in algorithms \ref{alg:x} and \ref{alg:Ztagg}; we write these as: $\bs \lambda_{t,n}(\bs \theta)$, $\bar{\bs \lambda}_{t,n}(\bs \theta), \bs\mu_{t,n}(\bs\theta)$; and $\bs \Lambda_{t,n}(\bs \theta)$, $\bar{\bs \Lambda}_{t,n}(\bs \theta)$, $\+M_{r,n}(\bs \theta)$. 

In either case (I) or (II), one can think of the expected initial population size $n$ as a global model index. We write $(\Omega_n, \mathcal{F}_n, \mathbb{P}^{\bs \theta}_n)$ for a probability space underlying each of these cases with expected initial population size $n$; in the context of case (I), $ \mathbb{P}^{\bs \theta}_n$ is the joint distribution of $(\+x_t)_{t\geq 0}$ and $(\+y_t)_{t\geq 1}$ (as formulated in section \ref{sec:models})  whilst in the context of case (II), $ \mathbb{P}^{\bs \theta}_n$ is the joint distribution of $(\+Z_t)_{t\geq 1}$ and $(\bar{\+Y}_r)_{r\geq 1}$. In either case the overall probability space we shall work with is $(\Omega, \mathcal{F}, \mathbb{P}^{\bs \theta}) \coloneqq (\prod_{n \geq 1}\Omega_n, \bigotimes_{n \geq 1}\mathcal{F}_n, \bigotimes_{n \geq 1}\mathbb{P}^{\bs \theta}_n)$. From henceforth we denote by  $\bs \theta^*\in\Theta$ an arbitrarily chosen but then fixed data-generating parameter (DGP). Almost sure convergence under $\mathbb{P}^{\bs\theta^*}$ is denoted $\asls$.

We now fix a time horizon $T\geq 1$ where for case (I), $T$ is any positive integer, whilst for case (II), we assume $T=\tau_{R}$ for some $R\geq 1$. Since this time horizon is fixed, it will not appear explicitly in some of the notation for our consistency results. However, in order to state and prove various results, we need to make the dependence on $\bs\theta$ and  $n$ of the PALs computed using algorithms \ref{alg:x} and \ref{alg:Ztagg} explicit. To do so we define
\begin{equation*}\label{eq:log_like_defns_consistency}
\ell_n(\bs\theta)\coloneqq\sum_{t=1}^T \ell(\+y_t|\+y_{1:t-1}),\qquad \mathcal{L}_n(\bs\theta)\coloneqq \sum_{r=1}^{R} \mathcal{L}(\bar{\+Y}_r |\bar{\+Y}_{1:r-1}),
\end{equation*}
where it is to be understood that each of the terms $ \ell(\+y_t|\+y_{1:t-1})$ and $\mathcal{L}(\bar{\+Y}_r |\bar{\+Y}_{1:r-1})$ are computed using respectively algorithms \ref{alg:x} and \ref{alg:Ztagg} with parameter value $\bs\theta$ and expected initial population size $n$, and where the distribution of the random variables $\+y_{1:T}$ and $\bar{\+Y}_{1:R}$ is specified by the DGP $\bs\theta^*$ and the expected initial population size $n$. The fact that $\ell_n(\bs\theta)$ and $\mathcal{L}_n(\bs\theta)$ are functions of respectively $\+y_{1:T}$ and $\bar{\+Y}_{1:R}$ is not shown in the notation.  
\subsection{Assumptions}
\begin{assumption}\label{as:compact}
The parameter space $\Theta \subset \mathbb{R}^d$ is compact.
\end{assumption}
\begin{assumption} \label{as:params}  For all probability vectors $\bs \eta$, $t \geq 1$, and $n \geq 1$,
$\+K_{t, \bs \eta}(\bs \theta),\+q_t(\bs \theta), \+Q_t(\bs \theta),$ $\+G_t(\bs \theta),$ $\bs \delta_t(\bs \theta),\bs \kappa_{t,n}(\bs \theta)$ and $\bs \alpha_{t,n}(\bs \theta)$ are continuous functions of $\bs \theta$, and the supports of these vectors and the supports of each matrix row do not depend on $\bs \theta$ or $n$. For all $\bs \theta \in \Theta$ and $t\geq1$,  $\mathrm{supp}(\bs \delta_t(\bs \theta)) = [m]$, i.e. $\bs \delta_t(\bs \theta)$ has no entries equal to $0$. Furthermore, there exist continuous functions of $\bs \theta$ mapping $\Theta\rightarrow\mathbb{R}_{\geq 0}^m$, $\bs \kappa_{t, \infty}({\bs \theta})$ and $\bs \alpha_{t, \infty}({\bs \theta})$,  such that $\mathrm{supp}(\bs \kappa_{t, \infty}({\bs \theta})) = \mathrm{supp}(\bs \kappa_{t, n}({\bs \theta}))$ and $\mathrm{supp}(\bs \alpha_{t, \infty}({\bs \theta})) = \mathrm{supp}(\bs \alpha_{t, n}({\bs \theta}))$ for all $n$, and for each $\bs \theta \in \Theta$ there exist $a_1>0$, $a_2>0$, $\gamma_1>0$, and $\gamma_2>0$ such that:
\begin{align*}
\|n^{-1}\bs \kappa_{t,n}({\bs \theta}) -\bs \kappa_{t, \infty}({\bs \theta}) \|_\infty &< a_1n^{-( \frac{1}{4}+ \gamma_1)},\\
\|n^{-1}\bs \alpha_{t,n}({\bs \theta}) -\bs \alpha_{t, \infty}({\bs \theta})\|_\infty &< a_2n^{-( \frac{1}{4}+ \gamma_2)}.
\end{align*}
\end{assumption}
\begin{assumption}\label{as:1}
For all $\bs \theta \in \Theta$, there exists a constant $c>0$  such that for all $t \geq 1$, all vectors $\bs f_1, \bs f_2 \in \mathbb{R}^m$, and all probability vectors $\bs \eta, \bs \eta'$:
\begin{equation*}|\bs f_1^\top \+K_{t, \bs \eta}(\bs \theta) \bs f_2 - \bs f_1^\top \+K_{t, \bs \eta'}(\bs \theta)\bs f_2 | \leq c \|\bs f_1\|_\infty\|\bs f_2\|_\infty\|\bs \eta - \bs {\eta'}\|_\infty.
\end{equation*}
Furthermore, if $\mathrm{supp}(\bs \eta) \subseteq \mathrm{supp}(\bs \eta')$ then $\mathrm{supp}(\+K^{i,\cdot}_{t,\bs \eta}(\bs \theta) )\subseteq \mathrm{supp}(\+K^{i,\cdot}_{t,\bs \eta'}(\bs \theta))$ for all  $i\in[m]$.
\end{assumption}
\begin{assumption} \label{as:init}
Let $\bs \theta \in \Theta$, $n \in \mathbb{N}$, and $\+x_{0} \sim \mathbb{P}^{\bs \theta}_{0,n}$. There exists $\bs \lambda_{0,\infty}({\bs \theta})$ which is a continuous mapping $\Theta\to\mathbb{R}^m_{\geq 0}$ such that the support of $\bs \lambda_{0,\infty}({\bs \theta})$, which is not the empty set, does not depend on $\bs \theta$, and  there exists $\gamma_0>0$ such that for any $\bs f \in \mathbb{R}^m$ there exists  a $c_0>0$ such that:
\begin{equation*}
\E \left[\left| n^{-1}\bs f^\top\+x_{0} - \bs f^\top \bs \lambda_{0,\infty}({\bs \theta})\right|^4\right]^\frac{1}{4} < c_0n^{-(\frac{1}{4} + \gamma_0)}.
\end{equation*}
Furthermore, there exists some $c>0$ and $\gamma>0$ such that:
\begin{equation*}
\|n^{-1} \bs \lambda_{0,n}(\bs \theta) -  \bs \lambda_{0,\infty}(\bs \theta)\|_\infty <cn^{-(\frac{1}{4}+\gamma)}
\end{equation*}
and $\mathrm{supp}(\bs \lambda_{0,n}(\bs \theta) ) = \mathrm{supp}(\bs \lambda_{0,\infty}(\bs \theta))$ for all $\bs \theta \in \Theta$ and $n \in \mathbb{N}$.
\end{assumption}

The compactness of $\Theta$ in assumption \ref{as:compact} and the continuity in $\bs\theta$ of various quantities in assumption \ref{as:params} are fairly standard assumptions in proofs of consistency of maximum likelihood estimators.  The conditions on the supports of various vectors in assumptions \ref{as:params}-\ref{as:init} are used to rule out the possibility that different parameter values may induce mutually singular distributions over observations, this helps us ensure well-defined contrast functions in our consistency proofs. Assumption \ref{as:init} asserts that the scaled initial population configuration, $ n^{-1}\+x_{0}$, obeys a law of large numbers.
\subsection{Main consistency theorem and outline of the proof}\label{sec:consistency_result}
In order to state and explain our main consistency result, theorem  \ref{theo:consistency}, we now summarize some intermediate results concerning the asymptotic behaviour of the models and quantities calculated using algorithms \ref{alg:x} and \ref{alg:Ztagg}. Precise statements and proofs of these intermediate results are in section \ref{sec:proofs} of the supplementary material.
\paragraph*{Laws of large numbers.} The first step is to establish laws of large numbers for the Latent Compartmental Model, and hence for the observations, these results are stated and proved in section \ref{sec:lln} of the supplementary material. In case (I) we show that for certain deterministic vectors $\bs\nu_t(\bs\theta^*)$, $t\geq1$,
\begin{align}
&\frac{1}{n} \+x_t\asls \bs\nu_t(\bs\theta^*),\qquad \frac{1}{n} \+y_t\asls [(\bs \nu_t(\bs \theta^*) \circ \+ q_t(\bs \theta^*) )^\top \+ G_t(\bs \theta^*)]^\top + \bs \kappa_{t,\infty}(\bs \theta^*),\label{eq:lln_summary_x}
\end{align}
and in case (II), for certain deterministic matrices $\+N_t(\bs\theta^*)$, $t\geq1$, 
\begin{align}
&\frac{1}{n}\+Z_t \asls \+N_t(\bs\theta^*),\qquad \frac{1}{n} \bar{\+Y}_r\asls \sum_{t=\tau_{r-1}+1}^{\tau_r}\+N_t(\bs\theta^*)\circ \+Q_t(\bs\theta^*).\label{eq:lln_summary_z}
\end{align}
The vectors $\bs\nu_t(\bs\theta^*)$ and matrices $\+N_t(\bs\theta^*)$ satisfy recursive (in time) formulae and the convergence of $\frac{1}{n} \+x_t$ and $\frac{1}{n}\+Z_t$ as $n\to\infty$ is a discrete time analogue of the convergence of the continuous time, stochastic model to the solution of the ODE in \eqref{eq:kurtz_lln}, i.e. a discrete-time counterpart of the results  of \citep{kurtz1970solutions}. 
\paragraph*{Filtering intensity limits and asymptotic filtering accuracy.} Making use of the laws of large numbers for the observations, the next step is to establish convergence to deterministic limits of intensity vectors and matrices computed using respectively algorithms \ref{alg:x} and \ref{alg:Ztagg} and which thus define the PALs \eqref{eq:pal_case_I} and \eqref{eq:pal_case_II}. This is the subject of section \ref{sec:filtering_limits} of the supplementary material. In case (I) we find deterministic vectors $\bs\lambda_{t,\infty}(\bs\theta^*, \bs\theta)$ and $\bs\mu_{t,\infty}(\bs\theta^*, \bs\theta)$, $t\geq1$, $\bs\theta\in\Theta$, where $\bs\mu_{t,\infty}(\bs\theta^*, \bs\theta)$ is a function of $\bs\lambda_{t,\infty}(\bs\theta^*, \bs\theta)$, such that: 
\begin{align*}
&\frac{1}{n} \bs\lambda_{t,n}(\bs\theta)\asls \bs\lambda_{t,\infty}(\bs\theta^*, \bs\theta),\qquad \frac{1}{n} \bs\mu_{t,n}(\bs\theta)\asls \bs\mu_{t,\infty}(\bs\theta^*, \bs\theta).
\end{align*}
In case (II) we find deterministic matrices $\bs\Lambda_{t,\infty}(\bs\theta^*, \bs\theta)$ and $\+M_{r,\infty}(\bs\theta^*, \bs\theta)$, $t\geq1$, $r\geq1$, $\bs\theta\in\Theta$, where $\+M_{r,\infty}(\bs\theta^*, \bs\theta)$ is a function of   $\bs\Lambda_{t,\infty}(\bs\theta^*, \bs\theta)$ for $t=\tau_{r-1}+1,\ldots,\tau_r$, such that:
\begin{align*}
&\frac{1}{n} \bs\Lambda_{t,n}(\bs\theta)\asls \bs\Lambda_{t,\infty}(\bs\theta^*, \bs\theta),\qquad \frac{1}{n} \+M_{r,n}(\bs\theta)\asls \+M_{r,\infty}(\bs\theta^*, \bs\theta).
\end{align*}
A notable fact about the limiting filtering intensities $\bs\lambda_{t,\infty}(\bs\theta^*, \bs\theta)$ and $\bs\Lambda_{t,\infty}(\bs\theta^*, \bs\theta)$ that we uncover (see remarks \ref{rem:lam_mu} and \ref{rem:Lam_M} in section \ref{sec:filtering_limits} of the supplementary material) is that:
\begin{equation*}
\bs\lambda_{t,\infty}(\bs\theta^*, \bs\theta^*) = \bs\nu_t(\bs\theta^*),\qquad \bs\Lambda_{t,\infty}(\bs\theta^*, \bs\theta^*) = \+N_t(\bs\theta^*),
\end{equation*}
where $\bs\nu_t(\bs\theta^*)$ and $\+N_t(\bs\theta^*)$ are as in \eqref{eq:lln_summary_x} and \eqref{eq:lln_summary_z}. In this sense, running algorithms \ref{alg:x} and \ref{alg:Ztagg} with the model specified by the DGP $\bs\theta\leftarrow\bs\theta^*$ is asymptotically accurate as $n\to\infty$, in spite of the recursive Poisson approximations involved in these procedures.
\paragraph*{Contrast functions.} We then construct contrast functions associated with the PALs. This is the subject of section \ref{sec:contrast_functions} of the supplementary material. The contrast functions turn out to be in the form of Kullback-Liebler divergences.  In case (I),
\begin{equation}\label{eq:contx}
\frac{1}{n}\ell_n(\bs\theta) - \frac{1}{n}\ell_n(\bs\theta^*) \asls -\sum_{t=1}^T\mathrm{KL}\left(\mathrm{Pois}[\bs\mu_{t,\infty}(\bs\theta^*,\bs\theta^*)]\,\|\,\mathrm{Pois}[\bs\mu_{t,\infty}(\bs\theta^*,\bs\theta)]\right),
\end{equation}
and in case (II),
\begin{equation} \label{eq:contZ}
\frac{1}{n}\mathcal{L}_n(\bs\theta) - \frac{1}{n}\mathcal{L}_n(\bs\theta^*) \asls -\sum_{r=1}^R\mathrm{KL}\left(\mathrm{Pois}[\+M_{r,\infty}(\bs\theta^*,\bs\theta^*)]\,\|\,\mathrm{Pois}[\+M_{r,\infty}(\bs\theta^*,\bs\theta)]\right),
\end{equation}
where in each case the convergence is established to be uniform in $\bs\theta$. 
\paragraph*{Convergence of the maximum PAL estimators.}With:
\begin{align*}
\Theta^{*}_{(I)} &\coloneqq \{\bs \theta \in \Theta: \bs \mu_{t, \infty}(\bs \theta^*, \bs \theta) = \bs \mu_{t, \infty}(\bs \theta^*,\bs \theta^*) \text{ for all } t = 1, \dots T\},\\
\Theta^{*}_{(II)} &\coloneqq \{\bs \theta \in \Theta: \+ M_{r, \infty}(\bs \theta^*, \bs \theta) = \+ M_{r, \infty}(\bs \theta^*,\bs \theta^*) \text{ for all } r = 1, \dots R\},
\end{align*}
uniform convergence to the contrast functions as well as standard continuity and compactness arguments are used to complete the proof of our main consistency result:
\begin{theorem}\label{theo:consistency}
Let assumptions \ref{as:compact}-\ref{as:init} hold and let $\hat{ \bs\theta}_n$ be a maximiser of $  \ell_n(\bs \theta)$ (resp. $\mathcal{L}_n(\bs \theta)$). Then  $\hat{ \bs\theta}_n$ converges to $\Theta^*_{(I)}$ (resp. $\Theta^*_{(II)}$) as $n \rightarrow \infty$, $\mathbb{P}^{\bs \theta^*}$-almost surely.
\end{theorem}
\noindent The proof is in section \ref{sec:conv_pal_estimators} of the supplementary material. Section \ref{sec:interpretation} of the supplementary material illustrates the main results through a simulation study.
\paragraph*{Identifiability.}
We now provide some further insight into the sets $\Theta^{*}_{(I)} $ and $\Theta^{*}_{(II)}$ in order to explain in what sense the model is identified under theorem \ref{theo:consistency}. In section \ref{sec:identifiabiliy} of the supplementary material we show that for any $\bs \theta\in\Theta$,
\begin{align*}
\bs \theta \in \Theta^*_{(I)}&\quad\Longleftrightarrow\quad \bs \mu _{t,\infty}(\bs \theta,\bs \theta) = \bs \mu _{t,\infty}(\bs \theta^*,\bs \theta^*),\quad\forall t=1,\ldots,T,\\ 
\bs \theta \in \Theta^*_{(II)}&\quad\Longleftrightarrow\quad \+{M}_{r,\infty}(\bs \theta,\bs \theta) = \+{M} _{r,\infty}(\bs \theta^*,\bs \theta^*),\quad\forall r=1,\ldots,R.
\end{align*}
The vector $\bs \mu_{t,\infty}(\bs \theta^*,\bs \theta^*)$ turns out (see remark \ref{rem:lam_mu}) to be equal to the r.h.s. of the second $\mathbb{P}^{\bs\theta^*}$-almost sure limit in \eqref{eq:lln_summary_x}. Thus for case (I),  the convergence to $\Theta^*_{(I)}$ in theorem \ref{theo:consistency} tells us that as $n\to\infty$,  $\hat{ \bs\theta}_n$ approaches the set of $\bs\theta$ such that the $\mathbb{P}^{\bs\theta}$-almost sure limit of $\frac{1}{n} \+y_t$ is the same as the $\mathbb{P}^{\bs\theta^*}$-almost sure limit of $\frac{1}{n} \+y_t$, for all $t=1,\ldots,T$. Similarly for case (II), $\+M_{t,\infty}(\bs \theta^*,\bs \theta^*)$  turns out (see remark \ref{rem:Lam_M}) to be equal to the r.h.s. of the second limit in \eqref{eq:lln_summary_z}, and the convergence to $\Theta^*_{(II)}$ in theorem \ref{theo:consistency} tells us that as $n\to\infty$,  $\hat{ \bs\theta}_n$ approaches the set of $\bs\theta$ such that the $\mathbb{P}^{\bs\theta}$-almost sure limit of $\frac{1}{n} \bar{\+Y}_r$ is the same as the $\mathbb{P}^{\bs\theta^*}$-almost sure limit of $\frac{1}{n} \bar{\+Y}_r$, for all $r=1,\ldots,R$.


\section{Dealing with over-dispersion}\label{sec:over-dispersion}
Over-dispersion is an important modelling consideration in many epidemiological contexts and may have substantial implications for model fit and predictive uncertainty.  The models we have considered so far are equi-dispersed in the sense of \cite{breto2011compound}. For compartmental models in general, over-dispersion can be incorporated in  either the transition or observation models, or both, see for example \citep{stocks2020model}. In the context of the models from section \ref{sec:models}, a natural approach would be to replace the binomial and Poisson-distributed elements of the latent compartmental model (section \ref{sec:lat_comp_model}) and/or observation models (sections \ref{sec:obs_model_x}-\ref{sec:obs_model_Z_agg}) with over-dispersed counterparts, such as beta-binomial and negative binomial distributions. It appears that analytically tractable PAL-style approximations cannot be derived for such models. However, one can often construct over-dispersed distributions as compound distributions through introduction of latent variables, e.g. placing a beta prior on $q_t^{(i)}$ in \eqref{eq:obsx} and then integrating out would result in a marginally beta-binomial observation model. Similarly priors could be placed on parameters which specify the matrix $\+ K_{t ,\bs \eta}$, the immigration and emigration parameters $\bs \alpha_t$, $\bs \delta_t$, the spurious observation intensity $\bs \kappa_t$,  and so on. It is through this latent variable perspective that we extend the use of the PAL to deal with over-dispersion.

Consider  the latent compartmental model from section \ref{sec:lat_comp_model} combined with observation mechanism from section \ref{sec:obs_model_x} with parameter $\bs \theta$ (the observation models from sections \ref{sec:obs_model_Z} and \ref{sec:obs_model_Z_agg} can be handled in a very similar manner). We consider $\bs \theta$ to be partitioned into two components: $\bs \theta = [{\bs \vartheta} \; \bar{\bs \theta}_{1:T}]$, where ${\bs \vartheta}$ consists of parameters which are either fixed or to be estimated, and $\bar{\bs \theta}_{1:T}\sim f(\cdot | \bs \varphi)$ are to be integrated out, for some density $f$ and hyperparameter $\bs \varphi$. A default approach would be for $\bar{\bs \theta}_{1:T}$ to be independent under $f(\cdot | \bs \varphi)$, but Markovian or other dependence could be incorporated. 

We assume that the elements of the model are parameterised such that:
\begin{align*}
\bs \alpha_t(\bs \theta) &= \bs \alpha(\bs \vartheta, \bar{\bs \theta}_{t}),\quad \bs \delta_t(\bs \theta) = \bs \delta(\bs \vartheta, \bar{\bs \theta}_t), \quad\+ K_{t ,\bs \eta}(\bs \theta) = \+ K_{\bs \eta}(\bs \vartheta, \bar{\bs \theta}_t),\\
\quad \bs \kappa_t(\bs \theta) &= \bs \kappa(\bs \vartheta, \bar{\bs \theta}_{t}),  \quad \+q_t(\bs \theta) = \+q(\bs \vartheta, \bar{\bs \theta}_t),\quad\+G_t(\bs \theta) = \+G(\bs \vartheta, \bar{\bs \theta}_t),
\end{align*}
for some given functions $\bs \alpha$, $\bs \delta$, etc., which implies that:
$$
p(\+{x}_t | \+{x}_{t-1},\bs\theta) = p(\+{x}_t | \+{x}_{t-1},\bs\vartheta, \bar{\bs \theta}_t),\qquad p(\+{y}_t | \+{x}_{t},\bs\theta) = p(\+{y}_t | \+{x}_{t},\bs\vartheta, \bar{\bs\theta}_t),
$$
and in turn that $\bar{\bs \theta}_{t}$ is conditionally independent of ${\+ y}_{1:t-1}$ given $\bar{\bs \theta}_{1:t-1}$, $\bs\vartheta$ and $\bs\varphi$.


Let us derive the marginal likelihood for the parameters $[\bs \vartheta \; \bs \varphi]$ with $\bar{\bs \theta}_{1:T}$ integrated out. Momentarily regarding $[\bs \vartheta \; \bs \varphi]$ as fixed and suppressing it from notation, consider the recursive relationship:
\begin{align}
p(\+ y_{1:t}, \bar{\bs \theta}_{1:t}) &= p(\+ y_{t},  \bar{\bs \theta}_{t}  | \+ y_{1:t-1}, \; \bar{\bs \theta}_{1:t-1} )p(\+ y_{1:t-1}, \; \bar{\bs \theta}_{1:t-1} ) \\
&= p(\+ y_{t} |  \bar{\bs \theta}_{t}  , \+ y_{1:t-1},  \bar{\bs \theta}_{1:t-1})p(\bar{\bs \theta}_t | \+ y_{1:t-1},  \bar{\bs \theta}_{1:t-1} )p(\+ y_{1:t-1},\bar{\bs \theta}_{1:t-1} ) \\
&= p(\+ y_{t} |  \bar{\bs \theta}_{t}  , \+ y_{1:t-1},  \bar{\bs \theta}_{1:t-1})f(\bar{\bs \theta}_t | \bar{\bs \theta}_{1:t-1} )p(\+ y_{1:t-1},\bar{\bs \theta}_{1:t-1}),
\end{align}
where the third equality holds due to the aforementioned conditional independence.  Now, re-introducing $[\bs \vartheta \; \bs \varphi]$ to the notation, we have:
\begin{align}
p(\+ y_{1:T}| \bs \vartheta, \bs \varphi) &= \int p(\+ y_{1:T}, \; \bar{\bs \theta}_{1:T}  | \bs \vartheta, \bs \varphi )d\bar{ \bs \theta}_{1:T}  \\
 &= \int \prod _{t=1}^Tp(\+ y_{t}| \+ y_{1:t-1}, {\bs \vartheta},\bar{\bs \theta}_{1:t})f(\bar{\bs \theta}_t | \bar{\bs \theta}_{1:t-1}, \bs \varphi)d\bar{\bs \theta}_{1:T}.
\end{align}
We can approximate this using the PAL:
\begin{align}
\label{eq:SMCapproxlik}
p(\+ y_{1:T}| \bs \vartheta, \bs \varphi) \approx \int \prod _{t=1}^T \exp \left\{ \ell(\+y_{t}| \+ y_{1:t-1}, {\bs \vartheta},\bar{\bs \theta}_{1:t})\right\}f( \bar{\bs \theta}_t | \bar{\bs \theta}_{1:t-1}, \bs \varphi)d\bar{\bs \theta}_{1:T},
\end{align}
where $\ell$ is defined as per algorithm \ref{alg:x}. The right-hand side of \eqref{eq:SMCapproxlik} can be efficiently numerically approximated by embedding PAL computations within sequential Monte Carlo -- see \citep{chopin2020introduction} for an introduction to this family of Monte Carlo algorithms. Such a scheme is given by algorithm \ref{alg:overSMC} and its subroutine algorithm \ref{alg:smcupdate}.

In line \ref{propline} of algorithm \ref{alg:overSMC} we take the convention $\pi(\cdot | \bar{\bs \theta}^{(i)}_{1:0},  \bar{\bs \lambda}_{t-1}^{(i)}, \+y_{1:1}) \coloneqq \pi(\cdot | \bar{\bs \lambda}_{0}^{(i)}, \+y_{1})$. Algorithm \ref{alg:overSMC} yields a Monte Carlo approximation to the r.h.s. of \eqref{eq:SMCapproxlik}, so overall we obtain:
\begin{align}
\log p(\+y_{1:t}| \bs \vartheta, \bs \varphi) \approx \sum_{s=1}^t \widehat\ell(\+ y_s| \+ y_{1:s-1},\bs \vartheta, \bs \varphi).
\end{align}
We stress there are two ingredients to this approximation: the Monte Carlo approximation and the PAL approximation. Whilst the main emphasis above regarding $\bar{\bs\theta}_{1:t}$ is that they are to be integrated out, a benefit of algorithm \ref{alg:overSMC} is that it also yields the approximation:
\begin{equation}
p(\bar{\bs \theta}_{t} |\+{y}_{1:t},\bs\vartheta, \bs\varphi) \approx \sum_{i=1}^{n_{part}} \bar{w}_t^{(i)}\delta_{\bar{\bs\theta}_{t}^{(i)}},\label{eq:SMC_filtering_approx}
\end{equation}
which enables inference for $\bar{\bs\theta}_t$ on the basis of observations $\+y_{1:t}$.
\begin{algorithm}[H]
\caption{PALSMC subroutine}\label{alg:smcupdate}
\begin{algorithmic}[1]
  \Statex {\bf input:} $\bar{\bs \lambda}_{t-1}$ and $[\bs \vartheta \; \bar{\bs \theta}_t]$
  \State $\bs \alpha_t\leftarrow \bs \alpha_t(\bs \vartheta, \bar{\bs \theta}_{t}), \bs \delta_t\leftarrow \bs \delta_t(\bs \vartheta, \bar{\bs \theta}_t), \+ K_{t ,\bs \eta} \leftarrow \+ K_{t ,\bs \eta}(\bs \vartheta, \bar{\bs \theta}_t),
  \+q_t \leftarrow \+q_t(\bs \vartheta, \bar{\bs \theta}_t),\newline\bs \kappa_t \leftarrow \bs \kappa_t(\bs \vartheta, \bar{\bs \theta}_{t}),\+G_t\leftarrow \+G_t(\bs \vartheta, \bar{\bs \theta}_t)$ 
  \State  $ \bs \lambda_{t} \leftarrow [(\bar{\bs {\lambda}}_{t-1}\circ \bs{\delta}_t)^\top\+{K}_{t, \bs{\eta}( \bar{\bs \lambda}_{t-1}\circ \bs{\delta}_t)}]^\top + \bs{\alpha}_t $
   \State $\bar{ \bs \lambda}_{t} \leftarrow [\+ 1_m - \+q_t +  (\{\+ y_t^\top \oslash[(\+ q_t \circ \bs \lambda_t)^\top \+ G_t + \bs \kappa_t^\top ] \}[(\+ 1_m \otimes \+ q_t)\circ \+ G_t^\top] )^\top ]\circ \bs \lambda_t$
   \State  $\bs\mu_t\leftarrow [(\bs \lambda_t\circ \+ q_t)^\top \+G_t]^\top + \bs \kappa_t$
   \State $\ell(\+y_t | \+y_{1:t-1}) \leftarrow -\bs\mu_t^\top\+1_m + \+y_t^\top \log(\bs\mu_t) - \+1_m^\top\log(\+y_t!)$
   \Statex {\bf return} $\ell(\+y_t | \+y_{1:t-1})$ and $\bar{ \bs \lambda}_{t}$
\end{algorithmic}
\end{algorithm}
\begin{algorithm}[H]
\caption{PALSMC}\label{alg:overSMC}
\begin{algorithmic}[1]
 \Statex {\bf input:} proposal distribution $\pi(\cdot | \cdot)$, number of particles $n_{part}$, parameter $[\bs \vartheta \; {\bs \varphi}]$.
  \Statex {\bf initialize:} $\bar{\bs \lambda}_{0}^{(i)} \leftarrow {\bs \lambda}_0$ for $i = 1,\dots, n_{part}$
  \State {\bf for}  $t  \geq 1$: 
  \State \quad  {\bf for}  $i = 1,\dots, n_{part}$: 
  \State \quad \quad \label{propline} $\bar{\bs \theta}^{(i)}_t \sim \pi(\cdot | \bar{\bs \theta}^{(i)}_{1:t-1},  \bar{\bs \lambda}_{t-1}^{(i)}, \+y_{1:t})$
  \State \quad \quad Obtain  $\ell(\+y_{t}| \+ y_{1:t-1},{\bs \vartheta}, \bar{\bs \theta}^{(i)}_{1:t})$ and $\bar{\bs \lambda}_t^{(i)}$ from algorithm \ref{alg:smcupdate} with input $\bar{\bs \lambda}_{t-1}^{(i)}$ and $\left[{\bs \vartheta},\bar{\bs \theta}_t^{(i)}\right]$  
 \State \quad \quad  $\log w_{t}^{(i)} \leftarrow\ell(\+y_{t}| \+ y_{1:t-1},{\bs \vartheta}, \bar{\bs \theta}^{(i)}_{1:t}) + \log f(\bar{\bs \theta}^{(i)}_t  | \bar{\bs \theta}^{(i)}_{1:t-1} , \bs \varphi) - \log\pi(\bar{\bs \theta}^{(i)}_t  | \bar{\bs \theta}^{(i)}_{1:t-1},\bar{\bs \lambda}_{t-1}^{(i)}, \+y_{1:t})$
  \State \quad {\bf end for}
 \State \quad $\widehat\ell(\+ y_t| \+ y_{1:t-1}, \bs \vartheta, \bs \varphi) \leftarrow \log \left( \frac{1}{n_{part}}\sum_{i=1}^{n_{part}} w_{t}^{(i)}\right)$
 \State \quad $\bar w_{t}^{(i)} \leftarrow w_{t}^{(i)}/\sum_{j=1}^{n_{part}}w_{t}^{(j)}$
 \State \quad resample $\left\{ \bar{\bs \theta}^{(i)}_{1:t}, \bar{\bs \lambda}_{t}^{(i)} \right\}_{i=1}^{n_{part}}$ according to the weights $\left\{\bar w_t^{(i)}\right\}_{i=1}^{n_{part}}$
 \State {\bf end for}
\end{algorithmic}
\end{algorithm}
In section \ref{sec:examples} we explore  ways in which the large population theory from section \ref{sec:consistency} is relevant to the construction and behaviour of PALSMC algorithms for over-dispersed models:
\begin{itemize}
\item It is well known that the efficiency of sequential Monte Carlo  methods can be highly sensitive to the choice of the proposal distribution, $\pi$ in algorithm \ref{alg:overSMC}.  If we could choose $\pi(\bar{\bs \theta}_t  | \bar{\bs \theta}^{(i)}_{1:t-1},\bar{\bs \lambda}_{t-1}^{(i)} \+y_{1:t})$ to be proportional (as a function of $\bar{\bs \theta}_t$) to:
\begin{equation}
\exp\left[\ell(\+y_{t}| \+ y_{1:t-1},{\bs \vartheta}, \bar{\bs \theta}^{(i)}_{1:t-1}, \bar{\bs \theta}_{t})\right]f(\bar{\bs \theta}_t  | \bar{\bs \theta}^{(i)}_{1:t-1} , \bs \varphi), \label{eq:optimal_proposal}
\end{equation}
then the weight $w_{t}^{(i)}$ would have no dependence on $\bar{\bs \theta}_{t}^{(i)}$. Consequently the variability of the weight would be reduced and the overall efficiency of the PALSMC algorithm likely improved. This ``optimal'' choice of $\pi$ is often not analytically tractable, but inspired by our consistency theory we suggest Laplace approximation to it. We demonstrate such proposals in simulation-based and real data examples in sections \ref{sec:SEIRover} -- \ref{sec:measles} and find them to be very efficient in practice.
\item Through a simulation example in section \ref{sec:SEIRover}, we  illustrate that even for our over-dispersed models, where one might expect estimation consistency to be ruled out, increasing population size can in fact increase the accuracy of point estimates of  $\bar{\bs \theta}_{t}$ obtained from the r.h.s. of \eqref{eq:SMC_filtering_approx}. The explanation for this is that whilst the model may be over-dispersed once $\bar{\bs \theta}_{1:t}$ are integrated out, it is equi-dispersed \emph{conditional} on $\bar{\bs \theta}_{1:t}$.
\end{itemize}
In the examples in section \ref{sec:examples} we also expand on  algorithm \ref{alg:overSMC}  to include sophisticated resampling schemes and block particle filtering techniques \citep{rebeschini2015can}.

\section{Discussion and examples}\label{sec:examples}
Code for all examples is available at: \href{https://github.com/LorenzoRimella/PAL}{https://github.com/LorenzoRimella/PAL}.
\subsection{Inference using automatic differentiation and HMC for an age-structured model of 'flu}\label{sec:age_struct}
In this example, we demonstrate PALs for an age-structured model of a 1957 outbreak of influenza in Wales. Computation is performed using the probabilistic programming language Stan \citep{carpenter2017stan}, taking advantage of automatic differentiation to implement Hamiltonian Monte Carlo (HMC). This example also highlights how the general Latent Compartmental Model can accommodate discrete or discretisable covariates associated with subpopulations: in this case the covariates are indicators of the age-group which individuals belong to and this is reflected in the compartment structure of the model.
\subsubsection*{Data and Model}
The data consist of 19 weeks of incidence data in the form of GP symptom reports for a town with population size $8000$ across $4$ age groups: $00-04$, $05-14$, $15-44$, and $45+$. The data were analysed by  \cite{vynnycky2008analyses} and are available via the Github page associated with \citep{andrade2020evaluation}. For each age group $k = 1, \dots, 4$, 
\begin{alignat*}{2}
S_{k,t+1} &= S_{k,t} - B_{k,t}, \qquad &&E_{k,t+1} = E_{k,t} + B_{k,t} - C_{k,t}, \\
I_{k,t+1} &= I_{k,t} + C_{k,t} - D_{k,t}, \qquad &&R_{k,t+1} = R_{k,t} - D_{k,t},
\end{alignat*}
with conditionally independent increments: $B_{k,t} \sim {\text{Bin}(S_{k,t}, 1-e^{-h\bar \beta_{k,t}})}$,  $C_{k,t} \sim {\text{Bin}(E_{k,t}, 1-e^{-h\rho})}$,  $D_{k,t} \sim \text{Bin}(I_{k,t}, 1-e^{-h\gamma}),$ where
\begin{equation}\label{eq:beta_bar}
\left [ \begin{array}{c}
\bar \beta_{1,t} \\
\vdots \\
\bar \beta_{4,t} \\
\end{array} \right ]
= 
\underbrace{
\left [ \begin{array}{ccc}
\beta_{11} & \dots &\beta_{14} \\
\vdots & \ddots & \vdots \\
\beta_{14} & \dots & \beta_{44}\\
\end{array} \right ]}_{=:\+B}
\left [ \begin{array}{c}
I_{1,t} \\
\vdots \\
I_{4,t} \\
\end{array} \right ]\frac{1}{n}.
\end{equation}
$\+B$ is a symmetric matrix with element $\beta_{ij}$ representing the rate at which two individuals, one from the susceptible compartment of the $i$th age group and the other from the infective compartment of the $j$th age group come into effective contact. 

The mean time spent in the exposed compartment $1/\rho$ and the mean recovery time $1/\gamma$ are taken to be independent of age group and set to be $1.5$ days, following \cite{andrade2020evaluation}. We assume that the model evolves daily with $h=1/7$ and that observations consist of cumulative weekly transitions from the $E$ to $I$ compartments for each age group, that is we have observation times at times $\tau_r = 7r$ for $r = 1, \dots,R$ corresponding to the end of each week. In the setting of case (II) we denote observations $\bar{\+Y}_{k,r} = \sum_{s=\tau_{r-1}+1}^{\tau_{r}} {\+Y}_{k,t}$ where each element of each $ {\+Y}_{k,t}$ is equal to zero except the $(2,3)$th element corresponding to transitions from compartment $E$ to $I$ which, conditional on $C_{k,t}$, is distributed ${Y}^{(2,3)}_{k,t} \sim \text{Bin}(C_{k,t}, Q_{k,t}^{(2,3)}),$ where $\+Q_{k,t} \in [0,1]^{4\times 4}$ has elements equal to zero except for the $(2,3)$th entry which is equal to an age group dependant under reporting parameter $q_k \in (0,1)$ which is to be estimated. We give details of how this model is written as an instance of the Latent Compartmental Model and the algorithm used to calculate the PAL in the supplementary material.
\subsubsection*{Hamiltonian Monte Carlo with automatic differentiation in Stan}
We now consider MCMC sampling to approximate the posterior $p(\bs \theta | \bar{\+Y}_{1:R})$.
The probabilistic programming language Stan  \citep{carpenter2017stan} provides a framework for implementing HMC -- a type of MCMC algorithm which uses auxiliary ``momentum'' variables to help explore the posterior -- in which the user only needs to specify priors and provide a function which evaluates the likelihood for the model. Stan uses Automatic Differentiation (AD) to compute gradients and update the auxiliary HMC variables without the need for user input. Since the PAL consists of recursive compositions of elementary linear algebra operations, it is a natural candidate for AD. 
\subsubsection*{Results}
We implemented a Stan program incorporating the PAL, details of which are given in section \ref{sec:age_struct_supp} of the supplementary material. We stress that here we do not correct for the fact that the PAL is only an approximation to the true likelihood, so Stan is targeting an approximation to the true intractable posterior, although in a separate example in the supplementary material we explore corrections using Delayed-Acceptance MCMC methods. 
 
The parameters to be estimated are $\bs \theta = [\beta_{11}\; \cdots\; \beta_{44}\; q_1\;\cdots\;q_4]^\top$, the initial state for each age group is assumed known as $\+x_{1,0}=[948\;0\;1\;0]^\top$, $\+x_{2,0}=[1689\;0\;1\;0]^\top$, $\+x_{3,0}=[3466\;0\;1\;0]^\top$, $\+x_{4,0}=[1894\;0\;1\;0]^\top$.  We used vague gamma priors $\beta_{ij} \sim \text{Gamma}(5,1)$ for ${i,j = 1, \dots,4}$ and a vague truncated normal prior $q_k \sim \mathcal{N}(0.5,0.5)_{\geq 0, \leq 1}$ for $k = 1,\dots,4$. The HMC sampler was run to produce a chain of length $5\times 10^{5}$ iterations, a burn-inperiod of size $10^5$ was discarded and the remaining was thinned to produce a sample of $2.5 \times 10^4$. We report approximate posterior distributions and trace plots in section \ref{sec:age_struct_supp} of the supplementary material, these show no signs of unsatisfactory mixing. Figure \ref{ppplotsstoch} reports the posterior predictive distributions and credible intervals, we see good coverage of observed data. 




We repeated the analysis using an ODE version of the same age-structured SEIR model, from \cite{andrade2020evaluation}, with a Poisson reporting model: we use as emission distribution a Poisson distribution with rate given by the ODE solution scaled by an under-reporting parameter. This was implemented in the Stan framework using the code available in \cite{andrade2020evaluation}, we again sampled a chain of length $5\times 10^{5}$ iterations, discarded a burn-inperiod of size $10^5$, and thinned the remaining to produce a sample of $2.5 \times 10^4$. To calculate the reproduction number $R_0$ for stratified models such as this, one must calculate the so called \textit{next generation matrix} \citep{van2002reproduction} which has elements given by $\frac{n_i\beta_{ij}}{n\gamma}$ where $n_i$ is the population size of the $i$th age group.  $R_0$ is then given by the largest modulus of the eigenvalues of the next generation matrix \citep{diekmann1990definition}. Using this definition, we can produce approximate posterior distributions of $R_0$ using each of the PAL and ODE procedures, which we report in figure \ref{R0age}. The approximate posteriors concentrate around $1.42$ using the PAL and $1.82$ using the ODE model. This disparity in estimates can be related to the features of the posterior predictive distributions reported in figures \ref{ppplotsstoch} and \ref{ppplotdet}: the distribution of trajectories in figure \ref{ppplotdet} appears to `overshoot' the data in comparison to those in figure \ref{ppplotsstoch}, reflecting the higher force of infection implied by the ODE procedure in contrast to the PAL procedure. These posterior predictive plots also exhibit the inherent inflexibility of the ODE model: since the latent process is deterministic, random variations in the data away from the ODE trajectory must be explained as observation error. As is apparent in the $45+$ age group, this rigidity in modelling results in overconfidence and a poor fit compared to that of the stochastic model combined with the PAL procedure.
%
%
%
%
%
%

\begin{figure}[h!]
\centering
    \includegraphics[width=\textwidth, clip=true, trim={40 10 30 40}]{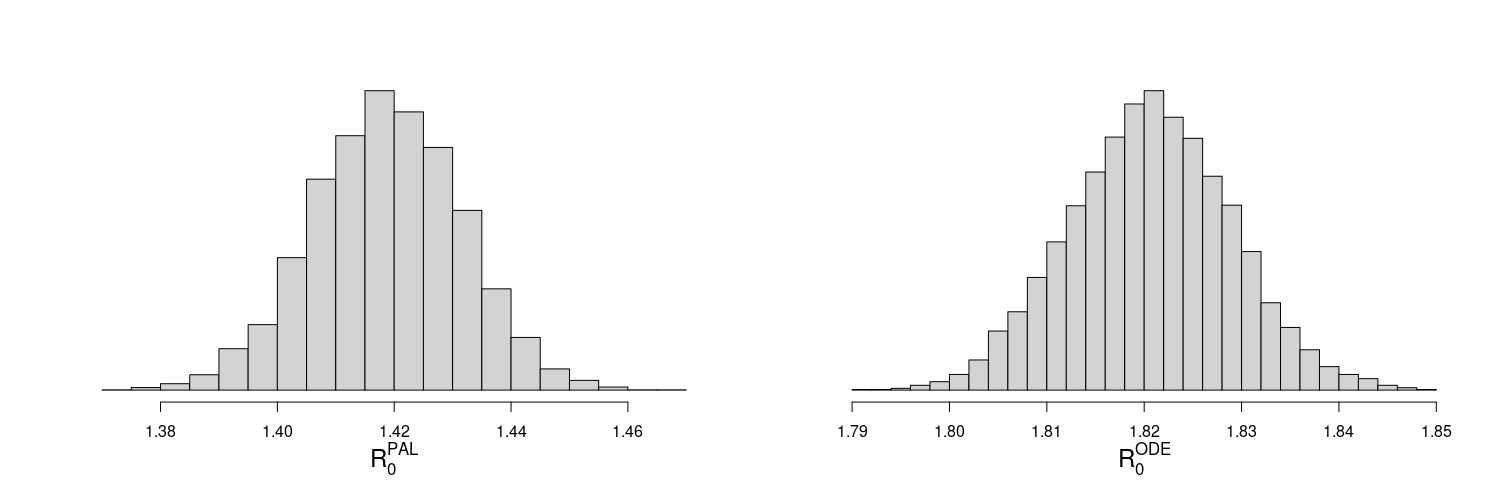}
    \caption{Age-structured 'flu example. Approximate posterior distributions for $R_0$ under the PAL and ODE procedures.\label{R0age}}
\end{figure}
\begin{figure}[h!]
    \centering
    \includegraphics[width=\textwidth]{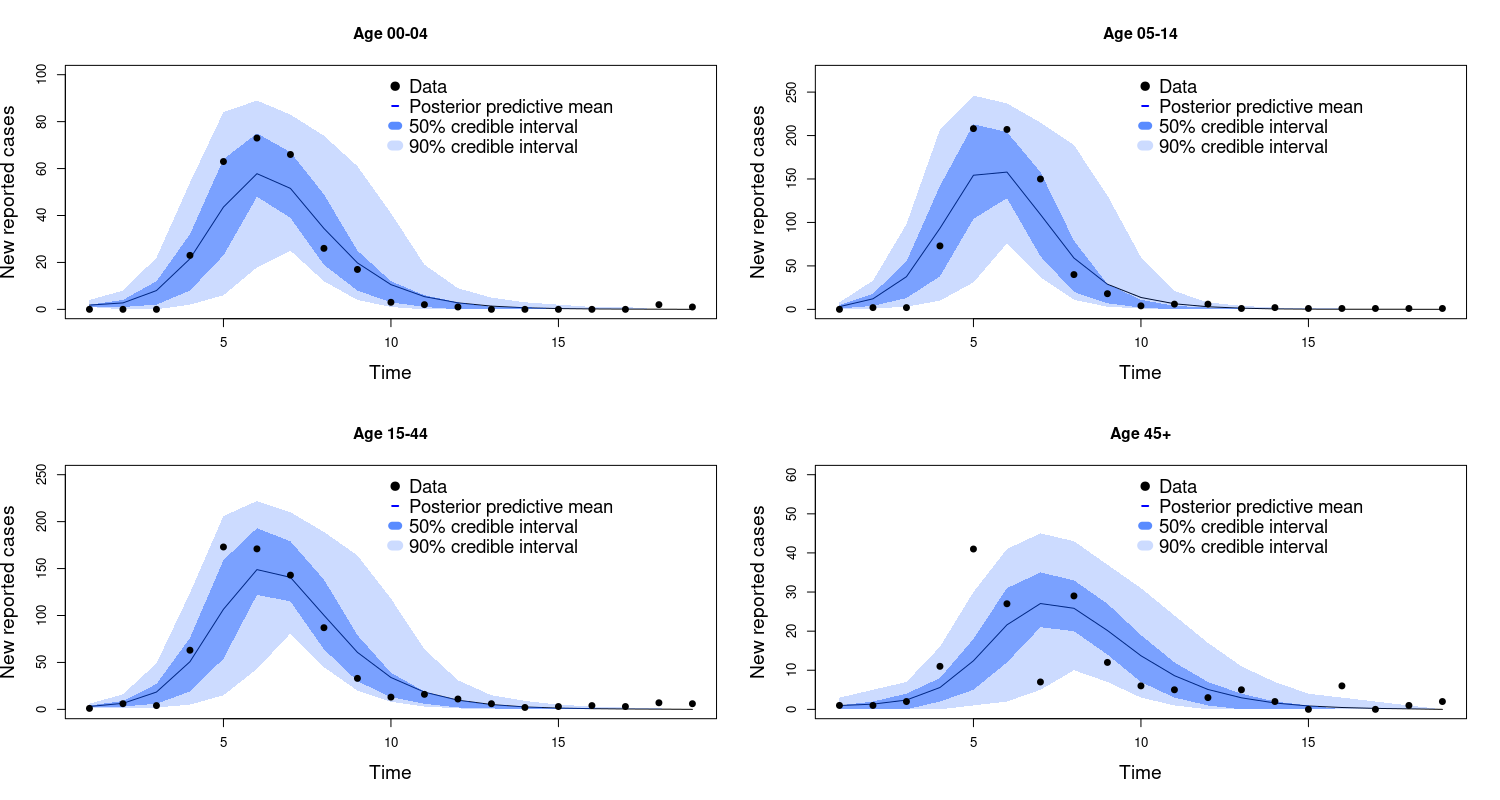}
    \caption{Age-structured 'flu example. Posterior predictive distributions obtained from inference under the stochastic model using the PAL within Stan.}
    \label{ppplotsstoch}
\end{figure}
\begin{figure}[h!]
    \centering
    \includegraphics[width = \textwidth]{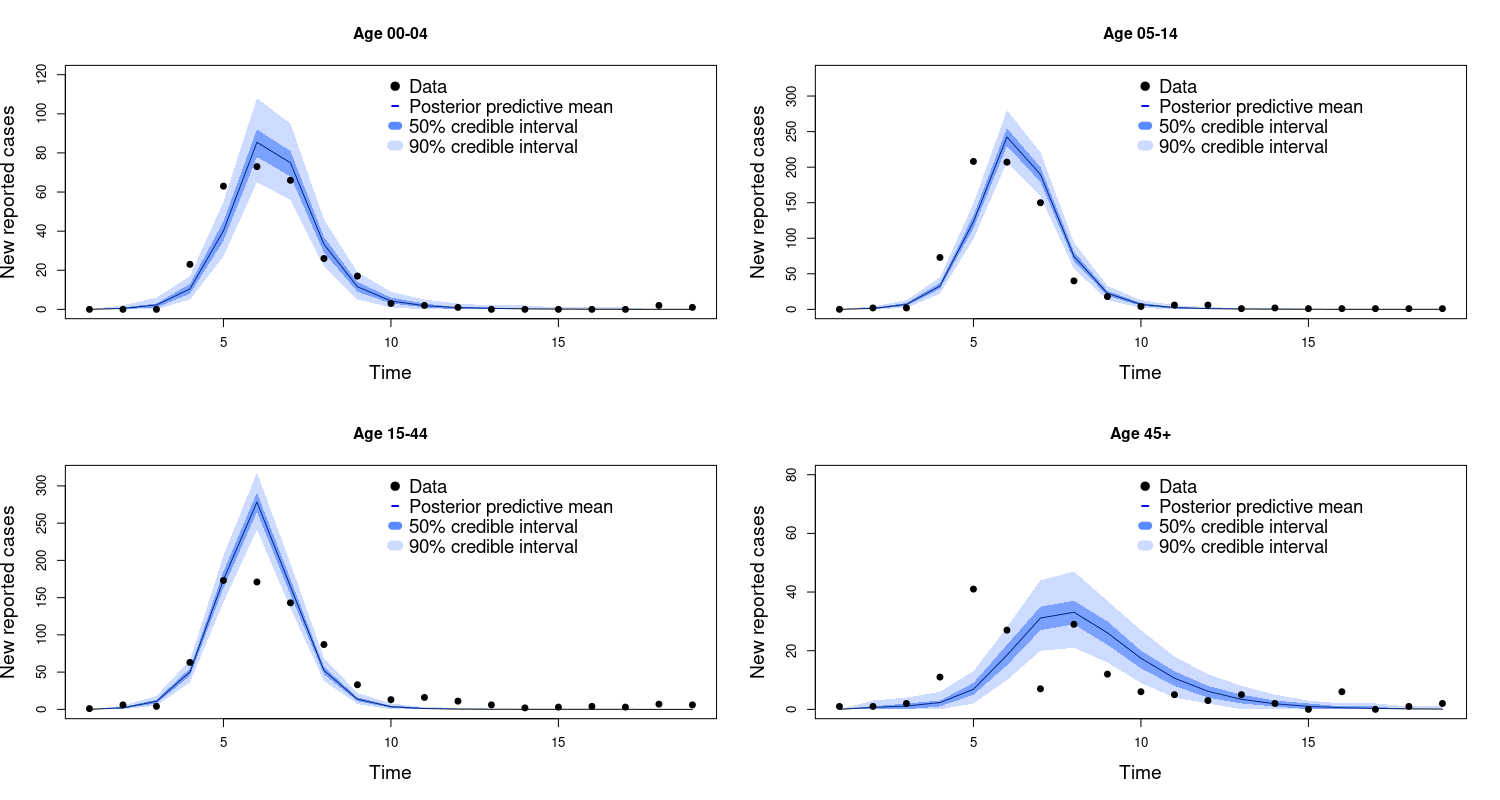}
    \caption{Age-structured 'flu example. Posterior predictive distributions obtained from inference under the ODE model using Stan.}
    \label{ppplotdet}
\end{figure}

\subsection{Pedagogical over-dispersed SEIR example} \label{sec:SEIRover}
To demonstrate inference for an over-dispersed model using PALSMC we consider a simple SEIR model for which the latent population $\+{x}_t\equiv[S_t\;E_t\;I_t\;R_t]^\top$ evolves according to transition matrix \eqref{eq:SEIRexample}, with immigration and emigration parameters, $\bs \alpha_t$ and $\bs \delta_t$, combined with the observation model $y_t\sim \text{Binom}(I_t,q_t)$. We assume $\bs \alpha_t$ and $\bs \delta_t$ are known. We can cast this model in the form discussed in section \ref{sec:over-dispersion} by identifying $\bs \vartheta = [\beta \; \rho\; \gamma ]$, $\bar{\bs \theta}_{1:T} = q_{1:T}$, and choosing $f(\cdot|\bs\varphi)$ to make $q_{1:T}$ i.i.d. according to a truncated normal distribution  $q_t\sim\mathcal{N}(\mu_q, \sigma^2_q)_{\geq 0, \leq 1}$,  with $\bs \varphi = [\mu_q \; \sigma_q^2]$, $\mu_q \in [0,1]$ and $\sigma^2_q>0$. We give the details of a PALSMC scheme for this model in section \ref{sec:SIRov_supp} of the supplementary material, including the design efficient, data-informed proposals by Laplace approximation to \eqref{eq:optimal_proposal}, inspired by the theory from section \ref{sec:consistency}.

\begin{figure}[h!]
\includegraphics[width = \textwidth]{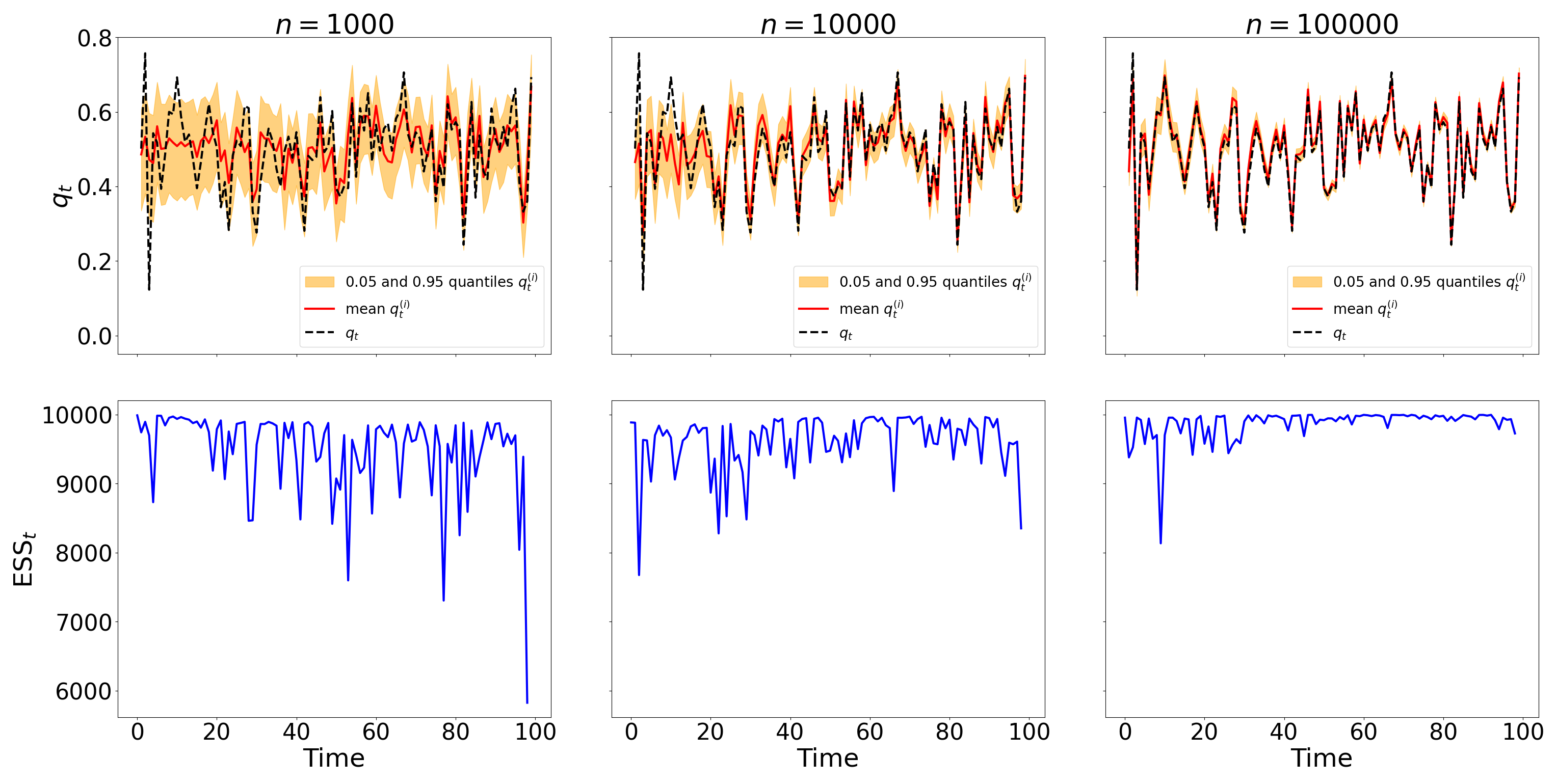}
\includegraphics[width = \textwidth]{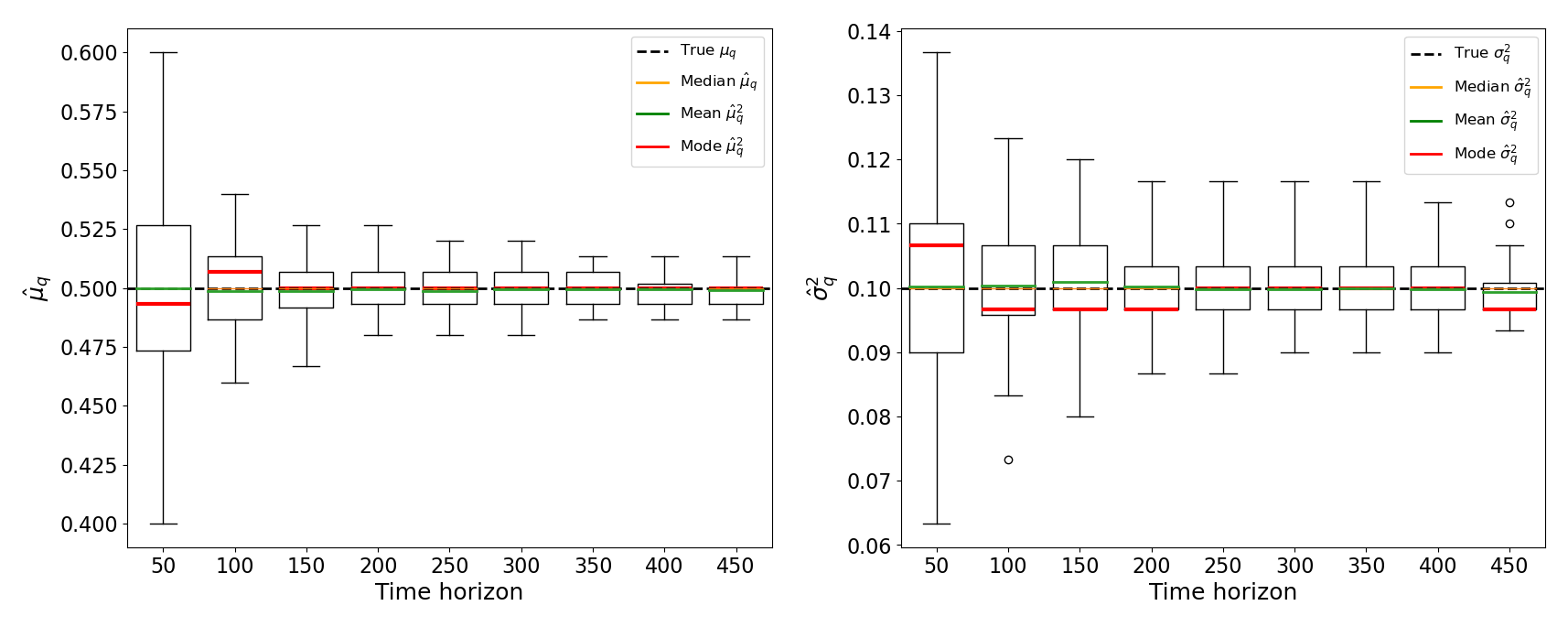}
\caption{Pedagogical over-dispersed SEIR example. Top two rows: filtering distribution approximations and ESS obtained from PALSMC with $n_{part}=10^4$ particles and increasing model population size $n$. Bottom row: maximum PALSMC estimation of hyper-parameters $\bs \varphi = [ \mu_q \; \sigma_q^2]$ over increasing time horizons. Each boxplot summarises $100$ hyper-parameter estimates.}
\label{SEIRoverplt}
\end{figure}

\subsubsection*{Filtering and parameter estimation simulation study}
To assess the ability of the PALSMC scheme to recover ground truth quantities, we simulated data from the model with $[\beta\; \rho \; \gamma \; \mu_q \; \sigma_q^2] = [0.8\; 0.1 \; 0.2 \; 0.5 \; 0.1]$, $\bs \pi_0 = [0.99 \;0 \; 0.01 \; 0]^\top$, $\bs \alpha_t = 0.05\bs \pi_0$ and $\bs \delta_t = [0.95\;0.95\;0.95\;0.95]^\top$. The first two rows of figure \ref{SEIRoverplt} explore the performance of PALSMC with increasing population size $n$ and using the data-generating values of $[\bs \vartheta \; \bs \varphi]$.  This collection of plots was created by first sampling a single draw of latent variables $q_{1:100}\sim f(\cdot|\bs\varphi)$, then for each value of $n=10^3, 10^4,10^5$, generating data $\+y_{1:100}$ from the model conditional on $q_{1:100}$, and running the PALSMC algorithm.  We see that the effective sample size (ESS) is high across all values of population size $n$, indicating a good approximation to the r.h.s. of \eqref{eq:SMCapproxlik}; this reflects the careful choice of proposal distribution. As in \eqref{eq:SMC_filtering_approx}, for each $t\geq1$, the PALSMC algorithm yields a Monte Carlo approximation $p(q_t | y_{1:t}) \approx\sum_{i=1}^{n_{part}}\bar{w}_{t}^{(i)}\delta_{ q_t^{(i)}}  $.  The first row of plots in figure \ref{SEIRoverplt} demonstrates that these PALSMC filtering approximations concentrate on the true $q_{1:t}$ as the population size $n$ grows.  This is in keeping with the theory of section \ref{sec:consistency}, which tells us that $\text{argmax}_{q_{1:t}, \bs \vartheta}\ell(y_{1:t} | q_{1:t}, \bs \vartheta)$ converges to the data generating $[q_{1:t} \; \bs \vartheta]$ in the large population limit $n \rightarrow \infty$.  

We also explored the ability of the procedure to recover the data generating hyperparameters $\bs \varphi = [\mu_q\;\sigma_q^2]$; in the bottom plots of figure \ref{SEIRoverplt}. Here each boxplot summarises $100$ estimates, each estimate was obtained as follows:  (1) simulate $q_{1:450}\sim f(\cdot|\bs\varphi)$ and data $\+y_{1:450}$ from the model with population size $n=10^4$, (2) construct a 2-dimensional grid of candidate values for estimation of $[ \mu_q \; \sigma_q^2]$, (3) run PALSMC with input $\+y_{1:450}$ for each grid point, with $n_{part}=10^4$ particles and $\bs \vartheta$ set to the DGP, (4) at time-steps $t=50,100,150,...$ report as an estimate of $[ \mu_q \; \sigma_q^2]$ the value on the grid for which the largest value of $\widehat{\ell}(y_{1:t}| \mu_q,  \sigma_q^2,\bs \vartheta)$ was obtained across the PALSMC runs. We see from these boxplots that, for increasing time horizon $T$, the maximum PALSMC estimators obtained across $100$ simulations converge towards the data generating $\bs \varphi$ with little bias. 

Overall, these simulation results illustrate that, even in an over-dispersed setting, a large population can be useful in estimating $\bar{\bs\theta}_{1:t}$, whilst a large time horizon can be useful in recovering hyperparameters $\bs \varphi$.

\subsection{Comparison of over-dispersion mechanisms in a model of rotavirus}\label{sec:rotavirus}
In this section we explore a model selection task in which an equi-dispersed model is nested within a larger class of models including over-dispersion, using the approach of section \ref{sec:over-dispersion}. The rotavirus data and model we consider are inspired by \cite{stocks2020model},  who assessed the fit of a family of continuous time, stochastic models with varying degrees of over-dispersion using the Akaike Information Criterion (AIC). 
\subsubsection*{Models}
The data considered consist of weekly incidence counts of rotavirus infections in Germany for $3$ age groups over the $8$ year period $2001$-$2008$. We consider a discrete-time version of the model of \cite{stocks2020model} which compartmentalises a population of $n=82,372,825$ into an age stratified SIR model $\{S_{k,t},I_{k,t},R_{k,t}\}_{k=1}^3$ comprising $3$ age groups: $0-4$, $5-59$, and $60-99$. Given the number of susceptibles in age group $k$ at time $t$ after immigration, which we denote $\bar{S}_{k,t}$, and the number of infected individuals in each age group $\bs I_t = [I_{t,1}\;I_{t,2}\;I_{t,3}]^\top$, for $t \geq1$ the number of new infected individuals in each age group $k=1,2,3$ at time step $t$ is conditionally distributed:
\begin{equation} \label{eq:eqrotavirus}
B_{k,t}\sim \text{Binom}\left(\bar S_{k,t},1-\exp\left\{- \frac{\bs{\beta}_k^\top\bs I_{t}}{n}\chi_t\right\} \right),
\end{equation}
where  $\bs \beta_k = {[\beta_k\; \beta_k\; \beta_k]^\top}$ with $ \beta_k >0$  denotes the force of infection experienced by age group $k$, and $\chi_t = {\left(1+\rho \cos(2\pi t /w + \phi)\right)}$ denotes a deterministic seasonality component with amplitude $\rho \in [0,1]$, phase $\phi\in [0,2\pi]$, and period length $w>0$, which we set to correspond to $1$ year. Other details of the latent compartmental model are given in section \ref{sup:rotavirus_supp} of the supplementary materials. We assume an aggregated transmission model, with weekly observations coming at times $\tau_r=4r$ for $r=1,\dots,R$. For each age group observations are conditionally distributed $Y_{k,r}\sim \text{Binom}\left(\sum_{t=\tau_{r-1}}^{\tau_r} B_{k,t}, q_{k,r}\right)$. 

We consider three variants of this model: 
\setlist[description]{font=\normalfont }
\begin{description}[itemsep=-3pt,topsep=2pt]
\item[EqEq:] a fully equi-dispersed model, in which  $q_{k,r} = \mu_{q} \in [0,1]$, and $\mu_{q}$ is assumed known as in \cite{stocks2020model};
\item[EqOv:]  an equi-dispersed latent compartmental model and an over-dispersed observation model, the same as EqEq except that $q_{k,r} \stackrel{\text{iid}}{\sim} \mathcal{N}(\mu_{q}, \sigma^2_q)_{\geq 0, \leq 1}$ where $\sigma^2_q >0$ is to be estimated;
\item[OvOv:]  over-dispersion in both the latent and observation models, the same as EqOv except that  we augment $\chi_t$ in equation \eqref{eq:eqrotavirus}  to $\chi_t\xi_r$
where for $r \geq1$, $\xi_{r} \stackrel{\text{i.i.d}}{\sim} \text{Gamma}(\sigma_\xi,\sigma_\xi)$ are multiplicative disturbances with mean $1$ and $\sigma_\xi>0$ is to be estimated.
\end{description}
\subsubsection*{Inference}
The parameters we estimate in each instance of the model are given by: EqEq: $\bs \theta = [\beta_1 \; \beta_2\; \beta_3 \; \phi \; \rho] $; EqOv: $\bs \vartheta = [\beta_1 \; \beta_2\; \beta_3 \; \phi \; \rho]$ with $\{ \bar{\bs \theta}_r \}_{r\geq 1} = \{\+ q_r\}_{r \geq 1}$ and $\bs \varphi = \sigma_q$; OvOv: $\bs \vartheta = [\beta_1 \; \beta_2\; \beta_3 \; \phi \; \rho]$ with $\{ \bar{{\bs \theta}}_r \}_{r\geq 1} = \{ [ \+ q_r \; \xi_r]\}_{r \geq 1}$ and $\bs \varphi = [\sigma_q \; \sigma_\xi]$. The PALSMC algorithm for this model is given section \ref{sup:rotavirus_supp} in the supplementary material. For parameter estimation the approximate likelihoods of each of the models EqEq, EqOv, OvOv, obtained from PALSMC were maximised using a finite-difference coordinate ascent algorithm; we ran the optimisation $100$ times, initialised randomly over a range of feasible values. Plots evidencing convergence are in section \ref{sup:rotavirus_supp} of the supplementary materials. The algorithm was implemented using R and Rcpp on a node of the University of Bristol’s BluePebble cluster, although we did exploit parallelization.

We note that a PAL, e.g. the exponential of the r.h.s. of \eqref{eq:pal_case_I}, is a valid likelihood function associated with a product of vector-Poisson distributions whose intensity parameters are defined through the corresponding filtering algorithm, e.g. algorithm \ref{alg:x}. Similarly the output from PALSMC, e.g. \eqref{eq:SMC_filtering_approx} from algorithm \ref{alg:overSMC}, is a Monte Carlo approximation to a valid likelihood for a mixture of products of vector Poisson distributions. This validity justifies the use of AIC for model comparison but with the log-PAL, or the log-output from PALSMC, substituted in place of the usual log-likelihood.

\hspace*{-0.1\parindent}%
\begin{minipage}[!h]{\textwidth}
 \noindent\hspace{-0.7cm}\begin{minipage}{0.58\textwidth}

 \begin{longtable}{ c c c c}
\captionsetup{justification=centering}
 \\

 \hline
 Model & AIC & Ave. comp. time\\
 \hline

 EqEq 	& 98866.65	& 30 sec\\
 EqOv 	    & 15154.75  & 2 hr  \\
 OvOv		& 13778.08    & 3 hr \\
\cite{stocks2020model}	    & 20134.38 & 11 hr  \\

 \hline
  \caption{Rotavirus example. Model assessment \\ and computation time.}\label{rotavirusscores}
 \end{longtable}
    \end{minipage}
\begin{minipage}{0.3\textwidth}
    \centering

 \begin{longtable}{ c c c c}\\
\captionsetup{justification=centering}

 \hline
 Parameter & EqEq & EqOv & OvOv \\
 \hline

 $\beta_{1}$ 	& 12.15	& 12.74	& 11.48 \\
 $\beta_2$ 	    & 0.22  & 0.21  & 0.25  \\
 $\beta_3$ 		& 0.34  & 0.31  &  0.35 \\
 $\phi$ 	    & 0.017 & 0.14  &  0.14 \\
 $\rho$ 		& 0.022 & 0.19  &  0.16 \\
 $\sigma_q^2$ 	& n/a   & 0.042 &  0.021 \\
 $\sigma_\xi$ 	& n/a   & n/a   &  66.89 \\

 \hline
  \caption{Rotavirus example. Parameter estimates. \label{rotavirusest}}
 \end{longtable}

  \end{minipage}
  \end{minipage}

As a benchmark comparison, we fitted an ARMA(2,0,1) model to the log-transformed data, which gives an AIC of 23043 (details are given  in the supplementary materials section \ref{sup:rotavirus_supp}). Table \ref{rotavirusscores} gives the AIC values for each of our models, along with the best AIC value reported by \cite{stocks2020model}, which was for a model with over-dispersion in the transition model in the form of multiplicative gamma distributed noise, and over-dispersion in the observation model through negative-binomial reporting. This model of \cite{stocks2020model} is therefore qualitatively most similar to our model OvOv. The average computation times were calculated over $100$ runs of the coordinate ascent procedure. We find that, whilst we can fit EqEq with high computational efficiency, our two over-dispersed models achieve a substantially better AIC score, indicating a much better fit with increasing over-dispersion. Both EqOv and OvOv outperform \cite{stocks2020model} AIC and computation time, although of course the latter is implementation-dependent. Figure \ref{fig:rotacoverage} demonstrates the increase in goodness of fit that an over-dispersed model provides for the rotavirus data, we see that prediction intervals for OvOv drastically outperform those for EqEq in terms of coverage. 

The estimated values of $\beta_1$ and $\beta_2$ we find for all three models EqEq, EqOv and OvOv (table \ref{rotavirusest}) are quite similar to those reported by \cite{stocks2020model}, but we find a slightly lower value of $\beta_3$. For EqOv  and OvOv we find a similar seasonal amplitude $\rho$ but slightly larger phase $\phi$ than \cite{stocks2020model}. The seasonal $R_0$ ranges for each model are:  EqEq $(0.98,1.027)$ ,EqOv  $(0.83,1.22)$, and OvOv  $(0.82,1.14)$ compared to $(0.855, 1.152)$ obtained by \cite{stocks2020model}. The better fit of EqOv and OvOv compared to \citep{stocks2020model} may thus be attributed to some combination of quite subtle differences in estimates of parameters related to disease transmission, together with the difference between the negative binomial observation model in \citep{stocks2020model} and the way EqOv and OvOv treat the $q_{k,r}$ as latent variables.

\begin{figure} 
\includegraphics[width = \textwidth]{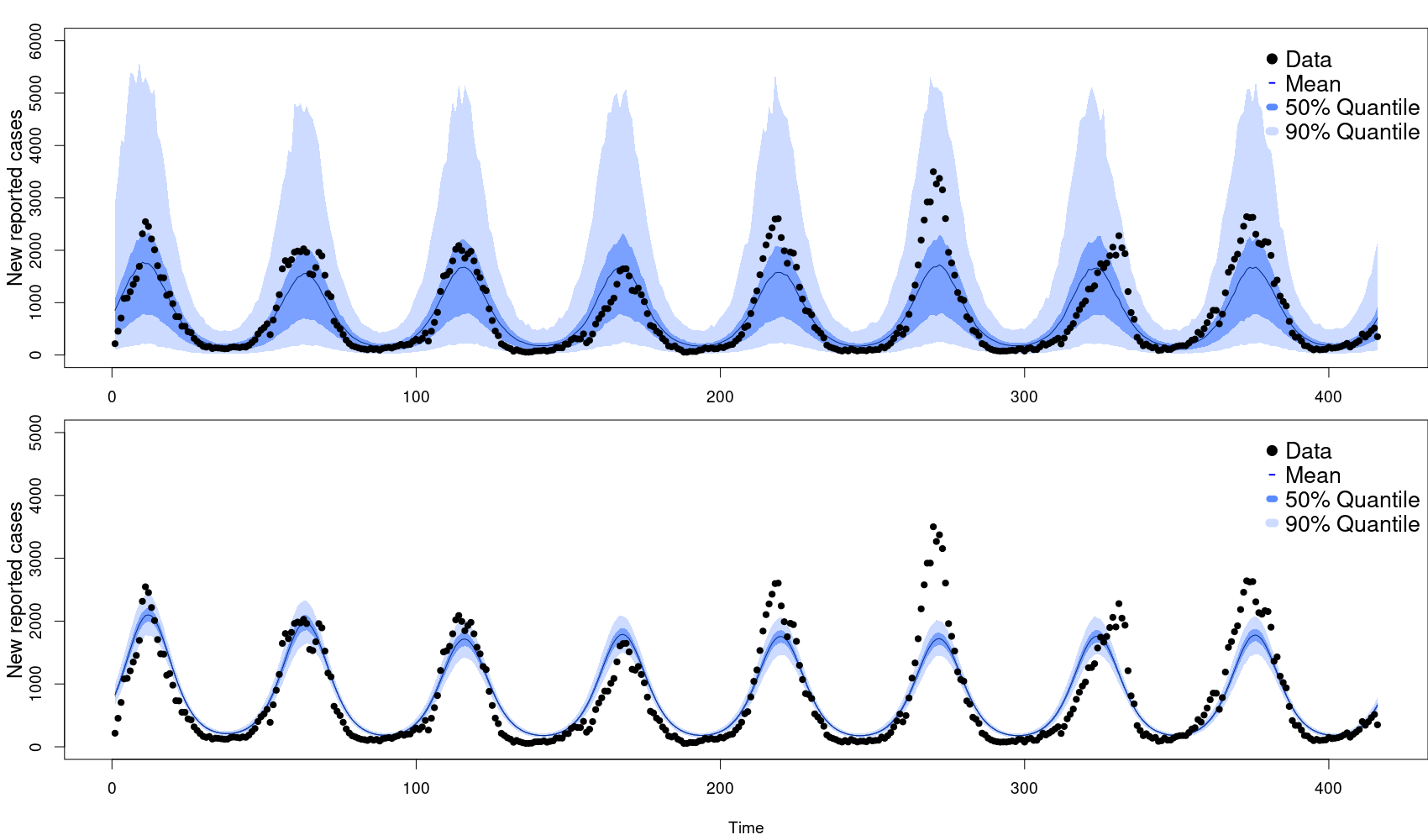}
\caption{Rotavirus example. Prediction intervals for age group $0-4$ corresponding to $1000$ realisations of OvOv (top panel) and EqEq (bottom panel), using maximum PALSMC parameter estimates.}
\label{fig:rotacoverage}
\end{figure}

\subsection{Evaluating the role of unit-specific parameters in a meta-population model of measles}\label{sec:measles}
In this section we illustrate how the PAL framework can be used to calibrate a more complex, larger-scale model, and compare the fit of sub-models with different levels of unit-specific parameters.
\subsubsection*{Model}
We consider a discrete time version of a measles model originally presented by \cite{xia2004measles}, subsequently extended into a spatio-temporal framework by \cite{he2010plug} and recently explored by \cite{park2020inference} using guided intermediate resampling filter (GIRF) techniques. 

The model describes the evolution of recurrent pre-vaccination measles epidemics in $J=40$ cities across the UK over the $15$ year period $1950-1965$. The model has susceptible $(S)$, exposed $(E)$, infective $(I)$, and removed $(R)$ compartments for each of the $J=40$ cities. For each city $k = 1, \dots, J$ the initial state of the epidemic is given by  $[S_{k,0}\;E_{k,0}\;I_{k,0}\;R_{k,0}]^\top \sim \text{Mult}(n_{k,0},\bs \pi_{k,0})$, where the probability vector $\bs \pi_{k,0}$ is a possibly city -specific initial distribution parameter, and $n_{k,t}$ for $t\geq 0$ denotes time varying population size. For each city $k=1, \dots, J$ the population  evolves twice per week with the following dynamic:
\begin{alignat*}{2}
    S_{k,t+1} &= S_{k,t} - B_{k,t}- F^{(S)}_{k,t} + A_{k,t}, \qquad &&E_{k,t+1} = E_{k,t} + B_{k,t} - C_{k,t}- F^{(E)}_{k,t}, \\
    I_{k,t+1} &= I_{k,t} + C_{k,t} - D_{k,t}- F^{(I)}_{k,t}, \qquad &&R_{k,t+1} = R_{k,t} + D_{k,t}- F^{(R)}_{k,t},
\end{alignat*}
where $F^{(\cdot)}_{t,k}$ and $A_{k,t}$ model emigration (deaths) and immigration (births), respectively; and $C_{k,t}$ and $D_{k,t}$ are binomially distributed (details in the supplementary material). The term $B_{k,t}$ represents the number of new infections in the $k$th city and is distributed 
$$B_{k,t} \sim \text{Bin}\left(S_{k,t} - F^{(S)}_{k,t}, 1 - e^{ -hb_{k,r} } \right),$$ 
where:
\begin{equation}\label{eq:betameasles}
b_{k,r} = {\beta_{k,r}\xi_{k,r} \cdot\left[ \left( \frac{I_{k,\tau_r}}{n_{k,\tau_r}} \right)+ \sum_{l \neq k} \frac{v_{k,l}}{n_{k,\tau_r}}\left\{ \left(\frac{I_{l,\tau_r}}{n_{l,\tau_r}} \right) -  \left(\frac{I_{k,\tau_r}}{n_{k,\tau_r}} \right) \right\} \right]},
\end{equation}
for $r\geq 1$, $t = \tau_r, \dots, \tau_{r+1}-1$. Here $\beta_{k,r} $ denotes a possibly city-specific seasonal transmission coefficient and $\xi_{k,r} \stackrel{\text{iid}}{\sim} \text{Gamma}(\sigma_\xi,\sigma_\xi)$, for $\sigma_\xi>0$, is mean-$1$ multiplicative noise which achieves over-dispersion in the marginal distribution of $B_{k,t}$. 

The summation term in \eqref{eq:betameasles} encodes the intercity interaction under a `gravity model' -- see \citep{truscott2012evaluating} for background on these kind of models in epidemiology. The strength of the interaction $v_{k,l}$ is computed as:
\begin{equation*}
v_{k,l} = g  \frac{\bar s}{\bar n} \frac{n_{k,0}n_{l,0}}{s_{k,l}},
\end{equation*}
where $g$ is called the `gravitational' constant parameter, $\bar n$ is the average of the initial populations, $\bar s$ is the average inter-city distance and $s_{k,l}$ denotes the distance between cities $k$ and $l$. The interpretation of the gravity model is thus that the strength of the interaction between two cities is directly proportional to their populations and inversely proportional to their distance. 

The observations are aggregated incidence data in the form of cumulative fortnightly transitions from infective to recovered for each of the $40$ cities, at times $\tau_r=4r$ for $r = 1,\dots, R$. Our observation model, which allows for over-dispersion, is described in section \ref{sec:measles_supp} of the supplementary material, along with the distributions of $C_{k,t}$, $D_{k,t}$, $F^{(\cdot)}_{k,t}$, and $A_{k,t}$, and an explanation of how we write the model as an instance of the Latent Compartmental model with $h=3.5$ days, corresponding to bi-weekly transitions. 

We consider three variants of this model all with over-dispersion in both the dynamics and observation mechanisms, but with increasing levels of city-specific parameters:
\setlist[description]{font=\normalfont }
\begin{description}[itemsep=-3pt,topsep=2pt]
\item[A:] the initial distribution vectors $\bs \pi_{k,0}$ and force of infection parameters $\beta_{k,r}$ are shared across cities, i.e. constant in $k$;
\item[B:] $\bs \pi_{k,0}$ is city-specific and $\beta_{k,r}$ is shared across cities;
\item[C:] $\bs \pi_{k,0}$ and $\beta_{k,r}$ are city-specific.
\end{description}
Here we are inspired by an investigation conducted by \cite{ionides2022iterated}, where sub-models with  increasing numbers of city-specific parameters were fitted to a dataset on a smaller spatial scale, comprising $20$ cities compared to the $40$ we consider. \cite{ionides2022iterated} suggested that approximation techniques may be needed to analyse larger data sets, our application of the PAL framework is a step in that direction. However we note that the $20$-city dataset analysed by \cite{ionides2022iterated} is not a subset of the $40$-city dataset we consider here, so direct comparisons of model fit may not be made. Never-the-less we shall compare our results to those obtained by \citep{park2020inference} for a model in which parameters are shared across cities, fitted to the same $40$-city dataset we consider.

\subsubsection*{Inference}
In section \ref{sec:measles_supp} of the supplementary material we give the details of a PALSMC algorithm in which the PAL is embedded within a block particle filter \citep{rebeschini2015can,ionides2022iterated}, to numerically approximate the log-likelihood. We used data-informed proposals and lookahead resampling to improve efficiency.  For each of the models A,B,C, the approximate log-likelihood obtained from this PALSMC algorithm with $5000$ particles was maximized with respect to the model parameters through Sequential Least Squares Programming. The procedures were implemented using Python and TensorFlow on a 32gb Tesla V100 GPU available on the HEC (High-End Computing) facility from Lancaster University.

Table \ref{measlesresults} details PALSMC approximate log-likelihood and AIC values for each of the models A,B,C, along with an approximate log-likelihood reported by \cite{park2020inference} for comparison. The GIRF used by \cite{park2020inference} consists of a simulator for a continuous in time latent process combined with a particle filter  which uses guide functions for intermediate propagation and resampling, parameters of the model are estimated via an iterated filtering scheme. Together with Monte Carlo adjusted profile methodology \citep{ionides2017monte} they are able to generate profile likelihood estimates for confidence interval estimation. Frequentist uncertainty interval calculation is out of the scope of the current work and would require results on the asymptotic distribution of the maximum PAL estimator, see section \ref{sec:future} for a discussion.

Our model A is similar to that of \cite{park2020inference} in the sense that both these models have parameters shared across cities, but we find model A performs better  in terms of log-likelihood and AIC. As we move to from model A to models B and C, by making more parameters city-specific, we see an improvement in  log-likelihood and AIC. We also note that the computation time for fitting model A is orders of magnitude smaller than that of \cite{park2020inference}. The computation time is of course implementation-dependent, but we note that we have not devised a bespoke optimization algorithm to maximize the PALSMC approximation, but rather applied a standard ‘black-box’ optimizer. As prompted by an anonymous reviewer, we fitted an ARMA(2,0,1) model to the log-transformed data for a benchmark comparison; this gave a log-likelihood of -69168 (details are given  in the supplementary materials section \ref{sec:measles_supp}).

Estimates of the city-specific parameters $\beta_{k,r}$ in model C can be used to estimate city-specific $R_0$ values, calculated as in \cite{ionides2022iterated}. We find that across the $40$ cities these estimated $R_0$ values lie in the range $5.63-16.65$. The fitted mean latent and infective periods for model $C$ were $8.49$ and $9.53$  respectively; these values are in line with previous inferences on the behaviour of measles epidemics \citep{guerra2017basic}, \citep{delamater2019complexity}. Full details of our numerical results are in the supplementary materials section \ref{sec:measles_supp}.
\vspace{-0.5cm}
\begin{figure}[h!]
 \begin{longtable}{ c c c c c}
\label{measlesresults}\\
 \hline
 Model & No. parameters& Log-likelihood\;(sd) & AIC& Comp. time \\
 \hline

 A & $11$	& $-63579\;(62)$  & $127180$ & 45 min\\
 B &	 $128$   &  $-61257\;(28)$  & $122770$ & 10 hr  \\
 C &	$167$	&   $-61169\;(34)$ & $122672$& 24 hr  \\
\cite{park2020inference}	& $12$   & $-70000^\dagger$  & $140024^\dagger$ & 30 hr$^*$\\

 \hline
 \caption{Measles example. Mean log-likelihood values for models A, B, and C, with Monte Carlo standard deviation (sd) over $100$ runs of PALSMC with $5000$ particles. `No. parameters' is the number of parameters estimated by maximising the log-likelihood for each model. $\dagger$Approximate values read from figure 3 in \citep{park2020inference}. $^*$We note that the $30$hr reported by  \cite{park2020inference} includes confidence interval calculation via Monte Carlo adjusted profile methodology.}

\end{longtable}
\end{figure}

Figure \ref{measlesmap} shows projected case numbers for the $4$ fortnights following the end of the data record, obtained using model C with  parameters fixed to the estimated values, full details are in the supplementary materials section \ref{sec:measles_supp}. We see a general increase in forecast uncertainty as the time horizon increases, this reflecting the over-dispersed nature of model C. We also see that the forecasts generally exhibit higher certainty for cities with a larger population, as might be expected if a larger sub-population size allows   latent variables and parameters which are specific to that sub-population to be estimated more accurately. 
\begin{figure}[h!]
    \centering
    \includegraphics[width=\textwidth]{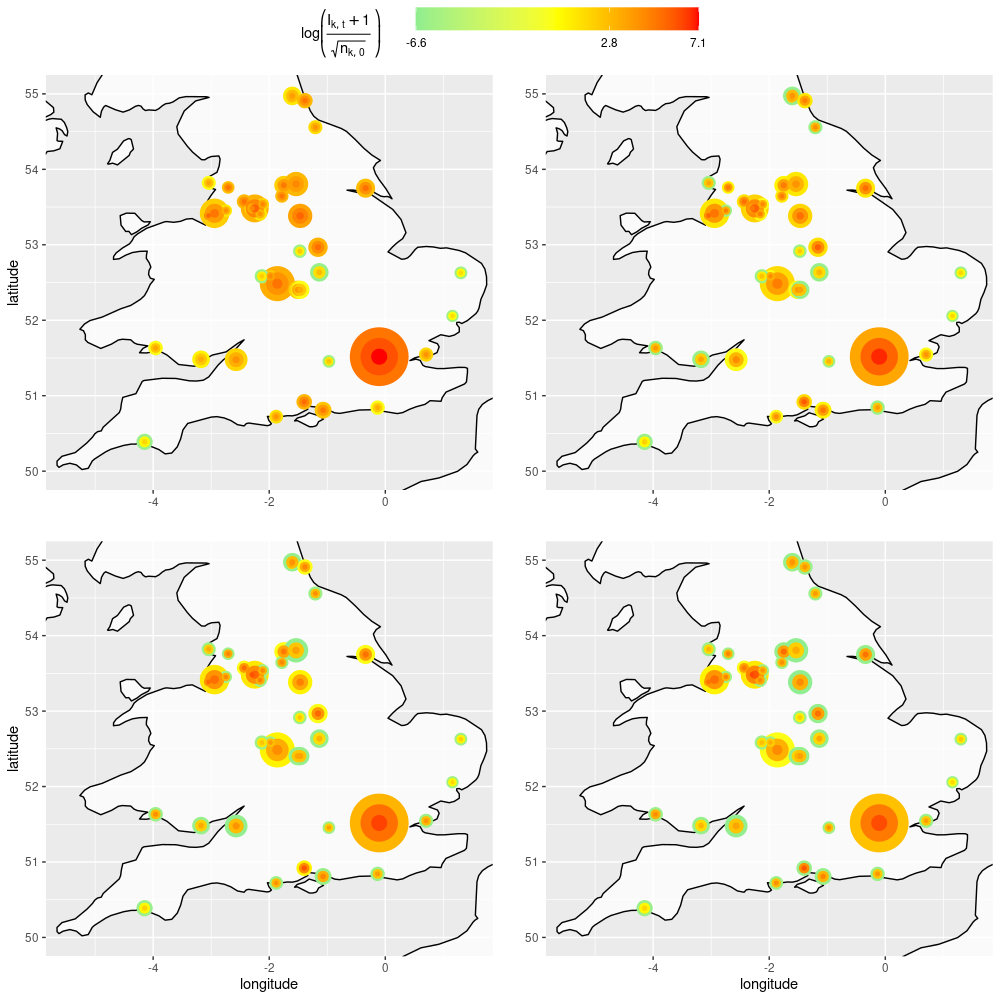}
    \caption{Measles example. Projected case numbers for the $4$ fortnights (ordered top-left, top-right, bottom-left, bottom-right) following the end of the data record. For each town/city, the diameter of the outer-most concentric ring represents log-population size. The shade of the outer concentric ring corresponds to the lower $5\%$ quantile of the simulated case numbers, the shade of the middle concentric ring to the mean, and the inner concentric ring to the upper $95\%$ quantile.  }
    \label{measlesmap}
\end{figure}

\subsection{Opportunities for further research}\label{sec:future}

In the examples from section \ref{sec:rotavirus} and \ref{sec:measles} we have not devised special optimisation techniques to estimate parameters, but rather just plugged PALSMC likelihood function evaluations into `black-box' optimisers. There may be opportunities here to even further increase computational efficiency, for example by embedding PALSMC within an iterated filtering scheme \citep{ionides2011iterated}.

Recently \cite{ju2021sequential} devised sophisticated SMC algorithms to fit agent-based models in which individuals in the population each carry covariates influencing, for example, the probabilities that they come into contact, and hence the probabilities of disease spreading from one individual to the next. When these covariates are discrete and take only finitely many distinct values, or can be discretised into that form, for example the subdivision of the population into age groups as in the age-structured example from section \ref{sec:rotavirus},   they can be handled in the latent compartmental modelling framework by introducing extra compartments and specifying an appropriate observation model. However, covariates taking infinitely many distinct values cannot be handled this way, or necessitate further approximations. \cite{rimella2023optimal}  have suggested methods related to PALs to construct efficient proposal distributions for SMC in individual-based models. Further research may expand the applicability of PAL-like approximations in this direction.

\section{Acknowledgements}
Michael Whitehouse is supported by a studentship from Compass – the EPSRC Centre for Doctoral
Training in Computational Statistics and Data Science. Lorenzo Rimella is supported by EP-
SRC Grant EP/R018561/1 (Bayes for Health). The authors are grateful to Nikolas Kantas for pointing them to the measles model, Patrick Cannon and colleagues at Improbable for discussion of automatic differentiation and David Greenwood for discussions of agent-based models.
\newpage
\bibliographystyle{plain}
\bibliography{consistency_paper}
\newpage
\appendix

\section{Proofs and supporting results for section \ref{sec:filtering}}\label{sec:proofs_filtering}

\begin{proof}[Proof of Lemma \ref{lem:xpred}]
For the first result, consider the probability mass function of ${\+x}$:
\begin{equation*}
p({\+x}) = \prod_{j=1}^m \frac{e^{-\lambda^{(j)}} (\lambda^{(j)})^{(x^{(j)})}}{x^{(j)}!},
\end{equation*}
and for $0 \leq \bar x^{(j)} \leq x^{(j)}$ $j = 1, \dots m$,
\begin{equation*}
p(\bar{\+x}\mid \+x) = \prod_{j=1}^m \frac{x^{(j)}!}{\bar{x}^{(j)}!(x^{(j)} - \bar{x}^{(j)})!}(\delta^{(j)})^{\bar{x}^{(j)}}(1 - \delta^{(j)})^{x^{(j)} - \bar{x}^{(j)}},
\end{equation*}
So that
\begin{equation*}
p(\+x, \bar{\+x}) = \prod_{j=1}^m \frac{e^{-\lambda^{(j)}}(\lambda^{(j)})^{x^{(j)}}(\delta^{(j)})^{\bar{x}^{(j)}}(1 - \delta^{(j)})^{x^{(j)} - \bar{x}^{(j)}}}{\bar{x}^{(j)}!(x^{(j)} - \bar{x}^{(j)})!},
\end{equation*}
and
\begin{align*}
p(\bar{\+x}) &= \sum_{x^{(i)} \geq \bar{x}^{(i)};i\in[m]} \prod_{j=1}^m  \frac{e^{-\lambda^{(j)}}(\lambda^{(j)})^{x^{(j)}}(\delta^{(j)})^{\bar{x}^{(j)}}(1 - \delta^{(j)})^{x^{(j)} - \bar{x}^{(j)}}}{\bar{x}^{(j)}!(x^{(j)} - \bar{x}^{(j)})!} \\
&=\left( \prod_{j=1}^m \frac{e^{-\lambda^{(j)}}(\delta^{(j)}\lambda^{(j)})^{\bar{x}^{(j)}}}{\bar{x}^{(j)}!} \right) \sum_{x^{(i)} \geq \bar{x}^{(i)};i\in[m]} \prod_{j=1}^m\frac{(\lambda^{(j)})^{(x^{(j)} - \bar x^{(j)})}(1 - \delta^{(j)})^{(x^{(j)} - \bar x^{(j)})}}{{(x^{(j)} - \bar{x}^{(j)})!}} \\
& = \left( \prod_{j=1}^m \frac{e^{-\lambda^{(j)}}(\delta^{(j)}\lambda^{(j)})^{\bar{x}^{(j)}}}{\bar{x}^{(j)}!} \right) e^{\lambda^{(j)}(1 - \delta^{(j)})} \\
&= \prod_{j=1}^m \frac{e^{-\lambda^{(j)}\delta^{(j)}}(\delta^{(j)}\lambda^{(j)})^{\bar{x}^{(j)}}}{\bar{x}^{(j)}!},
\end{align*}
which is the probability mass function associated with $\mathrm{Pois}(\bs \lambda \circ \bs \delta)$.
\par
Now consider $\+x' \sim {M}_t(\bar{\+x}, \bs \eta(\bs \lambda \circ \bs \delta), \cdot)$ where $\bar{\+ x} \sim \mu$, so that  $\sum_{\bar{\+x}\in \mathbb{N}_{0}^m} \mu(\bar{\+x})M_t(\bar{\+x}, \bs \eta (\bs \lambda \circ \bs \delta), \cdot )$ is the marginal probability mass function of $\+x'$. By the definition of $M_t$, $\+x' = (\+1_m^\top\+Z)$, where the rows of $\+Z$ are conditionally independent given $\bar{\+x}$, and the $i$th row of $\+Z$ is distributed $\mathrm{Mult}(\bar{x}^{(i)}, \+ K_{t,\bs \eta (\bs \lambda \circ \bs \delta)}^{(i, \cdot)})$. Now we can write the moment generating function (m.g.f.) of $\+ x'$ as:
\begin{equation*}
\begin{aligned}\mathcal{M}_{\+x'}(\+b) &= \E\left[\exp\left(\+1_m^\top \+ Z^\top\+ b \right) \right] \\
 &= \E \left[ \exp\left(\sum^m_{i,j=1} Z^{(i,j)}b^{(j)}\right) \right]\\
 & = \E \left[\prod_{j=1}^m \exp\left(\sum^m_{i=1} Z^{(i,j)}b^{(j)}\right) \right]\\
 &= \E \left\{ \prod_{i=1}^m  \E \left[ \exp\left(\sum^m_{j=1} Z^{(i,j)}b^{(j)}\right) \Bigg| \bar{\+ x}  \right] \right\}.
 \end{aligned}
\end{equation*}
Now we notice that $ \E \left[  \exp\left(\sum^m_{i=1} Z^{(i,j)}b^{(i)}\right) \Bigg| \bar{\+ x}  \right]$ is the m.g.f. of $\mathrm{Mult}(\bar{x}^{(i)}, \+ K_{t,\bs \eta (\bs \lambda \circ \bs \delta)}^{(i, \cdot)})$ so that
\begin{equation*}
\begin{aligned}
\mathcal{M}_{\+x'}(\+b) &= \E \left\{\prod_{i=1}^m  \left[ \sum_{j=1}^m  K_{t,\bs \eta (\bs \lambda \circ \bs \delta)}^{(i, j)}e^{b^{(i)}} \right]^{\bar{ x} ^{(j)}} \right\} \\
 &= \sum_{\bar{ x} ^{(1)}, \dots, \bar{ x} ^{(m)} \in \mathbb{N}_0^m} \prod_{i=1}^m  \left[ \sum_{j=1}^m  K_{t,\bs \eta (\bs \lambda \circ \bs \delta)}^{(i,j)}e^{b^{(j)}} \right]^{\bar{x} ^{(i)}} \frac{ e^{- \lambda^{(i)} \delta^{(i)}}( \lambda^{(i)} \delta ^{(i)})^{\bar{ x} ^{(i)}}}{\bar{ x} ^{(i)}!} \\
 &= \left(\prod_{i=1}^m e^{-\lambda^{(i)}\delta^{(i)}}\right)\sum_{\bar{ x} ^{(1)}, \dots, \bar{ x} ^{(m)} \in \mathbb{N}_0^m}\prod_{i=1}^m\frac{1}{\bar{x}^{(i)}!}\left[ \sum_{j=1}^m  \lambda^{(i)} \delta^{(i)} K_{t,\bs \eta (\bs \lambda \circ \bs \delta)}^{(i,j)}e^{b^{(j)}} \right]^{\bar{x} ^{(i)}} \\Y
 &=\prod_{i=1}^m\exp\left(- \lambda^{(i)}\delta^{(i)} + \sum_{j=1}^m \lambda^{(i)} \delta^{(i)} K_{t,\bs \eta (\bs \lambda \circ \bs \delta)}^{(i,j)}e^{b^{(j)}}\right) \\
 &= \exp\left\{ \sum_{i=1}^m\left( -\lambda^{(i)} \delta^{i)} \sum_{j=1}^m K_{t,\bs \eta (\bs \lambda \circ \bs \delta)}^{(i,j)} + \sum_{j=1}^m \lambda^{(i)}\delta^{(i)} K_{t,\bs \eta (\bs \lambda \circ \bs \delta)}^{(i,j)} e^{b^{(i)}}\right)\right\} \\
 &= \prod_{j=1}^m \exp\left\{\left( (\bs \lambda \circ \bs \delta)^\top \+ K^{(\cdot, j)}_{t, \bs \eta(\bs \lambda \circ \bs \delta)}\right)\left(e^{b^{(j)}}-1\right)\right\}.
\end{aligned}
\end{equation*}
We recognise this is the moment generating function of a $\mathrm{Pois}\left((\bs \lambda \circ \bs \delta)^\top \+ K_{t, \bs \eta(\bs \lambda \circ \bs \delta)} \right)$ random vector.
\end{proof}

\begin{proof}[Proof of Lemma \ref{lem:xupdate}]
We have $\bar{\+ y} \sim \mathrm{Pois}(\bs \lambda \circ \+q)$ by the same reasoning as lemma \ref{lem:xpred}. By definition $\hat{\+y} =\+1_m^\top \+M $, hence the moment generating function of $\hat{\+y}$ is:
\begin{equation*}
\begin{aligned}
\mathcal{M}_{\hat{\+y}}(\+b) &= \E\left[\exp\left(\+1_m^\top \+ M^\top\+ b \right) \right] \\
 &= \E \left\{ \prod_{i=1}^m  \E \left[ \exp\left(\sum^m_{j=1} M^{(i,j)}b^{(j)}\right) \Bigg| \bar{\+ y}  \right] \right\} \\
&= \E \left\{\prod_{i=1}^m  \left[ \sum_{j=1}^m  G^{(i, j)}e^{b^{(i)}} \right]^{\bar{ y} ^{(j)}} \right\} \\
 &= \sum_{\bar{ x} ^{(1)}, \dots, \bar{ x} ^{(m)} \in \mathbb{N}_0^m} \prod_{i=1}^m  \left[ \sum_{j=1}^m  G^{(i,j)}e^{b^{(j)}} \right]^{\bar{x} ^{(i)}} \frac{ e^{- \lambda^{(i)} q^{(i)}}( \lambda^{(i)} q ^{(i)})^{\bar{ x} ^{(i)}}}{\bar{ x} ^{(i)}!} \\
 &= \prod_{j=1}^m \exp\left\{\left( (\bs \lambda \circ \bs q)^\top \+ G^{(\cdot, j)}\right)\left(e^{b^{(j)}}-1\right)\right\}.
\end{aligned}
\end{equation*}
Which we recognise as the moment generating function of the $\mathrm{Pois}((\bs \lambda \circ \+ q)^\top \+G)$, the first result of the lemma then follows from applying element-wise the fact that the intensity of the sum of two independent Poisson random variables is the sum of the intensities.

We start the proof of \eqref{eq:condexp} by considering the decomposition of $\+x$ into the sum of random variables $\bar{\+ y}$ and $\breve{\+x}$ where $\breve{\+x} = \+x - \bar{\+y}$. Then, $\bar{\+ y}$ and $ \breve{ \+ x}$ are independent Poisson with intensity vectors $\+ q \circ \bs\lambda$ and $( \+ 1_m - \+ q) \circ \bs \lambda$ respectively, see \cite{kingman1992poisson}[Sec. 1.2]. Since $\breve{\+x}$ is independent of $\+y$, we have that:
\begin{equation}\label{exgivy}
\E\left[\+x \mid \+y \right] = \left[\+ 1_m - \+q\right]\circ \bs \lambda + \E\left[\bar{\+y} \mid \+y \right].
\end{equation}
So, we need to characterise the distribution of $\bar{\+y}$ given $\+y$. Construct the random variable $\+ \Xi \in \mathbb{N}_0^{(m+1) \times m}$ such that for $i,j\in[m]$, $\Xi^{(i,j)} = M^{(i,j)}$ and row $m+1$ of $\+ \Xi$ are the counts $\hat{\+ y}\sim \mathrm{Pois}(\bs \kappa)$. By this construction, $\sum_{j=1}^m \Xi^{(i,j)} = \bar{y}^{(i)}$ for $i= 1, \dots, m$ and $\sum_{i=1}^{m+1} \Xi^{(i,j)} = y^{(j)}$ for $j=1 ,\dots, m$. Furthermore, the elements of $\+ \Xi$ are independently Poisson, see \cite{kingman1992poisson}[Sec. 1.2], with intensity matrix $\+\Lambda \in \mathbb{R}^{(m+1) \times m}$ defined as follows:
\begin{align*}
\Lambda^{(i,j)} &=  \lambda^{(i)} q^{(i)} G^{(i,j)} \quad \mathrm{for } i = 1, \dots m, \quad j = 1, \dots m \\
\Lambda^{(m+1,j)} &=  \kappa^{(j)} \quad \mathrm{for } j = 1, \dots m.
\end{align*}
If, for some $j,k \in[m]$, $\sum_{i=1}^{m+1}  \Lambda^{(i,j)} = 0$, then we must have that $ \Lambda^{(i,j)} = 0$ for all $i = 1, \dots, m+1$ so that $\Xi^{(i,j)} = 0 \quad a.s.$. Otherwise we have that for $i = 1, \dots, m+1$ and $j\in[m]$, $\Xi^{(i,j)}$ conditioned on $\sum_{k=1}^{m+1} \Xi^{(k,j)} = y^{(j)}$ is distributed \begin{equation*}\mathrm{Bin}\left( y^{(j)}, \frac{ \Lambda^{(i,j)}}{\sum_{k=1}^{m+1}  \Lambda^{(k,j)}}\right).\end{equation*}
Hence, given $\+y$, $\bar y^{(i)}$ has a Poisson-Binomial distribution with mean:
\begin{equation}\label{eq:x_update_proof_y}
\E\left[\bar{y}^{(i)} \mid \+y \right] = \E\left[\sum_{j=1}^m \Xi^{(i,j)} \mid \+y \right] = \sum_{j=1}^m y^{(j)}\frac{ \lambda^{(i)} q^{(i)} G^{(i,j)}}{\sum_{k=1}^{m+1} \lambda^{(k)} q^{(k)} G^{(k,j)} + \kappa^{(j)}},
\end{equation}
for $i = 1, \dots, m$, where we set the $j$th term of the outer sum on the r.h.s to $0$ if ${\sum_{k=1}^{m+1} \lambda^{(k)} q^{(k)} G^{(k,j)} + \kappa^{(j)}} = 0$ since that achieves \begin{equation*}
    \E\left[\Xi^{(i,j)}\mid \+y\right] = 0.
\end{equation*}
Writing \eqref{eq:x_update_proof_y} in vector form and substituting into \eqref{exgivy} completes the proof.

\end{proof}

\begin{proof}[Proof of Lemma \ref{lem:Zpred}]
Note $\bs \eta(\E_{\bar \mu}[\+1_m^\top\+Z]) = \bs \eta(\+ 1_m^\top \bs \Lambda) = \bs \eta(\bs \lambda^\top)$. Let $\tilde{\+ Z} \sim \bar M_t(\+Z, \bs \eta(\bs \lambda), \cdot)$, then the moment generating function for $\tilde{\+Z}$ is:

\begin{equation*}
\begin{aligned}
\E \left[ \exp(\+{1}_m^\top (\tilde{\+Z} \circ \+{B}) \+1_m) \right] &= \E \left[ \prod_{i=1}^m \exp \left(\sum_{j=1}^m \tilde{Z}^{(i,j)}b^{(i,j)}\right) \right] \\
&= \E \left[ \E \left[ \prod_{i=1}^m \exp \left(\sum_{j=1}^m \tilde{Z}^{(i,j)}b^{(i,j)}\right)  \bigg| \+ Z \right] \right] \\
&= \E \left[ \prod_{i=1}^m \E \left[  \exp \left(\sum_{j=1}^m \tilde{Z}^{(i,j)}b^{(i,j)}\right)  \bigg| x^{(i)} \right] \right] \\
&= \E \left[ \prod_{i=1}^m  \left(\sum_{j=1}^m K_{t, \bs \eta(\bs \lambda)}^{(i,j)}e^{b^{(i,j)}}\right)^{x^{(i)}}\right] \\
&=\left(\prod_{i=1}^m e^{-\lambda^{(i)}} \right) \sum_{(x^{(1)}, \dots, x^{(m)}) \in \mathbb{N}_0} \left(\sum_{j=1}^m K_{t, \bs \eta(\bs \lambda)}^{(i,j)}e^{b^{(i,j)}}\lambda^{(i)}\right)^{x^{(i)}} \frac{1}{x^{(i)}!}  \\
&= \prod_{i=1}^m\exp\{- \lambda^{(i)}\sum_{j=1}^m K_{t, \bs \eta(\bs \lambda)}^{(i,j)}(1- e^{b^{(i,j)}})\} \\
&= \prod_{i,j=1}^m\exp\{- \lambda^{(i)}K_{t, \bs \eta(\bs \lambda)}^{(i,j)}(1- e^{b^{(i,j)}})\},
\end{aligned}
\end{equation*}
which we recognise as the moment generating function of a $\text{Pois}((\bs \lambda \otimes \+1_m)\circ \+K_{t,\bs \eta(\bs \lambda)})$ random matrix.
\end{proof}

\begin{proof}[Proof of Lemma \ref{lem:Zupdate}]
We have:
\begin{equation*}
p(\+ Z) = \prod_{i,j=1}^m \frac{e^{- \Lambda^{(i,j)}}( \Lambda^{(i,j)})^{Z^{(i,j)}}}{Z^{(i,j)}!},
\end{equation*}
furthermore:

\begin{equation*}
p( \+Y | \+ Z) = \prod_{i,j=1}^m \frac{Z^{(i,j)}!}{Y^{(i,j)}!(Z^{(i,j)}-Y^{(i,j)})!}{Q^{(i,j)}}^{Y^{(i,j)}}(1-Q^{(i,j)})^{Z^{(i,j)}-Y^{(i,j)}}.
\end{equation*}
So that:

\begin{equation*}
p(\+ Z, \+Y) = \prod_{i,j=1}^m \frac{{Q^{(i,j)}}^{Y^{(i,j)}}(1-Q^{(i,j)})^{Z^{(i,j)}-Y^{(i,j)}}e^{-\Lambda^{(i,j)}}( \Lambda^{(i,j)})^{Z^{(i,j)}}}{Y^{(i,j)}!(Z^{(i,j)}-Y^{(i,j)})!},
\end{equation*}
and

\begin{equation*}
\begin{aligned}
p(\+Y) &= \sum_{\{Z^{(i,j)} : Z^{(i,j)} \geq Y^{(i,j)}\}} \prod_{i,j=1}^m\frac{{Q^{(i,j)}}^{Y^{(i,j)}}(1-Q^{(i,j)})^{Z^{(i,j)}-Y^{(i,j)}}e^{-\Lambda^{(i,j)}}( \Lambda^{(i,j)})^{Z^{(i,j)}}}{Y^{(i,j)}!(Z^{(i,j)}-Y^{(i,j)})!} \\
&= \prod_{i,j=1}^m\frac{e^{-\Lambda^{(i,j)}}({Q^{(i,j)}  \Lambda^{(i,j)}})^{Y^{(i,j)}}}{Y^{(i,j)}!}   \sum_{Z^{(i,j)} -Y^{(i,j)} \geq 0}\frac{( \Lambda^{(i,j)}(1-Q^{(i,j)}))^{Z^{(i,j)}-Y^{(i,j)}}}{(Z^{(i,j)} - Y^{(i,j)})!} \\
&= \prod_{i,j=1}^m\frac{e^{- \Lambda^{(i,j)}}({Q^{(i,j)}  \Lambda^{(i,j)}})^{Y^{(i,j)}}}{Y^{(i,j)}!} e^{ \Lambda^{(i,j)}(1-Q^{(i,j)})}\\
&= \prod_{i,j=1}^m\frac{e^{- \Lambda^{(i,j)}Q^{(i,j)}}({Q^{(i,j)}  \Lambda^{(i,j)}})^{Y^{(i,j)}}}{Y^{(i,j)}!} .
\end{aligned}
\end{equation*}
Dividing $p(\+Z, \+ Y)$ by $p(\+Y)$ gives:

\begin{equation*}
p(\+ Z \mid \+ Y) = \prod_{i,j = 1}^m \frac{e^{-\Lambda^{(i,j)}(1 - Q^{(i,j)})}}{(Z^{(i,j)}- Y^{(i,j)})!}(\Lambda^{(i,j)}(1 - Q^{(i,j)}))^{Z^{(i,j)}- Y^{(i,j)}}.
\end{equation*}
Giving the desired probability mass function of $\+Y + \+Z^*$.
\end{proof}

\begin{proof}[Proof of Lemma \ref{lem:Zaggup}]
By lemma \ref{lem:Zupdate} we have that for each $t=1,\ldots,\tau$, $\+Y_t \sim \text{Pois}(\bs \Lambda_t \circ \+Q_t)$ and:
\begin{equation*}
\E\left[ \+ Z_\tau \mid \+ Y_\tau\right] = (\+1_m \otimes \+1_m - \+Q_\tau)\circ \bs \Lambda_{\tau} + \+Y_\tau.
\end{equation*}
Since $\bar{Y}^{(i,j)}$ is the sum of independent Poisson random variables $Y_t^{(i,j)}$,  we have $\+Y\sim \mathrm{Pois}(\sum_{t=1}^\tau \bs \Lambda_t\circ \+ Q_t  )$ and given $\bar{Y}^{(i,j)}$, $Y_\tau^{(i,j)}$ is distributed  $\mathrm{Bin}(\bar{Y}^{(i,j)}, \Lambda_\tau^{(i,j)}Q^{(i,j)}_t/ \sum_{t=1}^\tau \Lambda_t^{(i,j)}Q^{(i,j)}_t)$. Hence by the tower law:
\begin{align*}
 \E\left[ \+ Z_\tau \mid \+ Y\right] &=  \E\left[ \E\left[ \+ Z_\tau \mid \+Y_\tau \right] \mid \+ Y\right] \\
 &=  (\+1_m \otimes \+1_m - \+Q_t)\circ \bs \Lambda_{\tau} +  \E\left[ \+ Y_\tau \mid \+ Y\right] \\
 &=  (\+1_m \otimes \+1_m - \+Q_t)\circ \bs \Lambda_{\tau} + \bar{\+Y} \circ \bs \Lambda_\tau\circ \+ Q_\tau  \oslash \left( \sum_{t=1}^\tau \bs \Lambda_t\circ \+ Q_t  \right),
\end{align*}
in the case that all elements of $\sum_{t=1}^\tau \bs \Lambda_t\circ \+ Q_t$ are strictly positive. Otherwise we have
$ \E[ Z^{(i,j)}_\tau \mid \+ Y] = (1 - Q^{(i,j)}_t) \Lambda^{(i,j)}_{\tau} $ for any $(i,j)$ such that  $ \left[\sum_{t=1}^\tau  \bs\Lambda_t\circ \+Q_t \right]^{(i,j)}= 0$, since 
the latter equality implies $\left[ \bs\Lambda_\tau\circ \+Q_\tau\right]^{(i,j)}=0$, which in turn implies $Y_{\tau}^{(i,j)} = 0$ almost surely.
%


\end{proof}

\section{Proofs and supporting results for section \ref{sec:consistency}}\label{sec:proofs}

\subsection{Laws of Large Numbers}\label{sec:lln}
\subsubsection{Preliminaries}
\begin{lemma} \label{poissonlemma}
Let $\lambda,\lambda_n \in \mathbb{R}_{\geq 0}$ and $X_n \sim \text{Pois}(\lambda_n)$ for $n = 1,2,\dots$. Assume that for all $n$, $|\frac{\lambda_n}{n} - \lambda|< cn^{-(\frac{1}{4}+\gamma)}$ for some $c>0$ and $\gamma>0$, then there exist constants $b$ and $\bar \gamma$ such that:

\begin{equation*}
\E\left[ \left|\frac{X_n}{n} - \lambda  \right|^4\right]^\frac{1}{4} \leq b  n^{-(\frac{1}{4}+ \bar \gamma)},
\end{equation*}
furthermore
\begin{equation*}
\frac{X_n}{n} \overset{a.s.}{\rightarrow} \lambda.
\end{equation*}
\end{lemma}

\begin{proof}
By recurrence relations for the central moments of Poisson random variables, see e.g \cite{kendall1946advanced}, we can write
\begin{align}
\E\left[ \left|\frac{X_n}{n} - \frac{\lambda_n}{n}  \right|^4\right] &= n^{-4}\lambda_n \sum_{k=0}^2 \binom{3}{k}\E\left[ \left({X_n} - {\lambda_n}  \right)^k\right] \\
&= n^{-4}\lambda_n\left\{1+0+\lambda_n \frac{3!}{2!1!}\right\} \\
&=n^{-2}\left\{ n^{-2}\lambda_n + 3\left(\frac{\lambda_n}{n}\right)^2 \right\}
\leq an^{-2}, \label{convseq}
\end{align}
for some $a>0$ since the curly bracketed term in \eqref{convseq} defines a convergent sequence. Hence, by the Minkowski inequality:
\begin{align*}
\E\left[ \left|\frac{X_n}{n} - \lambda  \right|^4\right]^\frac{1}{4} &\leq \E\left[ \left|\frac{X_n}{n} - \frac{\lambda_n}{n}  \right|^4\right]^\frac{1}{4} + \left|\frac{\lambda_n}{n} - \lambda \right| \\
&\leq a^{\frac{1}{4}}n^{-\frac{1}{2}} + cn^{-(\frac{1}{4}+ \gamma)} \\
&\leq b \max\left(n^{-\frac{1}{2}}, n^{-(\frac{1}{4}+ \gamma)}\right) \\
&\leq b  n^{-(\frac{1}{4}+ \bar \gamma)},
\end{align*}
where b = $a^{\frac{1}{4}}+c$ and $\bar \gamma = \min\left(\gamma,\frac{1}{4}\right)$. Now let $\varepsilon>0$, by Markov's inequality:
\begin{equation*}
\mathbb{P}\left(\left|\frac{X_n}{n} - \lambda\right| > \varepsilon\right) \leq \varepsilon ^{-4}\E\left[\left|\frac{X_n}{n} - \lambda \right|^4\right] \leq  \epsilon^{-4}b^4  n^{-(1+ \bar 4\gamma)},
\end{equation*}
So that:
\begin{equation*}\label{finsumPois}
\sum_{n=1}^\infty \mathbb{P}\left(\left|\frac{X_n}{n} - \lambda \right| > \varepsilon\right)  \leq \epsilon^{-4}b^4  \sum_{n=1}^\infty n^{-(1+ 4\bar \gamma)} < \infty.
\end{equation*}
Then $n^{-1}X_n \rightarrow \lambda$ almost surely by the Borel-Cantelli lemma.
\end{proof}

\begin{corollary} \label{vecpoicor}
If $\+x_n \sim \mathrm{Pois}(\bs \lambda_n)$ for a sequence $(\bs \lambda_n)_{n\geq 1} \in \mathbb{R}^m$ such that there exists $c>0$ and $\gamma>0$ such that $\|n^{-1}\bs \lambda_n - \bs \lambda\|_\infty< cn^{-(\frac{1}{4}+\gamma)}$ for some $\bs \lambda \in \mathbb{R}^m$, then for any vector $\bs f \in \mathbb{R}^m$ there exists constants $b>0$ and $\bar{\gamma}>0$ such that:
\begin{equation*}
\E\left[ \left|\frac{1}{n}\+ x_n^\top \bs f - \bs \lambda^\top \bs f  \right|^4\right]^\frac{1}{4} \leq bn^{-(\frac{1}{4} + \bar \gamma)}.
\end{equation*}
\begin{proof}
Apply lemma \ref{poissonlemma} in an element-wise fashion.
\end{proof}
\end{corollary}

\begin{lemma} \label{MZlemma}
Let  $\mathcal{F}$ be a filtration and $\Delta^{(i)}$ for $i= 1,2,\dots$ be random variables which are conditionally independent given  $\mathcal{F}$, are bounded by a constant $|\Delta^{(i)}|\leq M<\infty$ almost surely, and satisfy   $\E\left[\Delta^{(i)} \mid \mathcal{F}\right]=0$. Let $a_n$ be a non-negative integer valued random variable such that $\sigma(a_n) \subseteq \mathcal{F}$ and assume there exist constants $a>0$, $b>0$, and $\gamma>0$ such that for all $n \in \mathbb{N}$:
\begin{equation*}
    \E\left[\left|\frac{a_n}{n} - a \right|^4\right]^\frac{1}{4} \leq bn^{-(\frac{1}{4} + \gamma)}.
\end{equation*}
Then there exists a constant $d>0$ such that:
\begin{equation*}
    \E\left[\left|\frac{1}{n}\sum_{i=1}^{a_n} \Delta^{(i)}\right|^4\right] \leq  d n^{-2}.
\end{equation*}
\end{lemma}
\begin{proof}
Recalling that a sum over an empty set is equal to zero by convention, we have that: 
\begin{align*}
\E\left[\left( \sum_{i=1}^{a_n} \Delta^{(i)} \right)^4 \Big{|} \mathcal{F}  \right] &= \E\left[ \sum_{\{k_1+\dots+k_{a_n}=4 ;  k_i\geq 0 \} }\binom{4}{k_1,\dots,k_{a_n}}\prod_{i=1}^{a_n}\left(\Delta^{(i)}\right)^{k_i}  \Big{|} \mathcal{F} \right]\\
&= \sum_{\{k_1+\dots+k_{a_n}=4  ;   k_i\geq 0 \} }\binom{4}{k_1,\dots,k_{a_n}}\prod_{i=1}^{a_n}\E\left[\left(\Delta^{(i)}\right)^{k_i}  \Big{|} \mathcal{F} \right]\\
& = \sum_{i=1}^{a_n}\E\left[\left(\Delta^{(i)}\right)^{4}\Big{|} \mathcal{F} \right] + 6\sum_{\{(i,j)\in [a_n]^2; i \neq j \} }\E\left[\left(\Delta^{(i)}\right)^{2}\Big{|} \mathcal{F} \right]\E\left[\left(\Delta^{(j)}\right)^{2}\Big{|} \mathcal{F} \right] \\
&\leq  a_nM^4 + 3a_n(a_n-1)M^4 \\
&\leq ca_n^2,
\end{align*}
for some constant $c>0$. The first equality holds by the multinomial theorem. The second equality holds through conditional independence of the $\Delta^{(i)}$. The third equality comes from the fact that all terms of the sum where $k_i=1$ for some $i$ disappear since $\E\left[\left(\Delta^{(i)}\right)^{1}  \Big{|} \mathcal{F} \right] = 0$;  hence, we need only count the terms with exclusively even $k_i$'s. The first term after the inequality arises since there are $a_n$ terms with a $4$th power, each of which we can bound $\E\left[\left(\Delta^{(j)}\right)^{4}  \Big{|} \mathcal{F} \right] <M^4$ . The term $3a_n(a_n-1)$ comes from counting the number of terms with exactly 2 of the $k_i$'s equal to $2$ with the rest equalling $0$; there are $\binom{a_n}{2} = a_n(a_n-1)/2$ such pairs, multiplying this by $\binom{4}{k_1,\dots,k_{a_n}} = \binom{4}{2,2,0,\dots} = 6$ gives a total of $3a_n(a_n-1)$, then we bound each of the $\E\left[\left(\Delta^{(i)}\right)^{2}  \Big{|} \mathcal{F} \right]\E\left[\left(\Delta^{(j)}\right)^{2}  \Big{|} \mathcal{F} \right]\leq M^2M^2 =M^4$ for all $(i,j)\in [a_n]^2$. So we have:

\begin{equation}
\E\left[\left.\left|\frac{1}{n} \sum_{i=1}^{a_n} \Delta^{(i)} \right|^4 \right| \mathcal{F}    \right]\leq c \left( \frac{a_n}{n} \right)^{2}n^{-2},
\end{equation}
by the Lyapunov inequality:

\begin{equation}
\begin{split}
\E\left[ \bigg| \frac{a_n}{n} \bigg|^2 \right]^\frac{1}{2} &\leq \E\left[ \bigg|\frac{a_n}{n} -a + a\bigg|^4\right]^\frac{1}{4} \\
&\leq \E\left[ \bigg|\frac{a_n}{n} -a \bigg|^4\right]^\frac{1}{4}  +  a  \\
&\leq  bn^{-(\frac{1}{4}+ \gamma)} +   a  \\
&\leq b +  a
 < \infty.
\end{split}
\end{equation}
We now apply a tower law argument to the above to see that, for constant ${d = c\left(b+a \right)^2}$:

\begin{equation} \label{-range}
\E\left[\left|\frac{1}{n} \sum_{i=1}^{a_n} \Delta^{(i)} \right|^4  \right]\leq d n^{-2}.
\end{equation}
\end{proof}

\begin{lemma} \label{lem:ratiolem}
Let $\+ x\in \mathbb{R}^m_{\geq 0}$, $\bs f \in \mathbb{R}^m$ , $c>0$, and $n \in \mathbb{N}$. Then:
\begin{equation*}
\left|\bs \eta(\+x)^\top \bs f -  \frac{n^{-1}\+x^\top \bs f}{c}\right| \leq \left|\bs \eta(\+x)^\top \bs f \right| c^{-1}\left| n^{-1}\+1_m^\top\+x - c\right|
\end{equation*}
\end{lemma}
\begin{proof}
If $\+x = \bs 0$ then the result is trivial. Now, for $\+x \neq \bs 0$ we have:

\begin{align*}
\left|\bs \eta(\+x)^\top \bs f -  \frac{n^{-1}\+x^\top \bs f}{c}\right| &= \left|\frac{\+x^\top}{\+1_m^\top \+x} \bs f -  \frac{n^{-1}\+x^\top \bs f}{c}\right| \\
&\leq\left|\+x^\top \bs f \right| \left| \frac{1}{\+1_m^\top \+x} - \frac{n^{-1}}{c} \right| \\
&= \left|\+x^\top \bs f \right| \left|\frac{c - n^{-1}\+1_m^\top\+x}{c\+1_m^\top \+ x} \right| \\
&\leq \left|\frac{\+x^\top \bs f}{\+1_m^\top \+x} \right|c^{-1}\left|n^{-1}\+1_m^\top\+x - c\right| \\
&= \left|\bs \eta(\+x)^\top \bs f \right| c^{-1}\left| n^{-1}\+1_m^\top\+x - c\right|.
\end{align*}
\end{proof}

\subsubsection{Case (I)}
 Define the sequence of vectors:
\begin{equation*}
\begin{aligned}
\bs \nu_0(\bs \theta^*) &:= {\bs \lambda}_{0,\infty}(\bs \theta^*),\\
\bs \nu_{t+1}(\bs \theta^*) &:= \left[\left( \bs \nu_t(\bs \theta^*) \circ \bs \delta_{t+1}(\bs \theta^*) \right)^\top\mathbf{K}_{t+1, \boldsymbol{\eta}( \bs \nu_t(\bs \theta^*) \circ \bs \delta_{t+1}(\bs \theta^*))}\right]^\top + {\bs \alpha}_{t+1, \infty}(\bs \theta^*).
\end{aligned}
\end{equation*}

\begin{lemma}\label{lem:llnx}
Let assumptions \ref{as:params}-\ref{as:init} hold. For all $t \geq 0$ there exists $\gamma_t>0$ and for all $\bs f \in \mathbb{R}^m$ there exists $c_t >0$ such that:
\begin{equation} \label{xconvstat}
\mathbb{E}\left[  \left| \frac{ \+ x_{t}^\top}{n}\bs f  - \bs \nu_t(\bs \theta^*)^\top \bs f \right|^4  \right]^{\frac{1}{4}} \leq c_tn^{- (\frac{1}{4} + \gamma_t)}.
\end{equation}
\end{lemma}
\begin{proof}
Explicit dependence of some quantities on $\bs \theta^*$ and $n$ is omitted throughout the proof to avoid over-cumbersome notation where the dependence is unambiguous. We proceed to prove the above by induction on $t$. At time $0$ we have for some $c_0>0$ and $\gamma_0>0$:

\begin{equation}
\E\left[ \left|\frac{\+x_0^\top}{n}\bs f - \bs{\nu}_0^\top \bs f \right|^4\right]^\frac{1}{4} \leq c_0n^{-(\frac{1}{4} + \gamma_0)},
\end{equation}
by assumption \ref{as:init}. Now for $t \geq 1$ assume \eqref{xconvstat} holds for $t-1$. Recall $\+x_t = \tilde{\+x}_t + \hat{\+x}_t$, $\tilde{ x}^{(j)}_t = \sum_{i=1}^{n_{t-1}}\mathbb{I}\{\phi_t^{(i)} = 1 \}\mathbb{I}\{\xi_t^{(i)} = j \}$ and $\bar{x}_{t-1}^{(j)} =\sum_{i=1}^{n_{t-1}} \mathbb{I}[\xi_{t-1}^{(i)}=j]\mathbb{I}[\phi^{(i)}_t=1]$.  We make the following decomposition:
\begin{align}\label{timetpart1}
\frac{ \+ x_t^\top}{n}\bs f  - \bs \nu_t^\top \bs f &= \frac{ \+ x_{t}^\top}{n}\bs f  -\left[ \frac{ \+ {\bar x}_{t-1}}{n} ^\top\+ K_{t, \bs \eta(\bar{\+ x}_{t-1})} + {\bs \alpha}_{t,\infty}^\top \right] \bs f \\
&+ \left[\frac{\+ {\bar x}_{t-1}}{n} - \bs{\nu}_{t-1}\circ \bs \delta_t\right]^\top\left[\+ K_{t, \bs \eta(\bar{\+ x}_{t-1})} \bs f \right] \label{timetpart2} \\
 & +(\bs \nu_{t-1} \circ \bs \delta_t )^\top\left[\+ K_{t, \bs \eta(\bar{\+ x}_{t-1})} - \+ K_{t, \bs \eta(\bs \nu_{t-1}\circ \bs \delta_t)} \right] \bs f \label{timetpart3}.
\end{align}
Consider \eqref{timetpart1}. We make the further decomposition:
\begin{align}
\frac{ \+ x_{t}^\top}{n}\bs f &-\left[ \frac{\+ {\bar x}_{t-1}}{n}^\top\+ K_{t, \bs \eta(\bar{\+ x}_{t-1})} + {\bs \alpha}_{t,n}^\top \right] \bs f \\
&= \frac{ (\tilde{\+ x}_t+ \hat{\+x}_t)^\top}{n}\bs f -\left[ \frac{\+ {\bar x}_{t-1}}{n}^\top\+ K_{t, \bs \eta(\bar{\+ x}_{t-1})} + {\bs \alpha}_{t,n}^\top \right] \bs f \\
&= \frac{ \tilde{\+ x}_{t}^\top}{n}\bs f -\left[ \frac{\+ {\bar x}_{t-1}}{n}^\top\+ K_{t, \bs \eta(\bar{\+ x}_{t-1})} \right] \bs f + \left[\frac{\hat{\+ x}^\top_t}{n} - {\bs \alpha}_{t,n}^\top \right]\bs f\\
&= \frac{1}{n}\sum_{j=1}^m \left(\sum_{i=1}^{n_{t-1}}\mathbb{I}\{\phi_t^{(i)}=1\} \mathbb{I}\{\xi_t^{(i)} = j\}\right)f^{(j)} - \frac{1}{n}\sum_{j=1}^m \left( \sum_{i=1}^{n_{t-1}}\mathbb{I}\{\xi_{t-1}^{(i)} = j\}\mathbb{I}\{\phi_t^{(i)}=1\}\right)\+K_{t, \bs \eta(\bar{\+ x}_{t-1})}^{j, \cdot} \bs f \\
&+ \left[\frac{\hat{\+ x}^\top_t}{n} - {\bs \alpha}_{t,n}^\top \right]\bs f.\\
&= \frac{1}{n}\sum_{j=1}^m \left(\sum_{i=1}^{n_{t-1}}\mathbb{I}\{\phi_t^{(i)}=1\} \mathbb{I}\{\xi_t^{(i)} = j\}\right)f^{(j)} - \frac{1}{n} \sum_{i=1}^{n_{t-1}}\mathbb{I}\{\phi_t^{(i)}=1\}\+K_{t, \bs \eta(\bar{\+ x}_{t-1})}^{\xi_{t-1}^{(i)}, \cdot} \bs f + \left[\frac{\hat{\+ x}^\top_t}{n} - {\bs \alpha}_{t,n}^\top \right]\bs f.\\
&= \frac{1}{n}\sum_{i=1}^{n_{t-1}}\sum_{j=1}^m \left \{  \mathbb{I}\{\phi_t^{(i)}=1\}\mathbb{I}\{\xi_t^{(i)} = j\}- \mathbb{I}\{\phi_t^{(i)}=1\}K_{t, \bs \eta(\bar{\+ x}_{t-1})}^{(\xi_{t-1}^{(i)}),j}\right \}f ^{(j)} \label{decomp} \\
&+ \left[\frac{\hat{\+ x}^\top_t}{n} - {\bs \alpha}_{t,n}^\top \right]\bs f.
\end{align}
 The term $\left[\frac{\hat{\+ x}^\top_t}{n} - {\bs \alpha}_{t,n}^\top \right]\bs f$  converges to $0$ in $L^4$ by assumption \ref{as:params} and lemma \ref{poissonlemma}, that is there exists an $\hat{c}_t>0$ and $\hat{\gamma}_t>0$ such that:

\begin{equation}
\E\left[ \left| \left[\frac{\hat{\+ x}^\top_t}{n} - {\bs \alpha}_{t,n}^\top \right]\bs f \right|^4 \right]^\frac{1}{4} \leq \hat{c}_tn^{-(\frac{1}{4} + \hat{\gamma}_t)}.
\end{equation}
Now, turning to \eqref{decomp}, let $\mathcal{G}_t := \sigma(\{\xi_t^{(i)} \}_{i=1,\dots, n_t})$ and $\mathcal{F}_t := \sigma(\{\phi_t^{(i)}\}_{i=1,\dots, n_t})$. See that:
\begin{equation}
\begin{aligned}
&\E\left[ \sum_{j=1}^m \left \{  \mathbb{I}\{\phi_t^{(i)}=1\}\mathbb{I}\{\xi_t^{(i)} = j\}- \mathbb{I}\{\phi_t^{(i)}=1\}K_{t, \bs \eta(\+{\bar{x}}_{t-1})}^{(\xi_{t-1}^{(i)}),j}\right \}f ^{(j)} \Big{|} \mathcal{G}_{t-1} \vee \mathcal{F}_{t} \right] \\
= & \sum_{j=1}^m \left \{ \E\left[ \mathbb{I}\{\phi_t^{(i)}=1\}\mathbb{I}\{\xi_t^{(i)} = j\}\Big{|} \mathcal{G}_{t-1} \vee \mathcal{F}_{t} \right] - \mathbb{I}\{\phi_t^{(i)}=1\}K_{t, \bs \eta(\+{\bar{x}}_{t-1})}^{(\xi_{t-1}^{(i)}),j}\right \}f ^{(j)} \\
= & \sum_{j=1}^m \left \{  \mathbb{I}\{\phi_t^{(i)}=1\}K_{t, \bs \eta(\+{\bar{x}}_{t-1})}^{(\xi_{t-1}^{(i)}),j} - \mathbb{I}\{\phi_t^{(i)}=1\}K_{t, \bs \eta(\+{\bar{x}}_{t-1})}^{(\xi_{t-1}^{(i)}),j}\right \}f ^{(j)} \\
= &0,
\end{aligned}
\end{equation}
 since, given $\xi_{t-1}^{(i)}$, $\phi_t^{(i)}\sim \mathrm{Bernoulli}\left(\delta_t^{(\xi_{t-1}^{(i)})}\right)$ and, conditional on $\phi_t^{(i)} = 1$ and $\mathcal{G}_{t-1}$,  $\xi^{(i)}_{t}$ is a draw from the $\xi^{(i)}_{t-1}$th row of $\+ K_{t, \bs \eta(\bar{x}_{t-1})}$; and if $\phi_t^{(i)} = 0$ then $\xi_{t}^{(i)}=0$. Moreover:
\begin{equation} \label{Mbound}
\bigg|\sum_{j=1}^m \left \{  \mathbb{I}\{\phi_t^{(i)}=1\}\mathbb{I}\{\xi_t^{(i)} = j\}- \mathbb{I}\{\phi_t^{(i)} = 1 \}K_{t, \bs \eta(\bar{\+x}_{t-1})}^{(\xi_{t-1}^{(i)}),j}\right \}f ^{(j)} \bigg|\leq m\|\bs f \|_\infty =: M.
\end{equation}

Define:
\begin{equation*}\Delta_t^{(i)} := \sum_{j=1}^m \left \{ \mathbb{I}\{\phi_t^{(i)}=1\} \mathbb{I}\{\xi_t^{(i)} = j\}- \mathbb{I}\{\phi_t^{(i)} = 1 \}K_{t, \bs \eta(\bar {\+x}_{t-1})}^{(\xi_{t-1}^{(i)}),j}\right \}f ^{(j)}.
\end{equation*}
The $\Delta_t^{(i)}$ are conditionally independent and mean zero given $\mathcal{G}_{t-1}\vee\mathcal{F}_t$, and $\sigma(n_{t-1}) \subset \mathcal{G}_{t-1}\vee\mathcal{F}_t$. Also note that, since $\frac{n_{t-1}}{n}$ is equal to $\frac{\+x_{t-1}^\top}{n}\+1_m$, we can invoke the induction hypothesis with test vector $\+1_m$ to see there exist constants $c_{t-1}$ and $\gamma_{t-1}$ such that:
\begin{equation*}
\E\left[ \left|\frac{n_{t-1}}{n} - \+1_m^\top \bs \nu_{t-1} \right|^4 \right]^\frac{1}{4} \leq c_{t-1}n^{-(\frac{1}{4} + \gamma_{t-1})}
\end{equation*}
so that we satisfy the conditions of lemma \ref{MZlemma}. Hence there exists a constant $\tilde{c}_t>0$:

\begin{equation}
\E\left[\left|\frac{1}{n} \sum_{i=1}^{n_{t-1}} \Delta_t^{(i)} \right|^4  \right]^\frac{1}{4}\leq \tilde{c}_t n^{-\frac{1}{2}}.
\end{equation}
Before analysing \eqref{timetpart2} and \eqref{timetpart3} we will prove an intermediary result. Consider the decomposition:

\begin{align}
\label{deltap1}
\left| \frac{\bar{\+x}_{t-1}}{n}^\top \bs f - (\bs{\nu}_{t-1}\circ \bs \delta_t)^\top \bs f \right|  &\leq   \bigg|\frac{\bar{\+x}_{t-1}}{n}^\top \bs f - \left(\frac{{\+x}_{t-1}}{n}\circ \bs \delta_t \right)^\top \bs f \bigg| \\
 \label{deltap2}
 &+ \bigg|\left(\frac{{\+x}_{t-1}}{n}\circ \bs \delta_t \right)^\top \bs f - (\bs{\nu}_{t-1}\circ \bs \delta_t)^\top \bs f \bigg|.
 \end{align}
The term in \eqref{deltap2} converges to 0 in $L^4$ at the required rate by the induction hypothesis with test vector $\bs \delta_t \circ \bs f$. Now for \eqref{deltap1} see that:

\begin{equation} \label{xbar}
\frac{\bar{\+x}_{t-1}}{n}^\top \bs f - \left(\frac{{\+x}_{t-1}}{n}\circ \bs \delta_t \right)^\top \bs f = \frac{1}{n}\sum_{i=1}^{n_{t-1}} \underbrace{\sum_{j=1}^m \mathbb{I}\left\{ \xi_{t-1}^{(i)} = j\right\}\left(
\mathbb{I}\left\{ \phi_t^{(i)} = 1\right\}
  - \delta_t^{(\xi_{t-1}^{(i)})}\right)f^{(j)}}_{\bar \Delta_t^{(i)}},
\end{equation}
and note that the $\bar \Delta_t^{(i)}$ are mean $0$, bounded, and independent given $\mathcal{G}_{t-1}$ so that by lemma \ref{MZlemma} we have that \eqref{xbar} converges to $0$ in $L^4$ at the required rate. Combining this with the Minkowski inequality, we have that for some positive constants $\bar{c}_t$ and $\bar{\gamma}_t$:

\begin{equation} \label{nudelta}
\E\left[\left| \frac{\bar{\+x}_{t-1}}{n}^\top \bs f - (\bs{\nu}_{t-1}\circ \bs \delta_t)^\top \bs f \right|^4 \right]^\frac{1}{4} \leq \bar{c}_tn^{-(\frac{1}{4}+ \bar{\gamma})}.
\end{equation}
Now we look at \eqref{timetpart2}:

\begin{equation}
\begin{aligned}
 \left| \left[\frac{\+{\bar x}_{t-1}}{n} - \bs{\nu}_{t-1}\circ \bs \delta_t\right]^\top\left[\+ K_{t, \bs \eta(\bar{\+ x}_{t-1})} \bs f\right] \right| &\leq \left\| \frac{\bar{\+ x}_{t-1}}{n} - \bs{\nu}_{t-1}\right\|_1\left\|\+ K_{t, \bs \eta(\bar{\+ x}_{t-1})} \bs f\right\|_\infty \\
 &\leq \|\bs f \|_\infty \sum_{i=1}^m\left|\frac{\bar{x}_{t-1}^{(i)}}{n} - \nu_{t-1}^{(i)}\delta_t^{(i)}\right|\\
 &\leq  \|\bs f \|_\infty \sum_{i=1}^m\left|\frac{\+ {\bar{x}}_{t-1}^\top}{n}\bs e_i - (\bs \nu_{t-1}\circ \bs \delta_t)^\top \bs e_i\right|.
\end{aligned}
\end{equation} 
The first inequality here uses Holder's inequality and the second uses the fact that the row sums of the matrix $\+ K_{t, \bs \eta(\bar{\+ x}_{t-1})}$ are equal to $1$.  By \eqref{nudelta} with $\bs f = \bs e_i$ in conjunction with the Minkowski inequality there exists $\breve c_t>0$ and $\breve \gamma_t>0$ such that:

\begin{equation}
\E\left[  \left| \left[\frac{\+{\bar x}_{t-1}}{n} - \bs{\nu}_{t-1}\circ \bs \delta_t\right]^\top\left[\+ K_{t, \bs \eta(\bar{\+ x}_{t-1})} \bs f\right] \right|^4 \right]^\frac{1}{4} \leq \breve c_tn^{-(\frac{1}{4}+ \breve \gamma_t)}.
\end{equation}
Now looking at \eqref{timetpart3} we see using assumption \ref{as:1} that there exists a $c>0$ such that:

\begin{equation*}
\begin{aligned}
\left|\left[\bs \nu_{t-1} \circ \bs \delta_t \right]^\top\left[\+ K_{t, \bs \eta(\bar{\+ x}_{t-1})} - \+ K_{t, \bs \eta(\bs \nu_{t-1}\circ \bs \delta_t)} \right] \bs f\right| &\leq c \| \bs \nu_{t-1} \circ \bs \delta_t \|_\infty\|\bs f \|_\infty\|\bs \eta(\bar{\+ x}_{t-1}) -   \bs \eta(\bs \nu_{t-1}\circ \bs \delta_t)\|_\infty \\
&\leq c \| \bs \nu_{t-1} \circ \bs \delta_t \|_\infty\|\bs f \|_\infty \sum_{i=1}^m\left|(\bs \eta(\bar{\+ x}_{t-1}) -   \bs \eta(\bs \nu_{t-1}\circ \bs \delta_t) )^\top \bs e_i \right| \\
\end{aligned}
\end{equation*}
If $\+1_m^\top(\bs{\nu}_{t-1}\circ \bs \delta_t) = 0$ then $\+1_m^\top \bar{\+x}_{t-1} = 0$ $\mathbb{P}^{\theta^*}-a.s.$ by lemmas \ref{muylemma} and \ref{mumunlemma}, which we state and prove in section \ref{sec:filtering_limits}, in which case the right hand side of the above is $0$ and therefore satisfies all positive bounds. Henceforth, assume $\+1_m^\top(\bs{\nu}_{t-1}\circ \bs \delta_t) > 0$. Consider:

\begin{align}
\bigg| \bs \eta (\+{\bar x_{t-1}})^\top\bs f &-   \bs \eta(\bs{\nu}_{t-1}\circ \bs \delta_t)^\top \bs f\bigg| \\
& = \bigg| \bs \eta (\+{\bar x_{t-1}})^\top\bs f - \frac{(\bs{\nu}_{t-1}\circ \bs \delta_t)^\top \bs f}{\+1_m^\top (\bs{\nu}_{t-1}\circ \bs \delta_t)}\bigg|  \\ 
&=\bigg| \bs \eta (\+{\bar x_{t-1}})^\top\bs f  +\frac{n^{-1}\bar{\+x}_{t-1}^\top }{\+1_m^\top(\bs{\nu}_{t-1}\circ \bs \delta_t)}\bs f  - \frac{n^{-1}\bar{\+x}_{t-1}^\top }{\+1_m^\top(\bs{\nu}_{t-1}\circ \bs \delta_t)}\bs f  -   \bs \eta(\bs{\nu}_{t-1}\circ \bs \delta_t)^\top \bs f \bigg| \\
&\leq \bigg| \bs \eta (\+{\bar x_{t-1}})^\top\bs f  -\frac{n^{-1}\bar{\+x}_{t-1}^\top }{\+1_m^\top(\bs{\nu}_{t-1}\circ \bs \delta_t)}\bs f\bigg| + \bigg |\frac{n^{-1}\bar{\+x}_{t-1}^\top }{\+1_m^\top(\bs{\nu}_{t-1}\circ \bs \delta_t)}\bs f  -   \bs \eta(\bs{\nu}_{t-1}\circ \bs \delta_t)^\top \bs f \bigg| \\
&\leq  \left|\bs \eta (\+{\bar x_{t-1}})^\top \bs f \right|(\+1_m^\top (\bs{\nu}_{t-1}\circ \bs \delta_t))^{-1}\left|\frac{\+1_m^\top \bar{\+x}_{t-1}}{n} - \+1_m^\top(\bs{\nu}_{t-1}\circ \bs \delta_t)\right| \label{eq:ratiolem} \\
 &+ (\+1_m^\top(\bs{\nu}_{t-1}\circ \bs \delta_t))^{-1}\left|\frac{\bar{\+x}_{t-1}^\top }{n}\bs f - (\bs{\nu}_{t-1}\circ \bs \delta_t)^\top \bs f\right| \\
 & \leq \frac{m\|\bs f \|_\infty}{\+1_m^\top(\bs{\nu}_{t-1}\circ \bs \delta_t)}\left|\frac{\+1_m^\top \bar{\+x}_{t-1}}{n} - \+1_m^\top(\bs{\nu}_{t-1}\circ \bs \delta_t)\right| \label{xbar1}\\
  &+ (\+1_m^\top(\bs{\nu}_{t-1}\circ \bs \delta_t))^{-1}\left|\frac{\bar{\+x}_{t-1}^\top }{n}\bs f - (\bs{\nu}_{t-1}\circ \bs \delta_t)^\top \bs f\right|.\label{xbarf}
\end{align}
Where we use lemma \ref{lem:ratiolem} in line \eqref{eq:ratiolem}. We can again invoke \eqref{nudelta} to give $L^4$ convergence of \eqref{xbar1} and \eqref{xbarf} at the required rate. We can now combine all of the above, along with the Minkowski inequality to show that:

\begin{equation}
\E\left[ \left| \frac{ \+ x_t^\top}{n}\bs f  - \bs \nu_t^\top \bs f \right|^4 \right]^\frac{1}{4} \leq c_tn^{-(\frac{1}{4} + \gamma_t)},
\end{equation}
where $c_t = \hat{c}_t + \tilde{c}_t + \bar{c}_t + \breve{c}_t$, and $\gamma_t = \min(\hat{\gamma}_t, \frac{1}{4}, \bar{\gamma}_t, \breve{\gamma}_t)$.
\end{proof}


\begin{lemma}\label{lem:lpy}
Let assumptions \ref{as:params} - \ref{as:init} hold. Then there exists a constant $\rho_t>0$ for each $\bs f \in \mathbb{R}^m$ and $t\geq 1$, and a constant $a_t>0$ such that:
\begin{equation*}
\E\left[\left|\frac{\+y^\top_{t}}{n}\bs f - \left([\bs \nu_{t}(\bs \theta^*) \circ \+ q_t(\bs \theta^*) ]^\top \+ G_t(\bs \theta^*) + \bs \kappa_{t, \infty}(\bs \theta^*)^\top\right) \bs f\right|^4\right]^\frac{1}{4} \leq a_tn^{-(\frac{1}{4}+ {\rho}_t)}.
\end{equation*}
\end{lemma}
\begin{proof}
Explicit dependence of some quantities on $\bs \theta^*$ and $n$ is omitted throughout the proof to avoid over-cumbersome notation where the dependence is unambiguous. First note that:
\begin{equation}
\frac{\+y_t}{n} = \frac{\tilde{\+y}_t}{n} + \frac{\hat{\+y}_t}{n},
\end{equation}
and
\begin{equation} \label{ybound1}
\E\left[\left|\frac{\hat{\+y}^\top_t}{n} \bs f - \bs \kappa_{t,\infty}^\top \bs f\right|^4\right]^\frac{1}{4}< \hat a_{t}n^{-(\frac{1}{4}+ \hat \rho_{t})},
\end{equation}
for some $\hat a_t>0$ and $\hat \rho_t>0$ by corollary \ref{vecpoicor} and assumption \ref{as:params}. Write:
\begin{align}
\label{part1G}\frac{\tilde{\+y}^\top_t}{n}\bs f - (\bs \nu_t \circ \+ q_t )^\top \+ G_t\bs f&= \frac{\tilde{\+y}^\top_t}{n}\bs f -  \left(\frac{{\+ x}_t}{n}\circ \+ q_t \right)^\top \+ G_t \bs f \\
\label{part2G} &+ \left(\frac{{\+ x}_t}{n}\circ \+ q_t \right)^\top \+ G_t \bs f -  \left(\bs \nu_t\circ \+ q_t \right)^\top \+ G_t\bs f.
\end{align}
We have that
\begin{equation} \label{ybound2}
\E\left[ \left| \left(\frac{{\+ x}_t}{n}\circ \+ q_t \right)^\top \+ G_t \bs f -  \left(\bs \nu_t\circ \+ q_t \right)^\top \+ G_t \bs f \right|^4\right]^\frac{1}{4} \leq \bar a_t n^{-(\frac{1}{4} + \bar \rho_t)}
\end{equation}
for some $\bar a_t >0$ and $\bar \rho_t>0$ by lemma \ref{lem:llnx} using test function $[\left(\+q_t \otimes \+1_m\right)\circ \+G_t] \bs f$.

 Furthermore, we have:
\begin{equation}
\begin{split}
\frac{\tilde{\+y}_t^\top}{n}\+ f &= \frac{1}{n}\sum_{j=1}^m \tilde{\+y}^{(j)}\bs f^{(j)} \\
&= \frac{1}{n}\sum_{i=1}^{n_{t}}\sum_{j=1}^m\sum_{k=1}^m\mathbb{I}\{\xi_t^{(i)} = k\}\mathbb{I}\{\zeta_t^{(i)} = 1\}\mathbb{I}\{\varsigma_t^{(i)} = j \}\bs f^{(j)},
\end{split}
\end{equation}
where $\zeta_t^{(i)} \sim \text{Bernoulli}(\+q_t^{(\xi_t^{(i)})})$ and $\varsigma_t^{(i)} \sim \text{Categorical}(\+G^{(\xi_t^{(i)},\cdot)})$ indicates the compartment in which it is observed.
Notice that for \eqref{part1G} :

\begin{equation}
\begin{aligned}
\frac{\tilde{\+y}^\top_t}{n}\bs f -  &\left(\frac{{\+ x}_t}{n}\circ \+ q_t \right)^\top \+ G_t \bs f  \\
&= \frac{1}{n}\sum_{i=1}^{n_t}\underbrace{\sum_{j=1}^m\sum_{k=1}^m\left[\mathbb{I}\{\xi_t^{(i)} = k\}\mathbb{I}\{\zeta_t^{(i)} = 1\}\mathbb{I}\{\varsigma_t^{(i)} = j \} -  \mathbb{I}\{\xi_t^{(i)} = k\}q_t^{(k)}G_t^{(k,j)} \right]\bs f^{(j)}}_{=: \Xi_t^{(i)}}.
\end{aligned}
\end{equation}
The $\Xi_t^{(i)}$ are mean zero and independent conditioned on $\mathcal{G}_t$. Furthermore:
\begin{equation}
|\Xi_t^{(i)}|\leq m^2\|\bs f\|_\infty,
\end{equation}
almost surely and $\sigma(n_t) \subseteq \mathcal{G}_t$ where $ \mathcal{G}_t$ is defined as in lemma \ref{lem:llnx}. An application of lemma \ref{MZlemma}, yields:
\begin{equation} \label{ybound3}
\E \left[\left|\frac{\tilde{\+y}^\top_t}{n}\bs f -  \left(\frac{{\+ x}_t}{n}\circ \+ q_t \right)^\top \+ G_t \bs f\right|^4\right]^\frac{1}{4} \leq \tilde a_{t}n^{(\frac{1}{4} + \tilde \rho_t)},
\end{equation}
for some constants $\tilde a_t>0$ and $\tilde{\rho}_t>0$. Combining \eqref{ybound1}, \eqref{ybound2}, and \eqref{ybound3} with the Minkowski inequality yields:

\begin{equation}
\E\left[\left|\frac{\+y^\top_t}{n}\bs f - \left((\bs \nu_t \circ \+ q_t )^\top \+ G_t + \bs \kappa_t^\top\right) \bs f\right|^4\right]^\frac{1}{4} \leq a_tn^{-(\frac{1}{4}+ {\rho}_t)},
\end{equation}
where $a_t = \hat a_t + \bar{a}_t + \tilde{a}_t$ and $\rho_t = \min(\hat{\rho}_t, \bar \rho_t, \tilde \rho_t)$.

\end{proof}
\begin{proposition}\label{prop:yas}
Let assumptions \ref{as:params} - \ref{as:init} hold. Then for all $t\geq 1$:
\begin{equation}\label{eq:y_lln}
\frac{ \+ y_t^\top}{n}  \asls [\bs \nu_t(\bs \theta^*) \circ \+ q_t(\bs \theta^*) ]^\top \+ G_t(\bs \theta^*) + \bs \kappa_{t,\infty}(\bs \theta^*)^\top.
\end{equation}
\end{proposition}

\begin{proof}
By lemma \ref{lem:lpy} there exists constants $a_t>0$ and $\rho_t>0$ such that:
\begin{equation*}
\E\left[ \left| \frac{ \+ y_t^\top}{n}\bs f - ((\bs \nu_t \circ \+ q_t )^\top \+ G_t + \bs \kappa_t)^\top \bs f \right|^4 \right] \leq a_t^4n^{-(1 + 4\rho_t)}.
\end{equation*}
By Markov's inequality:
\begin{equation*}
\begin{split}
\mathbb{P}^{\bs\theta^*} \left[ \left| \frac{ \+ y_t^\top}{n}\bs f - ((\bs \nu_t \circ \+ q_t )^\top \+ G_t + \bs \kappa_t^\top) \bs f  \right| > \varepsilon \right] &\leq \varepsilon^{-4}\E\left[ \left| \frac{ \+ y_t^\top}{n}\bs f - ((\bs \nu_t \circ \+ q_t )^\top \+ G_t + \bs \kappa_t^\top) \bs f  \right|^4 \right] \\
 &\leq \varepsilon^{-4}a_t^4n^{-(1+4\rho_t)}.
\end{split}
\end{equation*}
This implies that:
\begin{equation} \label{BClem}
\sum_{n=1}^\infty \mathbb{P}^{\bs\theta^*} \left[ \left| \frac{ \+ y_t^\top}{n}\bs f - ((\bs \nu_t \circ \+ q_t )^\top \+ G_t + \bs \kappa_t^\top) \bs f \right| > \varepsilon \right] < \infty.
\end{equation}
We now appeal to the Borel-Cantelli lemma which tells us that \eqref{BClem} implies  the event:
\begin{equation*}\left \{  \left| \frac{ \+ y_t^\top}{n}\bs f - ((\bs \nu_t \circ \+ q_t )^\top \+ G_t + \bs \kappa_t^\top) \bs f \right| > \varepsilon \right \}, \end{equation*}
happens for infinitely many $n$ with probability 0, and that:

\begin{equation*}\mathbb{P}^{\bs\theta^*}\left( \lim_{n \rightarrow \infty} \left| \frac{ \+ y_t^\top}{n}\bs f - ((\bs \nu_t \circ \+ q_t )^\top \+ G_t + \bs \kappa_t^\top) \bs f \right| > \varepsilon \right) = 0,\end{equation*}
for all $\varepsilon>0$. Hence we have shown that:

\begin{equation*}
\frac{ \+ y_t^\top}{n} \bs f \asls ((\bs \nu_t \circ \+ q_t )^\top \+ G_t + \bs \kappa_t^\top)\bs f .
\end{equation*}

\end{proof}

\subsubsection{Case (II)}

Define:
\begin{align*}
\bs \nu_0(\bs \theta^*) &\coloneqq \bs \lambda_{0,\infty}(\bs \theta^*), \\
\+ N_t(\bs \theta^*) &\coloneqq \left(\bs \nu_{t-1}(\bs \theta^*) \otimes \+1_m \right) \circ \+K_{t,\bs \eta\left(\bs \nu_{t-1}(\bs \theta^*)\right)}(\bs \theta^*), \\
\bs \nu_t(\bs \theta^*) &\coloneqq (\+1_m^\top \+ N_t(\bs \theta^*))^\top.
\end{align*}

\begin{lemma}\label{lem:llnZ}
Let assumptions \ref{as:params} - \ref{as:init} hold. For all $t \geq 1$ there exists a $\gamma_{t_z}>0$, and for all vectors $\bs f_1, \bs f_2 \in \mathbb{R}^m$ a constant $b_t>0$, such that:
\begin{equation*}
\E \left[ \left| n^{-1} \bs f_1^\top \+Z_t \bs f_2 - \bs f_1^\top\+ N_t(\bs \theta^*) \bs f_2 \right|^4 \right]^\frac{1}{4} \leq b_tn^{-(\frac{1}{4} + \gamma_{t_z})}, \text{ for all } t \geq 0.
\end{equation*}
\end{lemma}

\begin{proof}
Recall from section \ref{sec:case_II} that in case (II) there is no immigration or emigration,   $n_t=n$  and hence also $\+x_t = \+{\bar{x}}_t$ with probability $1$ for all $t \geq 0$. 

Consider the decomposition:
\begin{align}
\big|n^{-1} \bs f_1^\top \+Z_t \bs f_2 &- \bs f_1^\top\left[\bs \nu_{t-1} \otimes \+1_m \right] \circ \+K_{t,\bs \eta( \bs \nu_{t-1})}\bs f_2\big| \\ \label{Zp1}
&\leq \left|n^{-1} \bs f_1^\top \+Z_t \bs f_2 -  \bs f_1^\top\left[\frac{\+x_{t-1}}{n} \otimes \+1_m \right] \circ \+K_{t, \bs \eta(\+ x_{t-1})}\bs f_2\right|\\ \label{Zp2}
&+ \left|\bs f_1^\top\left[\frac{\+x_{t-1}}{n} \otimes \+1_m \right] \circ \+K_{t, \bs \eta(\+ x_{t-1})}\bs f_2 - \bs f_1^\top\left[\frac{\+x_{t-1}}{n} \otimes \+1_m \right] \circ \+K_{t, \bs \eta(\bs \nu_{t-1})}\bs f_2 \right|\\ \label{Zp3}
&+ \left|\bs f_1^\top\left[\frac{\+x_{t-1}}{n} \otimes \+1_m \right] \circ \+K_{t, \bs \eta(\bs \nu_{t-1})}\bs f_2 - \bs f_1^\top\left[\bs \nu_{t-1} \otimes \+1_m \right] \circ \+K_{t, \bs \eta(\bs \nu_{t-1})}\bs f_2 \right|.
\end{align}
Notice that by assumption \ref{as:params} with vectors $\bs f_1$ and $\bs f_2$, there exists a constant $c>0$ such that the term \eqref{Zp2} satisfies:
\begin{align}
    &\left|\bs f_1^\top\left[ \left(\frac{\+x_{t-1}}{n} \otimes \+1_m \right)\circ\left( \+K_{t, \bs \eta(\+ x_{t-1})} - \+K_{t, \bs \eta(\bs \nu_{t-1})} \right) \right]\bs f_2 \right| \\
    =&\left| \left( \bs f_1\circ \frac{\+x_{t-1}}{n}\right)^\top\left( \+K_{t, \bs \eta(\+ x_{t-1})} - \+K_{t, \bs \eta(\bs \nu_{t-1})} \right) \bs f_2 \right| \\
    \leq &c\|\bs f_1\|_\infty \|\bs f_2\|_\infty \left\| \bs \eta(\+x_{t-1}) - \bs \eta(\bs\nu_{t-1}) \right\|_\infty\\
    \leq &c\|\bs f_1\|_\infty \|\bs f_2\|_\infty \left\| \frac{\+x_{t-1}}{n} - \bs \nu_{t-1}  \right\|_\infty \\
    \leq&c\|\bs f_1\|_\infty \|\bs f_2\| \sum_{i=1}^m\left|\left(\frac{\+x_{t-1}}{n} - \bs \nu_{t-1} \right)^\top \bs e_i \right|.
\end{align}
By the Minkowski inequality and lemma \ref{lem:llnx} there exist constants $\bar b_t>0$ and $\bar \gamma_{t_z}>0$ such that:
\begin{equation}
    \E\left[\left|\bs f_1^\top\left[ \left(\frac{\+x_{t-1}}{n} \otimes \+1_m \right)\circ\left( \+K_{t, \bs \eta(\+ x_{t-1})} - \+K_{t, \bs \eta(\bs \nu_{t-1})} \right) \right]\bs f_2 \right|^4\right]^\frac{1}{4} \leq \bar{b}_t n^{-(\frac{1}{4}+ \bar{\gamma}_{t_z})}
\end{equation}
Moreover, \eqref{Zp3} is equal to:

\begin{align}
    \left|\frac{\+x_{t-1}^\top}{n}\left[\bs f_1\otimes \+1_m \right] \circ \+K_{t, \bs \eta(\bs \nu_{t-1})}\bs f_2 - \bs \nu_{t-1} ^\top\left[\bs f_1\otimes \+1_m \right] \circ \+K_{t, \bs \eta(\bs \nu_{t-1})}\bs f_2 \right|,
\end{align}
therefore we can invoke lemma \ref{lem:llnx} with test vector $\left[\bs f_1\otimes \+1_m \right] \circ \+K_{t, \bs \eta(\bs \nu_{t-1})}\bs f_2 $, this tells us there exists constants $\hat b_t>0$ and $\hat \gamma_{t_z}>0$ such that:

\begin{equation}
    \E\left[ \left|\frac{\+x_{t-1}^\top}{n}\left[\bs f_1\otimes \+1_m \right] \circ \+K_{t, \bs \eta(\bs \nu_{t-1})}\bs f_2 - \bs \nu_{t-1} ^\top\left[\bs f_1\otimes \+1_m \right] \circ \+K_{t, \bs \eta(\bs \nu_{t-1})}\bs f_2 \right|^4\right]^\frac{1}{4}\leq \hat{b}_t n^{-(\frac{1}{4}+ \hat{\gamma}_{t_z})}.
\end{equation}
We now recall that $\+Z^{(j,k)}_t = \sum_{i=1}^n\mathbb{I}\{\xi_{t-1}^{(i)} = j, \xi_{t}^{(i)} = k \}$, so that the term \eqref{Zp1} is equal to:

\begin{equation}
 \frac{1}{n}\sum_{i=1}^{n} \underbrace{\left[ \sum_{j=1}^m\sum_{k=1}^m \left( \mathbb{I}\{\xi_{t-1}^{(i)} = j, \xi_{t}^{(i)} = k \} -  \mathbb{I}\{\xi_{t-1}^{(i)} = j\} K^{(j,k)}_{t,\bs \eta(\+x_{t-1})} \right) f_1^{(j)}f_2^{(k)} \right]}_{=:\Delta_t^{(i)}}.
\end{equation}
Since, conditioned on $\mathcal{G}_{t-1}\coloneqq \sigma(\{\xi_{t-1}^{(i)} \}_{i=1,\dots, n_t})$, $\xi_t^{(i)}$ is a draw from the $\xi_{t-1}^{(i)}$th row of $ \+ K_{t, \bs \eta(\+x_{t-1})} $, we have ${\E\left[ \Delta_t^{(i)} \mid \mathcal{G}_{t-1} \right] = 0}$. Furthermore, given $\mathcal{G}_{t-1}$ the $\Delta_t^{(i)}$ are independent and $\left| \Delta_t^{(i)} \right|\leq m^2\|\bs f\|_\infty^2$. An application of lemma \ref{MZlemma} yields that for some constants $\tilde b_t>0$ and $\tilde \gamma_{t_z}>0$:

\begin{equation}
\E\left[ \left|n^{-1} \bs f_1^\top \+Z_t \bs f_2 -  \bs f_1^\top\left[\frac{\+x_{t-1}}{n} \otimes \+1_m \right] \circ \+K_{t, \bs \eta(\+ x_{t-1})}\bs f_2\right|^4\right]^\frac{1}{4} \leq\tilde{b}_t n^{-(\frac{1}{4}+ \tilde{\gamma}_{t_z})}.
\end{equation}
Finally, use of the Minkowski inequality yields the result:

\begin{equation}
\E\left[\left|n^{-1} \bs f_1^\top \+Z_t \bs f_2 - \bs f_1^\top\left[\bs \nu_{t-1} \otimes \+1_m \right] \circ \+K_{t, \bs{\eta}(\bs \nu_t)}\bs f_2\right|^4\right]^\frac{1}{4} \leq {b}_t n^{-(\frac{1}{4}+ {\gamma}_{t_z})},
\end{equation}
where $b_t = \bar b_t+ \hat b_t+ \tilde b_t$ and $\gamma_{t_z}=\min(\bar \gamma_{t_z},\hat \gamma_{t_z}, \tilde \gamma_{t_z})$.

\end{proof}
\begin{lemma} \label{Ylp}
Let assumptions \ref{as:params} - \ref{as:init} hold. For all $t \geq 1$ there exists a $\bar{\gamma}_Y>0$, and for all vectors $\bs f_1, \bs f_2 \in \mathbb{R}^m$ a constant $c_Y>0$, such that:
\begin{equation*}
\E\left[\left| n^{-1} \bs f_1^\top \+Y_t \bs f_2 - \bs f_1^\top \left[\+ N_t(\bs \theta^*) \circ \+ Q_t(\bs \theta^*)\right]\bs f_2 \right| ^4\right]^\frac{1}{4} \leq c_Yn^{-(\frac{1}{4} + \bar{\gamma}_Y)}.
\end{equation*}
\end{lemma}

\begin{proof}
Write

\begin{align}
| n^{-1} \bs f_1^\top \+Y_t \bs f_2 &- \bs f_1^\top \left[\left(\bs \nu_{t-1} \otimes \+1_m \right) \circ \+K_{t, \bs \nu_{t-1}}\circ \+ Q_t\right]\bs f_2 | \\\label{Yp1}
&\leq \left| n^{-1} \bs f_1^\top \+Y_t \bs f_2 - n^{-1}\bs f_1^\top  \+Z_t\circ \+ Q_t\bs f_2 \right| \\ \label{Yp2}
&+\left| n^{-1}\bs f_1^\top  \+Z_t\circ \+ Q_t\bs f_2  - \bs f_1^\top \left[\left(\bs \nu_{t-1} \otimes \+1_m \right) \circ \+K_{t, \bs \nu_{t-1}}\circ \+ Q_t\right]\bs f_2 \right|.
\end{align}
By lemma \ref{lem:llnZ} there exists $a_Y>0$ and $\gamma_{Y_1}>0$ such that:

\begin{equation}
\E\left[\left| n^{-1}\bs f_1^\top  \+Z_t\circ \+ Q_t\bs f_2  - \bs f_1^\top \left[\left(\bs \nu_{t-1} \otimes \+1_m \right) \circ \+K_{t, \bs \nu_{t-1}}\circ \+ Q_t\right]\bs f_2 \right|^4\right]^\frac{1}{4} \leq a_Yn^{-(\frac{1}{4} + \gamma_{Y_1})}
\end{equation}
Now, we can write \eqref{Yp2} as:

\begin{equation*}
n^{-1}\sum_{i=1}^{n} \underbrace{\sum_{j=1}^m \sum_{k=1}^m \left[ \mathbb{I}\{\xi_{t-1}^{(i)} = j, \xi_{t}^{(i)} = k \}\mathbb{I}\{\zeta^{(i)} = 1\} - \mathbb{I}\{\xi_{t-1}^{(i)} = j, \xi_{t}^{(i)} = k \}Q^{(j,k)} \right]f_1^{(j)}f_2^{(k)}}_{=: \Xi_t^{(i)}}.
\end{equation*}
Where $\zeta^{(i)}$ given $\mathcal{G}_{t-1} \vee \mathcal{G}_t$ (where $\mathcal{G}_t$ is defined as in lemma \ref{lem:llnx}) is distributed $\text{Bernoulli}(Q^{(\xi_{t-1}^{(i)}, \xi_{t}^{(i)})})$. Hence, ${\E\left[ \Xi^{(i)}_t \mid \mathcal{G}_{t-1}\vee \mathcal{G}_t \right] = 0}$. Furthermore, given $\mathcal{G}_{t-1} \vee \mathcal{G}_t$ the $\Xi_t^{(i)} $ are independent and $\left| \Xi_t^{(i)} \right| < m^2\|\bs f\|_\infty^2$. An application of lemma \ref{MZlemma} yields:

\begin{equation}
\E\left[\left| n^{-1} \bs f_1^\top \+Y_t \bs f_2 - n^{-1}\bs f_1^\top  \+Z_t\circ \+ Q_t\bs f_2 \right|^4\right]^\frac{1}{4} \leq b_Yn^{-(\frac{1}{4} + \gamma_{Y_2})},
\end{equation}
for some $b_y>0$ and $\gamma_{Y_2}>0$. Use of the Minkowski inequality yields the result:

\begin{equation}
\E\left[\left| n^{-1} \bs f_1^\top \+Y_t \bs f_2 - \bs f_1^\top \left[\left(\bs \nu_{t-1} \otimes \+1_m \right) \circ \+K_{t, \bs \nu_{t-1}}\circ \+ Q_t\right]\bs f_2 \right| ^4\right]^\frac{1}{4} \leq (a_Y + b_Y)n^{-(\frac{1}{4} + \bar{\gamma}_Y)},
\end{equation}
for $\bar{\gamma}_Y = \min(\gamma_{Y_1}, \gamma_{Y_2}).$

\end{proof}
\begin{proposition}\label{prop:Yas}
Let assumptions \ref{as:params} - \ref{as:init} hold. Then for all $t\geq 1$:
\begin{equation*}
 n^{-1} \+Y_t  \asls \+ N_t(\bs \theta^*)\circ \+Q_t(\bs \theta^*),
\end{equation*}
and for all $r\geq 1$,
\begin{equation}\label{eq:Y_bar_lln}
    n^{-1} \bar{\+Y}_r  \asls \sum_{t = \tau_{r-1}+1}^{\tau_r}\+ N_t(\bs \theta^*)\circ \+Q_t(\bs \theta^*).
\end{equation}
\end{proposition}
\begin{proof}
We have that by lemma \ref{Ylp} for all $t$ there exists $c>0$ and $\gamma>0$ such that:
\begin{align*}
\mathbb{P}^{\bs\theta^*}(\big| n^{-1} \bs f_1^\top \+Y_t \bs f_2 &- \bs f_1^\top \left[\left(\bs \nu_{t-1} \otimes \+1_m \right) \circ \+K_{t, \bs \nu_{t-1}}\circ \+ Q_t\right]\bs f_2 \big|> \varepsilon) \\
&\leq \varepsilon^{-4}\E\left[\left| n^{-1} \bs f_1^\top \+Y_t \bs f_2 - n^{-1}\bs f_1^\top  \+Z_t\circ \+ Q_t\bs f_2 \right|^4\right] \\
&\leq \varepsilon^{-4}cn^{-(1+\gamma)}.
\end{align*}
It follows that:
\begin{equation*}
\sum_{n=1}^\infty \mathbb{P}^{\bs\theta^*}\left(\left| n^{-1} \bs f_1^\top \+Y_t \bs f_2 - \bs f_1^\top \left[\left(\bs \nu_{t-1} \otimes \+1_m \right) \circ \+K_{t, \bs \nu_{t-1}}\circ \+ Q_t\right]\bs f_2 \right|> \varepsilon\right)< \infty.
\end{equation*}
This result along with a Borel-Cantelli argument, as in proposition \ref{prop:yas}, completes the proof of the first claim of the proposition. The second claim follows from the first since $\bar{\+Y}_r = \sum_{t =\tau_{r-1}+1}^{\tau_r}\+Y_t$. 
\end{proof}
%

\subsection{Filtering intensity limits}\label{sec:filtering_limits}

\subsubsection{Case (I)}\label{sec:filtering_limits_case_I}

 Define the vectors, or $t\geq1$:


\begin{align}
\bar{\bs \lambda}_{0, \infty}(\bs \theta^*, \bs \theta) &\coloneqq \bs \lambda_{0, \infty}(\bs \theta),\label{eq:lam_0_inf_defn}\\
\bs \lambda_{t, \infty}(\bs \theta^*, \bs \theta)&\coloneqq \left[(\bar{\bs \lambda}_{t-1, \infty}(\bs \theta^*, \bs \theta)\circ \bs \delta_t(\bs \theta))^\top\+ K_{t, \bs \eta(\bar{\bs \lambda}_{t-1, \infty}\left(\bs \theta^*, \bs \theta)\circ \bs \delta_t(\bs \theta)\right)}(\bs \theta)\right]^\top + \bs \alpha_{t, \infty}(\bs \theta),\nonumber\\
    \bs \mu_{t, \infty}(\bs \theta^*, \bs \theta) &\coloneqq \left[\left(\bs \lambda_{t, \infty}(\bs \theta^*, \bs \theta)\circ \+q_t(\bs \theta) \right)^\top \+G_t(\bs \theta)\right]^\top +\bs \kappa_{t,\infty}(\bs \theta),\nonumber\\
     \bar{\bs \lambda}_{t, \infty}(\bs \theta^*, \bs \theta) &\coloneqq \bigg[\+ 1_m - \+q_t(\bs \theta)\nonumber\\
 &+\left(\left[ \bs \mu_{t, \infty}(\bs \theta^*, \bs \theta^*)\oslash\bs \mu_{t, \infty}(\bs \theta^*, \bs \theta)\right]^\top\left(\left[\+1_m\otimes\+q_t(\bs \theta)\right]\circ\+G_t(\bs \theta)^\top\right) \right)^\top\bigg]\circ \bs \lambda_{t, \infty}(\bs \theta^*,\bs \theta),\label{eq:lam_bar_t_inf_defn}
\end{align}
where by convention, if we encounter $0/0$ in the element-wise division operation we replace that ratio by $0$.

Our main objective in section \ref{sec:filtering_limits_case_I} is to show  these vectors are the $\mathbb{P}^{\bs \theta^*}$-a.s. limits of the corresponding finite-$n$ quantities evaluated at $\bs \theta$, computed using algorithm \ref{alg:x}. This is the subject of proposition \ref{vecfilt}.

\begin{proposition}\label{vecfilt}
Let assumptions \ref{as:params} - \ref{as:init} hold.  Then for all $\bs\theta \in \Theta$ and $t\geq1$:
\begin{align*}
n^{-1}\bs \mu_{t,n}( \bs \theta) &\asls \bs \mu_{t, \infty}(\bs \theta^*, \bs \theta),\\
n^{-1}\bs \lambda_{t,n}( \bs \theta) &\asls  \bs\lambda_{t, \infty}(\bs \theta^*, \bs \theta),\\
n^{-1} \bar{\bs\lambda}_{t,n}( \bs \theta) &\asls \bar{\bs\lambda}_{t,\infty}(\bs \theta^*, \bs \theta).
\end{align*}
\end{proposition}
\noindent The proof is postponed until later in section \ref{sec:filtering_limits_case_I}.

\begin{remark}\label{rem:lam_mu}
By writing out the above definitions it can be checked that $\bs\nu_t(\bs\theta^*)=\bs\lambda_{t, \infty}(\bs \theta^*, \bs \theta^*)$, hence lemma \ref{lem:llnx} implies by a Borel-Cantelli argument  $n^{-1}\+x_t\asls \bs\lambda_{t, \infty}(\bs \theta^*, \bs \theta^*)$; and that $ \bs \mu_{t, \infty}(\bs \theta^*, \bs \theta^*)$ is equal to the right hand side of \eqref{eq:y_lln} in  proposition \ref{prop:yas}, hence   $n^{-1}\+y_t \asls  \bs \mu_{t, \infty}(\bs \theta^*, \bs \theta^*)$. Therefore proposition \ref{vecfilt} implies that if algorithm \ref{alg:x} is run with the model specified by the DGP $\bs\theta^*$, thus computing $\bs \lambda_{t,n}( \bs \theta^*)$ and  $\bs \mu_{t,n}( \bs \theta^*) $, that when rescaled by $n^{-1}$ these vectors converge as $n\to\infty$ to the same $\mathbb{P}^{\bs\theta^*}$-almost sure limits as $n^{-1}\+x_t$ and  $n^{-1}\+y_t $. We provide empirical evidence for this remark in section \ref{sec:interpretation}.
\end{remark}

As preliminaries to the proof of proposition \ref{vecfilt} we need to verify that certain quantities in algorithm \ref{alg:x} and the vectors defined at the start of section \ref{sec:filtering_limits_case_I} are $\mathbb{P}^{\bs \theta^*}$-a.s. well-defined and finite. This is the purpose of lemma \ref{muylemma} and lemma \ref{mumulemma}.  In algorithm \ref{alg:x}, if $\bs\mu^{(i)}_{t,n}(\bs \theta) = 0 $ and $y_t^{(i)} > 0$, then line 3 would entail dividing a finite number by zero. Lemma \ref{muylemma} establishes that this happens with probability zero.
\begin{lemma}\label{muylemma}
Let assumptions \ref{as:params}-\ref{as:init} hold. For any $\bs \theta \in \Theta$, $n \in \mathbb{N}$, $i \in [m]$, and  $t \geq 1$,
$$\Ps\left(\mu_{t,n}^{(i)}( \bs \theta) = 0\right)>0 \implies y_t^{(i)} = 0, \quad \Ps\text{-}a.s.$$
\end{lemma}
\begin{proof}
Fix arbitrary $\bs \theta\in \Theta$ and $n \in \mathbb{N}$. All a.s. statements in the proof are with respect to $\mathbb{P}^{\bs \theta^*}_n$. We will show that for all $j \in[m]$ and $t\geq1$, the following two implications hold:
\begin{align}
\Ps\left(\lambda_{t,n}^{(j)}(\bs \theta) = 0\right)>0 & \implies x_t^{(j)} = 0,\quad a.s.,\label{eq:lamdbda_mu_induct1}\\
\Ps\left(\mu_{t,n}^{(i)}( \bs \theta) = 0\right)>0 &\implies y_t^{(i)} = 0, \quad a.s.\label{eq:lamdbda_mu_induct2}
\end{align}
The proof is inductive in $t$. To initialize the induction at $t=1$, let $j \in [m]$ and suppose that $\Ps(\lambda^{(j)}_{1,n}(\bs \theta) =0)>0$, i.e., 
\begin{equation*}
 \Ps\left( \sum_{k=1}^m \bar \lambda^{(k)}_{0,n}(\bs \theta)\delta_1^{(k)}(\bs \theta)K^{(k,j)}_{1, \bs \eta(\bar{\bs \lambda}_{0,n}(\bs \theta)\circ \bs \delta_1(\bs \theta))}(\bs \theta) + \alpha_{1,n}^{(j)}(\bs \theta) = 0\right) > 0,
\end{equation*}
then $\alpha_{1,n}^{(j)}(\bs \theta)=0$ which by assumption \ref{as:params} implies $\alpha_{1,n}^{(j)}(\bs \theta^*) = 0$ and hence $\hat{x}^{(j)}_1=0$, a.s. Furthermore, for all $k\in[m]$ we must have that either:
\begin{itemize}
    \item $\Ps\left(\lambda^{(k)}_{0,n}(\bs \theta) =0\right)>0$, which, since $\lambda^{(k)}_{0,n}(\bs \theta)$ is a deterministic quantity, implies $\lambda^{(k)}_{0,n}(\bs \theta)=0$, in turn by assumption \ref{as:init} this implies $\lambda^{(k)}_{0,n}(\bs \theta^*)=0$ so that $x_0^{(k)} = 0$ a.s. and $\bar x_0^{(k)} = 0$ a.s.; or
    \item $\delta_1^{(k)}(\bs \theta)= 0$,  which by assumption \ref{as:params} implies $\delta_1^{(k)}(\bs \theta^*)= 0$ which means $\bar x_0^{(k)} = 0$ a.s.; or
    \item $K^{(k,j)}_{1, \bs \eta(\bar{\bs \lambda}_{0,n}(\bs \theta)\circ \bs \delta_1(\bs \theta))}(\bs \theta)=0$, which by assumptions \ref{as:params}, \ref{as:1}, and \ref{as:init} implies $K^{(k,j)}_{1, \bs \eta(\bar{\+ x}_0)}(\bs \theta^*)=0$ a.s.
\end{itemize}
Hence we have for all $k\in[m]$ either $\bar x_0^{(k)} = 0$ a.s. or $K^{(k,j)}_{1, \bs \eta(\bar{\+ x}_0)}(\bs \theta^*)=0$ a.s. Since, given $\+x_0$, $\tilde{x}_1^{(j)}\sim\sum_{k=1}^m\text{Bin}\left(x_0^{(k)},K^{(k,j)}_{1, \bs \eta(\bar{\+ x}_0)}(\bs \theta^*)\right)$  we must have that $\tilde{x}_1^{(j)} = 0$ a.s., therefore we have that $x_1^{(j)} = \tilde x_1^{(j)} + \hat x_1^{(j)} = 0$ a.s. We have thus proved \eqref{eq:lamdbda_mu_induct1} in the case $t=1$. 

Now let us prove \eqref{eq:lamdbda_mu_induct2} in the case $t=1$. Suppose that for some $i \in [m]$, $\Ps\left(\mu^{(i)}_{1,n}(\bs \theta)= 0\right)>0$, i.e., 
\begin{equation*}
\Ps\left(\sum_{j=1}^m \lambda^{(j)}_{1,n}(\bs \theta)q_1^{(j)}(\bs \theta)G_1^{(j,i)}(\bs \theta) + \kappa^{(i)}_{1,n}(\bs \theta) = 0\right)>0.
\end{equation*}
Then $\kappa_{1,n}^{(i)}(\bs \theta)=0$ which by assumption \ref{as:params} implies $\kappa_{1,n}^{(i)}(\bs \theta^*) = 0$ and hence $\hat{y}^{(i)}_1=0$ a.s.  Furthermore, for all $j\in[m]$ we must have that either:
\begin{itemize}
    \item $\Ps\left(\lambda^{(j)}_{1,n}(\bs \theta) = 0\right)>0$, which  implies $x_1^{(j)} = 0$ a.s.  which implies $  \bar y_1^{(j)}=0$ a.s.; or
    \item $q_1^{(j)}(\bs \theta)= 0 $, which by assumption \ref{as:params} implies that $ q_1^{(j)}(\bs \theta^*)= 0 \implies \bar{y}_1^{(j)} = 0$ a.s.; or
    \item $G_1^{(j,i)}(\bs \theta)=0 $, which by assumption \ref{as:params} implies $G_1^{(j,i)}(\bs \theta^*)=0$.
\end{itemize}
Given, $\+ {\bar{y}}_1$ , $\tilde{y}_1^{(i)} \sim\sum_{j=1}^m \text{Bin}\left( \bar y_1^{(j)}, G_1^{(j,i)}(\bs \theta^*)\right)$. This means that $\tilde{y}_1^{(i)} = 0$  a.s., and furthermore that $y_1^{(i)} = \tilde{y}_1^{(i)} + \hat{y}_1^{(i)} = 0$ a.s. This completes the proof of \eqref{eq:lamdbda_mu_induct2} in the case $t=1$.

As an induction hypothesis suppose that \eqref{eq:lamdbda_mu_induct1} and \eqref{eq:lamdbda_mu_induct2} hold at $t$. We shall  show that $\Ps\left(\lambda^{(j)}_{t+1,n}(\bs \theta) = 0\right)>0 \implies x_{t+1}^{(j)} = 0$ a.s. Firstly we will show that, for all $k \in[m]$, $ \Ps\left(\bar\lambda^{(k)}_{t,n}(\bs \theta) = 0\right)>0 \implies \Ps\left(\lambda^{(k)}_{t,n}(\bs \theta) = 0\right)>0$ which, by the induction hypothesis, would imply $x^{(k)}_t = 0$ a.s. Suppose that for some $k \in[m]$ , $\Ps\left(\bar \lambda^{(k)}_{t,n}(\bs \theta) =0 \right)>0$, i.e., 
\begin{equation*}
\Ps\left((1 - q_t^{(k)}(\bs \theta)) \lambda^{(k)}_{t,n}(\bs \theta) + \sum_{j=1}^m y_t^{(j)}\frac{\lambda^{(k)}_{t,n}(\bs \theta)q_t^{(k)}(\bs \theta)G_t^{(k,j)}(\bs \theta)}{\mu^{(j)}_{t,n}(\bs \theta)} = 0\right)>0.
\end{equation*}
Firstly, $\bar \lambda^{(k)}_{t,n}(\bs \theta)$ is almost surely well defined by the induction hypothesis, since the event $\mu^{(j)}_{t,n}(\bs \theta) = 0$ and $y_t^{(j)} >0$ has probability $0$ for each $j \in[m]$. Now if the above displayed inequality holds we must have that either:
\begin{itemize}
    \item $q_t^{(k)}(\bs \theta)<1$, in which case we must have $\Ps\left(\lambda^{(k)}_{t,n}(\bs \theta) = 0\right)>0$; or
    \item $q_t^{(k)}(\bs \theta)=1$, in which case we must have $\Ps\left(\lambda^{(k)}_{t,n}(\bs \theta)G_t^{(k,j)}(\bs \theta) = 0\right) > 0$  for all $j$ so that the sum is equal to $0$ with positive probability, and  since $\+G_t$ is row-stochastic matrix,  there must exist a $j \in[m]$ such that $G_t^{(k,j)}(\bs \theta) > 0 $, hence $\Ps\left(\lambda^{(k)}_{t,n}(\bs \theta) = 0\right) > 0$.
\end{itemize}
We have thus shown $\Ps\left(\bar \lambda^{(k)}_{t,n}(\bs \theta) = 0\right)>0 \implies \Ps\left(\lambda^{(k)}_{t,n}(\bs \theta) = 0\right)>0$ which by the induction hypothesis implies $x_t^{(k)}=0$ a.s. so that further $\bar x_t^{(k)} = 0$ a.s. Now if for some $j \in [m]$, $\Ps\left(\lambda^{(j)}_{t+1,n}(\bs \theta) = 0\right) >0$, i.e.,
\begin{equation*}
\Ps\left(\sum_{k=1}^m \bar \lambda^{(k)}_{t,n}(\bs \theta)\delta_{t+1}^{(k)}(\bs \theta)K^{(k,j)}_{t+1, \bs \eta(\bar{\bs \lambda}_{t,n}(\bs \theta)\circ \bs \delta_{t+1}(\bs \theta))}(\bs \theta) + \alpha_{t+1,n}^{(j)}(\bs \theta) = 0\right)>0,
\end{equation*}
then $\alpha_{t+1,n}^{(j)}(\bs \theta)=0$, which by assumption \ref{as:params} implies $\alpha_{t+1,n}^{(j)}(\bs \theta^*) = 0$, hence $ \hat x_{t}^{(j)} = 0$ a.s. Furthermore, for all $k\in[m]$ we must have that either:
\begin{itemize}
    \item $\Ps\left(\bar \lambda^{(k)}_{t,n}(\bs \theta) = 0\right)>0$, which implies $x_{t}^{(k)} = 0$ a.s. $ \implies \bar x_{t}^{(k)} = 0$ a.s.; or
    \item $\delta_{t+1}^{(k)}(\bs \theta)= 0$, which by assumption \ref{as:params} implies $\delta_{t+1}^{(k)}(\bs \theta^*)= 0 \implies \bar x_{t}^{(k)} = 0$ a.s.; or
    \item $\Ps\left(K^{(k,j)}_{t+1, \bs \eta(\bar{\bs \lambda}_{t,n}(\bs \theta)\circ \bs \delta_{t+1}(\bs \theta))}(\bs \theta)=0\right)>0$. We claim this implies that $K^{(k,j)}_{t+1, \bs \eta(\bar{\+ x}_{t})}(\bs \theta)=0$, a.s. Suppose, for contradiction, that $\Ps\left(K^{(k,j)}_{t+1, \bs \eta(\bar{\bs \lambda}_{t,n}(\bs \theta)\circ \bs \delta_{t+1}(\bs \theta))}(\bs \theta)=0\right)>0$ and $\Ps\left(K^{(k,j)}_{t+1, \bs \eta(\bar{\+ x}_{t})}(\bs \theta)>0\right)>0$. Then there exist $E,E' \subseteq \Omega_n$ with $\Ps(E)>0$ and $\Ps(E')>0$ such that for all $\omega \in E$ and all $\omega' \in E'$:
    $$ K^{(k,j)}_{t+1, \bs \eta(\bar{\bs \lambda}_{t,n}(\bs \theta, \omega)\circ \bs \delta_{t+1}(\bs \theta))}(\bs \theta)=0 \text{ and } K^{(k,j)}_{t+1, \bs \eta(\bar{\+ x}_{t}(\omega'))}(\bs \theta)>0,$$
    which implies:
    $$ \text{supp}\left(\+K^{(k,\cdot)}_{t+1, \bs \eta(\bar{\bs \lambda}_{t,n}(\bs \theta, \omega)\circ \bs \delta_{t+1}(\bs \theta))}(\bs \theta)\right) \not \subseteq \text{supp}\left(\+K^{(k,\cdot)}_{t+1, \bs \eta(\bar{\+ x}_{t}(\omega'))}(\bs \theta) \right).$$
By assumption \ref{as:1} this implies:
$$ \text{supp}\left(\bar{\+ x}_{t}(\omega')\right) \not \subseteq \text{supp}\left(\bar{\bs \lambda}_{t,n}(\bs \theta, \omega)\circ \bs \delta_{t+1}(\bs \theta)\right),$$
i.e. there exists $l$ such that:
$$ \left(\bar{\bs \lambda}_{t,n}(\bs \theta, \omega)\circ \bs \delta_{t+1}(\bs \theta)\right)^{(l)} = 0 \text{ and } \left(\bar{\+ x}_{t}(\omega')\right)^{(l)} >0.$$
But since $\Ps(E)>0$ and $\Ps(E')>0$ this implies that:
$$\Ps(\bar{\lambda}^{(l)}_{t,n}(\bs \theta) = 0)>0 \text{ and } \Ps(\bar{x}^{(l)}_{t} > 0)>0.$$
This contradicts the observation in the first bullet point, hence $K^{(k,j)}_{t+1, \bs \eta(\bar{\+ x}_{t})}(\bs \theta)=0$ a.s. Then by assumption \ref{as:params} we have ${K^{(k,j)}_{t+1, \bs \eta(\bar{\+ x}_{t})}(\bs \theta^*)=0}$ a.s.
\end{itemize}
Hence, similarly to the argument used in the case $t=1$, we must have that $\tilde{x}_{t+1}^{(j)} =0$ a.s. so that $x_{t+1}^{(j)} = \tilde x_{t+1}^{(j)} + \hat x_{t+1}^{(j)} = 0$ a.s. Thus 
\eqref{eq:lamdbda_mu_induct1} holds with $t$ replace by $t+1$. 

It remains to show that \eqref{eq:lamdbda_mu_induct2} holds with $t$ replaced by $t+1$. So suppose that for some $i \in [m]$, $\Ps\left(\mu^{(i)}_{t+1,n}(\bs \theta)=0\right)>0$, i.e.,
\begin{equation*}
  \Ps\left(\sum_{j=1}^m \lambda^{(j)}_{t+1,n}(\bs \theta)q_{t+1}^{(j)}(\bs \theta)G_{t+1}^{(j,i)}(\bs \theta) + \kappa^{(i)}_{t+1,n}(\bs \theta) = 0\right)>0,
\end{equation*}
then we must have $\kappa^{(i)}_{t+1,n}(\bs \theta)=0$ which by assumption \ref{as:params} implies $\kappa^{(i)}_{t+1,n}(\bs \theta^*) = 0 $ hence $\hat{y}^{(i)}_{t+1} = 0$ a.s. Furthermore, for all $j\in[m]$ we must have either:
\begin{itemize}
    \item $\Ps \left(\lambda^{(j)}_{t+1,n}(\bs \theta) = 0\right)>0$, which implies that
    $x^{(j)}_{t+1} = 0 \implies \bar{y}_{t+1}^{(j)} = 0$ a.s.; or
    \item $q_{t+1}^{(j)}(\bs \theta)= 0 $, which by assumption \ref{as:params} implies $ q_{t+1}^{(j)}(\bs \theta^*)= 0 \implies \bar{y}_{t+1}^{(j)} = 0$ a.s.; or
    \item $G_{t+1}^{(j,i)}(\bs \theta)=0$,  which by assumption \ref{as:params} implies $G_{t+1}^{(j,i)}(\bs \theta^*)=0$.
\end{itemize}
Hence, using the same reasoning as in the $t=1$ case, we have $\tilde{y}_{t+1}^{(i)} = 0$  a.s. and furthermore ${y}_{t+1}^{(i)} = \tilde{y}_{t+1}^{(i)} + \hat{y}_{t+1}^{(i)} = 0$  a.s. This completes the proof of \eqref{eq:lamdbda_mu_induct2} with $t$ replaced by $t+1$. The induction is therefore complete.
\end{proof}

If $\mu^{(i)}_{t,\infty}(\bs \theta^*, \bs \theta) = 0$ and $\mu^{(i)}_{t, \infty}(\bs \theta^* , \bs \theta^*) > 0$ then $\bar{\bs \lambda}_{t, \infty}(\bs \theta^*, \bs \theta)$ would involve division of a finite number by zero. Lemma \ref{mumulemma} establishes that this situation cannot arise.
\begin{lemma}\label{mumulemma}
Let assumptions \ref{as:params} - \ref{as:init} hold. For any $\bs \theta,\bs\theta' \in \Theta$, $i \in[m]$ and $t\geq 1$,

\begin{align}
\bar \lambda^{(j)}_{t,\infty}(\bs \theta^*, \bs \theta) = 0 &\implies \bar \lambda^{(j)}_{t, \infty}(\bs \theta^*, \bs \theta') = 0,\\
\mu^{(i)}_{t,\infty}(\bs \theta^*, \bs \theta) = 0& \implies \mu^{(i)}_{t, \infty}(\bs \theta^*, \bs \theta') = 0.
\end{align}
\end{lemma}
\begin{proof}
Fix arbitrary $\bs \theta,\bs\theta' \in \Theta$. By symmetry we only need to prove the implication in one direction. We will show that the following two implications hold for all $i,j \in [m] $ and $t \geq 1$:
\begin{align}
\lambda^{(j)}_{t,\infty}(\bs \theta^*, \bs \theta) = 0 &\implies \lambda^{(j)}_{t, \infty}(\bs \theta^*, \bs \theta') = 0,\label{eq:mu_mu_induct_1}\\
\mu^{(i)}_{t,\infty}(\bs \theta^*, \bs \theta) = 0& \implies \mu^{(i)}_{t, \infty}(\bs \theta^*, \bs \theta') = 0.\label{eq:mu_mu_induct_2}
\end{align}
For the $t=1$ case, if $    \lambda^{(j)}_{1,\infty}(\bs \theta^*, \bs \theta)=0$, i.e., 
\begin{equation*}
 \sum_{k=1}^m\lambda^{(k)}_{0,\infty}(\bs \theta^*, \bs \theta)\delta_1^{(k)}(\bs \theta)K^{(k,j)}_{1, \bs \eta(\bs \lambda_{0,\infty}(\bs \theta^*, \bs \theta)\circ \bs \delta(\bs \theta))}(\bs \theta) + \alpha^{(j)}_{1,\infty}(\bs \theta) = 0,
\end{equation*}
then $\alpha^{(j)}_{1,\infty}(\bs \theta) = 0$ which by assumption \ref{as:params} implies $\alpha^{(j)}_{1,\infty}(\bs \theta') = 0$. Furthermore, for each $j\in [m]$ we must have either:
\begin{itemize}
    \item $\lambda^{(k)}_{0,\infty}(\bs \theta^*, \bs \theta) = \lambda^{(k)}_{0,\infty}(\bs \theta) = 0 $, which by assumption \ref{as:init} implies $ \lambda^{(k)}_{0,\infty}(\bs \theta') = 0$; or
    \item $\delta_1^{(k)}(\bs \theta) = 0$, which by assumption \ref{as:params} implies $\delta_1^{(k)}(\bs \theta')=0$; or
    \item $K^{(k,j)}_{1, \bs \eta(\bs \lambda_{0,\infty}(\bs \theta^*, \bs \theta)\circ \bs \delta(\bs \theta))}(\bs \theta) = 0$, which by assumptions \ref{as:params}, \ref{as:1}, and \ref{as:init} implies \newline $K^{(k,j)}_{1, \bs \eta(\bs \lambda_{0,\infty}(\bs \theta^*, \bs \theta')\circ \bs \delta(\bs \theta'))}(\bs \theta') = 0 $.
\end{itemize}
Hence we have:
\begin{equation*}
    \lambda^{(j)}_{1,\infty}(\bs \theta^*, \bs \theta') = \sum_{k=1}^m\lambda^{(k)}_{0,\infty}(\bs \theta^*, \bs \theta')\delta_1^{(k)}(\bs \theta')K^{(k,j)}_{1, \bs \eta(\bs \lambda_{0,\infty}(\bs \theta^*, \bs \theta')\circ \bs \delta_1(\bs \theta'))}(\bs \theta') + \alpha^{(j)}_{1,\infty}(\bs \theta') = 0,
\end{equation*}
so \eqref{eq:mu_mu_induct_1} holds with $t=1$. In order to establish \eqref{eq:mu_mu_induct_2} with $t=1$, consider
\begin{equation*}
    \mu^{(i)}_{1,\infty}(\bs \theta^*, \bs \theta) = \sum_{j=1}^m\lambda^{(j)}_{1,\infty}(\bs \theta^*, \bs \theta)q_1^{(j)}(\bs \theta)G_1^{(i,j)}(\bs \theta) + \kappa^{(i)}_{1,\infty}(\bs \theta) = 0,
\end{equation*}
hence $\kappa^{(i)}_{1,\infty}(\bs \theta) = 0 $, which by assumption \ref{as:params} implies  $\kappa^{(i)}_{1,\infty}(\bs \theta') = 0$. Furthermore, for each $j \in [m]$ we must have either:
\begin{itemize}
    \item $\lambda^{(j)}_{1,\infty}(\bs \theta^*, \bs \theta) = 0 $, which by the above implies $ \lambda^{(j)}_{1,\infty}(\bs \theta^*, \bs \theta') = 0$; or
    \item $q_1^{(j)}(\bs \theta)= 0 $, which by assumption \ref{as:params} implies $ q_1^{(j)}(\bs \theta') = 0$; or
    \item $G_1^{(i,j)}(\bs \theta) = 0 $, which by assumption \ref{as:params} implies $G_1^{(i,j)}(\bs \theta') = 0$.
\end{itemize}
Hence:
\begin{equation*}
    \mu^{(i)}_{1,\infty}(\bs \theta^*, \bs \theta') = \sum_{j=1}^m\lambda^{(j)}_{1,\infty}(\bs \theta^*, \bs \theta')q_1^{(j)}(\bs \theta')G_1^{(i,j)}(\bs \theta') + \kappa^{(j)}_{1,\infty}(\bs \theta') = 0.
\end{equation*}
Thus we have shown that \eqref{eq:mu_mu_induct_2} holds with $t=1$. 

For the induction hypothesis, assume that \eqref{eq:mu_mu_induct_1} and \eqref{eq:mu_mu_induct_2} with hold for some $t\geq1$. Then for each $k\in [m]$ write:
\begin{equation*}
    \bar \lambda^{(k)}_{t,\infty}(\bs \theta^*, \bs \theta) = (1- q^{(k)}_t(\bs \theta))\lambda^{(k)}_{t,\infty}(\bs \theta^*, \bs \theta) + \sum_{j=1}^m \mu^{(i)}_{t,\infty}(\bs \theta^*, \bs \theta^*) \frac{\lambda^{(k)}_{t,\infty}(\bs \theta^*, \bs \theta)q^{(k)}_t(\bs \theta)G^{(k,j)}_t(\bs \theta) }{\mu^{(i)}_{t,\infty}(\bs \theta^*, \bs \theta)} = 0.
\end{equation*}
This is well defined by the induction hypothesis choosing $\bs \theta' = \bs \theta^*$. Furthermore, we must have either:
\begin{itemize}
    \item $q_t^{(k)}(\bs \theta)<1$, in which case we must have $\lambda^{(k)}_{t,\infty}(\bs \theta^*, \bs \theta) = 0 \implies \lambda^{(k)}_{t,\infty}(\bs \theta^*, \bs \theta') = 0$; or
    \item $q_t^{(k)}(\bs \theta)=1$, in which case we must have $\lambda^{(k)}_{t,\infty}(\bs \theta^*, \bs \theta)G_t^{(k,j)}(\bs \theta) = 0$ a.s.  for all $j$ so that the sum is equal to $0$. Since $\+G_t$ is row-stochastic matrix, we know there must exist a $j \in[m]$ such that $G_t^{(k,j)}(\bs \theta) > 0 $ and so we must have $\lambda^{(k)}_{t,\infty}(\bs \theta^*, \bs \theta) = 0\implies \lambda^{(k)}_{t,\infty}(\bs \theta^*, \bs \theta') = 0$.
\end{itemize}
So we have $\bar \lambda^{(k)}_{t,\infty}(\bs \theta^*, \bs \theta) = 0\implies \lambda^{(k)}_{t,\infty}(\bs \theta^*, \bs \theta') = 0$, indeed the reverse implication is also true by definition of $\bar \lambda^{(k)}_{t,\infty}(\bs \theta^*, \bs \theta)$ so that $\bar \lambda^{(k)}_{t,\infty}(\bs \theta^*, \bs \theta) = 0\iff \lambda^{(k)}_{t,\infty}(\bs \theta^*, \bs \theta') = 0$. Now consider
\begin{equation*}
    \lambda^{(j)}_{t+1,\infty}(\bs \theta^*, \bs \theta) = \sum_{k=1}^m\bar \lambda^{(k)}_{t,\infty}(\bs \theta^*, \bs \theta)\delta_{t+1}^{(k)}(\bs \theta)K^{(k,j)}_{1, \bs \eta(\bs{\bar \lambda}_{t,\infty}(\bs \theta^*, \bs \theta)\circ \bs \delta_{t+1}(\bs \theta))}(\bs \theta) + \alpha^{(j)}_{t+1,\infty}(\bs \theta) = 0,
\end{equation*}
then $\alpha^{(j)}_{t+1,\infty}(\bs \theta) = \alpha^{(j)}_{t+1,\infty}(\bs \theta') = 0$ and for all $k \in [m]$ we must have either:
\begin{itemize}
    \item $\bar \lambda^{(k)}_{t,\infty}(\bs \theta^*, \bs \theta) = 0 $, which implies by the above that $ \lambda^{(k)}_{t,\infty}(\bs \theta^*, \bs \theta) = 0 \implies \lambda^{(k)}_{t,\infty}(\bs \theta^*, \bs \theta') = 0 \implies \bar \lambda^{(k)}_{t,\infty}(\bs \theta^*, \bs \theta')$; or
    \item $\delta_{t+1}^{(k)}(\bs \theta)= 0$, which by assumption \ref{as:params} that $ \delta_{t+1}^{(k)}(\bs \theta') = 0$; or
    \item $K^{(k,j)}_{1, \bs \eta(\bs{\bar \lambda}_{t,\infty}(\bs \theta^*, \bs \theta)\circ \bs \delta_{t+1}(\bs \theta))}(\bs \theta) = 0 $ which by assumptions \ref{as:params}, \ref{as:1}, and the induction hypothesis implies $ K^{(k,j)}_{1, \bs \eta(\bs{\bar \lambda}_{t,\infty}(\bs \theta^*, \bs \theta')\circ \bs \delta_{t+1}(\bs \theta'))}(\bs \theta') = 0$.
\end{itemize}
Hence
\begin{equation*}
    \lambda^{(j)}_{t+1,\infty}(\bs \theta^*, \bs \theta') = \sum_{k=1}^m\bar \lambda^{(k)}_{t,\infty}(\bs \theta^*, \bs \theta')\delta_{t+1}^{(k)}(\bs \theta')K^{(k,j)}_{1, \bs \eta(\bs{\bar \lambda}_{t,\infty}(\bs \theta^*, \bs \theta')\circ \bs \delta_{t+1}(\bs \theta'))}(\bs \theta') + \alpha^{(j)}_{t+1,\infty}(\bs \theta') = 0.
\end{equation*}
Now, if for some $i \in[m]$:
\begin{equation*}
    \mu^{(i)}_{t+1,\infty}(\bs \theta^*, \bs \theta) = \sum_{j=1}^m\lambda^{(j)}_{t+1,\infty}(\bs \theta^*, \bs \theta)q_{t+1}^{(j)}(\bs \theta)G_{t+1}^{(j,i)}(\bs \theta) + \kappa^{(i)}_{t+1,\infty}(\bs \theta) = 0,
\end{equation*}
then $\kappa^{(i)}_{t+1,\infty}(\bs \theta) = 0$, which by assumption \ref{as:params} implies  $\kappa^{(i)}_{t+1,\infty}(\bs \theta') = 0$ and for all $j \in [m]$ we have either:

\begin{itemize}
    \item $\lambda^{(j)}_{t+1,\infty}(\bs \theta^*, \bs \theta) = 0$, which by the above implies $\lambda^{(j)}_{t+1,\infty}(\bs \theta^*, \bs \theta') = 0$; or
    \item $q_{t+1}^{(j)}(\bs \theta) = 0$, which by assumption \ref{as:params} implies  $q_{t+1}^{(j)}(\bs \theta') = 0$; or
    \item $G_{t+1}^{(j,i)}(\bs \theta) = 0$, which by assumption \ref{as:params} implies  $G_{t+1}^{(j,i)}(\bs \theta') = 0$.
\end{itemize}
Hence
\begin{equation}
    \mu^{(i)}_{t+1,\infty}(\bs \theta^*, \bs \theta') = \sum_{j=1}^m\lambda^{(j)}_{t+1,\infty}(\bs \theta^*, \bs \theta')q_{t+1}^{(j)}(\bs \theta')G_{t+1}^{(i,j)}(\bs \theta') + \kappa^{(i)}_{t+1,\infty}(\bs \theta') = 0,
\end{equation}
and the inductive proof is complete.
\end{proof}
The following lemma will be used in the proof of lemma \ref{contrastfnx}.
\begin{lemma}\label{mumunlemma}
Let assumptions \ref{as:params}- \ref{as:init} hold. For all $\bs \theta  \in \Theta$, $n \in \mathbb{N}$, and $i \in[m]$:
\begin{align*}
\lambda^{(j)}_{t,\infty}(\bs \theta^*, \bs \theta) = 0& \implies \lambda^{(j
)}_{t, n}(\bs \theta) = 0 \quad a.s.,\\
  \mu^{(i)}_{t,\infty}(\bs \theta^*, \bs \theta) = 0 &\implies \mu^{(i)}_{t, n}(\bs \theta) = 0 \quad a.s..
\end{align*}
\end{lemma}
\begin{proof}
Fix arbitrary $\bs \theta,  \in \Theta$ and $n \in \mathbb{N} $. All almost sure statements in the proof are made with respect to $\Ps$.  We will show by induction  that the following two implications hold for all $t \geq 1$ and $i,j\in[m]$.
\begin{align}
\lambda^{(j)}_{t,\infty}(\bs \theta^*, \bs \theta) = 0& \implies \lambda^{(j
)}_{t, n}(\bs \theta) = 0 \quad a.s., \label{eq:mu_inf_mu_n_induct1}\\
  \mu^{(i)}_{t,\infty}(\bs \theta^*, \bs \theta) = 0 &\implies \mu^{(i)}_{t, n}(\bs \theta) = 0  \quad a.s.\label{eq:mu_inf_mu_n_induct2}
\end{align}
For  $t=1$ consider:
\begin{equation*}
    \lambda^{(j)}_{1,\infty}(\bs \theta^*, \bs \theta) = \sum_{k=1}^m\bar \lambda^{(k)}_{0,\infty}(\bs \theta^*, \bs \theta)\delta_1^{(k)}(\bs \theta)K^{(k,j)}_{1, \bs \eta(\bs \lambda_{0,\infty}(\bs \theta^*, \bs \theta)\circ \bs \delta_1(\bs \theta))}(\bs \theta) + \alpha^{(j)}_{1,\infty}(\bs \theta) = 0,
\end{equation*}
then $\alpha^{(j)}_{1,\infty}(\bs \theta) = 0$ which by assumption \ref{as:params} implies $\alpha^{(j)}_{1,n}(\bs \theta) = 0$, and for all $j \in [m]$ we must have either:
\begin{itemize}
    \item $\bar \lambda^{(k)}_{0,\infty}(\bs \theta^*, \bs \theta) = \lambda^{(k)}_{0,\infty}(\bs \theta) = 0 $, which by assumption \ref{as:init} implies $ \lambda^{(k)}_{0,n}(\bs \theta) = 0$; or
    \item $\delta_1^{(k)}(\bs \theta) = 0$; or
    \item $K^{(k,j)}_{1, \bs \eta(\bs \lambda_{0,\infty}(\bs \theta^*, \bs \theta)\circ \bs \delta(\bs \theta))}(\bs \theta) = 0 $ which by assumptions \ref{as:1} and  \ref{as:init} implies $ K^{(k,j)}_{1, \bs \eta(\bs \lambda_{0,n}(\bs \theta^*, \bs \theta)\circ \bs \delta(\bs \theta))}(\bs \theta) = 0 $.
\end{itemize}
Hence we have:
\begin{equation*}
    \lambda^{(j)}_{1,n}( \bs \theta) = \sum_{k=1}^m\lambda^{(k)}_{0,n}(\bs \theta')\delta_1^{(k)}(\bs \theta')K^{(k,j)}_{1, \bs \eta(\bs \lambda_{0,n}( \bs \theta')\circ \bs \delta_1(\bs \theta'))}(\bs \theta') + \alpha^{(j)}_{1,n}(\bs \theta') = 0,\quad a.s.
\end{equation*}
Now consider:
\begin{equation*}
    \mu^{(i)}_{1,\infty}(\bs \theta^*, \bs \theta) = \sum_{j=1}^m\lambda^{(j)}_{1,\infty}(\bs \theta^*, \bs \theta)q_1^{(j)}(\bs \theta)G_1^{(i,j)}(\bs \theta) + \kappa^{(i)}_{1,\infty}(\bs \theta) = 0,
\end{equation*}
then $\kappa^{(i)}_{1,\infty}(\bs \theta) = 0$, which by assumption \ref{as:params} implies $\kappa^{(i)}_{1,n}(\bs \theta) = 0$, furthermore for al $j \in[m]$ we must have either:
\begin{itemize}
    \item $\lambda^{(j)}_{1,\infty}(\bs \theta^*, \bs \theta) = 0$, which by the above implies $\lambda^{(j)}_{1,n}( \bs \theta') = 0 \quad \mathbb{P}^{\bs \theta^*}a.s.$; or
    \item $q_1^{(j)}(\bs \theta)= 0$; or
    \item $G_1^{(i,j)}(\bs \theta) = 0$.
\end{itemize}
Hence:
\begin{equation*}
    \mu^{(i)}_{1,n}( \bs \theta) = \sum_{j=1}^m\lambda^{(j)}_{1,n}( \bs \theta)q_1^{(j)}(\bs \theta)G_1^{(k,j)}(\bs \theta) + \kappa^{(i)}_{1,n}(\bs \theta) = 0.
\end{equation*}

For the induction hypothesis, assume \eqref{eq:mu_inf_mu_n_induct1} and \eqref{eq:mu_inf_mu_n_induct2} hold. Then for each $k \in [m]$,
\begin{equation*}
    \bar \lambda^{(k)}_{t,\infty}(\bs \theta^*, \bs \theta) = (1- q^{(k)}_t(\bs \theta))\lambda^{(k)}_{t,\infty}(\bs \theta^*, \bs \theta) + \sum_{j=1}^m \mu^{(i)}_{t,\infty}(\bs \theta^*, \bs \theta^*) \frac{\lambda^{(k)}_{t,\infty}(\bs \theta^*, \bs \theta)q^{(k)}_t(\bs \theta)G^{(k,j)}_t(\bs \theta) }{\mu^{(i)}_{t,\infty}(\bs \theta^*, \bs \theta)} = 0.
\end{equation*}
This is well defined by the induction hypothesis. Furthermore, in order for this equality with zero to hold we must have either:
\begin{itemize}
    \item $q_t^{(k)}(\bs \theta)<1$, in which case we must have $\lambda^{(k)}_{t,\infty}(\bs \theta^*, \bs \theta) = 0$ which by the induction hypothesis implies $\lambda^{(k)}_{t,n}(\bs \theta) = 0 \quad a.s.$; or
    \item $q_t^{(k)}(\bs \theta)=1$, in which case we must have $\lambda^{(k)}_{t,\infty}(\bs \theta^*, \bs \theta)G_t^{(k,j)}(\bs \theta) = 0 \quad a.s.$  for all $j$ so that the sum is equal to $0$. Since $\+G_t$ is row-stochastic matrix, we know there must exist a $j \in[m]$ such that $G_t^{(k,j)}(\bs \theta) > 0 $ and so we must have $\lambda^{(k)}_{t,\infty}(\bs \theta^*, \bs \theta) = 0$ which by the induction hypothesis implies $\lambda^{(k)}_{t,n}(\bs \theta) = 0$ a.s.
\end{itemize}
So we have $\bar \lambda^{(k)}_{t,\infty}(\bs \theta^*, \bs \theta) = 0\implies  \lambda^{(k)}_{t,n}(\bs \theta) = 0$ a.s., furthermore $\lambda^{(k)}_{t,n}(\bs \theta) = 0$ a.s. $\implies$ $\bar \lambda^{(k)}_{t,n}(\bs \theta) = 0$ a.s.. Now  if for some $j \in [m]$:
\begin{equation}
    \lambda^{(j)}_{t+1,\infty}(\bs \theta^*, \bs \theta) = \sum_{k=1}^m\bar \lambda^{(k)}_{t,\infty}(\bs \theta^*, \bs \theta)\delta_{t+1}^{(k)}(\bs \theta)K^{(k,j)}_{1, \bs \eta(\bs{\bar \lambda}_{t,\infty}(\bs \theta^*, \bs \theta)\circ \bs \delta_{t+1}(\bs \theta))}(\bs \theta) + \alpha^{(j)}_{t+1,\infty}(\bs \theta) = 0,
\end{equation}
then $\alpha^{(j)}_{t+1,\infty}(\bs \theta) = 0$ which by assumption \ref{as:params} implies $\alpha^{(j)}_{t+1,n}(\bs \theta) = 0$ and for each $k  \in [m]$ we must have either:
\begin{itemize}
    \item $\bar \lambda^{(k)}_{t,\infty}(\bs \theta^*, \bs \theta) = 0 $, which we have already shown implies $\lambda^{(k)}_{t,n}(\bs \theta) = 0 \quad a.s. \implies \bar \lambda^{(k)}_{t,n}(\bs \theta) = 0 \quad a.s.$; or
    \item $\delta_{t+1}^{(k)}(\bs \theta)= 0$; or
    \item $K^{(k,j)}_{1, \bs \eta(\bs{\bar \lambda}_{t,\infty}(\bs \theta^*, \bs \theta)\circ \bs \delta_{t+1}(\bs \theta))}(\bs \theta) = 0$, which together with assumption \ref{as:1} implies $K^{(k,j)}_{1, \bs \eta( \bar {\bs \lambda}_{t,n}(\bs \theta)\circ \bs \delta_{t+1}(\bs \theta))}(\bs \theta) = 0$.
\end{itemize}
Hence
\begin{equation*}
    \lambda^{(j)}_{t+1,n}( \bs \theta) = \sum_{k=1}^m\bar \lambda^{(k)}_{t,n} \bs \theta)\delta_{t+1}^{(k)}(\bs \theta)K^{(k,j)}_{1, \bs \eta(\bs{\bar \lambda}_{t,n}( \bs \theta)\circ \bs \delta_{t+1}(\bs \theta))}(\bs \theta) + \alpha^{(j)}_{t+1,n}(\bs \theta) = 0 \quad a.s.
\end{equation*}
Now, if
\begin{equation*}
    \mu^{(i)}_{t+1,\infty}(\bs \theta^*, \bs \theta) = \sum_{j=1}^m\lambda^{(j)}_{t+1,\infty}(\bs \theta^*, \bs \theta)q_{t+1}^{(j)}(\bs \theta)G_{t+1}^{(k,j)}(\bs \theta) + \kappa^{(i)}_{t+1,\infty}(\bs \theta) = 0,
\end{equation*}
then $\kappa^{(i)}_{t+1,\infty}(\bs \theta) = 0$ which by assumption \ref{as:init} implies $\kappa^{(i)}_{t+1,n}(\bs \theta) = 0$ and for all $j \in [m]$ we have either:
\begin{itemize}
    \item $\lambda^{(j)}_{t+1,\infty}(\bs \theta^*, \bs \theta) = 0$, which implies $ \lambda^{(j)}_{t+1,n}(\bs \theta) = 0$ a.s.; or
    \item $q_{t+1}^{(j)}(\bs \theta) = 0$; or
    \item $G_{t+1}^{(k,j)}(\bs \theta) = 0$.
\end{itemize}
Hence
\begin{equation*}
    \mu^{(i)}_{t+1,n}(\bs \theta) = \sum_{j=1}^m\lambda^{(j)}_{t+1,n}(\bs \theta^)q_{t+1}^{(j)}(\bs \theta)G_{t+1}^{(k,j)}(\bs \theta) + \kappa^{(j)}_{t+1,\infty}(\bs \theta) = 0 \quad a.s.,
\end{equation*}
and the inductive proof is complete.

\end{proof}

%

\begin{proof}[Proof of Proposition \ref{vecfilt}]
Fix any $ \bs \theta \in \Theta$. We proceed by induction to show that for all $t \geq1$,
\begin{equation*}
n^{-1} \bar{ \bs \lambda}_{t,n}( \bs \theta) \asls \bar{ \bs \lambda}_{t, \infty}(\bs \theta^*, \bs \theta),
\end{equation*}
with the other claims of the proposition proved along the way. 

Using assumption \ref{as:init} we have:
$${n^{-1}\bar{\bs \lambda}_{0,n}(\bs \theta) = n^{-1}{\bs \lambda}_{0,n}(\bs \theta)  \asls {\bs \lambda}_{0,\infty}(\bs \theta) =  \bs \lambda_{0, \infty}(\bs \theta^*, \bs \theta)}.$$ Now, for $t \geq 1$ assume that ${n^{-1}\bar{\bs \lambda}_{t-1,n}(\bs \theta) \asls \bar{\bs \lambda}_{t-1, \infty}(\bs \theta^*, \bs \theta)}$. We have:
\begin{equation*}
\begin{aligned} \label{lambdapred}
n^{-1}\bs \lambda_{t,n}(\bs \theta) &= \left[( n^{-1}\bar{\bs \lambda}_{t-1,n}(\bs \theta) \circ \bs \delta_t(\bs \theta))^\top \+K_{t, \bs \eta(\bar{\bs \lambda}_{t-1,n}(\bs \theta) \circ \bs \delta_t(\bs \theta))}\right]^\top + n^{-1}\bs \alpha_{t,n}(\bs \theta) \\
&\asls \left[(\bar{\bs \lambda}_{t-1,\infty}(\bs \theta^*, \bs \theta) \circ \bs \delta_t(\bs \theta))^\top \+K_{t, \bs \eta(\bar{\bs \lambda}_{t-1,\infty}(\bs \theta^*, \bs \theta) \circ \bs \delta_t(\bs \theta))}\right]^\top + \bs \alpha_{t,\infty}(\bs \theta) \\
&= \bs \lambda_{t,\infty}(\bs \theta^*, \bs \theta),
\end{aligned}
\end{equation*}
by the continuous mapping theorem (CMT) and assumptions \ref{as:params} and \ref{as:1}. A further application of the CMT and assumption \ref{as:params} yields:
\begin{equation*}
\begin{aligned}
    n^{-1}\bs \mu_{t,n}\left(\bs \theta \right) &= \left[\left(n^{-1}\bs \lambda_{t,n}(\bs \theta)\circ \+ q_t(\bs \theta) \right)^\top \+G_t(\bs \theta)\right]^\top + n^{-1} \bs \kappa_{t,n}(\bs \theta) \\
    &\asls\left[\left(\bs \lambda_{t,\infty}(\bs \theta^*,\bs \theta)\circ \+ q_t(\bs \theta) \right)^\top \+G_t(\bs \theta)\right]^\top + \bs \kappa_{t, \infty}(\bs \theta) \\
    &= \bs \mu_{t,\infty}\left(\bs \theta^*,\bs \theta \right)
\end{aligned}
\end{equation*}
Recalling from remark \ref{rem:lam_mu} that $n^{-1}\+y_t \asls\bs \mu_{t,\infty}\left(\bs \theta^*,\bs \theta^* \right)$ and applying the CMT, we have:
\begin{equation*}
\begin{aligned}
 n^{-1}\bar{\bs \lambda}_{t, n}( \bs \theta) = &\bigg[\+ 1_m - \+q_t(\bs \theta)  \\
 +&\left(\left[ n^{-1}\+ y_t \oslash n^{-1}\bs \mu_{t, n}( \bs \theta)\right]^\top\left(\left[\+1_m\otimes\+q_t(\bs \theta)\right)\circ\+G_t(\bs \theta)\right]^\top \right)^\top\bigg]\circ n^{-1} \bs \lambda_{t, n}(\bs \theta) \\
 \asls &\bigg[\+ 1_m - \+q_t(\bs \theta)\\
 &+\left(\left[ \bs \mu_{t, \infty}(\bs \theta^*, \bs \theta^*)\oslash\bs \mu_{t, \infty}(\bs \theta^*, \bs \theta)\right]^\top\left(\left[\+1_m\otimes\+q_t(\bs \theta)\right]\circ\+G_t(\bs \theta)^\top\right) \right)^\top\bigg]\circ \bs \lambda_{t, \infty}(\bs \theta^*,\bs \theta)\\
 & = \bar{\bs \lambda}_{t, \infty}(\bs \theta^*, \bs \theta)
\end{aligned}
\end{equation*}
We note this limit is almost surely well defined since by lemma \ref{mumulemma} for any $i \in[m]$, $$\mu^{(i)}_{t, \infty}(\bs \theta^*, \bs \theta^*) = 0 \iff \mu^{(i)}_{t, \infty}(\bs \theta^*, \bs \theta) = 0$$ and by lemma \ref{muylemma} if $\mu^{(i)}_{t, n}( \bs \theta) = 0$ with positive probability then $y_t=0$, $ \Ps$-a.s. In both these cases we are working under the convention $\frac{0}{0}\coloneqq0$. 
%
\end{proof}
\begin{lemma}\label{lem:contmu}
Let assumptions \ref{as:params}- \ref{as:init} hold. For all $\bs \theta^*  \in \Theta$ and $t\geq1$ the function $\bs \theta \mapsto \bs \mu_{t,\infty}(\bs \theta^*, \bs \theta)$ is continuous on $\Theta$.
\end{lemma}
\begin{proof}
Fix an arbitrary $\bs \theta^* \in \Theta$.  Note that $\bs {\bar \lambda}_{0, \infty}(\bs \theta^*,\bs \theta) := \bs {\bar \lambda}_{0, \infty}(\bs \theta)$ is continuous by assumption \ref{as:init}.  We will now show that for any $t\geq 1$,  continuity of $\bs {\bar \lambda}_{t-1, \infty}(\bs \theta^*,\bs \theta)$  implies continuity of $\bs { \lambda}_{t, \infty}(\bs \theta^*,\bs \theta)$, $\bs \mu_{t, \infty}(\bs \theta^*,\bs \theta)$, and $\bs {\bar \lambda}_{t, \infty}(\bs \theta^*,\bs \theta)$, from which the claim of the lemma follows. 

Henceforth assume that $\bs {\bar \lambda}_{t-1, \infty}(\bs \theta^*,\bs \theta)$  is continuous and recall that by definition of $\bs \lambda_{t, \infty}(\bs \theta^*, \bs \theta) $,
\begin{equation}
\bs \lambda_{t, \infty}(\bs \theta^*, \bs \theta) \coloneqq \left[(\bar{\bs \lambda}_{t-1, \infty}(\bs \theta^*, \bs \theta)\circ \bs \delta_t(\bs \theta))^\top\+ K_{t, \bs \eta(\bar{\bs \lambda}_{t-1, \infty}\left(\bs \theta^*, \bs \theta)\circ \bs \delta_t(\bs \theta)\right)}(\bs \theta)\right]^\top + \bs \alpha_{t, \infty}(\bs \theta).
\end{equation}
Continuity of $\bs \delta_t(\bs \theta)$ and $\bs \alpha_{t, \infty}(\bs \theta)$ in $\bs \theta$  holds directly by assumptions \ref{as:params} and \ref{as:init}. By assumption \ref{as:1} we know that $\+ K_{t,\bs \eta}(\bs \theta)$ is continuous in $\bs \theta$ and $\bs \eta$. Hence, to show continuity of $\bs \lambda_{t, \infty}(\bs \theta^*, \bs \theta)$ we shall show that $\bs \eta(\bar{\bs \lambda}_{t-1, \infty}\left(\bs \theta^*, \bs \theta)\circ \bs \delta_t(\bs \theta)\right)$ is continuous in $\bs \theta$. The function $\bs \eta:\mathbb{R}_{\geq 0}^m \rightarrow \mathbb{R}_{\geq 0}^m$ is continuous everywhere except at $\bs 0_m$, we now show that, by virtue of our assumptions, this discontinuity is immaterial. Consider the two following cases:
\begin{itemize}
\item There exists $\bs \theta' \in \Theta$ such that $\bar{\bs \lambda}_{t-1, \infty}\left(\bs \theta^*, \bs \theta'\right)\circ \bs \delta_t(\bs \theta') = \+ 0_m$. In this case,  by assumption \ref{as:params} and lemma \ref{mumulemma} we have that $\bar{\bs \lambda}_{t-1, \infty}(\bs \theta^*, \bs \theta)\circ \bs \delta_t(\bs \theta) = \+ 0_m$ for all $\bs \theta \in \Theta$, from which it follows that   $\bs \eta(\bar{\bs \lambda}_{t-1, \infty}\left(\bs \theta^*, \bs \theta)\circ \bs \delta_t(\bs \theta)\right)=\bs 0_m$ for all  $\bs \theta \in \Theta$, so that the continuity of $\bs \eta(\bar{\bs \lambda}_{t-1, \infty}\left(\bs \theta^*, \bs \theta)\circ \bs \delta_t(\bs \theta)\right)$ in $\bs\theta$ on $\Theta$ holds trivially;
\item For all $\bs \theta \in \Theta$, $\bar{\bs \lambda}_{t-1, \infty}\left(\bs \theta^*, \bs \theta \right)\circ \bs \delta_t(\bs \theta) \neq \+ 0_m$. In this case the continuity of $\bs\eta(\bar{\bs \lambda}_{t-1, \infty}\left(\bs \theta^*, \bs \theta \right)\circ \bs \delta_t(\bs \theta))$ in $\bs\theta$ on $\Theta$ follows from the continuity of $\bs\eta$ on $\mathbb{R}_{\geq 0}^m \setminus \{\bs 0_m\}$.
\end{itemize}
Hence, $\bs \theta \mapsto \bs \lambda_{t, \infty}(\bs \theta^*, \bs \theta)$ is continuous. Recall that:
$$    \bs \mu_{t, \infty}(\bs \theta^*, \bs \theta) \coloneqq \left[\left(\bs \lambda_{t, \infty}(\bs \theta^*, \bs \theta)\circ \+q_t(\bs \theta) \right)^\top \+G_t(\bs \theta)\right]^\top +\bs \kappa_{t,\infty}(\bs \theta).$$
Due to the continuity of $\lambda_{t, \infty}(\bs \theta^*, \bs \theta)$ and assumption \ref{as:params}, this is a composition of continuous functions and hence $\bs \theta \mapsto \bs \mu_ {t, \infty}(\bs \theta^*, \bs \theta)$  is itself continuous. Now consider 
\begin{align*}
\bar{\bs \lambda}_{t, \infty}(\bs \theta^*, \bs \theta) &\coloneqq \bigg[\+ 1_m - \+q_t(\bs \theta)\nonumber\\
 &+\left(\left[ \bs \mu_{t, \infty}(\bs \theta^*, \bs \theta^*)\oslash\bs \mu_{t, \infty}(\bs \theta^*, \bs \theta)\right]^\top\left(\left[\+1_m\otimes\+q_t(\bs \theta)\right]\circ\+G_t(\bs \theta)^\top\right) \right)^\top\bigg]\circ \bs \lambda_{t, \infty}(\bs \theta^*,\bs \theta).
\end{align*}
Each component of this function is trivially continuous on $\Theta$ except the $\bs \mu_{t, \infty}(\bs \theta^*, \bs \theta^*)\oslash\bs \mu_{t, \infty}(\bs \theta^*, \bs \theta)$ term, we will now prove its continuity. By lemma \ref{mumulemma}, for each $i \in [m]$ we need only consider the two cases:
\begin{itemize}
\item either $\mu^{(i)}_{t, \infty}(\bs \theta^*, \bs \theta) = 0$ for all $\bs \theta \in \Theta$, in which case we have by convention $\mu^{(i)}_{t, \infty}(\bs \theta^*, \bs \theta^*)/ \mu^{(i)}_{t, \infty}(\bs \theta^*, \bs \theta) \coloneqq  0$, which is continuous; or
\item $\mu^{(i)}_{t, \infty}(\bs \theta^*, \bs \theta) \neq 0$ for all $\bs \theta \in \Theta$, in which case $\mu^{(i)}_{t, \infty}(\bs \theta^*, \bs \theta^*)/ \mu^{(i)}_{t, \infty}(\bs \theta^*, \bs \theta)$ is continuous.
\end{itemize}
Hence we have elementwise continuity of $\bs \mu_{t, \infty}(\bs \theta^*, \bs \theta^*)\oslash\bs \mu_{t, \infty}(\bs \theta^*, \bs \theta)$ which gives us continuity of $\bs \theta \mapsto \bar{\bs \lambda}_{t, \infty}(\bs \theta^*, \bs \theta)$. 

We have shown that continuity of $\bs\theta\mapsto\bs {\bar \lambda}_{t-1, \infty}(\bs \theta^*,\bs \theta)$  on $\Theta$ implies continuity of   $\bs \theta \mapsto \bs { \lambda}_{t, \infty}(\bs \theta^*,\bs \theta)$, $ \bs \theta \mapsto \bs \mu_{t, \infty}(\bs \theta^*,\bs \theta)$ and $\bs \theta \mapsto \bar{\bs \lambda}_{t, \infty}(\bs \theta^*, \bs \theta)$ on $\Theta$, which completes the proof.
\end{proof}

%

\subsubsection{Case (II)}\label{sec:filtering_limits_case_II}
%
%
Define: 
 \begin{equation}
    \bs{\bar \lambda}_{0,\infty}(\bs \theta^*, \bs \theta) \coloneqq \bs \lambda_{0, \infty}(\bs \theta), \label{eq:Lam_0_inf_defn} 
\end{equation}
and for $r = 1, \dots, R$ and $t = \tau_{r-1}+1 , \dots, \tau_r -1 $, 
\begin{align*}
\bs \Lambda_{t, \infty}(\bs \theta^*, \bs \theta) &\coloneqq \left(\bs{\bar \lambda}_{\tau_r-1, \infty}(\bs \theta^*, \bs \theta) \otimes \+1_m \right) \circ \+K_{t,\bs \eta(\bs{\bar \lambda}_{t-1, \infty}(\bs \theta^*, \bs \theta))}(\bs \theta^*, \bs \theta),\\
\bs {\bar \lambda}_{t, \infty}(\bs \theta^*, \bs \theta)&\coloneqq (\+1_m^\top \bs \Lambda_{t, \infty}(\bs \theta^*, \bs \theta))^\top,
\end{align*}
and
\begin{align}
\bs \Lambda  _{\tau_r, \infty}(\bs \theta^*, \bs \theta)  &\coloneqq \left(\bar { \bs\lambda}_{\tau_r-1, \infty}(\bs \theta^*, \bs \theta)  \otimes \+1_m \right) \circ \+K_{t,\bs \eta\left(\bs {\bar \lambda}{\tau_r-1, \infty}(\bs \theta^*, \bs \theta)\right) }(\bs \theta), \nonumber\\
\+M_{r,\infty}(\bs \theta^*,\bs \theta) &\coloneqq \sum_{s = \tau_{r-1}+1}^{\tau_{r}} \+ \Lambda_{s,\infty}(\bs \theta) \circ \+ Q_s(\bs \theta), \nonumber\\
\bar{\bs \Lambda}_{\tau_r, \infty}(\bs \theta^*, \bs \theta)  &\coloneqq \left[\+1_m \otimes \+1_m - \+Q_{\tau_r}(\bs \theta) \right]\circ \bs \Lambda  _{\tau_r}(\bs \theta^*, \bs \theta) \nonumber\\
&+ \left[\+M_{r,\infty}(\bs \theta^*,\bs \theta^*)\oslash \+M_{r,\infty}(\bs \theta^*,\bs \theta)\right] \circ\left[ \bs \Lambda_{\tau_r, \infty}(\bs \theta^*, \bs \theta)\circ \+ Q  _{\tau_r}(\bs \theta) \right],\nonumber\\
\bar{\bs  \lambda}  _{\tau_r, \infty}(\bs \theta^*, \bs \theta)  &\coloneqq (\+1_m^\top \bar{\bs \Lambda}  _{\tau_r, \infty}(\bs \theta^*, \bs \theta) )^\top.\label{eq:Lam_bar_t_inf_defn} 
\end{align}
where if we encounter $0/0$ in the element-wise division operation we set the entry to $0$ by convention. The main result of section \ref{sec:filtering_limits_case_II} is proposition \ref{Yfilt} concerning the convergence to the above of the associated finite-$n$ quantities computed using algorithm \ref{alg:Ztagg}.

\begin{proposition}\label{Yfilt}
Let assumptions \ref{as:params} - \ref{as:init} hold. For any $\bs \theta \in \Theta$ and  $r\geq 1$ and $t\geq1$:
\begin{align*}
n^{-1}\+ M_{r,n}(\bs \theta) &\asls \+ M_{r, \infty}(\bs \theta^*, \bs \theta),\\
n^{-1}\bs \Lambda_{t,n}(\bs \theta) &\asls \+ \Lambda_{t, \infty}(\bs \theta^*,\bs \theta),
\end{align*}
\end{proposition}
\noindent The proof is postponed until later in section \ref{sec:filtering_limits_case_II}.

\begin{remark}\label{rem:Lam_M}
Similarly to properties of case (I) pointed out in remark \ref{rem:lam_mu}, by writing out the above definitions it can be checked that $\+N_t(\bs\theta^*)=\bs\Lambda_{t, \infty}(\bs \theta^*, \bs \theta^*)$,  thus $n^{-1}\+Z_t\asls \bs\Lambda_{t, \infty}(\bs \theta^*, \bs \theta^*)$; and that $ \+M_{r, \infty}(\bs \theta^*, \bs \theta^*)$ is equal to the right hand side of \eqref{eq:Y_bar_lln}, thus  $n^{-1}\bar{\+Y}_r \asls   \+M_{r, \infty}(\bs \theta^*, \bs \theta^*)$. 
\end{remark}

Similarly to as in section \ref{sec:filtering_limits_case_I}, in order to prove proposition \ref{Yfilt} we need to check that certain quantities are almost surely well defined. For the update step of algorithm \ref{alg:Ztagg} to be $\mathbb{P}^{\bs \theta^*} $-a.s. well  defined  for all $\bs \theta \in \Theta$ we need that if $M_{r,n}^{(i,j)}(\bs \theta) = 0$ occurs with positive probability then $\bar Y_r^{(i,j)} = 0$ $\mathbb{P}^{\bs \theta^*}$- a.s. This is established in the following lemma.

\begin{lemma}\label{MY}
Let assumptions \ref{as:params} - \ref{as:init} hold. For any $\bs \theta \in \Theta$, $n \in \mathbb{N}$, $(i,j) \in [m]^2$ and $r = 1, \dots, R$:
$$\Ps\left(M^{(i,j)}_{r,n}(\bs \theta) = 0 \right)>0 \implies \bar Y_r^{(i,j)} = 0, \quad  \Ps a.s.$$
\end{lemma}
\begin{proof}
Fix arbitrary $ \bs \theta \in \Theta$ and $n \in \mathbb{N}$. All almost sure statements made throughout the proof are with respect to $\Ps$. We will prove by induction that for all $r = 1, \dots, R$ we have that for all $s \in \{\tau_{r-1}+1,\dots, \tau_r\}$ and $(i,j) \in [m]^2$, the following two implications hold.
\begin{align}
\Ps\left(\Lambda_{s,n}^{(i,j)}(\bs \theta) = 0\right)>0 \implies Z_s^{(i,j)} = 0, \quad a.s.,\label{eq:M_Y_induct1}\\
\Ps\left(M^{(i,j)}_{r,n}(\bs \theta) = 0\right)>0\implies \bar Y_r^{(i,j)} = 0 \quad  a.s.\label{eq:M_Y_induct2}
\end{align}
 Consider the case $r=1$. We will first show that for all $s \in \{\tau_0+1,\dots, \tau_1\}$ if, for some $(i,j) \in [m]^2$, $\Ps\left(\Lambda_{s,n}^{(i,j)}(\bs \theta) = 0\right)>0$ then $Z_s^{(i,j)} = 0$ a.s. by induction on $s$. Suppose that for some $(i,j) \in [m]^2$, $\Ps\left(\Lambda_{1,\infty}^{(i,j)}(\bs \theta) =0 \right) >0$, i.e,
\begin{equation*}
\Ps\left(\lambda_{0,n}^{(i)}(\bs \theta)K^{(i,j)}_{1,\bs \eta(\bs \lambda_{0,n}(\bs \theta) )}(\bs \theta) = 0\right)>0.
\end{equation*}
This implies that either:
\begin{itemize}
    \item $\Ps\left(\lambda^{(i)}_{0,n}(\bs \theta) = 0\right)>0$, which since  $\lambda^{(i)}_{0,n}(\bs \theta)$ is deterministic implies that \newline $\lambda^{(i)}_{0,n}(\bs \theta) = 0$  which by assumption \ref{as:init} implies that $\lambda^{(i)}_{0,n}(\bs \theta^*) = 0 \implies x^{(i)}_0 = 0$ a.s.;  or
    \item $K^{(i,j)}_{1,\bs \eta(\bs \lambda_{0,n}^{(i)})}(\bs \theta) = 0 $,  which implies $K^{(i,j)}_{1,\bs \eta(\+ x_0)}(\bs \theta^*) = 0$ by assumptions \ref{as:params}, \ref{as:1}, and \ref{as:init}.
\end{itemize}
Together this implies imply $Z^{(i,j)}_1 = 0$ a.s..  Now let $s \in \{\tau_0+1,\dots, \tau_1\}$ and assume that if, for some $(i,j) \in [m]^2$, $\Ps\left(\Lambda_{s-1,n}^{(i,j)}(\bs \theta) = 0\right)>0$ then $Z_{s-1}^{(i,j)} = 0$ a.s.. Now suppose for some $(i,j) \in [m]^2$, $\Ps\left(\Lambda_{s,n}^{(i,j)}(\bs \theta) = 0\right)>0$, i.e.,
\begin{equation*}
\Ps\left(\left(\+1_m^\top \bs \Lambda_{s-1,n}^{(\cdot,i)}(\bs \theta)\right)K^{(i,j)}_{1,\bs \eta\left(\+1_m^\top \bs \Lambda_{s-1,n}^{(\cdot,i)}(\bs \theta)\right)}(\bs \theta) = 0\right) >0,
\end{equation*}
This implies that either:
\begin{itemize}
    \item $\Ps\left(\+1_m^\top \bs \Lambda_{0,n}^{(\cdot,i)}(\bs \theta)= 0\right)>0$, which by the induction hypothesis implies  $\+1_m^\top \+Z_{s-1}^{(\cdot,i)} = 0$
    a.s., which in turn implies $ x^{(i)}_{s-1} = 0$ a.s.; or
    \item $\Ps\left(K^{(i,j)}_{1,\bs \eta\left(\+1_m^\top \bs \Lambda_{s-1,n}^{(\cdot,i)}(\bs \theta)\right)}(\bs \theta)  = 0\right) >0$ which by assumptions \ref{as:params} and \ref{as:1} and the induction hypothesis implies $K^{(i,j)}_{1,\bs \eta(\+ x_{s-1})}(\bs \theta^*) = 0$ a.s.,
\end{itemize}
which together imply $Z^{(i,j)}_s = 0$ a.s.. Now suppose for some $(i,j) \in [m]^2$, $\Ps\left(M^{(i,j)}_{1,n}(\bs \theta) = 0\right)>0$, i.e.,
 \begin{equation*}
\Ps\left(\sum_{s = \tau_{0}+1}^{\tau_{1}}  \Lambda^{(i,j)}_{s,n}(\bs \theta) \circ Q^{(i,j)}_s(\bs \theta) = 0\right)>0,
\end{equation*}
then for all $s = \tau_{0}+1, \dots, \tau_1$ either:
\begin{itemize}
    \item $ \Ps\left(\Lambda^{(i,j)}_{s,n}(\bs \theta)=0\right)>0$, which  implies $Z^{(i,j)}_s = 0$ hence $Y^{(i,j)}_s = 0$ a.s.; or
    \item $Q^{(i,j)}_s(\bs \theta) =0$ which by assumption \ref{as:params} implies $Q^{(i,j)}_s(\bs \theta^*) = 0$ hence $Y^{(i,j)}_s = 0$ a.s.,
\end{itemize}
and hence $\bar Y_1^{(i,j)} =\sum_{s = \tau_{0}+1}^{\tau_{1}}  Y^{(i,j)}_s = 0$ a.s., this completes the proof of \eqref{eq:M_Y_induct1} and \eqref{eq:M_Y_induct2} for $r=1$. 

For the induction hypothesis, suppose that \eqref{eq:M_Y_induct1} and \eqref{eq:M_Y_induct2} hold for some $r\geq 1$. Notice that:
\begin{align*}
\bar{\Lambda}^{(i,j)}_{\tau_r, n}( \bs \theta) = \left[1 - Q^{(i,j)}_{\tau_r}(\bs \theta) \right] \Lambda^{(i,j)}_{\tau_r,n}( \bs \theta)
+ \frac{\bar Y^{(i,j)}_{r}}{M^{(i,j)}_{r,n}(\bs \theta)}\left[\Lambda^{(i,j)}_{\tau_r, n}( \bs \theta)\circ  Q^{(i,j)}_{\tau_r}(\bs \theta) \right] = 0
\end{align*}
is almost surely well defined by the induction hypothesis since we divide positive $\bar Y^{(i,j)}_{r}$ by $0$ with probability 0. Now suppose, for some $(i,j) \in [m]$, that  ${\Ps\left(\bar{\Lambda}^{(i,j)}_{\tau_r, n}( \bs \theta) =0\right)>0}$, then either:
\begin{itemize}
    \item $Q^{(i,j)}_{\tau_r}(\bs \theta)<1$, which implies $\Ps\left(\Lambda^{(i,j)}_{\tau_r,n}( \bs \theta) = 0\right)>0$, so that the first term of the sum is $0$ with positive probability, which then implies $ Z^{(i,j)}_{\tau_r} = 0 \quad a.s.$ by the induction hypothesis; or
    \item $Q^{(i,j)}_{\tau_r}(\bs \theta) =1$, which implies $\Ps\left(\Lambda^{(i,j)}_{\tau_r,n}( \bs \theta) = 0\right)>0$, so that the second term in the sum is $0$ with positive probability, which then implies $ Z^{(i,j)}_{\tau_r} = 0$ a.s. by the induction hypothesis.
\end{itemize}
Using this and identical reasoning  to that in the $r=1$ case completes the induction. 
\end{proof}
If, for some $i,j\in[m]$, $  M^{(i,j)}_{r, \infty}(\bs \theta^*, \bs \theta) = 0$ and $M^{(i,j)}_{r, \infty}(\bs \theta^*, \bs \theta^*) >0$, then $\bar{\bs \Lambda}_{\tau_r, \infty}(\bs \theta^*, \bs \theta)$ would involve division of a finite number by zero. The following lemma implies this situation does not arise.
\begin{lemma}\label{MM}
Let assumptions \ref{as:params} - \ref{as:init} hold. Then for all $ \bs \theta, \bs \theta' \in \Theta$, $(i,j) \in [m]^2$ and $r =  1, \dots, R$:
\begin{equation*}
    M^{(i,j)}_{r, \infty}(\bs \theta^*, \bs \theta) = 0 \iff M^{(i,j)}_{r, \infty}(\bs \theta^*, \bs \theta') = 0.
\end{equation*}
\end{lemma}
\begin{proof}
It is enough to establish the implication in one direction for arbitrary $ \bs \theta, \bs \theta' \in \Theta$. We will prove by induction that for all $r = 1, \dots, R$, $s = \tau_{r-1}+1,\dots, \tau_r $ and $i,j\in[m]$,
\begin{align*}
\Lambda_{s,\infty}^{(i,j)}(\bs \theta^*, \bs \theta) = 0 &\implies \Lambda_{s,\infty}^{(i,j)}(\bs \theta^*, \bs \theta') = 0,\\
M^{(i,j)}_{r,\infty}(\bs \theta^*, \bs \theta) = 0 &\implies M^{(i,j)}_{r,\infty}(\bs \theta^*, \bs \theta') = 0.
\end{align*}
 Consider the case $r=1$. We will first show that for all $s \in \{\tau_0+1,\dots, \tau_1\}$, $\Lambda_{s,\infty}^{(i,j)}(\bs \theta^*, \bs \theta) = 0 \implies \Lambda_{s,\infty}^{(i,j)}(\bs \theta^*, \bs \theta') = 0$ by induction on $s$. To this end suppose that for some $(i,j)\in [m]^2$:
\begin{equation*}
    \Lambda_{1,\infty}^{(i,j)}(\bs \theta^*, \bs \theta)  = \lambda_{0,\infty}^{(i)}(\bs \theta^*, \bs \theta)K^{(i,j)}_{1,\bs \eta(\bs \lambda_{0,\infty}^{(i)}(\bs \theta^*,\bs \theta))}(\bs \theta) = 0.
\end{equation*}
Then either:
\begin{itemize}
    \item $\lambda^{(i)}_{0,\infty}(\bs \theta^*,\bs \theta) =  \lambda^{(i)}_{0,n}(\bs \theta) = 0$, which by assumption \ref{as:init} implies $\lambda_{0,n}(\bs \theta') =  \lambda^{(i)}_{0,\infty}(\bs \theta^*,\bs \theta') = 0$; or
    \item $K^{(i,j)}_{1,\bs \eta(\lambda_{0,\infty}(\bs \theta^*,\bs \theta) }(\bs \theta) = 0$, which by assumptions \ref{as:params}, \ref{as:1}, and \ref{as:init} implies $K^{(i,j)}_{1,\bs \eta(\lambda_{0,\infty}(\bs \theta^*,\bs \theta')}(\bs \theta') = 0$.
\end{itemize}
Hence:
\begin{equation*}
    \Lambda_{1,\infty}^{(i,j)}(\bs \theta^*, \bs \theta')  = \lambda_{0,\infty}^{(i)}(\bs \theta^*, \bs \theta')K^{(i,j)}_{1,\bs \eta(\bs \lambda_{0,\infty}^{(i)}(\bs \theta^*,\bs \theta'))}(\bs \theta') = 0.
\end{equation*}
Now assume that for $s = tau_0+1,\dots, \tau_1 $ that $\Lambda_{s-1,\infty}^{(i,j)}(\bs \theta^*, \bs \theta) = 0 \implies \Lambda_{s-1,\infty}^{(i,j)}(\bs \theta^*, \bs \theta') = 0$, then if:
\begin{equation*}
    \Lambda_{s,\infty}^{(i,j)}(\bs \theta^*, \bs \theta) = \left(\+1_m^\top \bs \Lambda_{s-1,\infty}^{(\cdot,i)}(\bs \theta^*, \bs \theta)\right)K^{(i,j)}_{1,\bs \eta\left(\+1_m^\top \bs \Lambda_{s-1,\infty}(\bs \theta^*, \bs \theta)\right)}(\bs \theta) = 0,
\end{equation*}
we must have either:
\begin{itemize}
    \item $\left(\+1_m^\top \bs \Lambda_{s-1,\infty}^{(\cdot,i)}(\bs \theta^*, \bs \theta)\right)= 0$, which by the induction hypothesis implies $\left(\+1_m^\top \bs \Lambda_{s-1,\infty}^{(\cdot,i)}(\bs \theta^*, \bs \theta')\right) = 0$; or
    \item $K^{(i,j)}_{1,\bs \eta\left(\+1_m^\top \bs \Lambda_{s-1,\infty}^{(\cdot,i)}(\bs \theta^*, \bs \theta)\right)}(\bs \theta)  = 0$, which by the above and assumptions \ref{as:params} and \ref{as:1} implies  \newline $K^{(i,j)}_{1,\bs \eta\left(\+1_m^\top \bs \Lambda_{s-1,\infty}^{(\cdot,i)}(\bs \theta^*, \bs \theta')\right)}(\bs \theta') = 0$.
\end{itemize}
We therefore find:
\begin{equation*}
    \Lambda_{s,\infty}^{(i,j)}(\bs \theta^*, \bs \theta') = \left(\+1_m^\top \bs \Lambda_{s-1,\infty}^{(\cdot,i)}(\bs \theta^*, \bs \theta')\right)K^{(i,j)}_{1,\bs \eta\left(\+1_m^\top \bs \Lambda_{s-1,\infty}(\bs \theta^*, \bs \theta')\right)}(\bs \theta') = 0.
\end{equation*}
completing the intermediary induction on $s$. 
Now consider:
 \begin{align*}
M^{(i,j)}_{1,\infty}(\bs \theta^*, \bs \theta) &= \sum_{s = 1}^{\tau_{1}}  \Lambda^{(i,j)}_{s,\infty}(\bs \theta^*, \bs \theta) \circ Q^{(i,j)}_s(\bs \theta) = 0,
\end{align*}
then for all $s = \tau_{0}+1, \dots, \tau_1$ either:
\begin{itemize}
    \item $ \Lambda^{(i,j)}_{s,\infty}(\bs \theta^*, \bs \theta) =0 \implies \Lambda^{(i,j)}_{s,\infty}(\bs \theta^*, \bs \theta')  = 0$; or
    \item $Q^{(i,j)}_s(\bs \theta) =0 \implies Q^{(i,j)}_s(\bs \theta') = 0 $,
\end{itemize}
and hence:
 \begin{align*}
M^{(i,j)}_{1,\infty}(\bs \theta^*, \bs \theta') &= \sum_{s = 1}^{\tau_{1}}  \Lambda^{(i,j)}_{s,\infty}(\bs \theta^*, \bs \theta') \circ Q^{(i,j)}_s(\bs \theta') = 0,
\end{align*}
completing the $r=1$ case. 

Now assume that for all  $s \in \{\tau_{r-1}+1,\dots, \tau_r\}$ that  $\Lambda_{s,\infty}^{(i,j)}(\bs \theta^*, \bs \theta) = 0 \implies \Lambda_{s,\infty}^{(i,j)}(\bs \theta^*, \bs \theta') = 0$ and that $M^{(i,j)}_{r,\infty}(\bs \theta^*, \bs \theta) = 0 \implies M^{(i,j)}_{r,\infty}(\bs \theta^*, \bs \theta') = 0 $. Then we have that if:
\begin{align*}
\bar{ \Lambda}^{(i,j)}_{\tau_r, \infty}(\bs \theta^*, \bs \theta)  &= \left[1- Q^{(i,j)}_{\tau_r}(\bs \theta) \right]\circ  \Lambda^{(i,j)}  _{\tau_r, \infty}(\bs \theta^*, \bs \theta) \\
&+ \frac{M^{(i,j)}_{r,\infty}(\bs \theta^*,\bs \theta^*)}{M^{(i,j)}_{r,\infty}(\bs \theta^*,\bs \theta)} \circ\left[  \Lambda^{(i,j)}_{\tau_r, \infty}(\bs \theta^*, \bs \theta)\circ  Q ^{(i,j)} _{\tau_r}(\bs \theta) \right] = 0,
\end{align*}
which is well defined by the inductive hypothesis, then either:
\begin{itemize}
    \item $Q^{(i,j)}_{\tau_r}(\bs \theta)<1$, and then $\Lambda^{(i,j)}  _{\tau_r, \infty}(\bs \theta^*, \bs \theta) = 0 \implies  \Lambda^{(i,j)}  _{\tau_r, \infty}(\bs \theta^*, \bs \theta') \implies \bar   \Lambda^{(i,j)}  _{\tau_r,\infty}(\bs \theta^*, \bs \theta') = 0$, hence the first term is $0$, or
    \item $Q^{(i,j)}_{\tau_r}(\bs \theta) =1$, and then $  \Lambda^{(i,j)}  _{\tau_r}(\bs \theta^*, \bs \theta) = 0 \implies  \Lambda^{(i,j)}  _{\tau_r,\infty}(\bs \theta^*, \bs \theta') \implies \bar   \Lambda^{(i,j)}  _{\tau_r,\infty}(\bs \theta^*, \bs \theta') = 0$,  so that the right hand term is $0$.
\end{itemize}
This along with using the same reasoning used in the $r=1$ case gives:
\begin{equation*}
    \Lambda_{\tau_r,\infty}^{(i,j)}(\bs \theta^*, \bs \theta) = \left(\+1_m^\top \bs {\bar \Lambda}_{\tau_r-1,\infty}^{(\cdot,i)}(\bs \theta^*, \bs \theta)\right)K^{(i,j)}_{1,\bs \eta\left(\+1_m^\top \bs{\bar \Lambda}_{\tau_r-1,\infty}(\bs \theta^*, \bs \theta)\right)}(\bs \theta) = 0,
\end{equation*}
 implies
\begin{equation*}
    \Lambda_{\tau_r,\infty}^{(i,j)}(\bs \theta^*, \bs \theta') = \left(\+1_m^\top \bs {\bar \Lambda}_{\tau_r-1,\infty}^{(\cdot,i)}(\bs \theta^*, \bs \theta')\right)K^{(i,j)}_{1,\bs \eta\left(\+1_m^\top \bs{\bar \Lambda}_{\tau_r-1,\infty}(\bs \theta^*, \bs \theta')\right)}(\bs \theta') = 0.
\end{equation*}
Using this and further using identical inductive reasoning to the $r=1$ case we see that for all $s = \tau_{r}+1,\dots, \tau_{r+1}$, $\Lambda_{s,\infty}^{(i,j)}(\bs \theta^*, \bs \theta) = 0 \implies \Lambda_{s,\infty}^{(i,j)}(\bs \theta^*, \bs \theta') = 0$ and further that $M^{(i,j)}_{r+1,\infty}(\bs \theta^*, \bs \theta) = 0 \implies M^{(i,j)}_{r+1,\infty}(\bs \theta^*, \bs \theta') = 0 $. This completes the inductive proof.
\end{proof}
The following lemma is used in the proof of lemma \ref{Zcontfn}.
\begin{lemma}\label{MMn}
Let assumptions \ref{as:params} - \ref{as:init} hold. For all $ \bs \theta \in \Theta$, $n \in \mathbb{N}$, $(i,j) \in [m]^2$, and  $r \in \{1, \dots, R\}$:
\begin{equation*}
    M^{(i,j)}_{r, \infty}(\bs \theta^*, \bs \theta) = 0 \implies M^{(i,j)}_{r, n}(\bs \theta) = 0, \quad \Ps\text{-}a.s.
\end{equation*}
\end{lemma}

\begin{proof}
Fix arbitrary $ \bs \theta \in \Theta$ and $n \in \mathbb{N}$. We will prove that for all $r = 1, \dots, R$, $s = \tau_{r-1}+1,\dots, \tau_r$ and $i,j\in[m]$ the following two implications hold:
\begin{align*}
\Lambda_{s,\infty}^{(i,j)}(\bs \theta^*, \bs \theta) = 0 &\implies \Lambda_{s,n}^{(i,j)}(\bs \theta) = 0,\quad a.s.\\
M^{(i,j)}_{r,\infty}(\bs \theta^*, \bs \theta) = 0 &\implies M^{(i,j)}_{r,n}(\bs \theta^*, \bs \theta') = 0,\quad a.s.
\end{align*}
The induction is on $r$ and $s$. Consider $r=1$. We will first show that for all $s \in \tau_0+1,\dots, \tau_1$ that  $\Lambda_{s,\infty}^{(i,j)}(\bs \theta^*, \bs \theta) = 0 \implies \Lambda_{s,\infty}^{(i,j)}(\bs \theta^*, \bs \theta') = 0$ by induction on $s$. We have for $t-1$ case:
\begin{equation*}
    \Lambda_{1,\infty}^{(i,j)}(\bs \theta^*, \bs \theta)  = \lambda_{0,\infty}^{(i)}(\bs \theta^*, \bs \theta)K^{(i,j)}_{1,\bs \eta(\bs \lambda_{0,\infty}(\bs \theta^*,\bs \theta))}(\bs \theta) = 0,
\end{equation*}
which implies that either:
\begin{itemize}
    \item $\lambda^{(i)}_{0,\infty}(\bs \theta^*,\bs \theta) =  \lambda^{(i)}_{0,\infty}(\bs \theta) = 0$, in which case $  \lambda^{(i)}_{0,n}(\bs \theta) = 0$ or
    \item $K^{(i,j)}_{1,\bs \eta(\lambda_{0,\infty}(\bs \theta^*,\bs \theta) }(\bs \theta) = 0$, in which case $K^{(i,j)}_{1,\bs \eta(\bs \lambda_{0,n}(\bs \theta))}(\bs \theta) = 0$,
\end{itemize}
so that:
\begin{equation*}
    \Lambda_{1,n}^{(i,j)}(\bs \theta^*, \bs \theta')  = \lambda^{(i)}_{0,n}(\bs \theta))K^{(i,j)}_{1,\bs \eta(\bs \lambda_{0,n}(\bs \theta))}(\bs \theta) = 0.
\end{equation*}

Now assume that given $\Lambda_{s-1,\infty}^{(i,j)}(\bs \theta^*, \bs \theta) = 0 \implies \Lambda_{s-1,n}^{(i,j)}(\bs \theta) =  0 \quad a.s.$, then:
\begin{equation*}
    \Lambda_{s,\infty}^{(i,j)}(\bs \theta^*, \bs \theta) = \left(\+1_m^\top \bs \Lambda_{s-1,\infty}^{(\cdot,i)}(\bs \theta^*, \bs \theta)\right)K^{(i,j)}_{1,\bs \eta\left(\+1_m^\top \bs \Lambda_{s-1,\infty}(\bs \theta^*, \bs \theta)\right)}(\bs \theta) = 0,
\end{equation*}
which in turn implies either:
\begin{itemize}
    \item $\left(\+1_m^\top \bs \Lambda_{s-1,\infty}^{(\cdot,i)}(\bs \theta^*, \bs \theta)\right)= 0_m$, which implies $\left(\+1_m^\top \bs \Lambda_{s-1,n}^{(\cdot,i)}(\bs \theta)\right) = 0$ a.s.; or
    \item $K^{(i,j)}_{1,\bs \eta\left(\+1_m^\top \bs \Lambda_{s-1,\infty}^{(\cdot,i)}(\bs \theta^*, \bs \theta)\right)}(\bs \theta)  = 0$, which implies $ K^{(i,j)}_{1,\bs \eta\left(\+1_m^\top \bs \Lambda_{s-1,n}(\bs \theta)\right)}(\bs \theta) = 0$ a.s..
\end{itemize}
Together we find:
\begin{equation*}
    \Lambda_{s,n}^{(i,j)}(\bs \theta)= \left(\+1_m^\top \bs \Lambda_{s-1,n}^{(\cdot,i)}(\bs \theta)\right)K^{(i,j)}_{1,\bs \eta\left(\+1_m^\top \bs \Lambda_{s-1,n}(\bs \theta)\right)}(\bs \theta) = 0,
\end{equation*}
completing the intermediary induction on $s$. Now consider:
 \begin{align*}
M^{(i,j)}_{1,\infty}(\bs \theta^*, \bs \theta) &= \sum_{s = 1}^{\tau_{1}}  \Lambda^{(i,j)}_{s,\infty}(\bs \theta^*, \bs \theta) \circ Q^{(i,j)}_s(\bs \theta) = 0,
\end{align*}
then for all $s = \tau_{0}+1, \dots, \tau_1$ either
\begin{itemize}
    \item $ \Lambda^{(i,j)}_{s,\infty}(\bs \theta^*, \bs \theta) =0$, which implies $ \Lambda^{(i,j)}_{s,n}(\bs \theta) = 0$; or
    \item $Q^{(i,j)}_s(\bs \theta) = 0 $,
\end{itemize}
and hence:
 \begin{align*}
M^{(i,j)}_{1,n}(\bs \theta^*, \bs \theta') &= \sum_{s = 1}^{\tau_{1}}  \Lambda^{(i,j)}_{s,n}(\bs \theta)\circ Q^{(i,j)}_s(\bs \theta) = 0.
\end{align*}
This completes the case $r=1$. 

Now assume that for all  $s = \tau_{r-1}+1,\dots, \tau_r$,  $\Lambda_{s,\infty}^{(i,j)}(\bs \theta^*, \bs \theta) = 0$, which implies $\Lambda_{s,\infty}^{(i,j)}(\bs \theta^*, \bs \theta') = 0$ and that $M^{(i,j)}_{r,\infty}(\bs \theta^*, \bs \theta) = 0$, which implies $M^{(i,j)}_{r,\infty}(\bs \theta^*, \bs \theta') = 0 $. Then
\begin{align*}
\bar{ \Lambda}^{(i,j)}_{\tau_r, \infty}(\bs \theta^*, \bs \theta)  &= \left[1- Q^{(i,j)}_{\tau_r}(\bs \theta) \right]\circ  \Lambda^{(i,j)}  _{\tau_r}(\bs \theta^*, \bs \theta) \\
&+ \frac{M^{(i,j)}_{r,\infty}(\bs \theta^*,\bs \theta)}{M^{(i,j)}_{r,\infty}(\bs \theta^*,\bs \theta)} \circ\left[  \Lambda^{(i,j)}_{\tau_r, \infty}(\bs \theta^*, \bs \theta)\circ  Q ^{(i,j)} _{\tau_r}(\bs \theta) \right] = 0,
\end{align*}
and either:
\begin{itemize}
    \item $Q^{(i,j)}_{\tau_r}(\bs \theta)<1$, which implies $  \Lambda^{(i,j)}  _{\tau_r,\infty}(\bs \theta^*, \bs \theta) = 0 \implies  \Lambda^{(i,j)}  _{\tau_r,n}(\bs \theta) \implies \bar   \Lambda^{(i,j)}  _{\tau_r,n}(\bs \theta) = 0$ a.s., so that the first term is $0$; or
    \item $Q^{(i,j)}_{\tau_r}(\bs \theta) =1$ which implies $ \Lambda^{(i,j)}  _{\tau_r,\infty}(\bs \theta^*, \bs \theta) = 0 \implies  \Lambda^{(i,j)}  _{\tau_r,n}(\bs \theta) \implies \bar   \Lambda^{(i,j)}  _{\tau_r,n}(\bs \theta) = 0$,  a.s., so that the right hand term is $0$.
\end{itemize}
This along with using the same reasoning used in the $r=1$ case tells us that given:
\begin{equation*}
    \Lambda_{\tau_r,\infty}^{(i,j)}(\bs \theta^*, \bs \theta) = \left(\+1_m^\top \bs {\bar \Lambda}_{\tau_r-1,\infty}^{(\cdot,i)}(\bs \theta^*, \bs \theta)\right)K^{(i,j)}_{1,\bs \eta\left(\+1_m^\top \bs{\bar \Lambda}_{\tau_r-1,\infty}(\bs \theta^*, \bs \theta)\right)}(\bs \theta) = 0,
\end{equation*}
 which implies
\begin{equation*}
    \Lambda_{\tau_r, n}^{(i,j)}(\bs \theta) = \left(\+1_m^\top \bs {\bar \Lambda}_{\tau_r-1,n}^{(\cdot,i)}(\bs \theta)\right)K^{(i,j)}_{1,\bs \eta\left(\+1_m^\top \bs{\bar \Lambda}_{\tau_r-1,n}(\bs \theta)\right)}(\bs \theta) = 0 \quad a.s.
\end{equation*}

Using this and further using identical inductive reasoning to the $r=1$ case we see that for all $s \in \{\tau_{r}+1,\dots, \tau_{r+1}\}$ we have $\Lambda_{s,\infty}^{(i,j)}(\bs \theta^*, \bs \theta) = 0 \implies \Lambda_{s,n}^{(i,j)}(\bs \theta) = 0$ a.s. and further that $M^{(i,j)}_{r+1,\infty}(\bs \theta^*, \bs \theta) = 0$, which implies $M^{(i,j)}_{r+1,n}(\bs \theta) = 0$ a.s. This completes the inductive proof.
\end{proof}

\begin{lemma}\label{lem:contM}
Let assumptions \ref{as:params}- \ref{as:init} hold. For all $\bs \theta^*  \in \Theta$ and $r\geq 1$, the function $\bs \theta \mapsto \+M_{r,\infty}(\bs \theta^*, \bs \theta)$ is continuous on $\Theta$.
\end{lemma}
\begin{proof}
The arguments are very similar to those in the proof of \ref{lem:contmu}, but making use of lemma \ref{MM}, so we omit them.
\end{proof}

\begin{proof}[Proof of Proposition \ref{Yfilt}]
The proof is by induction on $r$. Consider $r=1$. Note that $n^{-1} \bs {\bar{ \lambda}}_{0 , n}(\bs \theta) := n^{-1}\bs \lambda_{0,n}(\bs \theta) \asls \bs \lambda_{0,\infty}(\bs \theta^*, \bs \theta)$ by assumption \ref{as:init}.
Now let  $ t = 1, \dots, \tau_1 -1$ and assume that $n^{-1}\bar{\bs \lambda}_{t-1,n}(\bs \theta) \asls \bar{\bs \lambda}_{t-1,\infty}(\bs \theta,\bs \theta^*)$. Then:
\begin{align*}
n^{-1}\bs \Lambda_{t,n}(\bs \theta) &= n^{-1}( \bar{\bs \lambda}_{t-1,n}(\bs \theta) \otimes \+1_m)\circ \+K_{t, \bs \eta \left(\bar{\bs\lambda}_{t-1,n}(\bs \theta) \right)}\\
&\asls \left( \bar{\bs  \lambda}_{t, \infty}(\bs \theta,\bs \theta^*) \otimes \+1_m \right) \circ \+ K_{t, \bs \eta  \left( \bar{\bs  \lambda}_{t, \infty}(\bs \theta^*,\bs \theta) \right)} \\
&= \+ \Lambda_{t, \infty}(\bs \theta^*,\bs \theta).
\end{align*}
By the CMT, a further application yields:
\begin{equation*}
n^{-1}\bar{\bs \lambda}_{t,n}(\bs \theta) = n^{-1}(\+1_m^\top \bs \Lambda_{t,n}(\bs \theta))^\top \asls (\+1_m^\top\+ \Lambda_{t, \infty}(\bs \theta^*,\bs \theta))^\top = \bar{\bs \lambda}_{t,\infty}(\bs \theta^*, \bs \theta).
\end{equation*}
Then by induction on $t$ we have that:
\begin{align*}
n^{-1}\bs \Lambda_{t,n}(\bs \theta) \asls \+ \Lambda_{t, \infty}(\bs \theta^*,\bs \theta),
\end{align*}
for all $t = 1, \dots, \tau_1$, this means that:
\begin{align*}
n^{-1}\+ M_{1,n}(\bs \theta) = n^{-1}&\sum_{s = 1}^{\tau_1}\bs \Lambda_{t,n}(\bs \theta)\circ \+Q_s(\bs \theta)  \\
\asls &\sum_{s = 1}^{\tau_1} \bs \Lambda_{t,\infty}(\bs \theta^*, \bs \theta)\circ \+Q_s(\bs \theta) =  \+ M_{1, \infty}(\bs \theta^*, \bs \theta).
\end{align*}
Now for general $r = 1, \dots, R$ assume that $\bar{\bs \lambda}_{\tau_{r-1},n}(\bs \theta) \asls \bar{\bs \lambda}_{\tau_{r-1},\infty}(\bs \theta^*, \bs \theta)$. Using identical reasoning to the $r=1$ case, we find that for all $t = \tau_{r-1}+1, \dots, \tau_r$:
\begin{align*}
n^{-1}\bs \Lambda_{t,n}(\bs \theta) \asls \+ \Lambda_{t, \infty}(\bs \theta^*,\bs \theta),
\end{align*}
which in turn implies by the CMT that:
\begin{align*}
n^{-1}\+ M_{r,n}(\bs \theta) \asls \+ M_{r, \infty}(\bs \theta^*,\bs \theta),
\end{align*}
Writing out the definition of $M_{r, \infty}(\bs \theta^*,\bs \theta^*)$, proposition \ref{prop:Yas} gives $n^{-1}\bar{\+Y}_{{r}} \asls \+  M_{r, \infty}(\bs \theta^*,\bs \theta^*)$. Then by the CMT,
\begin{align*}
n^{-1}\bar{\bs \Lambda}_{\tau_r ,n}(\bs \theta) &= \left(\+1_m \otimes \+1_m - \+Q_{\tau_r}^* \right)\circ n^{-1}\bs \Lambda_{\tau_r,n}(\bs \theta) \\
&+   \left[ n^{-1}\bar{\+Y}_{{r}}\oslash n^{-1}\+ M_{r,n}(\bs \theta)\right]  \circ \left(\left[ n^{-1}\bs \Lambda_{\tau_{r},n}(\bs \theta) \circ \+ Q_{\tau_r}(\bs \theta) \right] \right) \\
&\asls \left(\+1_m \otimes \+1_m - \+Q_{\tau_r}(\bs \theta) \right)\circ \+ \Lambda_{\tau_r, \infty}(\bs \theta^*,\bs \theta) \\
&+  \left[\+  M_{r, \infty}(\bs \theta^*,\bs \theta^*)\oslash \+  M_{r, \infty}(\bs \theta^*,\bs \theta)\right] \circ \left[  \+ \Lambda_{\tau_r, \infty}(\bs \theta^*,\bs \theta)\circ \+ Q_{\tau_r}(\bs \theta) \right]  \\
&=  \bar{\+  \Lambda}_{\tau_r, \infty}(\bs \theta^*,\bs \theta).
\end{align*}
We note here that the left hand side of the above display is almost surely well defined since, by lemma \ref{MY}, for all $n \in \mathbb{N}$ and $i,j\in[m]$ if there is positive probability that $M^{(i,j)}_{r,n}(\bs \theta) = 0$ then $\bar{Y}^{(i,j)}_{{r}} = 0$ $\Ps$-a.s., in which case we invoke the convention $\frac{0}{0} \coloneqq 0$. The right hand side of the limit is well defined since for all $i,j \in [m]$ we have $M^{(i,j)}_{r, \infty}(\bs \theta^*,\bs \theta^*) = 0 \iff  M^{(i,j)}_{r, \infty}(\bs \theta^*,\bs \theta) = 0$ by lemma \ref{MM}, in which case we again invoke the convection $\frac{0}{0}\coloneqq 0$ . A further application of the CMT gives:
\begin{equation*}
n^{-1} \bar{\bs \lambda}_{\tau_r,n}(\bs \theta) = n^{-1}\left( \+1_m^\top \bar{\bs \Lambda}_{\tau_r}(\bs \theta)\right)^\top \asls (\+ 1_m^\top\bar{\+  \Lambda}_{\tau_r, \infty}(\bs \theta^*,\bs \theta))^\top = \bar{\bs \lambda}_{\tau_r, \infty}(\bs \theta^*, \bs \theta).
\end{equation*}
This completes the proof.
\end{proof}


\subsection{Contrast functions}\label{sec:contrast_functions}

\begin{definition}
Let $(\mathcal{H}_n)_{n \geq 1}$ be a sequence of random functions $\mathcal{H}_n:\theta\in\Theta \mapsto \mathcal{H}_n(\theta) \in \mathbb{R}$ where $\Theta$ is a metric space. We say that $(\mathcal{H}_n)_{n \geq 1}$ are stochastically equicontinuous if there exists an event $M$ of probability $1$, such that for all $\varepsilon>0$ and $\omega \in M$, there exists $N(\omega)$ and $\delta>0$ such that $n>N(\omega)$ implies:

\begin{equation*}
\sup_{|\theta_1-\theta_2|<\delta}| \mathcal{H}_n(\omega, \theta_1) - \mathcal{H}_n(\omega, \theta_2)| < \epsilon.
\end{equation*}

\end{definition}

\begin{lemma}\label{stochascoli}
Assume $\Theta$ is a compact metric space and let $(\mathcal{H}_n)_{n \geq 1}$ be a sequence of random functions $\mathcal{H}_n:\theta\in\Theta \rightarrow \mathcal{H}_n(\theta)\in\mathbb{R}$. If there exists a continuous function $\mathcal{H}$ such that for all $\theta \in \Theta$ we have  $| \mathcal{H}_n(\theta) - \mathcal{H}(\theta)| \overset{a.s.}{\rightarrow} 0$, and $(\mathcal{H}_n)_{n \geq 1}$ are stochastically equicontinuous, then:
\begin{equation*}
\sup_{\theta \in \Theta}| \mathcal{H}_n(\theta) - \mathcal{H}(\theta)| \overset{a.s.}{\rightarrow} 0.
\end{equation*}
That is $\mathcal{H}_n(\theta)$ converges to $\mathcal{H}(\theta)$ almost surely as $n \rightarrow \infty$, uniformly in $\theta$.
\end{lemma}
\begin{proof}
See \cite{andrews1992generic}.
\end{proof}

\subsubsection{Case (I)}
We have:
\begin{equation}\label{contfin}
\begin{aligned}
    n^{-1} \ell_n(\bs \theta) - n^{-1}\ell_n(\bs \theta^*) = \sum_{t=1}^T\Bigg\{&\frac{\+y_{t}^\top}{n} \log\big(\bs \mu_{t,n}(\bs \theta) \oslash \bs \mu_{t,n}(\bs \theta^*) \big) \\
&- n^{-1} \+1_m^\top \left[\bs \mu_{t,n}( \bs \theta) - \bs \mu_{t,n}(\bs \theta^*)    \right]\Bigg\}.
\end{aligned}
\end{equation}
The following proposition details the limit of \eqref{contfin}.
\begin{proposition}\label{contrastfnx}
Let assumptions \ref{as:compact}-\ref{as:init} hold.  Then:
\begin{equation}\label{eq:case_I_contrast}
     n^{-1}  \ell_n(\bs \theta) - n^{-1}\ell_n(\bs \theta^*) \asls -\sum_{t=1}^T \mathrm{KL}\left(\mathrm{Pois}\left[\bs \mu_{t, \infty}(\bs \theta^*, \bs \theta^*)\right] \| \mathrm{Pois}\left[\bs \mu_{t, \infty}(\bs \theta^*,\bs \theta)\right]\right),
\end{equation}
 uniformly in $\bs \theta$.
\end{proposition}
\begin{proof}
For $n \in \mathbb{N}$,  we define a random function $\mathcal{C}_n:\Theta \rightarrow \mathbb{R}$, as $\mathcal{C}_{n}(\bs \theta) \coloneqq \sum_{t=1}^T \mathcal{C}_{t,n}(\bs \theta)$  where:
\begin{align*}
\mathcal{C}_{t,n}(\bs \theta) &\coloneqq \frac{\+y_{t}^\top}{n} \log\big(\bs \mu_{t,n}( \bs \theta) \oslash \bs \mu_{t,n}( \bs \theta^*) \big)
- n^{-1} \+1_m^\top \left[\bs \mu_{t,n}( \bs \theta) - \bs \mu_{t,n}( \bs \theta^*)    \right] \\
&= \sum_{i=1}^m \frac{y_t^{(i)}}{n}\log \frac{\mu_{t, n}^{(i)}(\bs \theta)}{\mu_{t, n}^{(i)}(\bs \theta^*)} - n^{-1}\left[\mu_{t, n}^{(i)}(\bs \theta) - \mu_{t, n}^{(i)}(\bs \theta^*)\right],
\end{align*}
with the convention $0\log 0 \coloneqq 0$. To see that $\mathcal{C}_{t,n}(\bs \theta)$ is almost surely well defined, consider the following cases for each $i \in [m]$. If both $\mu_{t,n}^{(i)}(\bs \theta) > 0$ and $\mu_{t,n}^{(i)}(\bs \theta^*) > 0$, $\Ps$-a.s., then the $\log$ of the ratio of these terms is almost surely well defined. If $\mu_{t,n}^{(i)}(\bs \theta) = 0$ or $\mu_{t,n}^{(i)}(\bs \theta^*) = 0$ with positive probability, then  $y^{(i)}_t = 0$ $\Ps$-a.s. by lemma \ref{muylemma} and we invoke the convention $0\log 0 \coloneqq 0$.

We shall show that for $t = 1,\dots, T$,
\begin{equation*}
    \mathcal{C}_{t,n}(\bs \theta) \asls - \mathrm{KL}(\mathrm{Pois}[\bs \mu_{t, \infty}(\bs \theta^*, \bs \theta^*)] \| \mathrm{Pois}[\bs \mu_{t, \infty}(\bs \theta^*,\bs \theta)]) , \text{ uniformly in } \bs \theta.
\end{equation*}
The proof consists of showing pointwise convergence and then stochastic equicontinuity of $\mathcal{C}_{t,n}(\bs \theta)$. Uniform almost sure convergence then follows by lemma \ref{stochascoli}.

Fix $t \in \{1, \dots, T\}$ and note that $\frac{\+y_t}{n}\asls \bs \mu_{t,\infty}(\bs \theta^*, \bs \theta^*)$ by proposition \ref{prop:yas} (see remark \ref{rem:lam_mu}), and by proposition \ref{vecfilt}, 
$n^{-1}\bs\mu_{t,n}(\bs \theta) \asls \bs \mu_{t}(\bs \theta^*, \bs \theta)
$. We claim that by the CMT:
\begin{align}
\mathcal{C}_{t,n}(\bs \theta) &= \frac{\+y_{t}^\top}{n} \log\left(\bs \mu_{t,n}(\bs \theta)\oslash \bs \mu_{t,n}(\bs \theta^*)\right) - n^{-1}\left[\bs \mu_{t,n}(\bs \theta) - \bs \mu_{t,n}(\bs \theta^*) \right]^\top \+1_m \\ \label{fincon}
&= \sum_{i=1}^m \frac{y_t^{(i)}}{n}\log \frac{\mu_{t, n}^{(i)}(\bs \theta)}{\mu_{t, n}^{(i)}(\bs \theta^*)} - n^{-1}\left[\mu_{t, n}^{(i)}(\bs \theta) - \mu_{t, n}^{(i)}(\bs \theta^*)\right] \\
\label{consum}
&\asls \sum_{i=1}^m \mu_{t,\infty}^{(i)}(\bs \theta^*, \bs \theta^*)\log \frac{\mu_{t, \infty}^{(i)}(\bs \theta^*, \bs \theta)}{\mu_{t, \infty}^{(i)}(\bs \theta^*, \bs \theta^*)} - \left[\mu_{t, \infty}^{(i)}(\bs \theta^*, \bs \theta) - \mu_{t, \infty}^{(i)}(\bs \theta^*, \bs \theta^*)\right] \\
&=-\mathrm{KL}\left(\text{Pois}\left[\bs \mu_{t, \infty}(\bs \theta^*, \bs \theta^*)]\right) \| \text{Pois}\left[\bs \mu_{t, \infty}(\bs \theta^*,\bs \theta)\right]\right).
\end{align}
To see that the limit is well defined consider the cases for each $i \in [m]$, either:
\begin{itemize}
    \item $ \mu_{t,\infty}^{(i)}(\bs \theta^*,\bs \theta^*)>0$ and $ \mu_{t,\infty}^{(i)}(\bs \theta^*,\bs \theta)>0.$ In this case all functions in the sequence $\{\mathcal{C}_{t,n}(\bs \theta)\}_{n\geq1}$ and its limit are well defined; or
    \item $\mu_{t,\infty}^{(i)}(\bs \theta^*,\bs \theta)>0$ and $\mu_{t,\infty}^{(i)}(\bs \theta^*,\bs \theta^*) = 0$, or $\mu_{t,\infty}^{(i)}(\bs \theta^*,\bs \theta)= 0$ and $ \mu_{t,\infty}^{(i)}(\bs \theta^*,\bs \theta^*)>0$. This case is prohibited by lemma \ref{mumulemma}; or
    \item $ \mu_{t,\infty}^{(i)}(\bs \theta^*,\bs \theta) = \mu_{t,\infty}^{(i)}(\bs \theta^*,\bs \theta^*) = 0$. In this case, by lemmas \ref{muylemma} and \ref{mumunlemma},  we have that for all $n \in \mathbb{N}$  $\mu_{t,n}^{(i)}(\bs \theta^*) = 0$, $\mu_{t,n}^{(i)}(\bs \theta) = 0$, and  $y^{(i)}_t = 0 \quad \mathbb{P}_n^{\bs \theta^*}a.s.$, so that the $i$th term disappears from \eqref{fincon}  and \eqref{consum} with probability $1$ by the convention $0\log 0 \coloneqq 0$.
\end{itemize}
Hence we have shown the convergence of:
\begin{equation*}
    \mathcal{C}_{t,n}(\bs \theta) \asls - \mathrm{KL}(\bs \mu_{t, \infty}(\bs \theta^*, \bs \theta^*) \| \bs \mu_{t, \infty}(\bs \theta^*,\bs \theta)),
\end{equation*}
point-wise in $\bs \theta \in \Theta$.

Next we show that $(\mathcal{C}_{t,n})_{n\geq1}$ are stochastically equicontinuous. Let $\bs f\in \mathbb{R}^m$ and $E \subset \Omega$ such that $\mathbb{P}^{\bs\theta^*}(E)=1$.  Let  $\bs \theta_1, \bs \theta_2 \in \Theta$, $\omega \in E$, and $\varepsilon>0$. Firstly we will show the stochastic equicontinuity of $n^{-1}\bs \mu_{t,n}(\bs \theta)^\top\bs f$ for any $\bs f\in \mathbb{R}^m$. Let $\varepsilon_0>0$ and write by the triangle inequality:
\begin{align}
\left|n^{-1}\bs \mu_{t,n}(\bs \theta_1)^\top\bs f - n^{-1}\bs \mu_{t,n}(\bs \theta_2)^\top\bs f \right| &\leq \left|n^{-1}\bs \mu_{t,n}(\bs \theta_1)^\top\bs f - \bs \mu_{t, \infty}(\bs \theta^*,\bs \theta_1)^\top\bs f \right| \\
&+ \left|n^{-1}\bs \mu_{t,n}(\bs \theta_2)^\top\bs f - \bs \mu_{t, \infty}(\bs \theta^*,\bs \theta_2)^\top\bs f \right| \\
&+ \left|\bs \mu_{t, \infty}(\bs \theta^*, \bs \theta_1)^\top\bs f - \bs \mu_{t, \infty}(\bs \theta^*, \bs \theta_2)^\top\bs f \right|.
\end{align}
 There exists $N(\omega)<\infty$ such that for $n> N(\omega)$ the first two terms are bounded by $\varepsilon_0/3$ by proposition \ref{vecfilt}. Furthermore, since $\bs \theta \mapsto \bs \mu_{t, \infty}(\bs \theta^*, \bs \theta)$ is continuous by lemma \ref{lem:contmu} there exists a $\delta_0>0$ such that:

$$\left\|\bs \theta_1 - \bs \theta_2 \right\|_\infty<\delta_0 \implies \left|\bs \mu_{t, \infty}(\bs \theta^*,\bs \theta_1)^\top\bs f - \bs \mu_{t, \infty}(\bs \theta^*,\bs \theta_2)^\top\bs f \right|< \varepsilon_0/3.$$
Hence we have shown stochastic equicontinuity of $(n^{-1}\bs \mu_{t,n}(\bs \theta)^\top\bs f)_{n\geq1}$. Now, consider $\mathcal{C}_{t,n}$:

\begin{align}
\left| \mathcal{C}_{t,n}(\bs \theta_1) - \mathcal{C}_{t,n}(\bs \theta_2) \right| &\leq \left|\sum_{i=1}^mn^{-1}y_{t}^{(i)} \log \frac{ \mu_{t,n}^{(i)}(\bs \theta_1)}{ \mu_{t,n}^{(i)}(\bs \theta_2)}\right| \\
\label{part2eq}
&+ \left|n^{-1}\left[\bs \mu_{t,n}(\bs \theta_1)^\top - \bs \mu_{t,n}(\bs \theta_2)^\top \right] \+1_m \right|.
\end{align}
By what has been proven already we can choose $\delta_1$ and $N_1(\omega)$ to bound \eqref{part2eq} by $\varepsilon/2$. Let $\varepsilon_2>0$, by proposition \ref{prop:yas} there exists $N_2(\omega)$ such that for $n> N_2(\omega)$:

\begin{equation}\label{eq:sum_ylog}
\left|\sum_{i=1}^mn^{-1}y_{t}^{(i)} \log \frac{ \mu_{t,n}^{(i)}(\bs \theta_1)}{ \mu_{t,n}^{(i)}(\bs \theta_2)}\right| < \sum_{i=1}^m |\mu_{t, \infty}^{(i)}(\bs \theta^*,\bs \theta^*) + \varepsilon_2| \left|\log \frac{n^{-1} \mu_{t,n}^{(i)}(\bs \theta_1)}{n^{-1} \mu_{t,n}^{(i)}( \bs \theta_2)}\right|
\end{equation}
Furthermore, for each $i \in [m]$ either:
\begin{itemize}
\item  there is positive probability that either $\mu^{(i)}_{t,n}(\bs \theta_1)=0$ or ${\mu^{(i)}_{t,n}(\bs \theta_2)=0}$, then the $i$th term of the sum on the l.h.s. of \eqref{eq:sum_ylog} disappears since $y_t^{(i)} = 0$ with probability $1$ by lemma \ref{muylemma}, and we invoke the convention $0\log 0 \coloneqq 0$; or
\item ${\mu^{(i)}_{t,n}(\bs \theta_1)> 0}$ and $ \mu^{(i)}_{t,n}(\bs \theta_2)>0$ almost surely. Then by continuity of $\log$ on $\mathbb{R}_{>0}$ there exists  a $\delta_3^{(i)}>0$ such that for  $| n^{-1}  \mu^{(i)}_{t,n}(\bs \theta_1) - n^{-1} \mu^{(i)}_{t,n}(\bs \theta_2)|<\delta_3^{(i)}$:
\begin{align*}
\left|\log \frac{n^{-1} \mu^{(i)}_{t,n}(\bs \theta_1)}{n^{-1} \mu^{(i)}_{t,n}(\bs \theta_2)}\right| &=  \left|\log n^{-1} \mu^{(i)}_{t,n}(\bs \theta_1) - \log{n^{-1} \mu_{t,n}^{(i)}(\bs \theta_2)}\right| \\
&< \frac{\varepsilon}{2m|\mu_{t, \infty}^{(i)}(\bs \theta^*,\bs \theta^*) + \varepsilon_2|}.
\end{align*}
\end{itemize}
By stochastic equicontinuity of $(n^{-1} \bs \mu_{t,n}^\top \bs f)_{n\geq1}$ there exists $N_3(\omega)$ and $\delta_2$ such that for \newline $n> N_3(\omega)$ and $\|\bs \theta_1 - \bs \theta_2\|_\infty< \delta_2$ we have that $\|n^{-1}  \bs \mu_{t,n}(\bs \theta_1) - n^{-1}  \bs \mu_{t,n}(\bs \theta_2)\|_\infty<\min_i \delta_3^{(i)}$ so that:
\begin{align*}
\left|\sum_{i=1}^mn^{-1}y_{t}^{(i)} \log \frac{ \mu_{t,n}^{(i)}(\bs \theta_1)}{ \mu_{t,n}^{(i)}(\bs \theta_2)}\right| &< \sum_{i=1}^m |\mu_{t, \infty}^{*(i)}(\bs \theta^*,\bs \theta^*) + \varepsilon_2|\frac{\varepsilon}{2m|\mu_{t, \infty}^{*(i)}(\bs \theta^*,\bs \theta^*) + \varepsilon_2|}\\
&=\varepsilon/2.
\end{align*}
Hence choosing $\delta = \min(\delta_1, \delta_2)$ and $N(\omega) = \max(N_1(\omega),N_2(\omega),N_3(\omega))$ we have that for \newline ${\|\bs \theta_1 - \bs \theta_2\|_\infty< \delta}$ and $n> N(\omega)$:

\begin{equation*}
\left| \mathcal{C}_{t,n}(\bs \theta_1) - \mathcal{C}_{t,n}(\bs \theta_2) \right| < \varepsilon/2 + \varepsilon/2 =  \varepsilon.
\end{equation*}
Hence we have established the stochastic equicontinuity of $\mathcal{C}_{t,n}$. This along with the already proven pointwise convergence establishes uniform almost sure convergence by lemma \ref{stochascoli} and completes the proof.
\end{proof}

%

\subsubsection{Case (II)}

 We have:
\begin{align*}
    n^{-1}\mathcal{L}_n(\bs \theta) - n^{-1}\mathcal{L}_n(\bs \theta^*) = \sum_{r=1}^R\Big\{ &\+1_m^\top \left[n^{-1}\bar{\+ Y}_{r} \circ \log\left({\+M_{r,n}(\bs \theta)}\oslash{ \+M_{r,n}(\bs \theta^*)}\right)\right]\+1_m \\
&+ n^{-1}\+1_m^\top\left[\+M_{r,n}(\bs \theta) - \+M_{r,n}(\bs \theta^*)\right]\+1_m\Big\}.
\end{align*}

\begin{proposition}\label{Zcontfn}
Let assumptions \ref{as:compact}-\ref{as:init} hold. Then
\begin{equation}\label{eq:case_II_contrast}
    n^{-1} \mathcal{L}_{n}(\bs \theta) - n^{-1}\mathcal{L}_{n}(\bs \theta^*) \asls -\sum_{r=1}^R \mathrm{KL}\left(\mathrm{Pois}\left[\+ M_{r, \infty}(\bs \theta^*,\bs \theta^*) \right]\| \mathrm{Pois}\left[\+ M_{r, \infty}(\bs \theta^*,\bs \theta)\right]\right) ,
\end{equation}
 uniformly in $\bs \theta$.
\end{proposition}

\begin{proof}
The details are similar to lemma \ref{contrastfnx}. For $n \in \mathbb{N}$ define the sequence of random functions $(\mathcal{D}_n( \bs \theta))_{n\geq1}$, $ \mathcal{D}_n( \bs \theta)\coloneqq \sum_{r=1}^R \mathcal{D}_{r,n}(\bs \theta)$, where:
\begin{equation*}
\begin{aligned}
\mathcal{D}_{r,n}(\bs \theta) &\coloneqq \+1_m^\top \left[n^{-1}\bar{\+ Y}_{r} \circ \log\left({\+M_{r,n}(\bs \theta)}\oslash{ \+M_{r,n}(\bs \theta^*)}\right)\right]\+1_m \\
&+ n^{-1}\+1_m^\top\left[\+M_{r,n}(\bs \theta) - \+M_{r,n}(\bs \theta^*)\right]\+1_m \\
 = &\sum_{i=1}^m\sum_{j=1}^m \bar Y_r^{(i,j)}\log \frac{M^{(i,j)}_{r, n}(\bs \theta)}{M^{(i,j)}_{r, n}(\bs \theta^*)} + \left[M^{(i,j)}_{r, n}(\bs \theta) - M^{(i,j)}_{r, n}(\bs \theta^*)\right]
\end{aligned}
\end{equation*}
With the convention $0\log 0 \coloneqq 0$.
To see that this mapping is almost surely well defined, consider the following cases for each $(i,j) \in [m]^2$. If both $M_{r,n}^{(i,j)}(\bs \theta) > 0$, or $M_{r,n}^{(i,j)}(\bs \theta^*) > 0$ $\Ps\text{-}a.s.$, then the $\log$ of each of these terms is almost surely well defined. If either $M_{r,n}^{(i,j)}(\bs \theta) = 0$, or $M_{r,n}^{(i,j)}(\bs \theta^*) = 0$ with positive probability, then  $\bar Y^{(i,j)}_r = 0 $ $\Ps\text{-}a.s.$ by lemma \ref{MY} and we invoke the convention $0\log 0 \coloneqq 0$.

It is enough to show that for each $r \in \{1, \dots, R \}$

$$\mathcal{D}_{r,n}(\bs \theta) \asls - \mathrm{KL}\left(\text{Pois}\left[\+ M_{r, \infty}(\bs \theta^*,\bs \theta^*) \right]\| \text{Pois}\left[\+ M_{r, \infty}(\bs \theta^*,\bs \theta)\right]\right), \text{ uniformly in } \bs \theta.
$$
We show pointwise almost sure convergence and then stochastic equicontinuity. Fix $r \in \{1, \dots, R\}$ and note that by proposition \ref{Yfilt}  $n^{-1}\+M_{r,n}(\bs \theta) \asls \+M_{r,\infty}(\bs \theta^*, \bs \theta)$ for all $\bs \theta \in \Theta$ and $r = 1, \dots, R$. Furthermore by proposition \ref{prop:Yas} $n^{-1}\bar{\+Y}_{r} \asls \+M_{r, \infty}(\bs \theta^*, \bs \theta^*)$ for all $r=1,\dots,R$. We claim that by the CMT:

\begin{align}
\mathcal{D}_{r,n}(\bs \theta) = &\+1_m^\top \left[n^{-1}\bar{\+ Y}_r \circ \log\left({\+M_{n,r}(\bs \theta)}\oslash{ \+M_{n,r}(\bs \theta^*)}\right)\right]\+1_m  \\
+ &n^{-1}\+1_m^\top\left[\+M_{n,r}(\bs \theta) - \+M_{n,r}(\bs \theta^*)\right]\+1_m \\ \label{finconZ}
 = &\sum_{i=1}^m\sum_{j=1}^m n^{-1}\bar Y_r^{(i,j)}\log \frac{M^{(i,j)}_{r, n}(\bs \theta)}{M^{(i,j)}_{r, n}(\bs \theta^*)} + n^{-1}\left[M^{(i,j)}_{r, n}(\bs \theta) - M^{(i,j)}_{r, n}(\bs \theta^*)\right] \\ \label{consumZ}
\asls &\sum_{i=1}^m\sum_{j=1}^m  M^{(i,j)}_{r,  \infty}(\bs \theta^*, \bs \theta^*)\log \frac{M^{(i,j)}_{r,  \infty}(\bs \theta^*, \bs \theta)}{M^{(i,j)}_{r,  \infty}(\bs \theta^*, \bs \theta^*)} + \left[M^{(i,j)}_{r,  \infty}(\bs \theta^*, \bs \theta) - M^{(i,j)}_{r, \infty}(\bs \theta^*, \bs \theta^*)\right] \\
= &-\mathrm{KL}(\+M_{r, \infty}(\bs \theta^*, \bs \theta^*)\| \+M_{r, \infty}(\bs \theta^*, \bs \theta)).
\end{align}
  To see that this limit is indeed almost surely well defined consider the cases for each $i = 1, \dots, m$ and  \newline $j\in[m]$, either:
\begin{itemize}
    \item $ M_{r,\infty}^{(i,j)}(\bs \theta^*,\bs \theta^*)>0$ and $ M_{r,\infty}^{(i,j)}(\bs \theta^*,\bs \theta)>0.$ In this case all functions in the sequence and its limit are well defined. Or
    \item $M_{r,\infty}^{(i,j)}(\bs \theta^*,\bs \theta)>0$ and $M_{r,\infty}^{(i,j)}(\bs \theta^*,\bs \theta^*) = 0$, or $M_{r,\infty}^{(i,j)}(\bs \theta^*,\bs \theta)= 0$ and $ M_{r,\infty}^{(i,j)}(\bs \theta^*,\bs \theta^*)>0$. This case is prohibited by lemma \ref{mumulemma}. Or
    \item $ M_{r,\infty}^{(i,j)}(\bs \theta^*,\bs \theta^*) = M_{r,\infty}^{(i,j)}(\bs \theta^*,\bs \theta^*) = 0$. In this case, by lemmas \ref{muylemma} and \ref{mumunlemma},  we have \newline $M_{r,n}^{(i,j)}(\bs \theta^*) = 0$ and  $\bar Y^{(i,j)}_r = 0 \quad \mathbb{P}_n^{\bs \theta^*}a.s.$, so that the $(i,j)$th term disappears from the sums in \eqref{finconZ} and \eqref{consumZ} by the convention $0\log 0 \coloneqq 0$.
\end{itemize}
Hence we have shown:
$$\mathcal{D}_{r,n}(\bs \theta) \asls - \mathrm{KL}\left(\text{Pois}\left(\+ M_{r, \infty}(\bs \theta^*,\bs \theta^*) \right)\| \text{Pois}\left(\+ M_{r, \infty}(\bs \theta^*,\bs \theta)\right)\right), \text{ pointwise in } \bs \theta.
$$
We now prove stochastic equicontinuity of $(\mathcal{D}_{r,n})_{n\geq1}$. Firstly we will show the stochastic equicontinuity of $(n^{-1} \bs f_1^\top \+ M_{r,n}(\bs \theta)\bs f_2)_{n\geq1}$ for any vectors $\bs f_1,\bs f_2 \in \mathbb{R}^m$. Let $\bs f_1,\bs f_2 \in \mathbb{R}^m$ and $E \subseteq \Omega$ such that $\mathbb{P}^{\theta^*}(E)=1$. Let  $\bs \theta_1, \bs \theta_2 \in \Theta$, $\omega \in E$, and $\varepsilon>0$. Let $\varepsilon_0>0$ and write by the triangle inequality:
\begin{align*}
\left|n^{-1}\bs f_1^\top \+ M_{r,n}(\bs \theta_1)\bs f_2 - n^{-1}\bs f_1^\top \+ M_{r,n}(\bs \theta_2)\bs f_2 \right| &\leq \left|n^{-1}\bs f_1^\top \+ M_{r,n}(\bs \theta_1)\bs f_2 - \bs f_1^\top \+ M_{r, \infty}(\bs \theta^*,\bs \theta_1)\bs f_2 \right| \\
&+ \left|n^{-1}\bs f_1^\top \+ M_{r,n}(\bs \theta_2)\bs f_2 -\bs f_1^\top  \+ M_{r, \infty}(\bs \theta^*,\bs \theta_2)\bs f_2 \right| \\
&+ \left|\bs f_1^\top\+ M_{r, \infty}(\bs \theta^*, \bs \theta_1)\bs f_2 - \bs f_1^\top\+ M_{r, \infty}(\bs \theta^*, \bs \theta_2)\bs f_2 \right|.
\end{align*}
There exists  $N(\omega)$ such that for $n> N(\omega)$ the first two terms are bounded by $\varepsilon_0/3$ by proposition \ref{vecfilt}. Furthermore, since $\bs \theta \mapsto \+ M_{r, \infty}(\bs \theta^*, \bs \theta)$ is continuous by lemma \ref{lem:contM} there exists a $\delta_0$ such that:
$$\left\|\bs \theta_1 - \bs \theta_2 \right\|_\infty<\delta_0 \implies \left|\bs f_1^\top\+ M_{r, \infty}(\bs \theta^*,\bs \theta_1)\bs f_2 - \bs f_1^\top\+ M_{r, \infty}(\bs \theta^*,\bs \theta_2)\bs f_2 \right|< \varepsilon_0/3.$$

Hence we have shown stochastic equicontinuity of $(n^{-1}\bs f_1^\top \+ M_{r,n}(\bs \theta)\bs f_2)_{n\geq1}$. Now, consider $\mathcal{D}_{r,n}$:
\begin{align}
\left| \mathcal{D}_{r,n}(\bs \theta_1) - \mathcal{D}_{r,n}(\bs \theta_2) \right| &\leq \left|n^{-1} \+1_m^\top\left[\bar{\+Y}_r \circ \log \left( {\+ M_{r,n}(\bs \theta_1)}\oslash{\+ M_{r,n}(\bs \theta_2)} \right)\right]\+1_m\right| \\ \label{part2eqZ}
&+ \left|n^{-1}\+1_m^\top\left[\+ M_{r,n}(\bs \theta_1) - \+ M_{r,n}(\bs \theta_2) \right] \+1_m \right|
\end{align}
By what has already been proven, for any $\varepsilon>0$  we can choose $\delta_1$ and $N_1(\omega)$ to bound \eqref{part2eqZ} by $\varepsilon/2$. Let $\varepsilon_1>0$,  by proposition \ref{prop:Yas} there exists  $N_2(\omega)$ such that for $n> N_2(\omega)$,
\begin{equation*}
\left|\sum_{i,j=1}^mn^{-1}\bar Y_{r}^{(i,j)} \log \frac{ M_{r,n}^{(i,j)}(\bs \theta_1)}{ M_{r,n}^{(i,j)}(\bs \theta_2)}\right| < \sum_{i,j=1}^m |M_{r, \infty}^{(i,j)}(\bs \theta^*, \bs \theta^*) + \varepsilon_1| \left|\log \frac{n^{-1} M_{r,n}^{(i,j)}(\bs \theta_1)}{n^{-1} M_{r,n}^{(i,j)}( \bs \theta_2)}\right|.
\end{equation*}

Furthermore, for each $(i,j) \in [m]^2$  either:
\begin{itemize}
\item $M_{r,n}^{(i,j)}(\bs \theta_1) = 0$ or $ M_{r,n}^{(i,j)}(\bs \theta_2)=0$ with positive probability. In this case the $(i,j)$th terms disappear from the sum on the left hand side since $\bar Y_r^{(i,j)} = 0$ with probability $1$ by lemma \ref{MY}; or
\item $ M_{r,n}^{(i,j)} (\bs \theta_1)> 0$ and $ M_{r,n}^{(i,j)}(\bs \theta_2)>0$ almost surely, by continuity of $\log$ on $\mathbb{R}_{>0}$ there exists  a $\delta_3^{(i,j)}>0$ such that if $| M^{(i,j)}_{r,n}(\bs \theta_1) -M^{(i,j)}_{r,n}(\bs \theta_2)|<\delta^{(i,j)}_3$ then:
\begin{align*}
 \left| \log \frac{n^{-1} M^{(i,j)}_{r,n}(\bs \theta_1)}{n^{-1} M^{(i,j)}_{r,n}(\bs \theta_2)}\right| &=  \left| \log {n^{-1} M^{(i,j)}_{r,n}(\bs \theta_1)} - \log{n^{-1} M^{(i,j)}_{r,n}(\bs \theta_2)}\right| \\
 &\leq \frac{\varepsilon}{2m^2|M^{(i,j)}_{r,\infty}(\bs \theta^*,\bs \theta^*) + \varepsilon_1|}.
\end{align*}
\end{itemize}
Then by stochastic equicontinuity of $(n^{-1}\+ M_{r,n})_{n\geq1}$ there exists $N_3(\omega)$ and $\delta_2$ such that for $n> \max(N_2(\omega),N_3(\omega))$ and $\|\bs \theta_1 - \bs \theta_2\|_\infty< \delta_2$ we have that \newline $\| \+ M_{r,n}(\bs \theta_1) - \+ M_{r,n}(\bs \theta_2)\|_\infty<\min_{(i,j)}\delta^{(i,j)}_3$ so that:

\begin{align*}
\sum_{i,j=1}^m n^{-1}\bar Y^{(i,j)}_r \left| \log \frac{n^{-1} M^{(i,j)}_{r,n}(\bs \theta_1)}{n^{-1} M^{(i,j)}_{r,n}(\bs \theta_2)}\right| & < \sum_{i,j=1}^m |M^{(i,j)}_{r,\infty}(\bs \theta^*,\bs \theta^*)+ \varepsilon_1|\frac{\varepsilon}{2m^2|M^{(i,j)}_{r,\infty}(\bs \theta^*,\bs \theta^*) + \varepsilon_1|} \\
  &= \varepsilon/2.
\end{align*}
Choosing $\delta = \min(\delta_1, \delta_2)$ and $N(\omega) = \max(N_1(\omega),N_2(\omega),N_3(\omega))$ we have that for $\|\bs \theta_1 - \bs \theta_2\|_\infty< \delta$ and $n> N(\omega)$:
\begin{equation*}
\left| \mathcal{D}_{r,n}(\bs \theta_1) - \mathcal{D}_{r,n}(\bs \theta_2) \right| < \varepsilon/2 + \varepsilon/2 =  \varepsilon.
\end{equation*}
Hence we have established the stochastic equicontinuity of $(\mathcal{D}_{r,n})_{n\geq1}$. This along with the already proven pointwise convergence establishes uniform almost sure convergence by lemma \ref{stochascoli}.

\end{proof}

\subsection{Convergence of Maximum PAL estimators}\label{sec:conv_pal_estimators}


\begin{proof}[Proof of Theorem \ref{theo:consistency}]


Let $\mathcal{C}(\bs \theta^*, \bs \theta)$ be defined to be the r.h.s. of \eqref{eq:case_I_contrast} and let $\mathcal{C}_n$ be as in the proof of proposition \ref{contrastfnx}.  We have that $\mathcal{C}_n(\hat{\bs \theta}_n) \geq \mathcal{C}_n({\bs \theta})$ for all $\bs \theta \in \Theta_{(I)}^*$. Furthermore  ${\mathcal{C}(\bs \theta^*, \bs \theta^*) - \mathcal{C}(\bs \theta^*, \bs \theta) >0}$ for all $\bs \theta \in \Theta$. We can combine these inequalities to obtain:
\begin{equation} \label{supconv}
    \begin{aligned}
    0 &\leq \mathcal{C}(\bs \theta^*, \bs \theta^*) - \mathcal{C}(\bs \theta^*, \hat{\bs \theta}_n) \\
    &\leq \mathcal{C}(\bs \theta^*, \bs \theta^*)  - \mathcal{C}_n( \bs \theta^*)  +\mathcal{C}_n( \bs \theta^*)
    -\mathcal{C}_n(  \hat{\bs \theta}_n)
    + \mathcal{C}_n(  \hat{\bs \theta}_n)
    - \mathcal{C}(\bs \theta^*, \hat{\bs \theta}_n) \\
    &\leq 2\sup_{\bs \theta \in \Theta}\left|\mathcal{C}(\bs \theta^*, \bs \theta) - \mathcal{C}_n(\bs \theta) \right| \asls 0.
    \end{aligned}
\end{equation}
Hence $\mathcal{C}(\bs \theta^*, \hat{\bs \theta}_n) \asls \mathcal{C}(\bs \theta^*, {\bs \theta^*})$.

Now assume for purposes of contradiction that there is some positive probability that $\hat{\bs \theta}_n$ does not converge to the set $\Theta^*_{(I)}$, i.e. assume that there is an event $E\subset \Omega$ with $\mathbb{P}^{\bs \theta^*}(E)>0$ such that for all $\omega \in E$ there exists a $\delta>0$ such that  for infinitely many $n \in \mathbb{N}$ we have ${\hat{\bs \theta}_n(\omega)}$  is not in the open neighbourhood $B_\delta(\Theta^*) = \{\bs \theta \in \Theta: \exists \bs \theta' \in \Theta^*: \|\bs \theta - \bs \theta' \|<\delta\}$.
Since $\Theta$ is compact, the set $B_\delta(\Theta_{(I)}^*)^c = \Theta \setminus B_\delta(\Theta_{(I)}^*)$ is closed, bounded, and therefore compact. Furthermore, $\mathcal{C}(\bs \theta^*,\bs \theta)$ is continuous in $\bs \theta$. By the extreme value theorem this means that there exists a $\bs \theta'\in B_\delta(\Theta_{(I)}^*)^c$ such that for all  $\bs \theta \in B_\delta(\Theta_{(I)}^*)^c$:
$$ \mathcal{C}(\bs \theta^*,\bs \theta) \leq \mathcal{C}(\bs \theta^*,\bs \theta')$$
Furthermore, since $\bs \theta' \notin \Theta_{(I)}^*$ there exists  $\varepsilon>0$ such that:
$$ \mathcal{C}(\bs \theta^*,\bs \theta') < \mathcal{C}(\bs \theta^*,\bs \theta^*) - \varepsilon.$$
By our assumption we have for each $\omega \in E$ there are infinitely many $n \in \mathbb{N}$ such that $\hat{\bs \theta}_n(\omega) \in B_\delta(\Theta_{(I)}^*)^c$. But this implies that for each $\omega \in E$ there are infinitely many $n \in \mathbb{N}$ such that:
$$ \mathcal{C}(\bs \theta^*, \hat{\bs \theta}_n(\omega)) \leq \mathcal{C}(\bs \theta^*,\bs \theta') < \mathcal{C}(\bs \theta^*,\bs \theta^*) - \varepsilon,$$
$$\implies |\mathcal{C}(\bs \theta^*, \bs \theta^*) - \mathcal{C}(\bs \theta^*, \hat{\bs \theta}_n(\omega))|> \varepsilon,$$
which contradicts \eqref{supconv}. Hence we must have that $\hat{\bs \theta}_n$ converges to the set $\Theta^*_{(I)}$ $\mathbb{P}^{\bs \theta^*}$-a.s. The proof for case (II) follows the same arguments but with $\mathcal{C}_n$ and $\mathcal{C}(\bs\theta^*,\bs\theta)$ replaced by $\mathcal{D}_n$ as in the proof of proposition \ref{Zcontfn} and $\mathcal{D}(\bs\theta^*,\bs\theta)$ defined to be the r.h.s. of \eqref{eq:case_II_contrast}.
\end{proof}


\subsection{Identifiability} \label{sec:identifiabiliy}

\begin{proposition}\label{prop:ident}
For any $\bs \theta\in\Theta$,
\begin{align*}
\bs \theta \in \Theta^*_{(I)}&\quad\Longleftrightarrow\quad \bs \mu _{t,\infty}(\bs \theta,\bs \theta) = \bs \mu _{t,\infty}(\bs \theta^*,\bs \theta^*),\quad\forall t=1,\ldots,T\\ 
\bs \theta \in \Theta^*_{(II)}&\quad\Longleftrightarrow\quad \+{M}_{r,\infty}(\bs \theta,\bs \theta) = \+{M} _{r,\infty}(\bs \theta^*,\bs \theta^*),\quad\forall r=1,\ldots,R.
\end{align*}
\end{proposition}
\begin{proof}
For the first equivalence in the statement, in order to prove the implication in the forward direction, assume that $\bs \theta \in \Theta^*_{(I)}$, i.e., $\quad \bs \mu _{t,\infty}(\bs \theta^*,\bs \theta) = \bs \mu _{t,\infty}(\bs \theta^*,\bs \theta^*)$, for all $t=1,\ldots,T$. Recall from the definitions in \eqref{eq:lam_0_inf_defn}-\eqref{eq:lam_bar_t_inf_defn} that $\bs\lambda_{0,\infty}(\bs \theta^* ,\bs \theta)$ does not depend on $\bs\theta^*$, hence neither does $\bs\lambda_{1,\infty}(\bs\theta^*,\bs\theta)$, and so: 
\begin{align*}
\bs \mu_{1,\infty}(\bs \theta^*,\bs \theta^*)^\top &= \bs \mu_{1,\infty}({\bs \theta^*},\bs \theta)^{\top} \\
&= (\bs \lambda_{1,\infty}({\bs \theta^*},\bs \theta)\circ \+ q_1({\bs \theta}))^\top \+G_1(\bs \theta) + {\bs \kappa_{1, \infty}(\bs \theta)}^{\top}  \\
&=  (\bs \lambda_{1,\infty}({\bs \theta},\bs \theta)\circ \+ q_1(\bs \theta))^\top \+G_1(\bs \theta) + {\bs \kappa_{1,\infty}(\bs \theta)}^{\top} \\
&= \bs \mu_{1,\infty}({\bs \theta},\bs \theta)^{\top}.
\end{align*}
Now, for $t>1$ assume that $\bs \lambda_{t-1,\infty}(\bs \theta,\bs \theta) = \bs \lambda_{t-1,\infty}(\bs \theta^*,\bs \theta)$ and $ \bs \mu_{t-1,\infty}(\bs \theta,\bs \theta) = \bs \mu_{t-1,\infty}(\bs \theta^*,\bs \theta)$. Then we have that:

\begin{align*}
\bar{\bs \lambda}_{t-1,\infty}(\bs \theta^*,\bs \theta) &= \bigg[\+ 1_m - \+q_{t-1}(\bs \theta) \\
&+ \bigg(\bs \mu_{t-1,\infty}(\bs \theta^*,\bs \theta)^\top \bigg\{[(\+ 1_m \otimes \+ q_{t-1}(\bs \theta))\circ \+ {G_{t-1}(\bs \theta)}^\top] \\
&\oslash\left[ \bs \mu_{t-1,\infty}(\bs \theta, \bs \theta)  \otimes \+1_m \right]\bigg\}\bigg)^\top\bigg]\circ \bs \lambda_{t-1,\infty}(\bs \theta^*,\bs \theta) \\
&= \left[1_m - \+q_{t-1}(\bs \theta) + \+q_{t-1}(\bs \theta) \right]\circ \bs \lambda_{t-1,\infty}(\bs \theta^*,\bs \theta) \\
&= \bs \lambda_{t-1,\infty}(\bs \theta^*,\bs \theta) \\
&= \bs \lambda_{t-1,\infty}(\bs \theta, \bs \theta),
\end{align*}
so that

\begin{align*}
\bs \lambda_{t,\infty}(\bs \theta^*,\bs \theta)^\top &= (\bar{\bs \lambda}_{t-1,\infty}(\bs \theta^*,\bs \theta) \circ \bs \delta_t(\bs \theta))^\top\+ K_{t, \bs \eta(\bar{\bs \lambda}_{t-1,\infty}(\bs \theta^*\bs \theta) \circ \bs \delta_t(\bs \theta))} + \bs \alpha_{t,\infty}(\bs \theta)^\top \\
&= (\bs  \lambda_{t-1,\infty}(\bs \theta,\bs \theta) \circ \bs \delta_t(\bs \theta))^\top\+ K_{t, \bs \eta(\bs \lambda_{t-1,\infty}(\bs \theta,\bs \theta) \circ \bs \delta_t(\bs \theta))} + \bs \alpha_{t,\infty}(\bs \theta)^\top \\
&= \bs \lambda_{t,\infty}(\bs \theta,\bs \theta)^\top,
\end{align*}
and

\begin{align*}
\bs \mu_{t,\infty}(\bs \theta^*,\bs \theta^*)^\top &= \bs \mu_{t,\infty}({\bs \theta^*},\bs \theta)^{\top} \\
&= (\bs \lambda_{t,\infty}({\bs \theta^*},\bs \theta)\circ \+ q_t({\bs \theta}))^\top \+G_t(\bs \theta) + {\bs \kappa_{t,\infty}(\bs \theta)}^{\top}  \\
&=  (\bs \lambda_{t,\infty}({\bs \theta},\bs \theta)\circ \+ q_t(\bs \theta))^\top \+G_t(\bs \theta) + {\bs \kappa_{t,\infty}(\bs \theta)}^{\top} \\
&= \bs \mu_{t,\infty}({\bs \theta},\bs \theta)^{\top}.
\end{align*}
By induction we have thus shown that $\bs \mu_{t,\infty}(\bs \theta^*,\bs \theta^*) = \bs \mu_{t,\infty}(\bs \theta,\bs \theta)$ for all $t = 1, \dots, T$ and have completed the proof for the forward direction of the first implication in the statement.

For the backwards direction we need to show that $\bs \mu_{t,\infty}(\bs \theta^*,\bs \theta^*) = \bs \mu_{t,\infty}({\bs \theta},\bs \theta)  \implies \bs \mu_{t,\infty}(\bs \theta^*,\bs \theta) = \bs \mu_{t,\infty}(\bs \theta^*,\bs \theta^*)$, for all $t=1,\ldots,T$ . Similarly as for the forwards direction:

\begin{align*}
\bs \mu_{1,\infty}(\bs \theta^*,\bs \theta)^\top &= (\bs \lambda_{1,\infty}({\bs \theta^*},\bs \theta)\circ \+ q_1({\bs \theta}))^\top \+G_1(\bs \theta) + {\bs \kappa_1(\bs \theta)}^{\top}  \\
&=  (\bs \lambda_{1,\infty}({\bs \theta},\bs \theta)\circ \+ q_1(\bs \theta))^\top \+G_1(\bs \theta) + {\bs \kappa_1(\bs \theta)}^{\top} \\
&= \bs \mu_{1,\infty}({\bs \theta},\bs \theta)^{\top} \\
&= \bs \mu_{1,\infty}({\bs \theta^*},\bs \theta^*)^{\top}.
\end{align*}
Now, for $t>1$ assume that $\bs \lambda_{t-1,\infty}({\bs \theta},\bs \theta) = \bs \lambda_{t-1,\infty}(\bs \theta^*,\bs \theta)$ and $ \bs \mu_{t-1,\infty}({\bs \theta},\bs \theta) = \bs \mu_{t-1,\infty}(\bs \theta^*,\bs \theta)$. Then we have that:

\begin{align*}
\bar{\bs \lambda}_{t-1,\infty}(\bs \theta^*,\bs \theta) &= \bigg[\+ 1_m - \+q_{t-1}(\bs \theta) \\
&+ \bigg(\bs \mu_{t-1,\infty}(\bs \theta^*,\bs \theta)^\top \bigg\{[(\+ 1_m \otimes \+ q_{t-1}(\bs \theta))\circ \+ {G_{t-1}(\bs \theta)}^\top]\\
&\oslash\left[ \bs \mu_{t-1,\infty}(\bs \theta, \bs \theta)  \otimes \+1_m \right]\bigg\}\bigg)^\top\bigg]\circ \bs \lambda_{t-1,\infty}(\bs \theta^*,\bs \theta)  \\
&= \left[1_m - \+q_{t-1}(\bs \theta) + \+q_{t-1}(\bs \theta) \right]\circ \bs \lambda_{t-1,\infty}(\bs \theta^*,\bs \theta) \\
&= \bs \lambda_{t-1,\infty}(\bs \theta^*,\bs \theta) \\
&= \bs \lambda_{t-1,\infty}(\bs \theta, \bs \theta),
\end{align*}
so that
\begin{align*}
\bs \lambda_{t,\infty}(\bs \theta^*,\bs \theta)^\top &= (\bar{\bs \lambda}_{t-1,\infty}(\bs \theta^*,\bs \theta) \circ \bs \delta_t(\bs \theta))^\top\+ K_{t, \bs \eta(\bar{\bs \lambda}_{t-1,\infty}(\bs \theta^*\bs \theta) \circ \bs \delta_t(\bs \theta))} + \bs \alpha_{t,\infty}(\bs \theta)^\top \\
&= (\bs  \lambda_{t-1,\infty}(\bs \theta,\bs \theta) \circ \bs \delta_t(\bs \theta))^\top\+ K_{t, \bs \eta(\bs \lambda_{1,\infty}(\bs \theta,\bs \theta) \circ \bs \delta_t(\bs \theta))} + \bs \alpha_{t, \infty}(\bs \theta)^\top \\
&= \bs \lambda_{t,\infty}(\bs \theta,\bs \theta)^\top,
\end{align*}
and
\begin{align*}
\bs \mu_{t,\infty}(\bs \theta^*,\bs \theta)^\top
&= (\bs \lambda_{t,\infty}({\bs \theta^*},\bs \theta)\circ \+ q_t({\bs \theta}))^\top \+G_t(\bs \theta) + {\bs \kappa_{t,\infty}(\bs \theta)}^\top  \\
&=  (\bs \lambda_{t,\infty}({\bs \theta},\bs \theta)\circ \+ q_t(\bs \theta))^\top \+G_t(\bs \theta) + {\bs \kappa_{t,\infty}(\bs \theta)}^\top \\
&= \bs \mu_{t,\infty}({\bs \theta},\bs \theta)^\top\\
&= \bs \mu_{t,\infty}({\bs \theta^*},\bs \theta^*)^\top
\end{align*}
This completes the proof of the first implication in the statement of the proposition.

For the second implication, we will first show $\+ M_{r\infty}({\bs \theta^*},\bs \theta) = \+ M_{r\infty}({\bs \theta^*},\bs \theta^*)\implies \+ M_{r\infty}({\bs \theta^*},\bs \theta^*) = \+ M_{r\infty}({\bs \theta},\bs \theta)$, for all $r=1,\ldots, R$. Recalling the definitions in \eqref{eq:Lam_0_inf_defn}-\eqref{eq:Lam_bar_t_inf_defn}, we have that for all $s \in \{1, \dots, \tau_1 \}$, $\bs \Lambda_{s, \infty}(\bs \theta^*, \bs \theta) = \bs \Lambda_{s, \infty}(\bs \theta, \bs \theta)$ and hence:

\begin{align}\label{ASEQM}
    \+M_{1, \infty}(\bs \theta^*, \bs \theta^*) = \+M_{1, \infty}(\bs \theta^*, \bs \theta) &= \sum_{s=1}^{\tau_1}\+\Lambda_{s,\infty}(\bs \theta^*, \bs \theta)\circ \+Q_s(\bs \theta) \\
    &=\sum_{s=1}^{\tau_1}\+\Lambda_{s,\infty}(\bs \theta, \bs \theta)\circ \+Q_s(\bs \theta) \\
    &=     \+M_{1, \infty}(\bs \theta, \bs \theta).
\end{align}

Now let $r\geq1$ and assume that, for all $s \in \{\tau_{r-1}+1, \dots, \tau_r\}$, $\bs \Lambda_{s, \infty}(\bs \theta^*, \bs \theta) = \bs \Lambda_{s, \infty}(\bs \theta, \bs \theta)$. Then:
\begin{align*}
    \bar{\bs \Lambda}_{\tau_r,\infty}(\bs \theta^*, \bs \theta) &= \left[\+1_m \otimes \+1_m - \+Q_{\tau_r}(\bs \theta)\right] \circ {\bs \Lambda}_{\tau_r,\infty}(\bs \theta^*, \bs \theta) \\
    &+ \frac{\+M_{r,\infty}(\bs \theta^*, \bs \theta^*)}{\+M_{r,\infty}(\bs \theta^*,\bs \theta)}\left[ \+\Lambda_{\tau_r,\infty}(\bs \theta^*, \bs \theta)\circ \+Q_{\tau_r}(\bs \theta)\right] \\
    &= \+\Lambda_{\tau_r,\infty}(\bs \theta^*, \bs \theta) \\
    &= \+\Lambda_{\tau_r,\infty}(\bs \theta, \bs \theta).
\end{align*}

This then implies that for all $s \in \{\tau_{r}+1, \dots, \tau_{r+1}\}$, $\bs \Lambda_{s, \infty}(\bs \theta^*, \bs \theta) = \bs \Lambda_{s, \infty}(\bs \theta, \bs \theta)$, which in turn implies, as in \eqref{ASEQM} that $\+M_{r+1, \infty}(\bs \theta^*,\bs \theta^*) = \+M_{r+1, \infty}(\bs \theta,\bs \theta)$. The reverse direction follows by similar reasoning, as in mirroring the proof of the backwards direction of the first implication in the statement of the proposition, so the details are omitted.
\end{proof}

\newpage

\section{Supplementary material for section \ref{sec:examples}}

\subsection{Supplementary material for the pedagogical SEIR example}
\label{sec:SIRov_supp}
In this section we present the PALSMC algorithm used in the pedagogical SEIR example, given by algorithm \ref{alg:SIRSMC}. For notation purposes, define $\+y_t = [0 \; 0 \; y_t \; 0]$. The following section describes how one can make proposals informed by observations.
\subsubsection{Deriving a proposal informed by observations}
Let $f(\cdot | \mu_q,  \sigma^2_q)$ be the density associated with a $\mathcal{N}(\mu_q,  \sigma^2_q)_{\geq 0, \leq 1}$ random variable. We would like to make proposals informed by observations, to that end we seek a Laplace approximation to:

$$\hat p(q_t \mid y_{1:t}, q_{1:t-1}) := \frac{\exp\ell(y_{t} \mid y_{1:t-1}, q_{1:t})f(q_t| \mu_q,  \sigma^2_q)}{\int\exp\ell(y_{t} \mid y_{1:t-1}, q_{1:t})f(q_t| \mu_q,  \sigma^2_q) dq_t}.$$ 
Suppressing dependence on the particle, let ${\bs \lambda}_{t}$ be calculated as per line \ref{seirpred} of algorithm \ref{alg:SIRSMC}. We have for some constant $C_1$ and $C_2$:

\begin{equation}
\begin{aligned} \label{seirpropeq}
\log \hat p(q_t | y_t) & = \ell(y_{t} \mid y_{1:t-1}, q_{1:t}) + f(q_t|\mu_q,  \sigma^2_q) + C_1 \\
&= y_t \log(q_t) + y_t \log(\lambda^{(3)}_t) - q_t - \log y_t!     - \frac{1}{2}\left(\frac{q_t - \mu_q}{\sigma_q} \right)^2 +C_2
\end{aligned} 
\end{equation}
To get the mean of a Laplace approximation to \eqref{seirpropeq} we must find it's maximum w.r.t. $q_t$, hence:

\begin{equation}
\begin{aligned}\label{SEIRSMCmean}
\frac{d \log \hat p(q_t \mid y_t)}{dq_t} &= \frac{y_t}{q_t} - \lambda^{(3)}_t - \frac{q_t - \mu_q}{\sigma_q^2} = 0\\ 
&\iff  (q_t)^2 + (\lambda^{(3)}_t\sigma_q^2 - \mu_q )q_t - y_t\sigma_q^2 = 0 \\
&\implies q_t = \frac{1}{2}\left(\mu_q  - \lambda^{(3)}_t\sigma_q^2 + \sqrt{(\lambda^{(3)}_t\sigma_q^2- \mu_q )^2 + 4y_t\sigma_q^2 }\right) =: \mu_{prop}
\end{aligned}
\end{equation}
For the variance we find the second derivative and evaluate it at $\mu_{prop}$:
\begin{equation}
\begin{aligned}\label{SEIRSMCvar}
\frac{d^2 \log \hat p(q_t \mid y_t)}{d(q_t)^2} = -  \frac{y_t}{(q_t)^2} - \frac{1}{\sigma_q^2} \\
\implies \sigma_{prop}^2 = \left(\frac{y_t}{\mu_{prop}^2} + \frac{1}{\sigma_q^2}\right)^{-1}.
\end{aligned}
\end{equation}
To be congruent with the support of $q_t$ we truncate the proposal to be $\mathcal{N}(\mu_{prop},\sigma_{prop}^2)_{\geq 0, \leq 1}$, denote its density as $\pi (\cdot | \mu_{prop},\sigma_{prop}^2)$.

\begin{algorithm}[H]
\caption{PAL within SMC}\label{alg:SIRSMC}
\begin{algorithmic}[1]
  \Statex {\bf initialize:} $\bar{\bs \lambda}_{0,i} \leftarrow \bs \lambda_0$ for $i = 1$ to $n_{part}$.
  \State {\bf for}  $t  \geq 1$: 
  \State \quad  {\bf for}  $i = 1,\dots, n_{part}$: 
  \State \quad \quad \label{seirpred}$\bs \lambda_{t}^{(i)} \leftarrow \left(\bar{\bs \lambda}_{t-1}^{(i)}\circ \bs \delta_t\right)^\top\+K_{t, \eta \left(\bs {\bar \lambda}_{t}^{(i)} \right)} + \bs \alpha_t$
  \State \quad \quad $q_{t}^{(i)}  \sim \mathcal{N}\left(\mu_{prop}, \sigma_{prop}^2  \right)_{\geq 0, \leq 1}$ calculated as per \ref{SEIRSMCmean} and \ref{SEIRSMCvar}.
   \State \quad \quad  $\+q_{t}^{(i)}  \leftarrow [0\;0\;q_{t}^{(i)} \;0]^\top$
 \State \quad \quad  $\log w_{t}^{(i)}  \leftarrow \+ y_t^\top \log \bs \lambda_{t}^{(i)}\circ \+ q_{t}^{(i)} - \bs \lambda_{t}^\top \+ q_{t}^{(i)}  - \log\+ y_t! + \log f(q_{t}^{(i)} |\mu_q,\sigma_q^2)- \log \pi(q_{t}^{(i)} |\mu_{prop},\sigma_{prop}^2)$
 \State \quad \quad $\bs{\bar \lambda}_{t}^{(i)} \leftarrow \left(\+1_m - \+ q_{t}^{(i)} \right)\circ \bs \lambda_{t}^{(i)} + \+y_t$ 
  \State \quad {\bf end for}
 \State \quad $\ell(y_t\mid y_{1:t-1}) \leftarrow \frac{1}{n_{part}}\sum_{i=1}^{n_{part}} w_{t}^{(i)} $
 \State \quad $\bar w_{t}^{(i)}  \leftarrow w_{t}^{(i)} /\sum_{j=1}^{n_{part}}w_{t}^{(j)} $
 \State \quad \textbf{resample} $\left\{\bs{\bar \lambda}_{t}^{(i)},q_{t}^{(i)}  \right\}_{i=1}^{n_{part}}$ according to a systematic resampling scheme with weights $\left\{\bar w_{t}^{(i)}\right\}_{i=1}^{n_{part}}$ .
 \State {\bf end for}
\end{algorithmic}
\end{algorithm}

\subsubsection{Comparison to a standard sequential Monte Carlo approach}
In this section we perform a routine comparison of PALSMC evaluations and `exact' particle filter likelihood estimates on a small simulated model for which a standard SMC approach, that is integrating out both the $\+ x_t$  and the $\bar{\bs \theta}_t$ processes, is still viable. The model we use is the same as in the pedagogical SEIR example, except we introduce over-dispersion into the infection rate, that is the number of new exposed at time $t$ is distributed:

\begin{equation}
B_t \sim \mathrm{Bin}(S_t, 1- e^{-h\xi_t\beta\frac{I_t}{n_t}}),
\end{equation}
where $\xi_t \sim \text{gamma}(\sigma_\xi,\sigma_\xi)$ is mean $1$ multiplicative noise. We simulated data using the same parameters as in the pedagogical SEIR example and $\sigma_\xi = 1$. For the standard SMC approach we used joint proposals inspired by the PALSMC derivations.

We ran each of the filters with the data generating parameters as input with 1000, 5000, and 10000 particles. For each procedure and each particle size we ran generated 100 likelihood estimates and calculated the standard deviation. We found that, as expected, the variance shrinks for both procedures as the number of particles increases, and that PALSMC had a systematically lower variance with little bias.

\begin{figure}[h!]
 \begin{longtable}{ c c c}
\label{measlesresults}\\
 \hline
 Number of particles & PALSMC\;(sd) & SMC\;(sd) \\
 \hline

 1000 & -585.44\;(0.26) & -585.51\;(0.48)  \\
 5000 & -585.39\;(0.15)	&  -585.38\;(0.27)  \\
 10000 & 	-585.06\;(0.11) & -585.10\;(0.17)   \\

 \hline
 \caption{Likelihood estimates calculated at the data generating parameters using the PALSMC and a standard SMC approach. Standard deviations are calculated from 100 runs of each procedure.}

\end{longtable}
\end{figure}

\subsection{A simulation example}\label{sec:interpretation}

Consider the following SEIR model with immigration and emigration: $\mathbb{P}_{0,n} = \mathrm{Mult} \left (n, \left [ 0.99\; \; 0\; \; 0.01 \; \; 0 \right ]^{\top} \right )$ and for all $t$: ${\bs\alpha_{t,n}} = \left [ \frac{4}{100} n\; \; \frac{4}{100} n\; \; \frac{4}{100} n \;\; \frac{4}{100} n \right ]^\top$, ${\bs \delta_t} = \left [ \frac{98}{100}\;\; \frac{98}{100}\;\; \frac{98}{100}\;\; \frac{98}{100} \right ]^\top$, ${\bs \kappa}_{t,n} = \left [ \frac{1}{100} n\;\; \frac{1}{100} n\;\; \frac{1}{100} n\;\; \frac{1}{100} n \right ]^\top$, ${\+q}_t = \left [ 0.1\; \; 0.1 \; \; 0.3 \; \; 0.2 \right ]^\top$
and
$$
\mathbf{K}_{t, \bs\eta} = 
    \left [ \begin{array}{cccc}
        e^{-\beta^* \eta^{(3)}} & 1 - e^{-\beta \eta^{(3)}} & 0 & 0 \\
        0 & e^{-\rho} & 1 - e^{-\rho} & 0 \\
        0 & 0 & e^{-\gamma} & 1 - e^{-\gamma} \\
        0 & 0 & 0 & 1
    \end{array} \right ],\quad{\+G}_t = 
    \left [ \begin{array}{cccc}
        0.95 & 0 & 0.05 & 0\\
        0.3 & 0 & 0.7 & 0\\
        0.15 & 0 & 0.85 & 0\\
        0 & 0 & 0 & 1\\
    \end{array} \right ]
$$
with DGP $\bs\theta^* = [\beta^*\;\rho^*\; \gamma^*]^\top =[ 0.5\;0.05\;0.1]^\top$.

This observation model can be interpreted as follows: with probability $q_t^{(i)}$ each individual in compartment $i$ is tested for disease. Allowing $q_t^{(i)}$ to vary across $i$ could model, for example, infective individuals being more likely to be tested. The above choice of ${\+G}_t $ allows for false-positives (first row) and false-negatives (third row), where those testing positive are considered infective, and those testing negative are considered susceptible. Of course, other choices are possible.

\begin{figure}
    \centering
    \includegraphics[width=\textwidth]{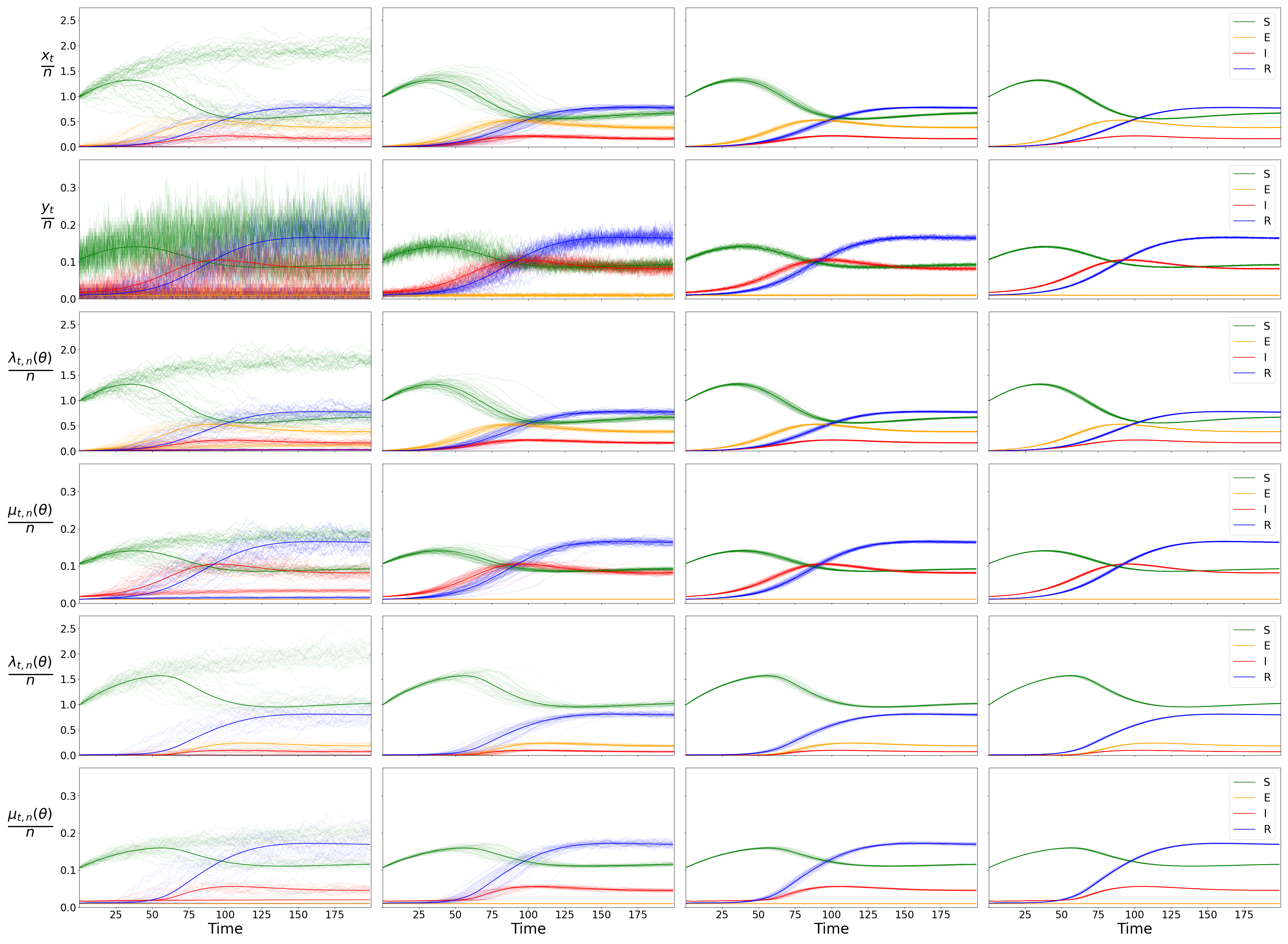}

    \caption{Simulation SEIR example. Top two rows: asymptotic behaviour of ${\+x}_t \slash n$ and ${\+y}_t \slash n$; $50$ simulations from the model (light lines) and theoretical deterministic $n\to\infty$ limits (bold line) for each population size (left to right), $n \in \{100, 1000, 10000, 100000\}$. Middle two rows: filtering intensities associated with the $50$ simulated data sets with $\bs\theta$ taken to be $\bs\theta^*$. Bottom two rows: filtering with $\bs\theta$ set erroneously $\beta=0.1$, $\gamma=0.3$, and all other parameters set as for the middle two rows. }
    \label{fig:extendedSEIR_LLN}
\end{figure}

The top two rows of plots in figure \ref{fig:extendedSEIR_LLN} show $n^{-1}\+x_t$ and $n^{-1}\+y_t$ simulated $50$ times from the model with population sizes $n \in \{100, 1000, 10000, 100000\}$. Note that in the top row, the fact that trajectories for compartment $S$ in $n^{-1}\+x_t$ are valued above $1$ in places is explained in terms of immigration into the $S$ compartment exceeding the combined effect of emigration from $S$ and individuals transitioning from $S$ to  $E$. With $n=100$, the fact that some trajectories for the $S$ compartment are roughly increasing over time corresponds to the lack of an outbreak; for other trajectories which rise and then fall, an outbreak does occur. 

Due to the choices of  $\mathbb{P}_{0,n} $, ${\bs\alpha_{t,n}} $ and ${\bs \kappa}_{t,n}$ set out above, it is immediate that the vectors $\bs\lambda_{0,\infty}$, ${\bs\alpha_{t,\infty}} $ and ${\bs \kappa}_{t,\infty}$ appearing in assumptions \ref{as:init} and \ref{as:params} exist. The convergence of $n^{-1}\+x_t$ and $n^{-1}\+y_t$  as $n\to\infty$ to deterministic limits as discussed in section \ref{sec:consistency_result} is evident in figure \ref{fig:extendedSEIR_LLN}.  

The middle two rows of figure \ref{fig:extendedSEIR_LLN} show the behaviour of the scaled filtering intensities $n^{-1}\bs{\lambda}_{t,n}(\bs\theta)$ and $n^{-1}\bs{\mu}_{t,n}(\bs\theta)$ obtained from algorithm \ref{alg:x} in the case of correctly specified parameters $\bs\theta\leftarrow
\bs\theta^*$.  It is evident that, as per the discussion of asymptotic filtering accuracy in section \ref{sec:consistency_result}, as $n\to\infty$ these quantities converge to the same deterministic limits as do $n^{-1}\+x_t$ and $n^{-1}\+y_t$, respectively. On the other hand, as illustrated in the bottom two rows of figure \ref{fig:extendedSEIR_LLN},  when the model is not correctly specified, then $\bs{\lambda}_{t,n}(\bs\theta)$ and $\bs{\mu}_{t,n}(\bs\theta)$ converge to limits which are not equal to the limits of $n^{-1}\+x_t$ and $n^{-1}\+y_t$.

\begin{figure}[httb!]
    \centering
    \includegraphics[width=\textwidth]{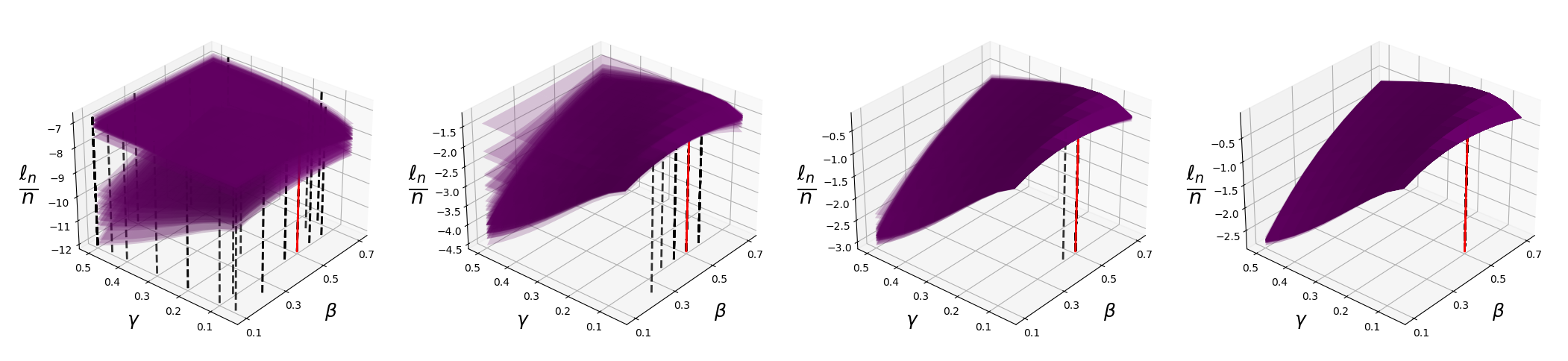}
    \includegraphics[width=\textwidth]{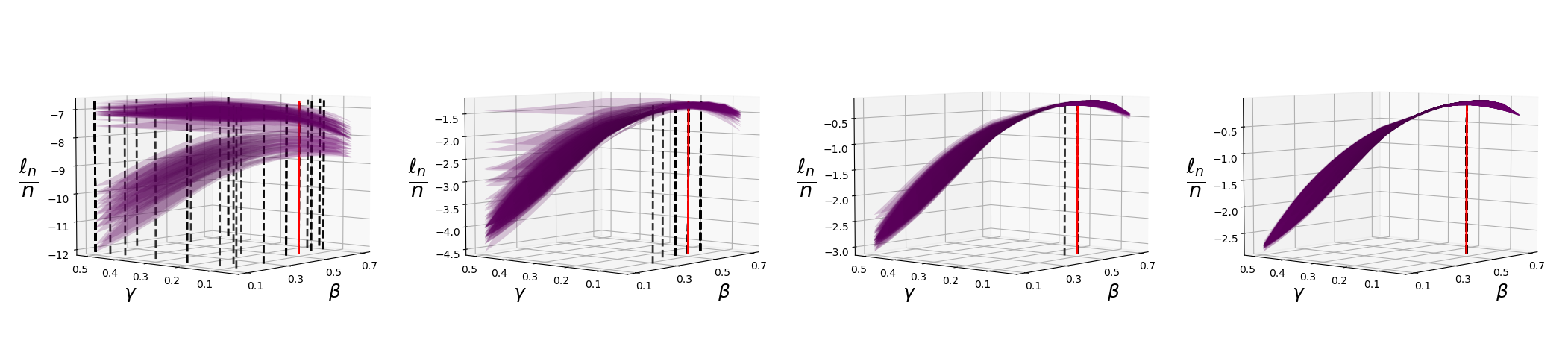}
    \caption{Simulation SEIR example. Purple surfaces within each plot are the scaled log-PAL surfaces associated with $50$ data sets simulated from the model with the DGP. From left to right: $n =100, 1000, 10000, 100000$. Vertical black dashed lines are the maximum PAL estimates for each surface, the vertical red line is the DGP. The two rows show the same 3-d plots from different viewing angles.}
    \label{fig:consistency}
\end{figure}

Figure \ref{fig:consistency} illustrates the behaviour of the scaled log-PAL $n^{-1}\ell_n(\bs\theta)$ evaluated over a find grid of values fo $\bs\theta=[\beta\; \gamma ]^\top$ (all other parameters held constant). Each purple surface in each plot corresponds to a different data set simulated from the model, as in the second row of figure \ref{fig:extendedSEIR_LLN}.  As $n$ grows, figure \ref{fig:extendedSEIR_LLN} evidences convergence of the maximum PAL estimates to the true parameter value, as per theorem \ref{theo:consistency}.

\subsection{Delayed Acceptance PMCMC for the boarding school influenza outbreak}\label{sec:boarding_school}
This example illustrates the use of the PAL within delayed acceptance PMCMC, specifically the delayed acceptance Particle Marginal Metropolis Hastings (daPMMH) algorithm of \cite{golightly2015delayed}. 

\subsubsection*{Data and model}
The data set is the well-known boarding school influenza outbreak data, recorded at a British boarding school in 1978 and reported in the British Medical Journal \cite{Anonbsflu, davies1982christ}. The data are available in the R package ``pomp'' \cite{pomp20126}. On day one there was one infection and over the course of the $14$ day epidemic a total of $512$ students reported symptoms from a population of $n=763$. The observations are prevalence data: daily counts of the total number of symptomatic individuals. We cast this an instance of case (I), using a simple SIR model, where the initial state of the population is fixed to $\left[763\;1\;0 \right]^\top$ and we define the matrix $\+K_{t, \bs \eta}$ as follows:
\begin{equation*}
\+K_{t, \bs \eta} = 
\left [ \begin{array}{ccc}
e^{-\beta \eta^{(2)}} & 1 - e^{-\beta \eta^{(2)}} & 0 \\
0 & e^{-\gamma} & 1 - e^{-\gamma} \\
0 & 0 & 1 
\end{array} \right ],
\end{equation*}
where $\beta$ and $\gamma$ are to be estimated. Observations $y_t$ are modelled as binomially under-reported counts of infected individuals, that is, given $x_t^{(2)}$, $y_t \sim \text{Bin}(x_t^{(2)},q)$ where $q\in[0,1]$ is unknown and to be estimated. To connect with the notation of algorithm \ref{alg:x} we have $\+y_t \equiv [0\;y_t\;0]^\top$ and $\+q_t \equiv [0\;q\;0]^\top$ for $t\geq 1$.

\subsubsection*{Delayed Acceptance Particle Marginal Metropolis Hastings}


In the standard PMMH algorithm \citep{andrieu2010particle}, one calculates a particle filter approximation to the likelihood for each proposed parameter value, which is typically a computationally intensive operation. The daPMMH algorithm introduces an additional `pre-screening' acceptance step based on  an approximate likelihood which is assumed to be cheap to evaluate. Only if the proposed parameter is accepted in this initial step is a particle filter approximation to the likelihood then evaluated; thus in performing this additional step, one seeks to avoid running a particle filter for proposals which are likely to be rejected. Details of the validity of the scheme, in the sense that it indeed targets the true posterior distribution over the parameters, can be found in \cite{golightly2015delayed}. Algorithm \ref{daPMMH} in section \ref{sec:boarding_school_supp} of the supplementary material illustrates how to use a PAL within a daPMMH.

We stress that, although for the SIR model the number of compartments is small ($m=3$), and for the data set in question the population size is fairly small ($n=763$), this actually presents a stern relative speed test for PALs versus particle filters: the particle filter element of the daPMMH and PMMH algorithms involves simulating from the latent compartmental model, and the overall cost of the particle filter, therefore, grows with both the number of compartments and the size of the population, as well as the number of particles. By contrast, evaluating the PAL involves no random number generation and has a cost independent of population size. Thus, if a relative speed gain using PALs can be demonstrated with a small population size and small number of compartments, it is reasonable to expect an even greater relative speed gain for models with larger numbers of compartments and larger populations.

\subsubsection*{Results}

We compare the performance of three algorithms: PALMH: a Metropolis-within-Gibbs algorithm with the PAL substituted in place of the exact likelihood, i.e., targeting an approximation to the exact posterior distribution; PMMH: a standard Particle Marginal Metropolis-Hastings within Gibbs; daPMMH: a delayed acceptance Particle Marginal Metropolis-Hastings within Gibbs, in which we use the PAL for the delayed acceptance step. We apply these three methods to both a synthetic and a real dataset.
For all three algorithms we use Gaussian random walk proposals independently for each element of $\bs \theta$. The random walk variances are tuned to ensure acceptance rates between 20\% and 40\%. The PMMH and daPMMH algorithms were each run with $1000$ particles. All experiments were run on a single core of a 1.90 GHz i7-8650U CPU.

The parameters of the model are collected in the vector $\bs \theta = [\beta\; \gamma\; q]^\top$. We consider a fairly vague prior $p(\bs \theta) = p(\beta)p(\gamma)p(q)$, where $p(\beta)$ and $p(\gamma)$ are truncated Gaussian densities $\mathcal{N}(0,1)_{\geq 0}$ and $p(q)$ is a truncated Gaussian density $\mathcal{N}(0.5,0.5)_{\geq 0, \leq 1}$. 

\paragraph{Simulated data.} We simulated an epidemic for 14 days with the parameter regime $\bs \theta^* = [\beta^*\; \gamma^*\; q^*]^\top = [2\; 0.5\; 0.8\;]^\top$. For each of the PALMH, PMMH, and daPMMH we ran a $5\times10^5$ length chain, discarded $10^5$ for burn in and then thinned to a sample of $2.5\times 10^5$. Trace plots, autocorrelation plots, and posterior sample histograms for each scheme are presented in  section \ref{sec:boarding_school_supp} of the supplementary material, the rates of decay of the ACFs with respect to lag for the daPMMH and PMMH algorithms are similar, the rate of decay for the PALMH algorithm is faster. The Monte Carlo approximations of the posterior marginals are closely matched across the three algorithms, see table \ref{bsflusynthtable} for summary statistics, and are concentrated around the data generating parameters. A single evaluation of the PAL took a mean time of $9.4\times 10^{-6}$ seconds, the particle filter approximation to the likelihood took a mean time of $4.5\times 10^{-3}$ seconds, both algorithms were implemented with Rcpp.

\paragraph{Real data.} On the real data we ran the PALMH, PMMH, and daPMMH for $5\times10^5$ iterations each, with run times of 12.2 minutes, 4.5 hours, and 2.8 hours respectively, exhibiting the speed benefits of the PAL approach. Trace plots, autocorrelation plots, and approximate posterior sample histograms for each scheme are presented in  section \ref{sec:boarding_school_supp}, the rate of decay of the ACF with lag is similar for the daPMMH and PMMH algorithms, the rate of decay for the PALMH algorithm is faster. The daPMMH and PMMH algorithms yield very similar approximate posterior marginals as expected -- see table \ref{bsflufitreal}. The posterior marginals obtained from the PALMH scheme exhibit modes in different locations to those from PMMH/daPMMH, with the following epidemiological interpretation. The approximate posterior marginals obtained from PALMH correspond to a fast growing outbreak (large $\beta$), with individuals spending longer in the infected state (small $\gamma$) and a relatively lower reporting rate (relatively small $q$).  By contrast, the PMMH/daPMMH marginals suggest a slower outbreak (smaller $\beta$) with less time spent in the infected compartment (larger $\gamma$), but with a higher case reporting rate (relatively high $q$). 
Posterior predictive checks \citep{gelman1995bayesian} show that, while having contrasting epidemiological interpretations (potentially due to model mis-specification), both PALMH and PMMH/daPMMH achieve good coverage of the data, see figure \ref{bsflupredictive}. The mean trajectories from these posterior predictive distributions reflect the above interpretations of posterior marginals. The posterior predictive means and credible regions were calculated from $10000$ samples from the posterior predictive distributions produced by the PALMH and daPMMH respectively; each sample from the posterior predictive distribution was generated by: sampling a parameter $\bs \theta'$ from the approximate posterior; then using $\bs \theta'$ to simulate an epidemic trajectory and data record from the model.

We performed inference on this data set using a Linear Noise Approximation to the likelihood as described in section 4 of \cite{fearnhead2014inference} within a Metropolis Hastings scheme, the full results can be found in the supplementary material. We find that, whilst the LNA and the PAL perform similarly in terms of parameter inference, the Latent Compartmental Model from which the PAL is derived is more congruent with reality than the SDE model, since the latter allows non-integer and negative counts of individuals in compartments. Furthermore, a single evaluation of the PAL was approximately $\sim100$ times faster than a single evaluation of the LNA marginal likelihood for this dataset.

\begin{longtable}[h!]{ c c c c c}

 \hline
 Parameter & True value & PALMH & PMMH & daPMMH  \\
 \hline

 $\beta$ & 2	& 2.10\;(1.88, 2.34)			& 2.08\;(1.85 2.35) 	& 2.08 \;(1.85,2.35) \\
 $\gamma$	 & 0.5	& 0.51\;(0.42, 0.63) 			& 0.53\;(0.44 0.65)	& 0.53\;(0.43, 0.65 )  \\
 $q$ 	& 0.8	& 0.81\;(0.70, 0.94)			& 0.82\;(0.71, 0.96) 	&  0.82\;(0.71, 0.96) \\

 \hline
  \caption{Boarding school model posterior means and 95\% credible interval, synthetic data. \label{bsflusynthtable}}
 \end{longtable}
 \vspace{-1em}
 \begin{longtable}[h!]{ c c c c}

 \hline
 Parameter & PALMH & PMMH & daPMMH \\
 \hline

 $\beta$ 	& 2.98\;(2.60,3.30) 	& 2.30\;(2.00,2.68) 	&  2.30\;(2.00,2.68)	 \\
 $\gamma$ 	& 0.406\;(0.35,0.47) 	& 0.58\;(0.47,0.68)	 	&  0.58\;(0.47,0.68)	 \\
 $q$ 		& 0.69\;(0.62,0.77) 	& 0.90\;(0.76,0.99)  	&  0.90\;(0.76,0.99)	 \\

 \hline
  \caption{Boarding school model posterior means and 95\% credible interval, real data. \label{bsflufitreal}}
 \end{longtable}

\begin{figure}
    \centering
    \includegraphics[width=\textwidth]{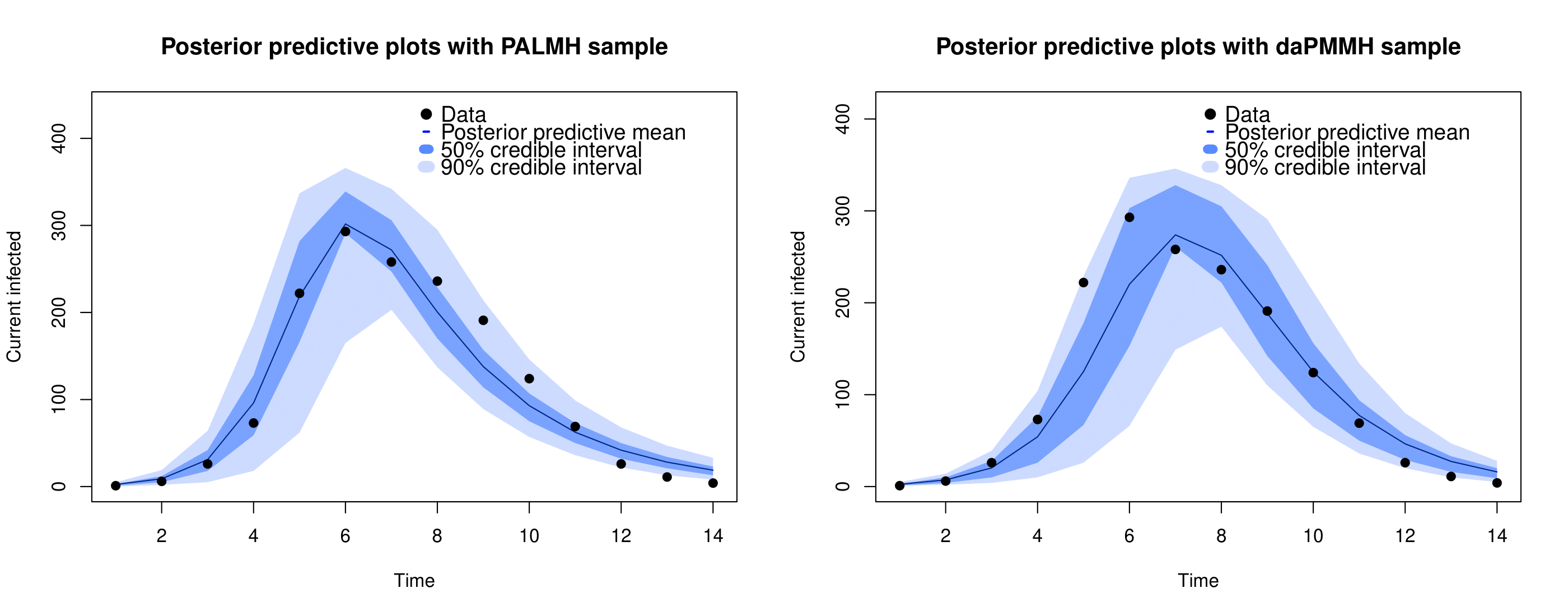}
    \caption{Boarding school influenza example. Means and credible intervals for posterior predictive distributions.}
    \label{bsflupredictive}
\end{figure}

\subsubsection*{PALMH: prior sensitivity analysis}
Section \ref{sec:boarding_school} explores Bayesian analysis on real data under vague priors, with results in table \ref{bsflufitreal}. Whilst both the PALMH and PMMH schemes result in identical inferences on simulated data, there are discrepancies on the real-world boarding school data - which could be attributed to a misspecified model. The estimated parameters under the PALMH suggest an $R_0$ of around $7.3$ (the PMMH estimates suggest an $R_0$ of around $4$) which is consistent with the entire population being infected at some point during the epidemic - one can question whether this is a realistic inference. Given the closed nature of this epidemic, along with the likelihood of close monitoring of the individuals in the system, one could afford to place stronger priors on $q$. In this section, we explore the inferences one can make using the PALMH scheme under more informative priors. We consider the following scenarios:

\begin{enumerate}
\item $\mathcal{N}(0.5,0.5)_{\geq 0, \leq 1}$ - the vague prior used in the original analysis.
\item $\text{Beta}(9,1)$ - an informative prior with mean $0.9$ and variance $0.0082$ and mode $1$.
\item $\text{Beta}(95,5)$ - a strongly informative prior with mean $0.95$ and variance $0.00047$ and mode $1$.
\item $\delta_{0.9}$ - an atomic prior on $0.9$ (the mean inferred $q$ under the PMMH analysis).
\item $\delta_{0.95}$ - an atomic prior on $0.95$.
\end{enumerate} 

For each of these we ran a $5 \times 10^5$ length chain, discarded $10^5$ for burn in and then thinned to a sample of $2.5 \times 10^5$, we summarise our findings in table \ref{bsflupriors}. We find that as the prior belief in a high reporting rate is strengthened, the resulting estimated $R_0$  lowers. If one places strong prior belief in a high reporting rate, see the  $\text{Beta}(95,5)$ and $\delta_{0.95}$  columns, then the inferred $R_0$ falls in line with our findings using the PMMH procedure.

\begin{longtable}[h!]{ c c c c c c}

 \hline
 Parameter & $\mathcal{N}(0.5,0.5)_{\geq 0, \leq 1}$ & $\text{Beta}(9,1)$ & $\text{Beta}(95,5)$  & $\delta_{0.9}$ & $\delta_{0.95}$ \\
 \hline

 $\beta$ 	& 2.98(2.60,3.30) 	& 2.77(2.29,3.22) &	2.39(2.01,2.84)&  2.46(2.08,2.92) & 2.35(1.98,2.77)\\
 $\rho$ 	& 0.41(0.35,0.47) 	& 0.44(0.35,0.57)	&   0.58(0.49,0.68) &  0.55(0.48,0.63)	 & 0.59(0.52,0.68)	 \\
 $q$ 		& 0.69(0.62,0.77) 	& 0.74(0.63,0.90) &   0.94(0.86,0.97)	&  0.90	 & 0.95	 \\
$R_0$     & 6.91 & 6.47 & 4.14 & 4.52 & 4.05	 \\
 \hline
  \caption{Boarding school model PALMH prior sensitivity analysis. Posterior means and 95\% credible interval, real data under various prior assumptions with $R_0$ posterior mean  point estimates. \label{bsflupriors}}
 \end{longtable}

To investigate the disparity between inferences using the PALMH procedure vs the PMMH procedure when applied to real data, exhibited in table \ref{bsflufitreal}, we repeated the analysis with a fixed $q=0.9$ (equivalent to the $\delta_{0.9}$ prior). The resulting posteriors for the PALMH and PMMH procedures still exhibited some differences, but were much more similar as a result of this stronger assumption:

\begin{itemize}
\item The prosterior mean and 95\% credible interval for $\beta$ under the PMMH procedure was 2.14 (1.91,2.40), to be compared with 2.46( 2.08,2.92) for PALMH.
\item The prosterior mean and 95\% credible interval for $\gamma$ under the PMMH procedure was 0.58 (0.53,0.64), to be compared with 0.55( 0.48,0.63) for PALMH.
\item The posterior mean estimates for $R_0$ under the PMMH and PALMH procedures were 4.52 and 3.70, respectively.
\end{itemize}

\subsubsection{Algorithm details for section \ref{sec:boarding_school}}
The following algorithm describes how the PAL can be used within a delayed acceptance pmcmc scheme.
\label{sec:boarding_school_supp}
\begin{algorithm}[H]
\caption{Delayed acceptance PMMH algorithm with PAL}\label{daPMMH}
\begin{algorithmic}[1]
  \Statex {\bf Initialize:} $i = 0$, set $\bs \theta_0$ arbitrarily.
  
  \State Run a particle filter to produce an approximation to  $p(\+y_{1:t}\mid \bs \theta_0)$ and denote this as $\hat p(\+y_{1:t}\mid \bs \theta_0)$.
  
  \State Run algorithm \ref{alg:x} to produce a PAL approximation  to $p(\+y_{1:t}\mid \bs \theta_0)$ and denote this as $\hat {p}_a(\+y_{1:t}\mid \bs \theta_0)$.
  \State {\bf for}  $i \geq 1:$ \label{step1}
  \State \quad sample $\bs \theta_* \sim q(\cdot \mid \bs \theta_{i-1}).$
  \State \quad\textbf{stage 1} 
  
  \begin{itemize}
     \item Run algorithm \ref{alg:x} to produce a PAL approximation to $p(\+y_{1:t}\mid \bs \theta_*)$ and denote this as $\hat {p}_a(\+y_{1:t}\mid \bs \theta_*)$.
     \item  With probability:
     \begin{equation*}
     \alpha_1(\bs \theta_{i-1}, \bs \theta_*) = \min\left\{ 1, \frac{\hat {p}_a(\+y_{1:t}\mid \bs \theta_*)p(\bs \theta_*)}{\hat {p}_a(\+y_{1:t}\mid \bs \theta_{i-1})p(\bs \theta_{i-1})}\frac{q(\bs \theta_{i-1} \mid \bs \theta_*)}{q(\bs \theta_* \mid \bs \theta_{i-1})}\right\},
     \end{equation*}
		run a particle filter to produce an approximation to $p(\+y_{1:t}\mid \bs \theta_*)$, denote this as $\hat p(\+y_{1:t}\mid \bs \theta_*)$ and go to \textbf{Stage 2}. Otherwise, set $\bs \theta_i = \bs \theta_{i-1}$,  set $i=i+1$ and return to 4.
  \end{itemize}
  \State \label{step2} \quad\textbf{stage 2} 
  
  With probability 
  \begin{equation*}
  \alpha_2(\bs \theta_{i-1}, \bs \theta_*) = \min \left\{1, \frac{\hat p(\+ y_{1:t}|\bs \theta_*)p(\bs \theta_*)}{\hat p(\+ y_{1:t}| \bs \theta_{i-1})p(\bs \theta_{i-1})}\frac{\hat {p}_a(\+y_{1:t}| \bs \theta_{i-1})p(\bs \theta_{i-1})}{\hat {p}_a(\+y_{1:t}|\bs \theta_*)p(\bs \theta_*)} \right\},
  \end{equation*}
  \quad \quad set $\bs \theta_i = \bs \theta_*$, otherwise set $\bs \theta_i = \bs \theta_{i-1}$. Set $i=i+1$ and return to 4.
  \State {\bf end for}
\end{algorithmic}
\end{algorithm}

\subsubsection*{Time comparisons with the Linear Noise Approximation}
For comparisons with the PAL we consider LNAMH: a Metropolis-within-Gibbs algorithm with a Linear noise approximation (LNA) to the likelihood of a stochastic differential equation model used in the accept/reject step, see \cite{fearnhead2014inference} for details. We apply the LNAMH to the real dataset to compare and contrast to the PAL, for these comparisons we implement the PAL in base R, whereas the LNA computations use base R interfaced with fortran for cumbersome ODE solving calculations. The LNAMH implementation introduces an extra parameter in the variance of a Gaussian obsevation model, analogous to $V(\bs \theta)$ in section 4.2 of \cite{fearnhead2014inference}, which we will denote as $v$; we consider a vague truncated Gaussian prior of  $\mathcal{N}(400,300)_{\geq 0}$. We ran the chain for 100k iterations, discarded the first 20k and thinned to a sample of 25k to produce the posterior histograms.

The posterior predictive plot associated with the LNAMH sample, figure \ref{LNApostpred}, demonstrates good coverage of the data, yet they help illustrate some of its shortfalls in comparison to the PAL approach: the Gaussian nature of the ingredients of the LNA permits non-integer and even allows negative valued observations, which is clearly not parsimonious with reality; further, modelling with a constant in time observation variance leads to underconfidence in the start and end of the data record. In order to circumvent these issues within the LNA framework, one would have to turn to sophisticaed and expensive methods; alternatively, one could avoid each of them for free through the use of PALs.

\begin{figure}
    \centering
    \includegraphics[width=0.7\textwidth]{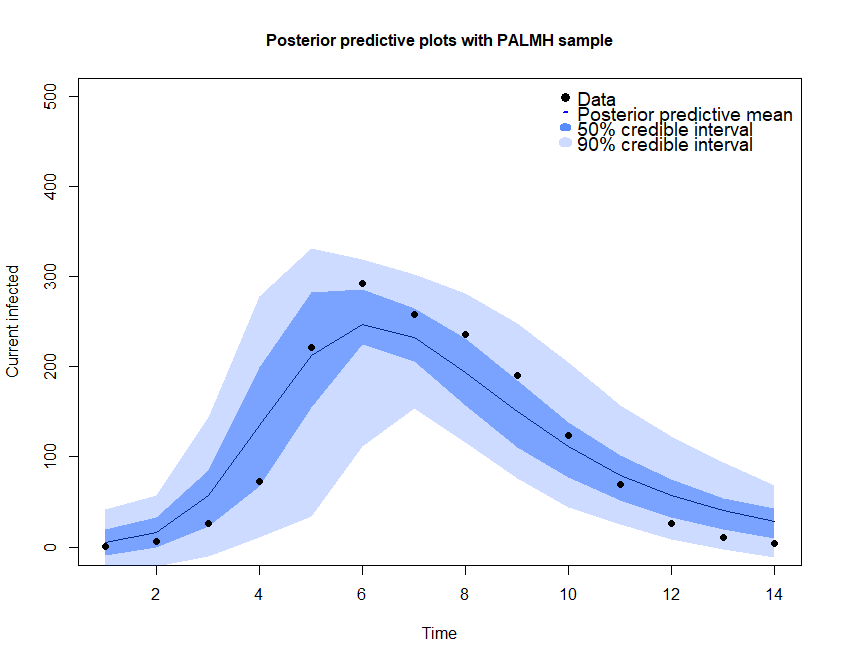}
    \caption{Posterior predictive distribution for LNAMH sample. To produce this plot we sampled a parameter from the approximate posterior and simulated from the SDE model 1000 times.}
    \label{LNApostpred}
\end{figure}

Figure \ref{LNAtime} reports the mean time ratio between a single evaluation of the LNA likelihood and a single evaluation of the PAL for varying ODE solver intermediate time step choices for the LNA and analogous choice of $h$ for the PAL. The order of magnitude of the speed gains is around 100 for the PAL, demonstrating the significant speed benefits given by the simplicity of computations needed to compute the PAL in comparison to cumbersome ODE solution calculations. Experiments were performed on an Intel(R) Core(TM) i7-7700HQ CPU @ 2.80GHz processor.

\begin{figure}
    \centering
    \includegraphics[width=0.7\textwidth]{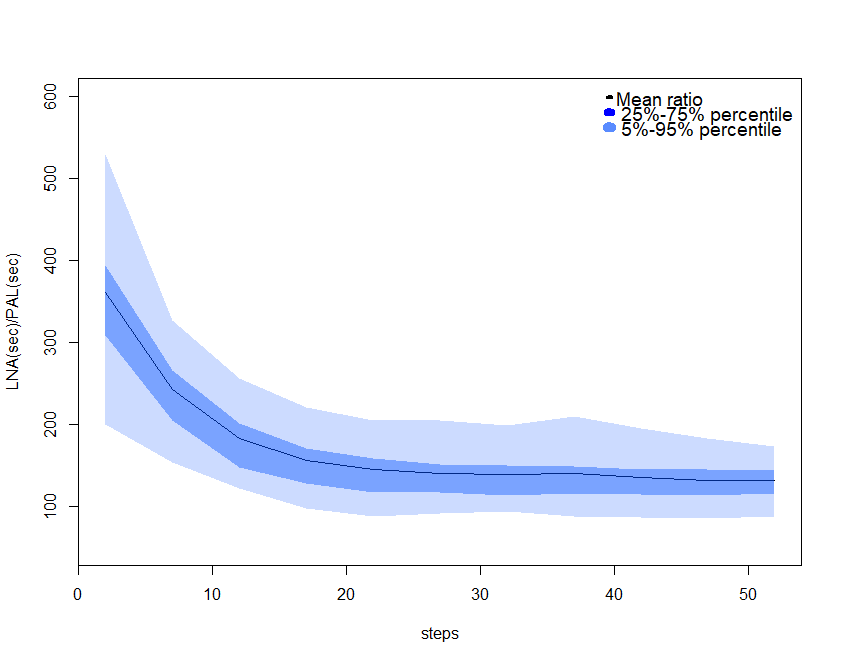}
    \caption{Time comparisons for the LNA. The ratio of one evaluation of the LNA liklihood to one evaluation of the PAL for varying ODE solver intermediate time steps, with the comparative PAL evaluation run with $h=1\slash \text{number of timesteps}$. Percentiles are based on 1000 runs.}
    \label{LNAtime}
\end{figure}

\begin{figure}
    \centering
    \includegraphics[width=\textwidth]{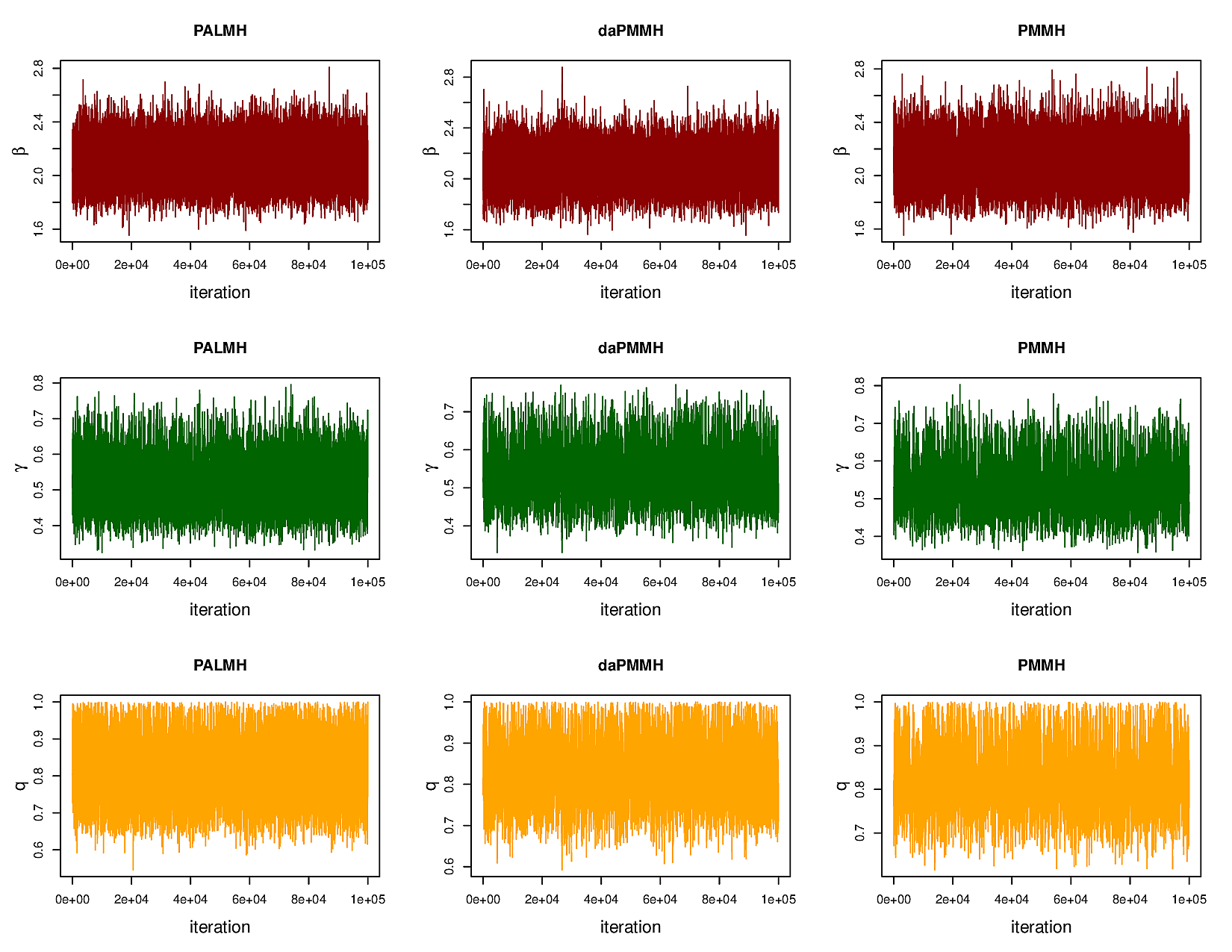}
    \caption{Boarding school influenza example. Traceplots produced by the 3 procedures we have considered when run using synthetic data generated with parameters $\bs \theta^* = (\beta^*, \gamma^*, q^*) = (2, 0.5, 0.8)$. The plots display the first $10^5$ iterations after the burn in period.}
    \label{syntheticdatatrace}
\end{figure}

\begin{figure}
    \centering
    \includegraphics[width=\textwidth]{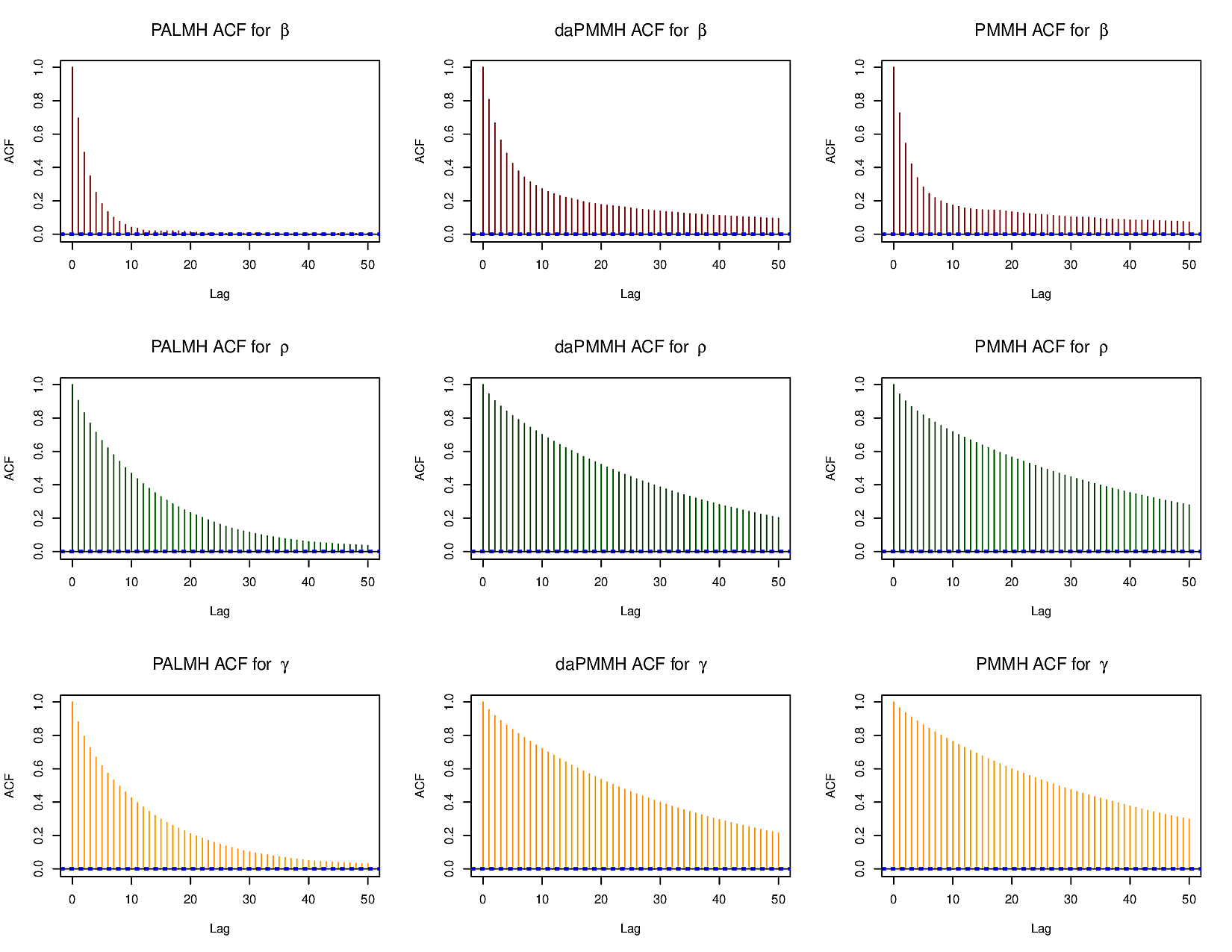}
    \caption{Boarding school influenza example. ACF plots for each considered scheme when run using synthetic data generated with parameters $\bs \theta^* = (\beta^*, \gamma^*, q^*) = (2, 0.5, 0.8)$.}
    \label{syntheticdataacf}
\end{figure}

\begin{figure}
    \centering
    \includegraphics[width=\textwidth]{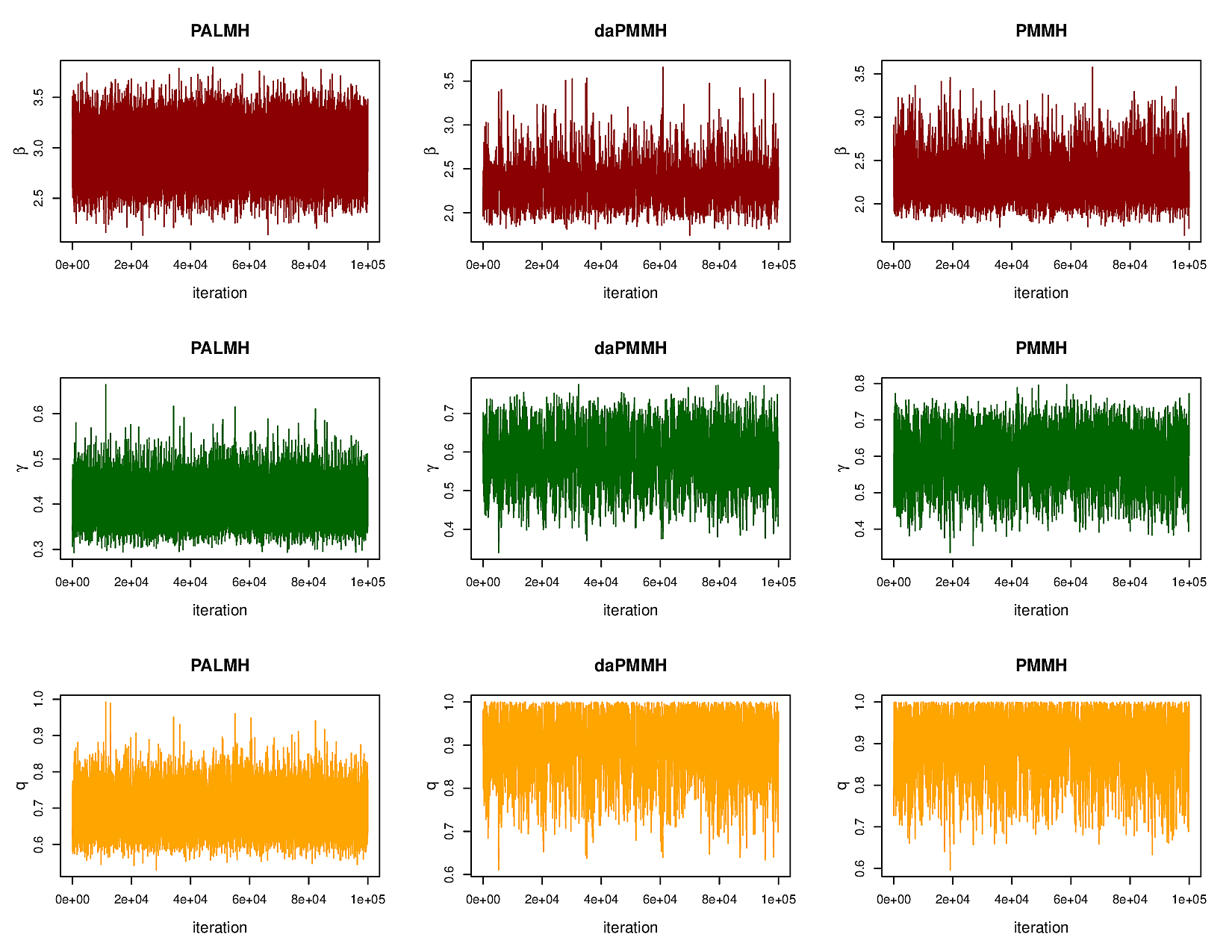}
    \caption{Boarding school influenza example. Traceplots produced by the three considered schemes run using real data. The plots display the first $10^5$ iterations after the burn in period.}
    \label{realdatatrace}
\end{figure}

\begin{figure}
    \centering
    \includegraphics[width=\textwidth]{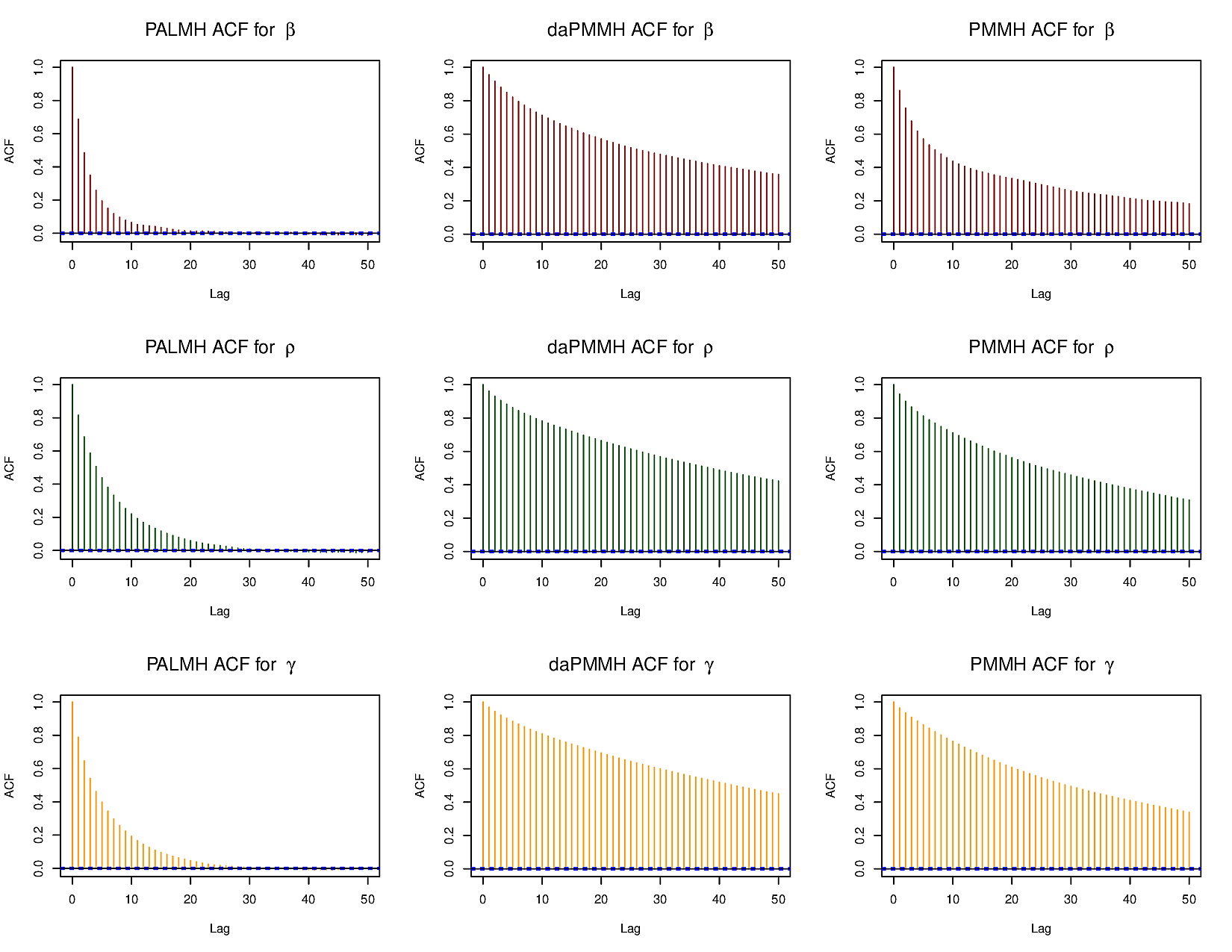}
    \caption{Boarding school influenza example. ACF plots produced by the three schemes run using real data.}
    \label{realdataacf}
\end{figure}

 \begin{figure}
    \centering
    \includegraphics[width=\textwidth]{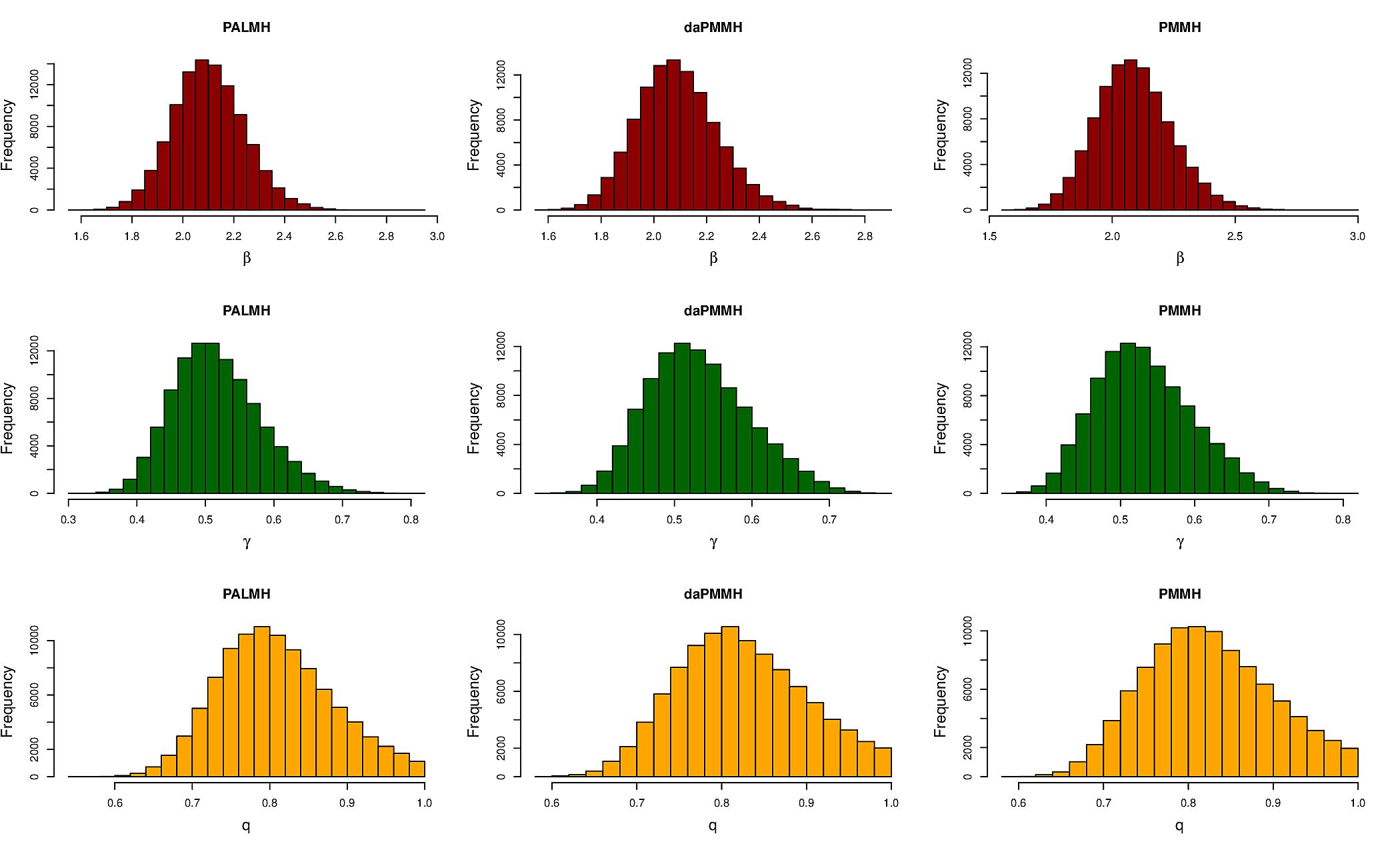}
    \caption{Boarding school influenza example. Posterior marginales produced by the three algorithms when run using synthetic data generated with parameters $\bs \theta^* = [\beta^*\; \gamma^*\; q^*]^\top = [2 \; 0.5\; 0.8]^\top$, the histograms are based on a thinned sample of $2.5 \times 10^4$.}
    \label{syntheticdatapost}
\end{figure}

\begin{figure}
    \centering
    \includegraphics[width=\textwidth]{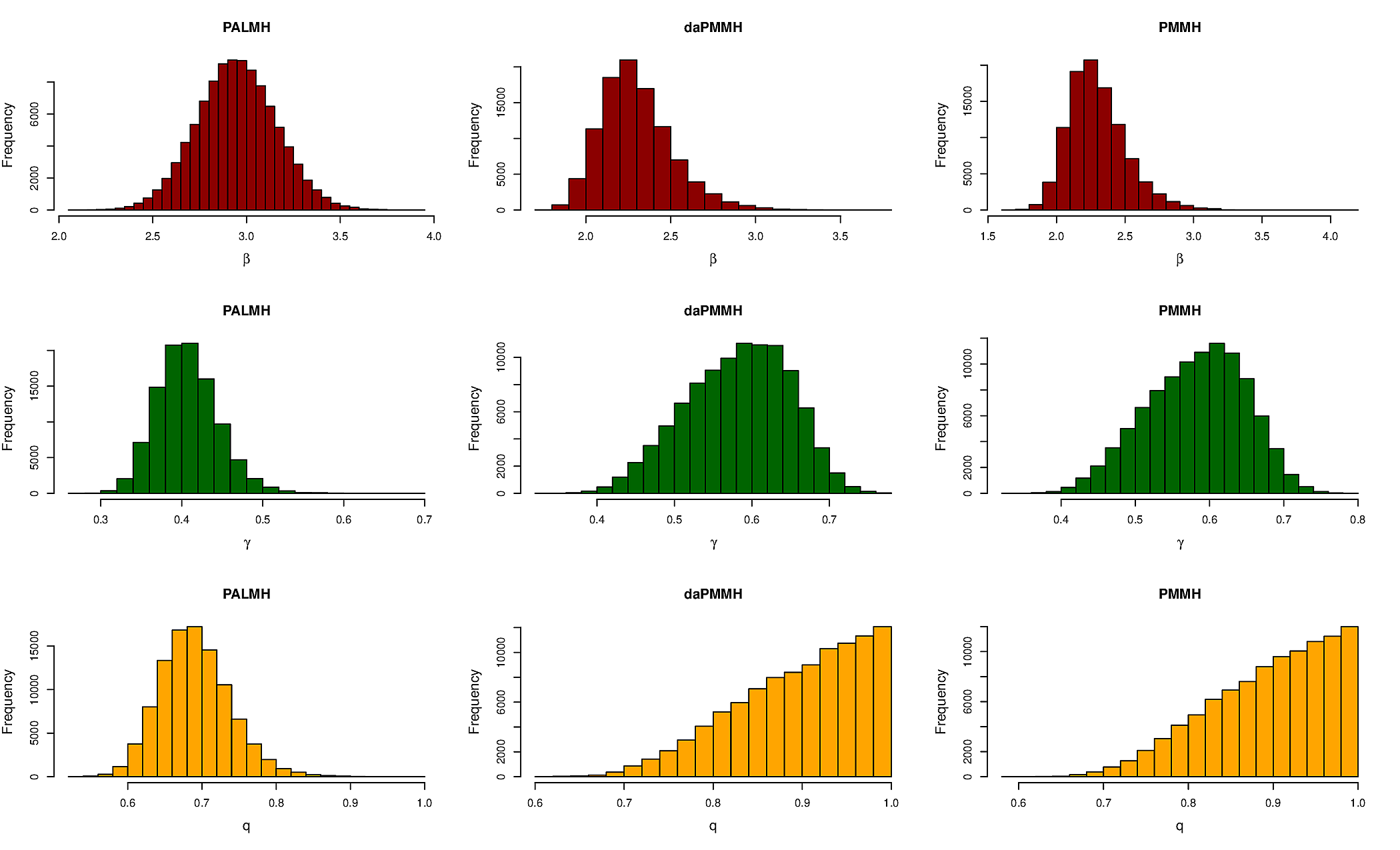}
    \caption{Boarding school influenza example. Posterior samples produced by 3 considered schemes run using real data, the histograms are based on a thinned sample of $2.5 \times 10^4$.}
    \label{realdatapost}
\end{figure}

\begin{figure}
    \centering
    \includegraphics[width=\textwidth]{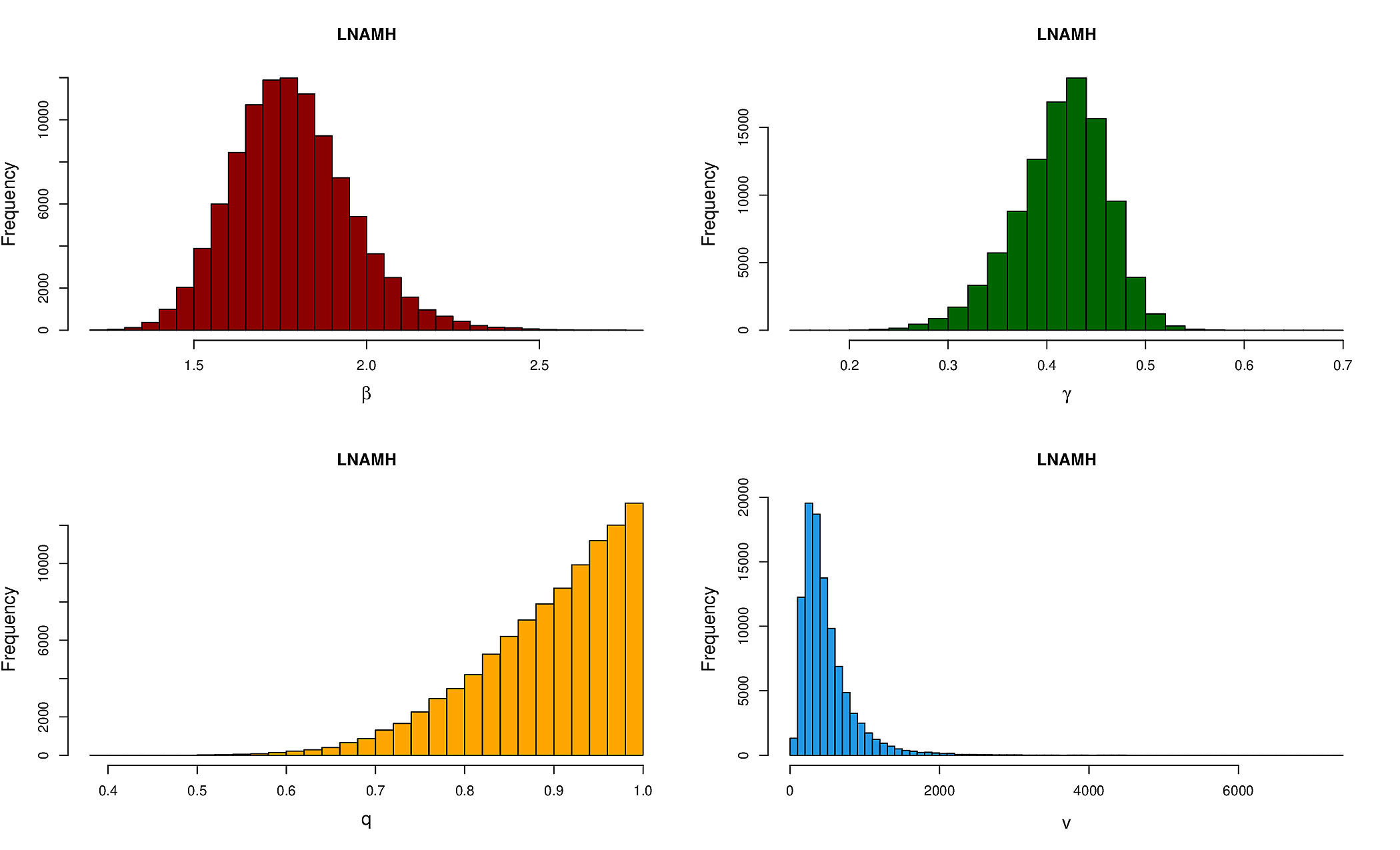}
    \caption{Boarding school influenza example. Posterior samples produced by the LNA procedure, the histograms are based on a thinned sample of $2.5 \times 10^4$.}
    \label{realdatapostlna}
\end{figure}
 \newpage
\subsection{Supplementary material for section \ref{sec:age_struct}} \label{sec:age_struct_supp}

To write the age structured model of section \ref{sec:age_struct} as an instance of the Latent Compartmental Model we take $m=16$ and identify vectors $\+x_{k,t} := \left[S_{k,t}\;E_{k,t}\;I_{k,t}\;R_{k,t} \right]^\top$, $\+x_t := \left[\+x_{1,t}^\top\;\dots\; \+x_{4,t}^\top\right]^\top$ and matrices:

\begin{equation*}
\+ Z_{k,t} :=\left[\begin{array}{cccc}
S_{k,t} -  B_{k,t} & B_{k,t} & 0 & 0 \\
0 & E_{k,t}-C_{k,t} & C_{k,t} & 0 \\
0 & 0 & I_{k,t}-D_{k,t}& D_{k,t} \\
0 & 0 & 0 & R_{k,t}
\end{array}\right],
\; 
\+Z_{t}:=\left[\begin{array}{cccc}
\+Z_{1,t}& & \dots & 0 \\
 & \+Z_{2,t,}&  &  \vdots \\
\vdots & & \ddots& \\
0 & \dots &  & \+Z_{4,t}
\end{array}\right].
\end{equation*}

\begin{equation*}
\mathbf{K}_{k,t, \boldsymbol{ \eta}(\+x_t)}:=\left[\begin{array}{cccc}
e^{-h \bar \beta_{k,t}} & 1-e^{-h \bar \beta_{k,t}} & 0 & 0 \\
0 & e^{-h \rho} & 1-e^{-h \rho} & 0 \\
0 & 0 & e^{-h \gamma} & 1-e^{-h \gamma} \\
0 & 0 & 0 & 1
\end{array}\right], 
\;
\mathbf{K}_{t, \boldsymbol{ \eta}}:=\left[\begin{array}{cccc}
\mathbf{K}_{1,t, \boldsymbol{ \eta}}& & \dots & 0 \\
 & \mathbf{K}_{2,t, \boldsymbol{ \eta}}&  &  \vdots \\
\vdots & & \ddots& \\
0 & \dots &  & \mathbf{K}_{4,t, \boldsymbol{ \eta}}
\end{array}\right],
\end{equation*}
where the $\bar \beta_{k,t}$ are the elements of the vector on the l.h.s. of \eqref{eq:beta_bar}. Due to the block-diagonal structure of the matrix $\mathbf{K}_{t, \boldsymbol{ \eta}}$ for this example, algorithm \ref{alg:Ztagg} can be simplified to avoid performing various multiplications by zero. The resulting procedure is algorithm \ref{alg:age}.

   \begin{longtable}[t!]{ c c c c c c c c c c c c}\\

 \hline
 Parameter & ODE & PAL  \\
 \hline

$q_1$ & $0.93\;(0.78,0.99)$& $0.71\;(0.53,0.97)$\\
$q_2$ & $0.96\;(0.86,0.99)$& $0.52\;(0.49,0.56)$\\
$q_3$ & $0.28\;(26,0.30)$& $0.84\;(0.61,0.99)$\\
$q_4$ & $0.28\;(0.22,0.34)$& $0.25\;(0.19,0.32)$\\
$\beta_{11}$ & $4.34\;(1.36,8.83)$& $1.26\;(0.44,2.44)$\\
$\beta_{12}$ & $2.91\;(1.09,5.45)$& $0.85\;(0.56,1.25)$\\
$\beta_{13}$ & $3.51\;(2.54,4.59)$& $0.26\;(0.09,0.52)$\\
$\beta_{14}$ & $1.33\;(0.58,2.29)$& $0.17\;(0.05,0.36)$\\
$\beta_{22}$ & $2.55\;(0.86,5.11)$& $4.21\;(3.98,4.37)$\\
$\beta_{23}$ & $6.89\;(5.88,8.18)$& $0.46\;(0.35,0.60)$\\
$\beta_{24}$ & $0.72\;(0.36,1.12)$& $0.09\;(0.04,0.16)$\\
$\beta_{33}$ & $18.08\;(17.54,18.50)$& $0.35\;(0.15,0.58)$\\
$\beta_{34}$ & $0.14\;(0.06,0.25)$& $0.10\;(0.01,0.33)$\\
$\beta_{44}$ & $21.34\;(20.41,22.26)$& $1.96\;(1.59,2.24)$\\

 \hline
 \caption{Age-structured 'flu example. Posterior means and
       95\%  credible intervals. \label{agestructfit}}
 \end{longtable}

\begin{algorithm}[!htb]
\caption{Filtering for the age-structured model}\label{alg:age}
\begin{algorithmic}[1]
  \Statex {\bf Initialize:} $\bs {\bar \lambda}_{0} \leftarrow \bs \lambda_{0}$.
  \State $\bs {\bar \lambda}_{0} \leftarrow \left[\bs {\bar \lambda}_{1,0}^\top\; \cdots\; \bs {\bar \lambda}_{4,0}^\top \right]^\top$
  \State {\bf for} $r\geq 1:$
  \State \quad {\bf for}  $t = \tau_{r-1}+1,\dots, \tau_r -1$: 
  \State \quad \quad{\bf for}  $k = 1, \dots, 4$ :

  \State \quad \quad \quad $\+ \Lambda_{k,t} \leftarrow (\bs {\bar \lambda}_{k,t-1} \otimes \+1_m)\circ \+K_{k,t,{\bs \eta}\left(\bs {\bar \lambda}_{t-1}\right)} $
  \State \quad \quad \quad $\bar{\bs \lambda}_{k,t} \leftarrow (\+ 1_m^\top \bar{ \+ \Lambda}_{k,t})^\top$
  \State \quad \quad{\bf end for}
      \State \quad\quad  $\bs {\bar \lambda}_{t} \leftarrow \left[\bs {\bar \lambda}_{1,t}^\top\; \cdots\; \bs {\bar \lambda}_{4,t}^\top \right]^\top$
  \State \quad {\bf end for}
  \State \quad {\bf for}  $k = 1, \dots, 4$: 
  \State \quad \quad  $\+ \Lambda_{k,\tau_r} \leftarrow (\bs {\bar \lambda}_{k,\tau_r-1}\otimes \+1_m)\circ \+K_{k,\tau_r,{\bs \eta}\left(\bs {\bar \lambda}_{t-1}\right)} $
  \State \quad \quad $\+M_{k,r} \leftarrow \sum_{t=\tau_{r-1}}^{\tau_r} \bs \Lambda_{k,t}\circ \+Q_{k,t}$
   \State \quad \quad  $\+ {\bar \Lambda}_{k,\tau_r} \leftarrow \left( \+1_m \otimes \+1_m - \+ Q_{k,\tau_r}  \right)\circ \bar{\bs \Lambda}_{k,\tau_r}+ \bar{\+ Y}_{k,r}\circ \bs \Lambda_{k,\tau_r}\circ \+Q_{k,\tau_r}\oslash\+M_{k,r}$
  \State \quad \quad $\bar{\bs \lambda}_{k,\tau_r} \leftarrow (\+ 1_m^\top \bar{ \+ \Lambda}_{k,\tau_r})^\top$
  \State \quad{\bf end for}
   \State \quad  $\bs {\bar \lambda}_{\tau_r} \leftarrow \left[\bs {\bar \lambda}_{1,\tau_r}^\top\; \cdots\; \bs {\bar \lambda}_{4,\tau_r}^\top \right]^\top$
    \State \quad   ${\mathcal{L}}(\bar{\+Y}_{1:4,r}|\bar{\+Y}_{1:4,1:r-1}) \leftarrow  \sum_{k=1}^{4}\+ - 1_m^\top\+M_{k,r}\+1_m +\+1_m^\top (\bar{\+ Y}_{k,r} \circ \+M_{k,r})\+1_m - \+1_m^\top \log(\bar{\+ Y}_{k,r}!)\+1_m $
    \State {\bf end for}
\end{algorithmic}
\end{algorithm}

\begin{figure}
    \centering
    \includegraphics[width=\textwidth]{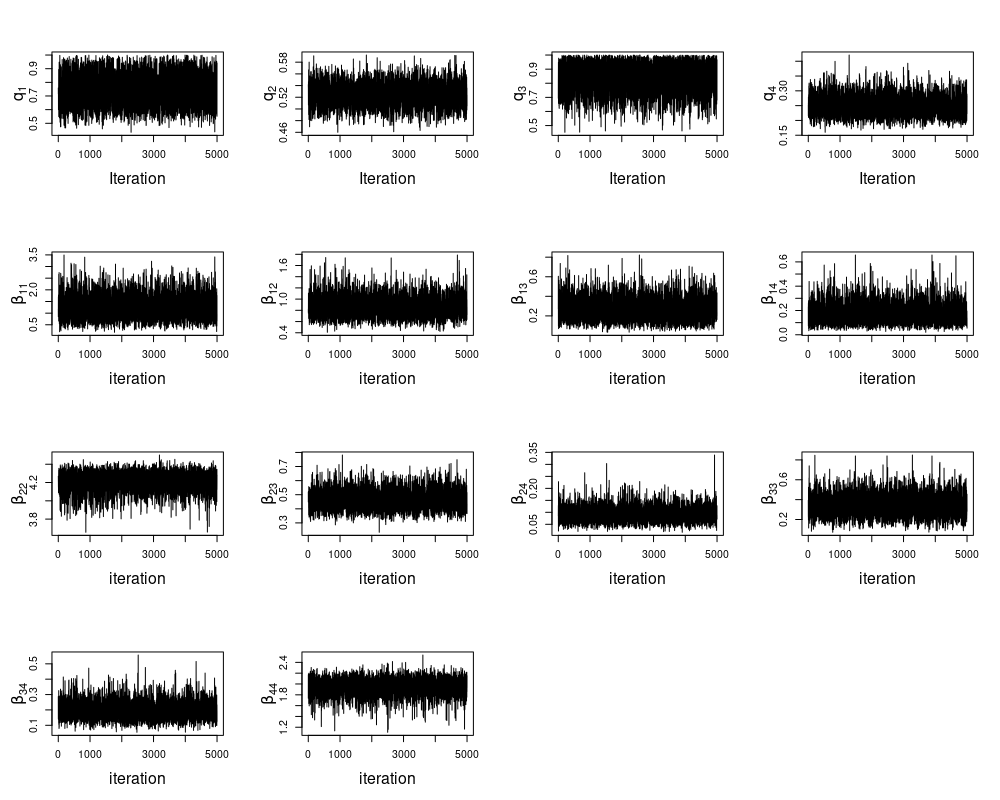}
    \caption{Age-structured example. HMC posterior trace plots for the parameters of the stochastic model produced using Stan. The plots show the first $5^5$ iterations after the burn in period.}
    \label{agestochtrace}
\end{figure}

\begin{figure}
    \centering
    \includegraphics[width=\textwidth]{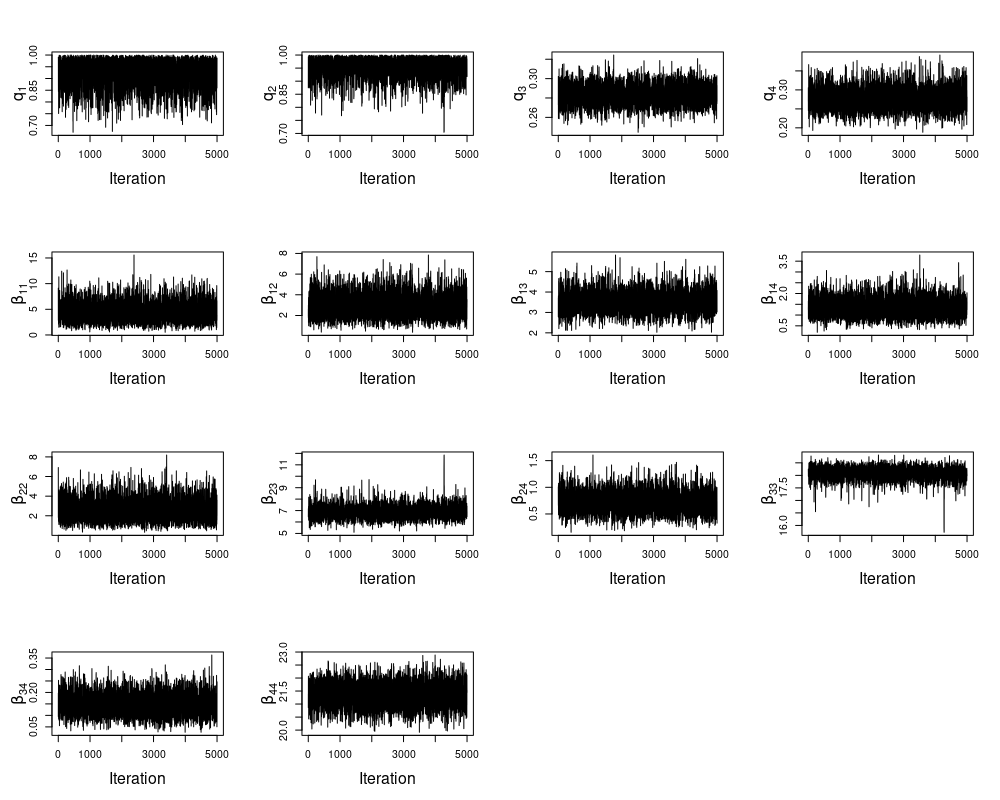}
    \caption{Age-structured example. HMC posterior trace plots for the parameters of the ODE model produced using Stan. The plots show the first $5^5$ iterations after the burn in period.}
	\label{agedettrace}
\end{figure}

\begin{figure}
    \centering
    \includegraphics[width=\textwidth]{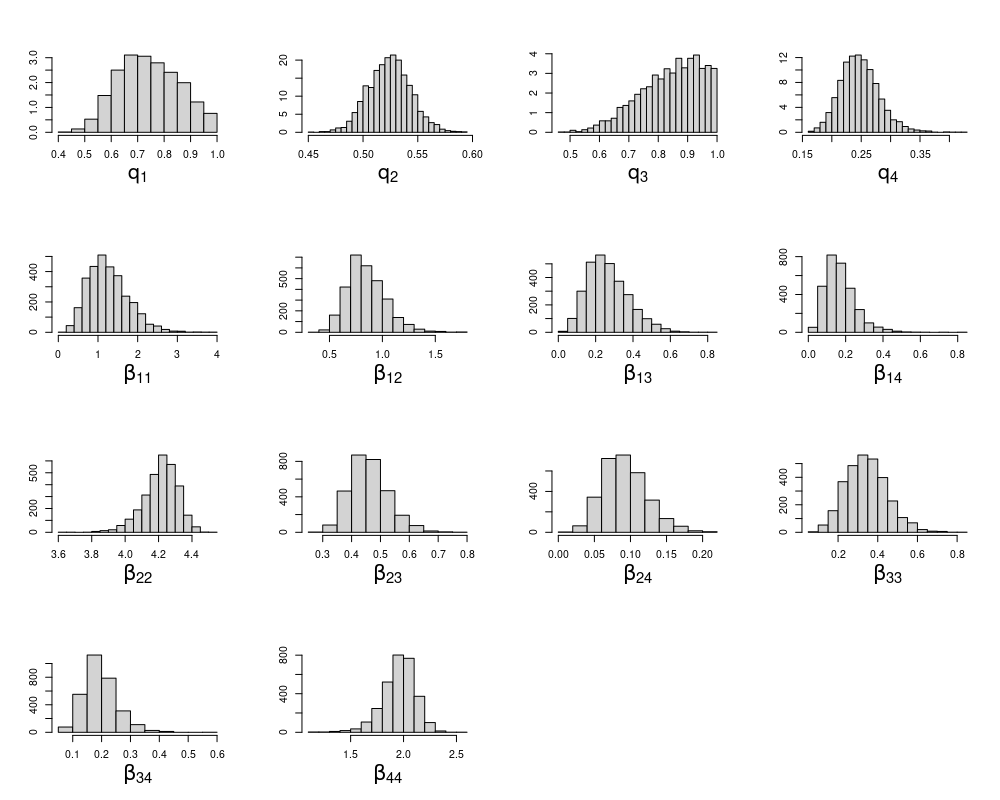}
    \caption{Age-structured example. HMC posterior histograms for the parameters of the stochastic model produced using Stan.}
    \label{agestochhist}
\end{figure}

\begin{figure}
    \centering
    \includegraphics[width=\textwidth]{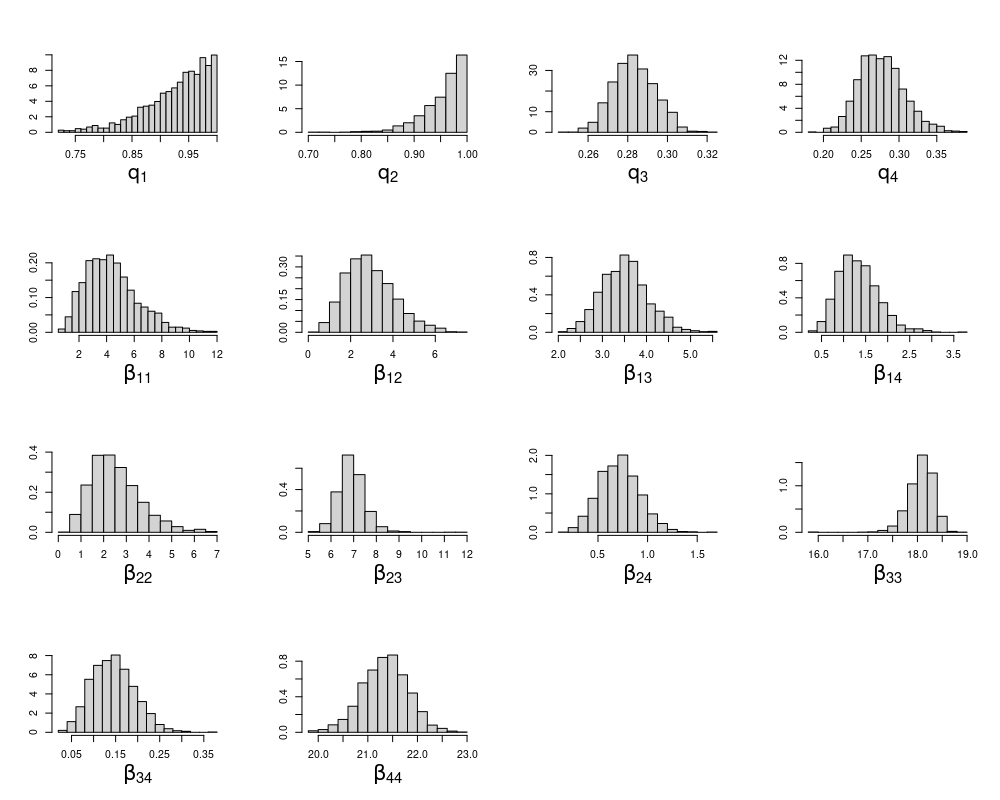}
    \caption{Age-structured example. HMC posterior histograms for the parameters of the ODE model produced using Stan.}
	 \label{ODEhist}
\end{figure}
\newpage

\subsection{Supplementary material for section \ref{sec:rotavirus}} \label{sup:rotavirus_supp}

\subsubsection{Model}
The evolution of the full age stratified rotavirus model at time $t$ is given by:

\begin{align}
S_{1,t+1}&= S_{1,t-1} + A_{1,t} + E_{1,t} - B_{1,t} - F_{1,t}^{(S)}, \\
I_{1,t+1}&= I_{1,t} + B_{1,t} - C_{1,t} - F_{1,t}^{(I)}, \\
R_{1,t+1}&= R_{1,t} + C_{1,t} - E_{1,t} - F_{1,t}^{(R)}, \\
S_{2,t+1}&= S_{2,t} + F_{1,t}^{(S)} + E_{2,t} - B_{2,t} - F_{2,t}^{(S)},\\
I_{2,t+1}&= I_{2,t}+ F_{1,t}^{(I)} + B_{2,t} - C_{2,t} - F_{2,t}^{(I)}, \\
R_{2,t+1}&=  R_{2,t} + F_{1,t}^{(R)} + C_{2,t} - E_{2,t} - F_{2,t}^{(R)},\\
S_{3,+1t}&= S_{3,t} + F_{2,t}^{(S)} + E_{3,t} - B_{3,t} - D_{t}^{(S)},\\
I_{3,t+1}&= I_{3,t}+ F_{2,t}^{(I)} + B_{3,t} - C_{3,t} - D_{t}^{(I)},\\
R_{3,t+1}&=  R_{3,t} + F_{2,t}^{(R)} + C_{3,t} - E_{3,t} - D_{t}^{(R)},\\
\end{align}
where at time $t$: $A_{1,t} \sim \text{Pois}(\alpha_t)$, for some $\alpha_t \in \mathbb{R}$ represents new births, which is chosen according to historical birth record data; $B_{\cdot,t}$ represents new infectives; $C_{\cdot,t}$ represents recovering individuals; $D_{t}^\cdot \sim \text{Binom}\left({\cdot_{t-1},1- \delta} \right)$ represents emigrating (dying) individuals; $E_{\cdot,t}$ represents individuals experiencing waning immunity; and $F_{\cdot,t}$ represents ageing individuals.

\begin{align}
\left[\begin{array}{c}
B_{1,t} \\
 F_{1,t}^{(S)} \\
 S_{1,t}- B_{1,t} - F_{1,t}^{(S)}
\end{array}\right]
&\sim\text{Mult}\left(S_{1,t},
\left[\begin{array}{c}
p_{1,t} \\
1 - e^{(-hd_1)} \\
 e^{(-hd_1)} - p_{k,t}
\end{array}\right]
\right) \\
\left[\begin{array}{c}
C_{1,t}\\
F_{1,t}^{(I)} \\
 I_{1,t}- C_{1,t} - F_{1,t}^{(I)}
\end{array}\right]
&\sim\text{Mult}\left(I_{1,t},
\left[\begin{array}{c}
1 - e^{-h\gamma} \\
1 - e^{-hd_1} \\
e^{-h\gamma} + e^{-hd_1} - 1
\end{array}\right]
\right) \\
\left[\begin{array}{c}
E_{1,t}\\
F_{1,t}^{(R)} \\
 R_{1,t}- E_{1,t} - F_{1,t}^{(R)}
\end{array}\right]
&\sim\text{Mult}\left(R_{1,t},
\left[\begin{array}{c}
1 - e^{-h\omega} \\
1 - e^{-hd_1} \\
e^{-h\omega} + e^{-hd_1} - 1
\end{array}\right]
\right) \\
\left[\begin{array}{c}
B_{2,t} \\
 F_{2,t}^{(S)} \\
 S_{2,t} - B_{2,t} - F_{2,t}^{(S)}
\end{array}\right]
&\sim\text{Mult}\left(S_{2,t},
\left[\begin{array}{c}
p_{2,t} \\
1 - e^{(-hd_2)} \\
 e^{(-hd_2)} - p_{2,t}
\end{array}\right]
\right) \\
\left[\begin{array}{c}
C_{2,t}\\
F_{2,t}^{(I)} \\
 I_{2,t} - C_{2,t} - F_{2,t}^{(I)}
\end{array}\right]
&\sim\text{Mult}\left(I_{2,t},
\left[\begin{array}{c}
1 - e^{-h\gamma} \\
1 - e^{-hd_2} \\
e^{-h\gamma} + e^{-hd_2} - 1
\end{array}\right]
\right) \\
\left[\begin{array}{c}
E_{2,t}\\
F_{2,t}^{(R)} \\
 R_{2,t}- E_{2,t} - F_{2,t}^{(R)}
\end{array}\right]
&\sim\text{Mult}\left(R_{2,t},
\left[\begin{array}{c}
1 - e^{-h\omega} \\
1 - e^{-hd_2} \\
e^{-h\omega} + e^{-hd_2} - 1
\end{array}\right]
\right) \\
\left[\begin{array}{c}
B_{3,t} \\
S_{3,t} - D^{(S)}_t- B_{3,t}
\end{array}\right]
&\sim\text{Mult}\left(S_{3,t} - D^{(S)}_t,
\left[\begin{array}{c}
p_{3,t} \\
1-p_{3,t}
\end{array}\right]
\right) \\
\left[\begin{array}{c}
C_{3,t}\\
 I_{3,t} - D_t^{(I)}- C_{2,t}
\end{array}\right]
&\sim\text{Mult}\left(I_{3,t} - D_t^{(I)},
\left[\begin{array}{c}
1 - e^{-h\gamma} \\
e^{-h\gamma}
\end{array}\right]
\right) \\
\left[\begin{array}{c}
E_{3,t}\\
 R_{3,t} - D_t^{(R)} - E_{2,t}
\end{array}\right]
&\sim\text{Mult}\left(R_{3,t} - D_t^{(R)},
\left[\begin{array}{c}
1 - e^{-h\omega} \\
e^{-h\omega}
\end{array}\right]\right)
\end{align}  

To align notation with the model descriptions in section \ref{sec:models} collect observations at time $r$ in the matrix $\bar {\+Y}_r \in \mathbb{N}^{9 \times 9}$ which has elements equal to zero except $\bar{Y}_r^{(3k-2,3k-1)} = Y_{r,k}$ for age groups $k = 1,2,3$, similarly collect reporting rates in ${\+Q}_r \in \mathbb{N}^{9 \times 9}$ which has elements equal to zero except $Q_r^{(3k-2,3k-1)} = q_{r,k}$ for $k=1,2,3$. Define ${\+x}_t = [S_{1,t} \; I_{1,t} \; R_{1,t}\; S_{2,t} \; I_{2,t} \; R_{2,t} \; S_{3,t} \; I_{3,t}  \; R_{3,t} ]$. Identify the matrix:


\begin{align}
\mathbf{K}_{t, \boldsymbol{\eta}}^{(1, \cdot)} &= 
\left[ e^{-hd_1} - p_{1,t}  \quad \quad p_{1,t} \quad \quad 0 \quad \quad 1 - e^{-hd_1} \quad \quad 0 \quad \quad 0 \quad \quad 0 \quad \quad 0 \quad \quad 0\right], \\
\mathbf{K}_{t, \boldsymbol{\eta}}^{(2, \cdot)} &=  \left[ 0 \quad \quad e^{-h\gamma} + e^{-hd_1} - 1 \quad \quad 1-e^{-h\gamma} \quad \quad 0 \quad \quad 1 - e^{-hd_1} \quad \quad 0 \quad \quad 0 \quad \quad 0 \quad \quad 0 \right], \\
\mathbf{K}_{t, \boldsymbol{\eta}}^{(3, \cdot)} &= \left[ 1 - e^{-h\omega} \quad \quad 0 \quad \quad 0 \quad \quad e^{-h\omega} +  e^{-hd_1} -1 \quad \quad 0 \quad \quad 1-e^{-hd_1} \quad \quad 0 \quad \quad 0 \quad \quad 0 \right], \\
\mathbf{K}_{t, \boldsymbol{\eta}}^{(4, \cdot)} & = \left[ 0 \quad \quad 0 \quad \quad 0 \quad \quad  e^{-hd_2} - p_{2,t} \quad \quad p_{2,t} \quad \quad 0 \quad \quad 1 - e^{-hd_2} \quad \quad 0 \quad \quad 0 \right], \\
\mathbf{K}_{t, \boldsymbol{\eta}}^{(5, \cdot)} & = \left[0 \quad \quad 0 \quad \quad 0 \quad \quad 0 \quad \quad e^{-h\gamma} + e^{-hd_2} - 1 \quad \quad 1 - e^{-h\gamma} \quad \quad 0 \quad \quad 1- e^{-hd_2} \quad \quad 0\right], \\
\mathbf{K}_{t, \boldsymbol{\eta}}^{(6, \cdot)} & = \left[0 \quad \quad 0 \quad \quad 0 \quad \quad 1 - e^{-h\omega} \quad \quad 0 \quad \quad e^{-h\omega} + e^{-hd_1} - 1 \quad \quad 0 \quad \quad 0 \quad \quad 1 - e^{-hd_1}\right], \\
\mathbf{K}_{t, \boldsymbol{\eta}}^{(7, \cdot)} & = \left[0 \quad \quad 0 \quad \quad 0 \quad \quad 0 \quad \quad 0 \quad \quad 0\quad \quad 1 - p_{3,t} \quad \quad p_{3,t}  \quad \quad 0\right], \\
\mathbf{K}_{t, \boldsymbol{\eta}}^{(8, \cdot)} & = \left[0 \quad \quad 0 \quad \quad 0 \quad \quad 0 \quad \quad 0 \quad \quad 0 \quad \quad 0\quad \quad e^{-h\gamma} \quad \quad 1 - e^{-h\gamma} \right] ,\\
\mathbf{K}_{t, \boldsymbol{\eta}}^{(9, \cdot)} & = \left[0 \quad \quad 0 \quad \quad 0 \quad \quad 0 \quad \quad 0 \quad \quad 0\quad \quad 1 - e^{-h\omega} \quad \quad 0 \quad \quad e^{-h\omega}\right].   
\end{align}
Where for models EqEq and EqOv we have $p_{k,t}=1-\exp\left\{-\bs{\beta}_k^\top \frac{\bs I_{t}}{n}\chi_t\right\}$ for $k=1,2,3$ , and for model OvOv we have $p_{k,t}=1-\exp\left\{-\bs{\beta}_k^\top \frac{\bs I_{t}}{n}\chi_t\xi_{r}\right\}$ for $k=1,2,3$, in which case we will write $\mathbf{K}_{t, \boldsymbol{\eta}} = \mathbf{K}_{t, \boldsymbol{\eta}, \xi}$ .

For models EqOv and OvOv we have for $k = 1,2,3$:

$$Q_r^{(3k-2,3k-1)} \sim \mathcal{N}(\mu_q, \sigma_q^2)_{\geq 0, \leq 1}$$
corresponding to the reporting rate of new infectived individuals for each age group. Denote this prior density of $\+Q_r$ as $f(\cdot \mid \mu_q, \sigma_q^2)$.

\begin{figure}[t]
\centering
\begin{tikzpicture}
\tikzstyle{new style 0}=[fill={rgb,255: red,171; green,255; blue,253}, draw=black, shape=rectangle, minimum size=1.5cm]
\tikzstyle{new edge style 0}=[->]
\tikzstyle{none}=[]
	\begin{pgfonlayer}{nodelayer}
		\node [style=new style 0] (0) at (-1.5, 28.5) {$R_1$};
		\node [style=new style 0] (1) at (-5, 28.5) {$I_1$};
		\node [style=new style 0] (2) at (-8.5, 28.5) {$S_1$};
		\node [style=new style 0] (3) at (-8.5, 25) {$S_2$};
		\node [style=new style 0] (4) at (-5, 25) {$I_2$};
		\node [style=new style 0] (5) at (-1.5, 25) {$R_2$};
		\node [style=new style 0] (6) at (-1.5, 21.5) {$R_3$};
		\node [style=new style 0] (7) at (-8.5, 21.5) {$S_3$};
		\node [style=new style 0] (8) at (-5, 21.5) {$I_3$};
		\node [style=none] (9) at (-8.5, 19.25) {};
		\node [style=none] (10) at (-5, 19.25) {};
		\node [style=none] (11) at (-1.5, 19.25) {};
		\node [style=none] (12) at (-3, 29.5) {$\omega$};
		\node [style=none] (13) at (-8, 27) {$d_1$};
		\node [style=none] (14) at (-4.5, 27) {$d_1$};
		\node [style=none] (15) at (-1, 27) {$d_1$};
		\node [style=none] (16) at (-8, 23.5) {$d_2$};
		\node [style=none] (17) at (-1, 23.5) {$d_2$};
		\node [style=none] (18) at (-4.5, 23.5) {$d_2$};
		\node [style=none] (20) at (-8, 20.25) {$\delta$};
		\node [style=none] (21) at (-4.5, 20.25) {$\delta$};
		\node [style=none] (22) at (-1, 20.25) {$\delta$};
		\node [style=none] (23) at (-3, 26) {$\omega$};
		\node [style=none] (24) at (-3, 22.5) {$\omega$};
		\node [style=none] (25) at (-6.5, 28) {$\beta_1$};
		\node [style=none] (26) at (-6.5, 24.5) {$\beta_2$};
		\node [style=none] (27) at (-6.5, 21) {$\beta_3$};
		\node [style=none] (28) at (-3, 28) {$\gamma$};
		\node [style=none] (29) at (-3, 28) {};
		\node [style=none] (30) at (-3, 24.5) {$\gamma$};
		\node [style=none] (31) at (-3, 21) {$\gamma$};
		\node [style=none] (32) at (-10.75, 28.5) {};
		\node [style=none] (33) at (-10, 28.75) {$\alpha$};
		\node [style=none] (34) at (-10.75, 28.5) {};
	\end{pgfonlayer}
	\begin{pgfonlayer}{edgelayer}
		\draw [style=new edge style 0, in=180, out=0] (2) to (1);
		\draw [style=new edge style 0] (1) to (0);
		\draw [style=new edge style 0] (3) to (4);
		\draw [style=new edge style 0] (4) to (5);
		\draw [style=new edge style 0] (7) to (8);
		\draw [style=new edge style 0] (8) to (6);
		\draw [style=new edge style 0] (2) to (3);
		\draw [style=new edge style 0] (1) to (4);
		\draw [style=new edge style 0] (0) to (5);
		\draw [style=new edge style 0] (3) to (7);
		\draw [style=new edge style 0] (4) to (8);
		\draw [style=new edge style 0] (5) to (6);
		\draw [style=new edge style 0] (7) to (9.center);
		\draw [style=new edge style 0] (8) to (10.center);
		\draw [style=new edge style 0] (6) to (11.center);
		\draw [style=new edge style 0, in=30, out=150] (0) to (2);
		\draw [style=new edge style 0, in=30, out=150] (5) to (3);
		\draw [style=new edge style 0, in=30, out=150] (6) to (7);
		\draw [style=new edge style 0] (32.center) to (2);
	\end{pgfonlayer}
\end{tikzpicture}
\caption{Schema for the latent compartmental model of rotavirus transmission.}
\end{figure}
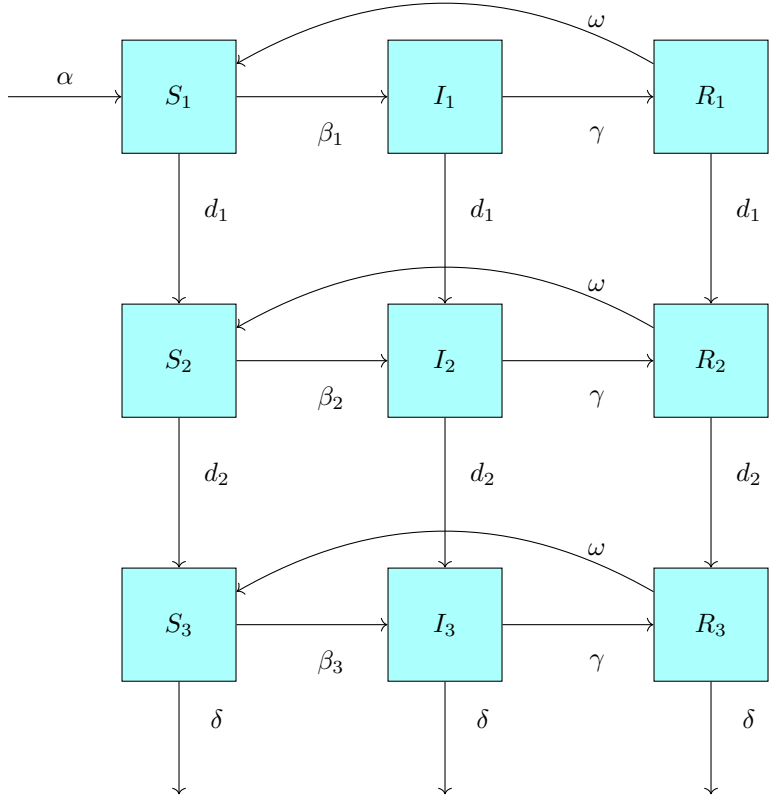
\subsubsection{Inference}
We assume that the values of $\alpha, d_1, d_2, \delta, \gamma,\omega$ and $\mu_q$ are known, we set them to the same values as assumed in \cite{stocks2020model}, these are available on the GitHub page. All other parameters are to be estimated.
\subsubsection*{Laplace approximation proposals for the rotavirus example}\label{Qproprota}
Consider algorithm \ref{alg:rota}. We factorise the proposal of particles at time $r$, $[\xi_r^{(i)}, \+Q_r^{(i)}]$, into sampling $\xi_r^{(i)}$ from its prior, then given this we seek a Laplace/PAL approximation to the distribution:

$$\hat p(\+Q_r \mid  \bar{\+Y}_{1:r}, \+Q_{1:r-1}, \bs \xi_{1:r}):= \frac{\exp\mathcal{L}(\bar{\+Y}_r \mid \bar{\+Y}_{1:r-1}, \+Q_{1:r}, \bs \xi_{1:r})f(\+Q_r \mid \mu_q, \sigma_q^2)}{\int \exp\mathcal{L}(\bar{\+Y}_r \mid \bar{\+Y}_{1:r-1}, \+Q_{1:r}, \bs \xi_{1:r})f(\+Q_r \mid  \mu_q, \sigma_q^2)d\+Q_{r}}$$
Surpressing dependence on the particle, let $\+L_r =  \sum_{t = \tau_{r-1}+1}^{\tau_{r}} \bs \Lambda_{t}$ with $\bs \Lambda_{t}$ calculated as per line \ref{rotapred} of algorithm \ref{alg:rota}, we have for some constants $C_1,C_2$:
\begin{equation}
\begin{aligned}
\log \hat p(\+Q_r \mid  \bar{\+Y}_{1:r}, \+Q_{1:r-1}, \bs{\xi}_{1:r}) & = \mathcal{L}(\bar{\+Y}_r \mid \bar{\+Y}_{1:r-1}, \+Q_{1:r}, \bs \xi_{1:r}) + \log f(\+Q_r \mid \mu_q, \sigma_q^2) + C_1 \\
&= \sum_{j=1}^3 \bigg\{ \bar Y_r^{(3j-2,3j-1)} \log(Q_r^{(3j-2,3j-1)}L_r^{(3j-2,3j-1)}) \\
&- L_r^{(3j-2,3j-1)}Q_r^{(3j-2,3j-1)} - \bar Y_r^{(3j-2,3j-1)}! \\
&- \frac{1}{2}\left(\frac{Q_r^{(3j-2,3j-1)} - \mu_q}{\sigma_q} \right)^2 \bigg\} +C_2
\end{aligned} 
\end{equation}
To get the mean of a Laplace approximation to the above we must find it's maximum w.r.t. $\+Q_r$, hence for $j=1,2,3$:

\begin{equation}
\begin{aligned}
\frac{d \log \hat p(\+Q_r \mid \+y_r)}{dQ_r^{(3j-2,3j-1)}} &= \frac{\bar Y_r^{(3j-2,3j-1)}}{Q_r^{(3j-2,3j-1)}} - L_r^{(3j-2,3j-1)} - \frac{Q_r^{(3j-2,3j-1)} - \mu_q}{\sigma_q^2} = 0\\ 
&\iff  (Q_r^{(3j-2,3j-1)})^2 + (L_r^{(3j-2,3j-1)}\sigma_q^2 - \mu_q )Q_r^{(3j-2,3j-1)} - \bar Y_r^{(3j-2,3j-1)}\sigma_q^2 = 0 \\
&\implies Q_r^{(3j-2,3j-1)} = \frac{1}{2}\left(\mu_q  - L_r^{(3j-2,3j-1)}\sigma_q^2 + \sqrt{(L_r^{(3j-2,3j-1)}\sigma_q^2- \mu_q )^2 + 4\bar Y_r^{(3j-2,3j-1)}\sigma_q^2 }\right)\\
 &=: \mu^{(j)}_r.
\end{aligned}
\end{equation}
For the variance we find the second derivative and evaluate it at $\mu^{(j)}_r$:

\begin{equation}
\begin{aligned}
\frac{d^2 \log \hat p(\+q_r \mid \+y_r)}{d(Q_r^{(3j-2,3j-1)})^2} = -  \frac{\bar Y_r^{(3j-2,3j-1)}}{(Q_r^{(3j-2,3j-1)})^2} - \frac{1}{\sigma_q^2} \\
\implies \left(\sigma_r^{(j)}\right)^2 = \left(\frac{\bar Y_r^{(3j-2,3j-1)}}{\left(\mu^{(j)}_r\right)^2} + \frac{1}{\sigma_q^2}\right)^{-1}.
\end{aligned}
\end{equation}
Hence, having proposed $\xi_r$ from its prior, we propose $\+ Q_r$  by setting all elements to be zero except:
\begin{equation}\label{eq:Qproprota}
\begin{aligned}
Q_r^{(3j-2,3j-1)} \sim \mathcal{N}\left(\mu_r^{(j)}, (\sigma_r^{(j)})^2\right)_{\geq 0, \leq 1} \quad \quad \text{ for } j=1,2,3.
\end{aligned}
\end{equation}
Let $ \pi(\cdot \mid \bar{\+Y}_{1:r}, \+Q_{1:r-1}, \xi_{1:r})$ be the proposal density associated with \eqref{eq:Qproprota}.
\begin{algorithm}
\caption{PAL within SMC for model of Rotavirus}\label{alg:rota}
\begin{algorithmic}[1]
  \Statex {\bf initialize:} $\bar{\bs \lambda}_{0}^{(i)} \leftarrow \bs \lambda_0$ for $i = 1$ to $n_{part}$.
  \State {\bf for}  $r  \geq 1$:
  \State \quad {\bf for}  $i = 1$ to $n_{part}$
  \State \quad \quad $\xi_r^{(i)} \sim \mathrm{Gamma}(\sigma_{\xi},\sigma_{\xi})$ 
  \State \quad \quad  {\bf for}  $  t = \tau_{r-1}+1, \dots,  \tau_r - 1$:
  \State \quad \quad \quad \label{rotapred} $\bs \Lambda_{t}^{(i)} \leftarrow  ((\bar{ \bs \lambda}^{(i)}_{t-1}\circ \bs \delta_{t}) \otimes \+1_m)\circ \+K_{t,\bs \eta\left(\bar{ \bs \lambda}_{t-1}^{(i)}\circ \bs \delta_{t}\right),\xi_r^{(i)}} + \bs \alpha_t$
  \State \quad \quad \quad $\bar{\bs \lambda}^{(i)}_{t} \leftarrow  (\+1_m^\top \bs \Lambda_{t}^{(i)})^\top$
  \State \quad \quad {\bf end for}
  \State  \quad \quad $\bs \Lambda^{(i)}_{\tau_r}\leftarrow  ((\bs \lambda^{(i)}_{\tau_r -1}\circ \bs \delta_{\tau_r}) \otimes \+1_m)\circ \+K_{\tau_r,\bs \eta(\bs \lambda^{(i)}_{\tau_r -1}\circ \bs \delta_{\tau_r}),\xi_r^{(i)}} + \bs \alpha_{\tau_r}$
  \State \quad \quad $\+Q_r^{(i)}\sim \pi(\cdot \mid \bar{\+Y}_{1:r}, \+Q_{1:r-1}^{(i)}, \xi_{1:r}^{(i)})$ calculated according to \eqref{eq:Qproprota}.
  \State \quad \quad $\+M_r^{(i)} \leftarrow   \sum_{t = \tau_{r-1}+1}^{\tau_{r}} \bs \Lambda^{(i)}_{t}\circ \+Q_r^{(i)}$
  \State \quad \quad $  \mathcal{L}(\bar{\+Y}_{r} \mid \bar{\+Y}_{1:r-1}, {\+Q}^{(i)}_{1:r}, \xi_{1:r}^{(i)}) \leftarrow \+1_m^\top\+M_r\+1_m +\+1_m^\top (\bar{\+ Y}_r \circ \log\+M_r)\+1_m - \+1_m^\top \log(\bar{\+ Y}_r!)\+1_m$
  \State \quad \quad $\log w^{(i)}_r \leftarrow  \mathcal{L}(\bar{\+Y}_{r} \mid \bar{\+Y}_{1:r-1}, {\+Q}^{(i)}_{1:r}, \xi_{1:r}^{(i)}) + f(\+Q_r^{(i)} \mid \mu_q, \sigma_q) - \pi(\+Q_r^{(i)}\mid \bar{\+Y}_{1:r}, \+Q_{1:r-1}^{(i)}, \xi_{1:r}^{(i)})  $
    \State \quad \quad $\bar{\bs \Lambda}^{(i)}_{\tau_r}\leftarrow (\+1_m \otimes \+1_m - \+Q_r^{(i)})\circ \bs \Lambda_{\tau_r}^{(i)}+ \bar{\+Y}_{{r}} \circ \bs \Lambda^{(i)}_{\tau_r} \circ \+ Q_r^{(i)}  \oslash \+M_r^{(i)} $
  \State \quad \quad $\bar{\bs \lambda}^{(i)}_{ \tau_{r}} \leftarrow (\+ 1_m^\top \bar{ \+ \Lambda}^{(i)}_{ \tau_{r}})^\top$
  \State \quad \quad {\bf end for}
  \State \quad $\hat{\mathcal{L}}(\bar{\+ Y}_r |\bar{\+ Y}_{1:r-1} )\leftarrow \log\left(n_{part}^{-1}\sum_{j=1}^{n_{part}} w_r^{(j)}\right)$
  \State \quad $\bar w_r^{(i)} \leftarrow w_r^{(i)}/\sum_j w_r^{(j)}$ for $i = 1$ to $n_{part}$
  \State \quad {\bf resample} $\left\{\bar{\bs \lambda}^{(i)}_{\tau_r} \right\}_{i=1}^{n_{part}}$ according to a systematic resampling scheme with weights $\{ w_r^{(i)}\}_{i=1}^{n_{part}}$
  \State {\bf end for}
\end{algorithmic}
\end{algorithm}
The resulting approximate likelihood estimate for algorithm \ref{alg:rota} is:

\begin{equation}
p(\bar{\+ Y}_{1:R}) \approx \sum_{r=1}^R \hat{\mathcal{L}}(\bar{\+ Y}_r |\bar{\+ Y}_{1:r-1} ).
\end{equation}
\subsubsection*{Convergence plots for coordinate ascent algorithm}
For each of EqEq EqOv, and OvOv, we performed a finite differencing coordinate ascent optimisation. That is, for each paramter: fix all others to their current value and approximate the sign of the gradient with finite differencing and take a step in positive gradient direction - cycle through parameters until convergence. Figures \ref{EqEqconvergence}, \ref{EqOvconvergence}, and \ref{OvOvconvergence} demonstrate the convergence of this procedure for each model EqEq, EqOv, and OvOv respectively.
\begin{figure}
\includegraphics[width = \textwidth]{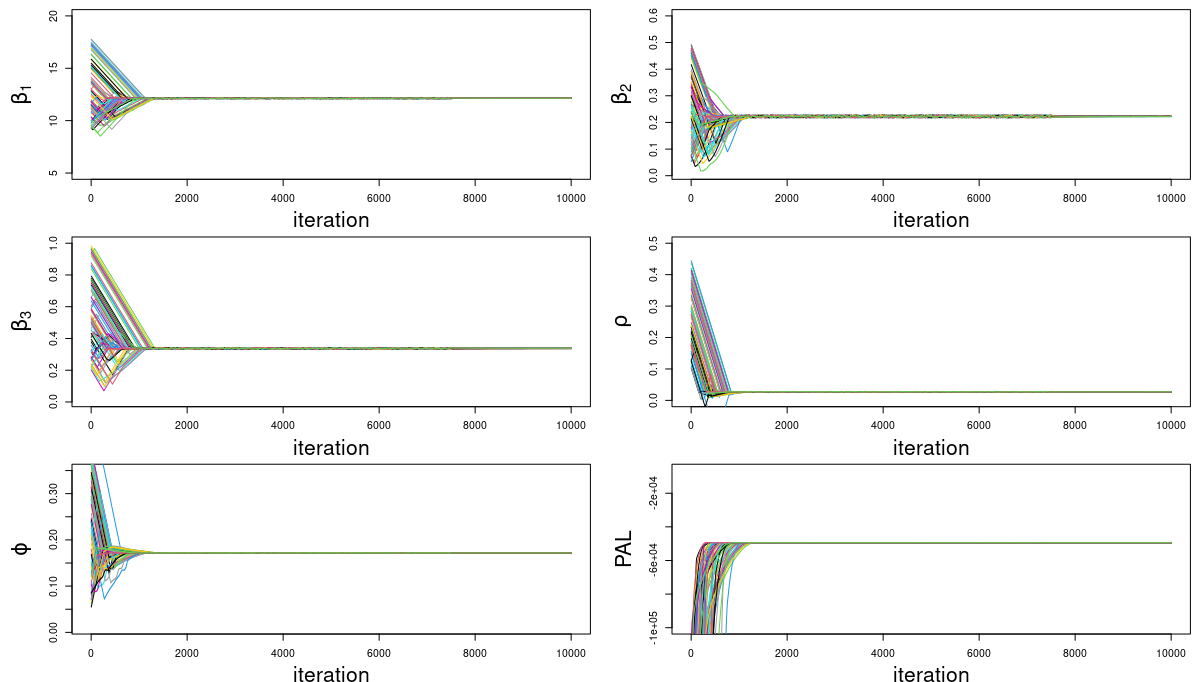}
\caption{Plots showing 100 runs of the optimization procedure the EqEq rotavirus model applied to real data.}
\label{EqEqconvergence}
\end{figure}
\begin{figure}
\includegraphics[width = \textwidth]{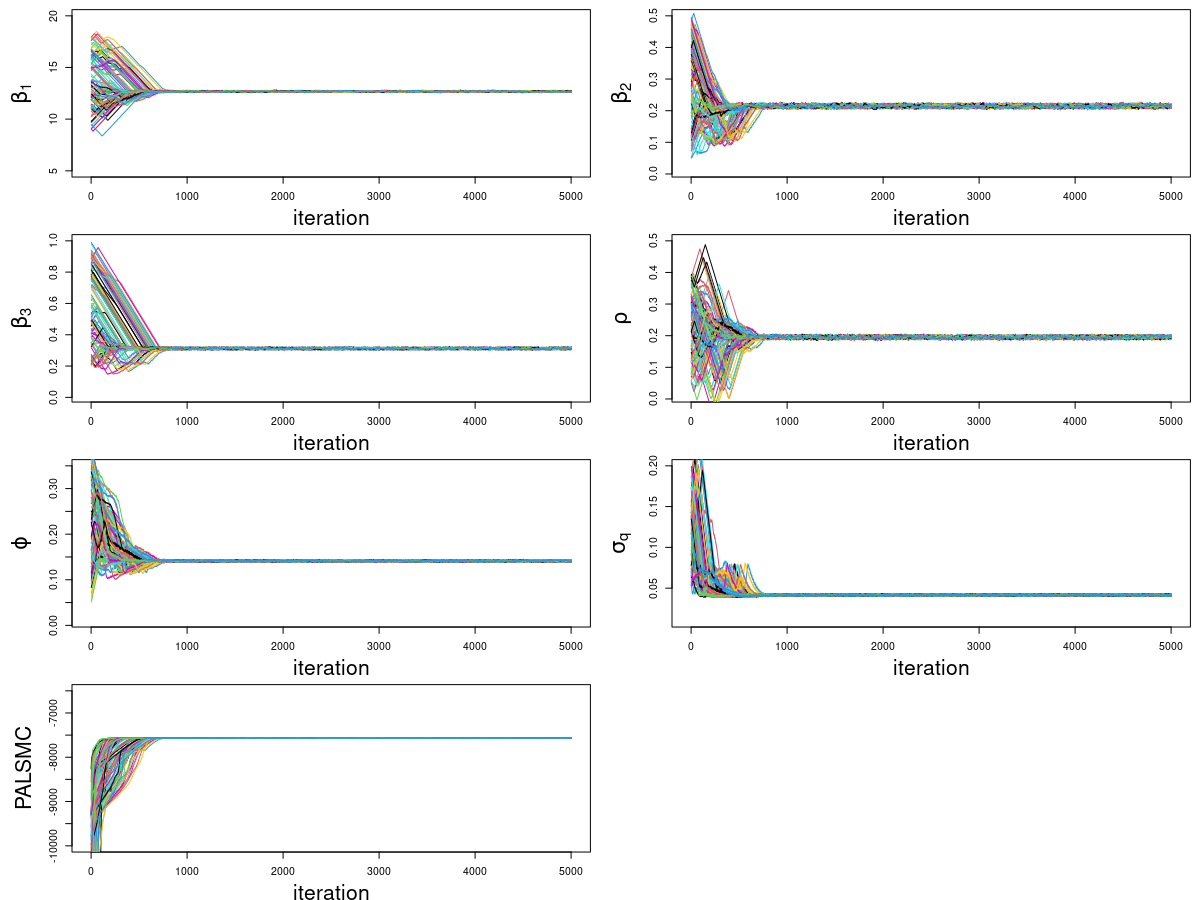}
\caption{Plots showing 100 runs of the optimization procedure the EqOv rotavirus model applied to real data.}
\label{EqOvconvergence}
\end{figure}
\begin{figure}
\includegraphics[width = \textwidth]{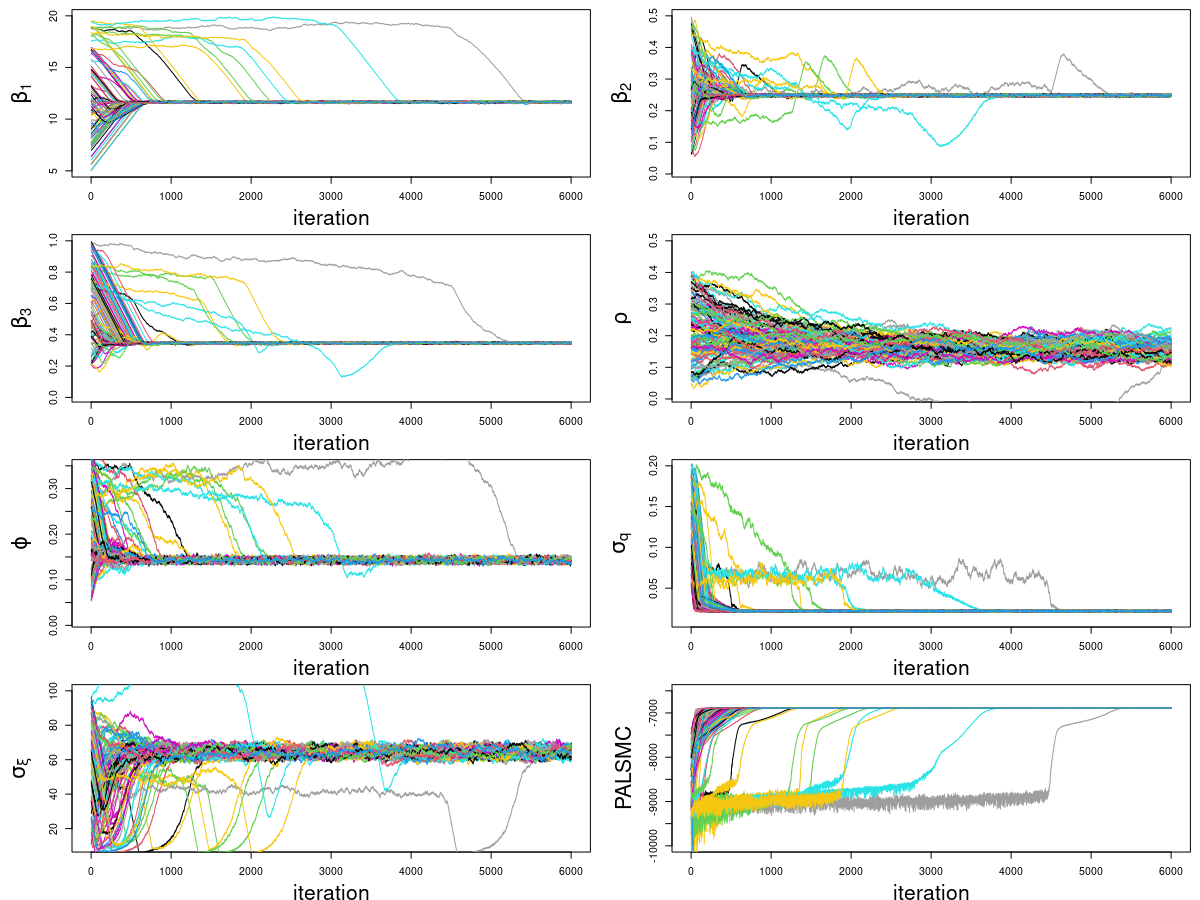}
\caption{Plots showing 100 runs of the optimization procedure the OvOv rotavirus model applied to real data.}
\label{OvOvconvergence}
\end{figure}

\subsubsection{Rotavirus ARMA comparison details.} \label{rotavirusARMA}
The benchmark model consists of an ARMA(2,0,1) model fit to the series $\log(cases +1)$ independently for each age group, taking care to apply the appropriate jacobian transform to the likelihood.

\subsection{Supplementary material for section \ref{sec:measles}} \label{sec:measles_supp}
\subsubsection{Model}
As in \cite{xia2004measles} and \cite{park2020inference}, we assume that $\beta_{r,k}$ follows the school year:
\begin{equation*}
\beta_{r,k} =
\begin{cases}
(1 + 2(1-p)a)\bar{\beta}_k, & \text{during school term,} \\
(1-2pa) \bar{\beta}_k, & \text{during school holidays,}
\end{cases}
\end{equation*}
where $p=0.759$ is the proportion of the year taken up by school terms, $\bar{\beta}_k>0$ is the mean transition rate for city $k$, and $a$ is the relative effect of holidays on transmission. Finally, the new infected and new removed are:
\begin{equation*}
C_{k,t} \sim \text{Bin}(E_{k,t} - F_{k,t}^{(E)}, 1- e^{-h \rho}), \qquad D_{k,t} \sim \text{Bin}(I_{k,t}- F_{k,t}^{(I)}, 1- e^{-h \gamma}),
\end{equation*}
with $h \slash \rho$ mean time spent in the exposed compartment and $h \slash \gamma$ mean recovery time. 
Given the vectors $\bs \delta_{k,t} = \left[\delta_{k,t}^{(S)}\;\delta_{k,t}^{(E)}\;\delta_{k,t}^{(I)}\;\delta_{k,t}^{(R)}\right]^\top \in \mathbb{R}^4_{\geq 0}$ and $\bs \alpha_{k,t} = \left[\alpha_{k,t}^{(1)}\;0\;0\;0 \right]^\top \in \mathbb{R}^4_{\geq 0}$, we have:
\begin{equation*}
F^{(\cdot)}_{k,t} \sim \text{Bin}\left(\cdot_{k,t}, 1- \delta_{k,t}^{(\cdot)} \right), \qquad 
A_{k,t} \sim \text{Pois}\left(\alpha_{k,t}^{(1)}\right),
\end{equation*}
modelling the new births (immigration) into the susceptible population and the deaths (emigration) across compartments. Since there is no reinfection mechanism in the model (a realistic assumption for measles modelling), it is important to have new individuals enter the population to model the recurrent epidemic peaks present in the data. As already mentioned, for the model to capture recurrent peaks, it must accommodate recruitment into the susceptible compartments. Birthrate data for each city of the model is used to do this --- as in \cite{xia2004measles} --- it is assumed that newborns enter the susceptible class after a delay of $4$ years, corresponding to the age an individual enters the high-risk school-age demographic. There is a further `cohort' effect aspect to the model: it is assumed that at the start of the school year, a fraction $c \in (0,1)$ of the lagged births enter the susceptible compartment, the remaining $1-c$ proportion enter at a constant rate throughout the year. This informs the assumed rate parameters $\bs \alpha_{k,t} = [\alpha_{k,t}^{(1)}\;0\;0\;0\;]^\top$ and, similarly, death rate records inform choice of $\bs \delta_{k,t}$. The values used for $\bs\alpha_{k,t}$ and $\bs\delta_{k,t}$ are reported in the data available on the GitHub page.

The observations are aggregated incidence data in the form of cumulative fortnightly transitions from infective to recovered for each of the $40$ cities subject to binomial under-reporting, at times $\tau_r=4r$ for $r = 1,\dots,R$. Observations are modelled as transitions from  infective to recovered compartments because, on discovery, cases are treated with bed rest and hence removed from the population \cite{park2020inference}. Denoting observations as $\bar{\+Y}_{k,r}=\sum_{t=\tau_{r-1}+1}^{\tau_{r}}\+ Y_{k,t}$ where each $\+ Y_{k,t}\in \mathbb{N}^{4 \times 4}$ has each element equal to zero except for the $(3,4)$th element which, conditional on $D_{k,t}$, is distributed:
\begin{equation*}
Y_{k,t}^{(3,4)} \sim \text{Bin}\left(D_{k,t}, Q_{k,r}^{(3,4)} \right), \text{ for } t = \tau_{r-1},\dots \tau_r, r\geq 1, \quad k=1,\dots,J,
\end{equation*}
where  $\+Q_{k,r} \in [0,1]^{4\times4}$ consists of all zeros apart from the $(3,4)$th entry, which is the reporting rate of transitions from infective to recovered. We assume that this rate follows $Q_{k,r}^{(3,4)} \sim \mathcal{N}(\mu_{q,k}, \sigma_q^2)_{\geq 0, \leq 1}$ for $k = 1, \dots , K$, denote this density with $f(\cdot \mid \cdot)$ for the purposes of algorithm \ref{alg:measles_block_lookahead}. The mean under-reporting rate parameters, $\mu_{q,k} \in [0,1], k = 1,\dots,J$, are assumed known for each city and are set to the same values as \cite{park2020inference}, which are available in the data on the GitHub page, $ \sigma_q^2>0$ is to be estimated.
\subsubsection{Inference}
To employ the algorithms described in the methodology section we need to specify the transition matrix $\+K_{r, \bs \eta}$, which in the case of this model is of size $4J \times 4J $. To be more succinct, we can write out a matrix  $\+K_{k,r,\tilde{ \bs \eta}}$ for each city $k= 1, \dots, 40$.
 We define our matrices $\+K_{r,\tilde{ \bs \eta}, \xi,k}$:

\begin{equation*}
\mathbf{K}_{r, \boldsymbol{ \eta}, \xi,k}=\left[\begin{array}{cccc}
e^{-h g_{k}(\beta_{k,r},\boldsymbol{ \eta}, \xi) } & 1-e^{-hg_{k}(\beta_{k,r},\boldsymbol{ \eta}, \xi) } & 0 & 0 \\
0 & e^{-h \rho} & 1-e^{-h \rho} & 0 \\
0 & 0 & e^{-h \gamma} & 1-e^{-h \gamma} \\
0 & 0 & 0 & 1
\end{array}\right],
\end{equation*}
where

\begin{equation*}
g_{k}(\beta,\boldsymbol{ \eta}, \xi) = \beta\xi \cdot\left[  { \eta}^{(k)}+ \sum_{l \neq k} \frac{v_{kl}}{n_{k}}\left\{ { \eta}^{(l)} - { \eta}^{(k)} \right\} \right],
\end{equation*}
%
One can identify matrices:
\begin{equation*}
\+ Z_{k,t} :=\left[\begin{array}{cccc}
S_{k,t}-F^{(S)}_{k,t}-B_{k,t} & B_{k,t} & 0 & 0 \\
0 & E_{k,t}-F^{(E)}_{k,t}-C_{k,t} & C_{k,t} & 0 \\
0 & 0 & I_{k,t}-F^{(I)}_{k,t}-D_{k,t}& D_{k,t} \\
0 & 0 & 0 & R_{k,t}-F^{(R)}_{k,t}
\end{array}\right],
\end{equation*}
and let $\+ Z_{t}$ be block-diagonal with blocks $\+ Z_{k,t} $, $k=1,\ldots,J$. One can take advantage of the block-diagonal structure of $\+K_{t, \boldsymbol{\eta}}$ to implement an efficient block particle filter, see \cite{rebeschini2015can} and \cite{ning2021iterated}, with lookahead resampling scheme \citep{lin2013lookahead}.
\subsubsection*{Proposals for algorithm \ref{alg:measles_block_lookahead}}

The details for the derivation of the proposals used in algorithm \ref{alg:measles_block_lookahead} lines \ref{prop1} and \ref{prop2} are similar to those of section \ref{Qproprota}, so we omit them. Suppressing dependence on the particle, let $\+L_{r,k} =  \sum_{t = \tau_{r-1}+1}^{\tau_{r}} \bs \Lambda_{t,k}$ with $\bs \Lambda_{t,k}$ calculated as per line \ref{measlespred1} (resp. \ref{measlespred2}) of algorithm \ref{alg:rota} and define:

\begin{align}
\hat{\mu}_{r,k} &= \frac{1}{2}\left(\mu_{q,k}  - L_{r,k}^{(3,4)}\sigma_q^2 + \sqrt{(L_{r,k}^{(3,4)}\sigma_q^2- \mu_{q,k} )^2 + 4\bar Y_r^{(3,4)}\sigma_q^2 }\right), \\
\left(\hat{\sigma}_{r,k}\right)^2 &= \left(\frac{\bar Y_r^{(3,4)}}{\left(\hat{\mu}_{r,k}\right)^2} + \frac{1}{\sigma_q^2}\right)^{-1}.
\end{align}
Then in line \ref{prop1} (resp. \ref{prop2}) we make the proposals:

\begin{equation}\label{eq:measlesproposal}
Q_{k,r}^{(3,4)} \sim \mathcal{N}\left(\hat{\mu}_{r,k}, \hat{\sigma}_{r,k}^2 \right)_{\geq 0, \leq 1}. 
\end{equation}

\begin{algorithm}
\caption{PAL within a lookahead block particle filter}\label{alg:measles_block_lookahead}
\begin{algorithmic}[1]
  \Statex {\bf initialize:} $\bar{\bs \lambda}_{0,k}^{(i)} \leftarrow  n_{k,0} \bs \pi_{k,0}$ set $\log \zeta_{0,k}^{(i)} \leftarrow 0$  and set $\log W_{0,k}^{(i)} \leftarrow 0$ for $i = 1$ to $n_{part}$ and $k = 1, \dots, K.$
  \State {\bf for}  $r  \geq 1$:
  \State \quad {\bf for}  $k = 1, \dots, J$:
  \State \quad \quad {\bf for}  $i = 1, \dots,n_{part}$
  \State \quad \quad \quad $\xi_{k,r}^{(i)} \sim \mathrm{Gamma}(\sigma_{\xi},\sigma_{\xi})$ 
  \State \quad \quad \quad  {\bf for}  $  t = \tau_{r-1}+1, \dots,  \tau_r - 1$:
  \State \quad \quad \quad \quad \label{measlespred1} $\bs \Lambda_{t,k}^{(i)} \leftarrow  ( (\bar{\bs \lambda}^{(i)}_{t-1,k} \circ \bs \delta_{t,k}) \otimes \+1_m)\circ \+K_{r,\bs \eta\left(\bar{ \bs \lambda}_{\tau_{r-1},1:J}^{(i)}\right),\xi_{k,r}^{(i)},k} $
  \State \quad \quad \quad \quad $\bar{\bs \lambda}^{(i)}_{t,k} \leftarrow  (\+1_m^\top \bs \Lambda_{t,k}^{(i)})^\top + \bs \alpha_{t,k}$
  \State \quad \quad \quad {\bf end for}
  \State  \quad \quad \quad $\bs \Lambda^{(i)}_{\tau_r,k}\leftarrow  ((\bs \lambda^{(i)}_{\tau_r -1,k} \circ \bs \delta_{\tau_r,k}) \otimes \+1_m)\circ  \+K_{r,\bs \eta\left(\bar{ \bs \lambda}_{\tau_{r-1},1:J}^{(i)}\right),\xi_{r,k}^{(i)},k}$
      \State \quad \quad \quad \label{prop1}$\+Q_{k,r}^{(i)} \sim \pi\left(\cdot \mid \left\{ \bs \Lambda^{(i)}_{t,k}\right\}_{t = \tau_{r-1}+1}^{\tau_{r}}, \bar{\+ Y}_{r,k}, \bs \varphi \right)$ as per \ref{eq:measlesproposal}
        \State \quad \quad \quad $\+M_{r,k}^{(i)} \leftarrow   \sum_{t = \tau_{r-1}+1}^{\tau_{r}} \bs \Lambda^{(i)}_{t,k}\circ\+Q_{k,r}^{(i)} $
  \State \quad \quad \quad $\mathcal{L}(\bar{\+ Y}_{r,k} |\bar{\+ Y}_{1:r-1,k} ) \leftarrow  - \+1_m^\top\+M_{r,k}^{(i)} \+1_m +    \+1_m^\top\left(\bar{\+ Y}_{r,k }\circ \log\+M_{r,k}^{(i)}\right)\+1_m - \+1_m^\top\left(\log\bar{\+Y}_{r,k}! \right)\+1_m$
    \State \quad \quad \quad $\small \log w^{(i)}_{r,k} \leftarrow \mathcal{L}(\bar{\+ Y}_{r,k} |\bar{\+ Y}_{1:r-1,k} ) + \log\left(f(\+Q_{k,r}^{(i)} \mid \+Q_{k,1:r-1}^{(i)})\right) - \log\left(\pi\left(\+Q_{k,r}^{(i)} \mid \left\{ \bs \Lambda^{(i)}_{t,k}\right\}_{t = \tau_{r-1}+1}^{\tau_{r}}, \bar{\+ Y}_{r,k}, \bs \varphi \right)\right)$
     \State \quad \quad \quad $\bar W_{r-1,k}^{(i)} \leftarrow W_{r-1,k}^{(i)}/\sum_j W_{r-1,k}^{(j)}$ for $i = 1$ to $n_{part}$
    \State \quad \quad \quad $\log W_{r,k}^{(i)} \leftarrow \log \bar W_{r-1,k}^{(i)} + \log w_{r,k}^{(i)}$
    \State \quad \quad \quad $\bar{\bs \Lambda}^{(i)}_{\tau_r,k}\leftarrow (\+1_m \otimes \+1_m -\+Q_{k,r}^{(i)} )\circ \bs \Lambda_{\tau_r,k}+ \bar{\+Y}_{{r,k}} \circ \bs \Lambda_{\tau_r,k}^{(i)} \circ \+Q_{k,r}^{(i)}  \oslash \+M_{r,k}^{(i)} $
  \State \quad \quad \quad $\bar{\bs \lambda}^{(i)}_{ \tau_{r},k} \leftarrow (\+ 1_m^\top \bar{ \+ \Lambda}^{(i)}_{ \tau_{r},k})^\top + \bs \alpha_{\tau_r,k}$
    \State \quad \quad {\bf end for}
      \State \quad  {\bf end for}
    \State \quad {\bf for}  $k = 1, \dots, J$:
  \State \quad \quad {\bf for}  $i = 1, \dots,n_{part}$
  \State \quad \quad \quad $\tilde{\xi}_{k,r+1}^{(i)} \sim \mathrm{Gamma}(\sigma_{\xi},\sigma_{\xi})$ 
  \State \quad \quad \quad  {\bf for}  $  t = \tau_{r}+1, \dots,  \tau_{r+1} -1$:
  \State \quad \quad \quad \quad \label{measlespred2} $\bs \Lambda_{t,k}^{(i)} \leftarrow  ( (\bar{\bs \lambda}^{(i)}_{t-1,k} \circ \bs \delta_{t,k}) \otimes \+1_m)\circ \+K_{r+1,\bs \eta\left(\bar{ \bs \lambda}_{\tau_{r},1:J}^{(i)}\right),\tilde{\xi}_{r+1,k}^{(i)},k} $
  \State \quad \quad \quad \quad $\bar{\bs \lambda}^{(i)}_{t,k} \leftarrow  (\+1_m^\top \bs \Lambda_{t,k}^{(i)})^\top + \bs \alpha_{t,k}$
  \State \quad \quad \quad {\bf end for}
  \State  \quad \quad \quad $\bs \Lambda^{(i)}_{\tau_{r+1},k}\leftarrow  ((\bs \lambda^{(i)}_{\tau_{r+1} -1,k} \circ \bs \delta_{\tau_{r+1},k}) \otimes \+1_m)\circ  \+K_{r+1,\bs \eta\left(\bar{ \bs \lambda}_{\tau_{r-1},1:J}^{(i)}\right),\tilde{\xi}_{r+1,k}^{(i)},k}$
      \State \quad \quad \quad \label{prop2} $\tilde{\+Q}_{k,r+1}^{(i)} \sim \pi\left(\cdot \mid \left\{ \bs \Lambda^{(i)}_{t,k}\right\}_{t = \tau_{r}+1}^{\tau_{r+1}}, \bar{\+ Y}_{r+1,k}, \bs \varphi \right)$ as per \ref{eq:measlesproposal}
        \State \quad \quad \quad $\+M_{r+1,k}^{(i)} \leftarrow   \sum_{t = \tau_{r}+1}^{\tau_{r+1}} \bs \Lambda^{(i)}_{t,k}\circ\tilde{\+Q}_{r+1,k}^{(i)} $
  \State \quad \quad \quad $\mathcal{L}(\bar{\+ Y}_{r+1,k} |\bar{\+ Y}_{1:r,k} ) \leftarrow  - \+1_m^\top\+M_{r+1,k}^{(i)} \+1_m +    \+1_m^\top\left(\bar{\+ Y}_{r+1,k }\circ \log\+M_{r+1,k}^{(i)}\right)\+1_m - \+1_m^\top\left(\log\bar{\+Y}_{r+1,k}! \right)\+1_m$
    \State \quad \quad \quad $\small \log w^{(i)}_{r+1,k} \leftarrow \mathcal{L}(\bar{\+ Y}_{r+1,k} |\bar{\+ Y}_{1:r,k} ) + \log\left(f(\+Q_{k,r+1}^{(i)} \mid \+Q_{k,1:r}^{(i)})\right) - \log\left(\pi\left(\+Q_{k,r+1}^{(i)} \mid \left\{ \bs \Lambda^{(i)}_{t,k}\right\}_{t = \tau_{r}+1}^{\tau_{r+1}}, \bar{\+ Y}_{r+1,k}, \bs \varphi \right)\right)$
    \State \quad \quad \quad  $\log \zeta_{r,k}^{(i)} \leftarrow \log W_{r,k}^{(i)} + \log w_{r+1,k}^{(i)}$
  \State \quad \quad {\bf end for}
  \State \quad \quad $\hat{\mathcal{L}}(\bar{\+ Y}_{r,k} |\bar{\+ Y}_{1:r-1,k} )\leftarrow \log \left(\sum_j W_{r,k}^{(j)} \right)$
   \State \quad \quad$\bar \zeta_{r,k}^{(i)} \leftarrow \zeta_{r,k}^{(i)}/\sum_j \zeta_{r,k}^{(j)}$ for $i = 1$ to $n_{part}$
  \State \quad \quad {\bf resample} $\left\{\bar{\bs \lambda}^{(i)}_{\tau_r, k}, W_{r,k}^{(i)}, \zeta_{r,k}^{(i)}\right\}_{i=1}^{n_{part}}$ with weights $\{ \bar{\zeta}_{r,k}^{(i)}\}_{i=1}^{n_{part}}$
   \State \quad \quad $\log W_{r,k}^{(i)} \leftarrow \log W_{r,k}^{(i)} - \log \zeta_{r,k}^{(i)}$
  \State {\bf end for}
\end{algorithmic}
\end{algorithm}

\subsubsection*{Inference}
In each model instance, A,B, and C, described in \ref{sec:measles}, we can define $\bs \vartheta$, $\bar{\bs \theta}_{1:T}$, and $\bs \varphi$:

\begin{itemize}
\item A: $\bs \vartheta = \left[\bs \pi_0 \; \bar \beta \; \rho \; \gamma \; g \; a \; c \right]$, $\left\{\bar{\bs \theta_r}\right\}_{r\geq 0 } = \left\{\left[\xi_{1,r} \; \dots \; \xi_{40,r}\; Q_{1,r}^{(3,4)}\; \dots\; Q_{40,r}^{(3,4)} \right]\right\}_{r\geq 0 }$, and $\bs \varphi = [\sigma_q^2, \sigma_\xi]$.
\item B: $\bs \vartheta = \left[\bs \pi_{1,0}\;\dots\;\bs \pi_{40,0} \; \bar \beta \; \rho \; \gamma \; g \; a \; c \right]$, $\left\{\bar{\bs \theta_r}\right\}_{r\geq 0 } = \left\{\left[\xi_{1,r}, \dots, \xi_{40,r}, Q_{1,r}^{(3,4)}, \dots, Q_{40,r}^{(3,4)} \right]\right\}_{r\geq 0 }$, and $\bs \varphi = [\sigma_q^2, \sigma_\xi]$.
\item C: $\bs \vartheta = \left[\bs \pi_{1,0}\;\dots\;\bs \pi_{40,0} \; \bar \beta_{1}\; \dots \; \bar \beta_{40} \; \rho \; \gamma \; g \; a \; c \right]$, $\left\{\bar{\bs \theta_r}\right\}_{r\geq 0 } = \left\{\left[\xi_{1,r}, \dots, \xi_{40,r}, Q_{1,r}^{(3,4)}, \dots, Q_{40,r}^{(3,4)} \right]\right\}_{r\geq 0 }$, and $\bs \varphi = [\sigma_q^2, \sigma_\xi]$.
\end{itemize}

Each block, labelled $k=1,\dots,J$, corresponds to a specific city. This block structure allows one to perform proposals and weighting locally to each block, avoiding explicit high-dimensional filtering. At time $r$, the lookahead scheme consists of: performing a  `regular' particle propagation and reweighting step (the usual SMC iteration), then we propagate again each particle and run a PAL iteration for time $r+1$, with  `dummy'  particles (used purely for weighting purposes, denoted with tildes in algorithm \ref{alg:measles_block_lookahead}), we then weight the original particles proportionally to the joint likelihood of the regular and dummy particles at times $r$ and $r+1$ - taking care to apply the appropriate correction in the likelihood calculation, dummy particles are then discarded. We found that this scheme greatly reduced Monte Carlo error. See algorithm \ref{alg:measles_block_lookahead} for our implementation.The resulting approximate log-likelihood estimate associated with algorithm \ref{alg:measles_block_lookahead} is:

\begin{equation}
\log p(\bar{\+Y}_{1:J,1:R}) \approx \sum_{r=1}^R \sum_{k=1}^J\hat{\mathcal{L}}(\bar{\+ Y}_{r,k} |\bar{\+ Y}_{1:r-1,k} ).
\end{equation}
The optimisation scheme we used is described in figure \ref{fig:measlesopt}. We report the inferences for model $C$ in table \ref{measles_C_results}.

\subsubsection{Measles ARMA comparison details.} \label{measlesARMA}
The benchmark model consists of an ARMA(2,0,1) model fit to the series $\log(cases +1)$ independently for each each city, taking care to apply the appropriate jacobian transform to the likelihood.

\subsubsection*{Measles projection details.}
The sample, size $300$, of projected case numbers used to produce figure \ref{measlesmap} in the main article were generated by the following workflow:
\begin{enumerate}
\item Presampling $\xi_{k,r}^{(i)}\sim \text{Gamma}(\sigma_\xi,\sigma_\xi)$ with $\sigma_\xi$ set to our point estimate, for $k=1,...,40$, $r = 1,...,4$, and $i = 1,\dots , 300$.
\item Running our PALSMC scheme on the original dataset with 300 particles and parameters set to our point estimates, taking as output a sample of final time-point population state intensity vectors $\bar{\bs \lambda}_{T,k}^{(i)}$.
\item For $i = 1,\dots,300$ and $k = 1,\dots, 40$, propagate the intensity vectors through the transition kernel using the iteration for $t = 1,\dots,16$ (corresponding to 8 weeks):

\begin{align}
\bs \Lambda_{t,k}^{(i)} &=  ( (\bar{\bs \lambda}^{(i)}_{t-1,k} \circ \bs \delta_{t,k}) \otimes \+1_m)\circ \+K_{r,\bs \eta\left(\bar{ \bs \lambda}_{\tau_{r-1},1:J}^{(i)}\right),\xi_{r,k}^{(i)},k} \\
\bar{\bs \lambda}^{(i)}_{t,k} &=  (\+1_m^\top \bs \Lambda_{t,k}^{(i)})^\top + \bs \alpha_{t,k}
\end{align}
Where $\bs \alpha_{t,k}$ and  $\bs \delta_{t,k}$ are chosen according to the assumption that birth rates and death rates remain constant.
\item Simulate $I^{(i)}_{k,t}\sim \text{Pois}(\bar{\bs \lambda}^{(i)}_{t,k})$ for $t$ corresponding to weeks $2,4,6,$ and $8$ for each sample $i = 1,\dots, 300$.
\end{enumerate}
\begin{figure}[h!]
\centering
\includegraphics[width = 0.8\textwidth]{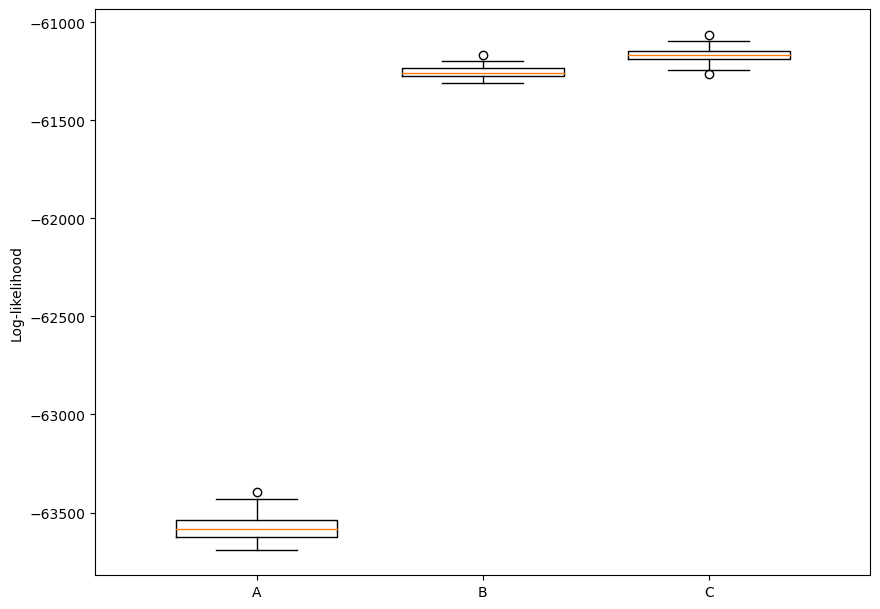}
\caption{Approximate log-likelihood values for the measles data under scenarios A, B, and C. For each scenario, the optimal combination of parameters was obtained through Sequential Least Squares Programming (SLSQP) with target function given by algorithm \ref{alg:measles_block_lookahead} with 5000 particles and lookahead resampling, this scheme was initialised randomly at 100 points over feasible values, we present the best attained values. After the optimization, algorithm \ref{alg:measles_block_lookahead} with 5000 particles and lookahead resampling is run 100 times on the optimized parameters to build the boxplots and estimate the variance of the approximate log-likelihood.}
\label{fig:measlesopt}
\end{figure}

\begin{figure}[h!]
\centering
\begin{longtable}{l c c c c c c c c}
\hline
City &    $\frac{n_{0,k}}{1000}$ &  $\pi_{0,k}^{(1)}$ &  $\pi_{0,k}^{(2)}$ & $\pi_{0,k}^{(3)}$ & $\pi_{0,k}^{(4)}$ &   $R_0$ &  $1/\rho$ &  $1/\gamma$ \\
\hline
BIRKENHEAD &  143 &      0.07594 &      0.00007 &      0.00013 &      0.92387 &  8.47 &         8.49 &           9.53 \\
    BIRMINGHAM & 1118 &      0.04575 &      0.00005 &      0.00013 &      0.95408 &  5.63 &         8.49 &           9.53 \\
     BLACKPOOL &  150 &      0.07210 &      0.00005 &      0.00259 &      0.92525 & 12.93 &         8.49 &           9.53 \\
        BOLTON &  169 &      0.09337 &      0.00007 &      0.00120 &      0.90537 &  9.44 &         8.49 &           9.53 \\
   BOURNEMOUTH &  140 &      0.12166 &      0.00006 &      0.00005 &      0.87822 & 10.62 &         8.49 &           9.53 \\
      BRADFORD &  294 &      0.08243 &      0.00004 &      0.00044 &      0.91708 & 10.22 &         8.49 &           9.53 \\
      BRIGHTON &  158 &      0.07625 &      0.00008 &      0.00035 &      0.92332 & 14.66 &         8.49 &           9.53 \\
       BRISTOL &  443 &      0.07355 &      0.00009 &      0.00206 &      0.92430 &  8.63 &         8.49 &           9.53 \\
       CARDIFF &  245 &      0.09190 &      0.00005 &      0.00058 &      0.90747 &  7.81 &         8.49 &           9.53 \\
      COVENTRY &  257 &      0.11602 &      0.00004 &      0.00018 &      0.88376 &  8.16 &         8.49 &           9.53 \\
         DERBY &  143 &      0.11061 &      0.00006 &      0.00008 &      0.88925 & 10.46 &         8.49 &           9.53 \\
     GATESHEAD &  115 &      0.08601 &      0.00007 &      0.00006 &      0.91386 &  8.28 &         8.49 &           9.53 \\
  HUDDERSFIELD &  130 &      0.09003 &      0.00007 &      0.00022 &      0.90968 & 10.78 &         8.49 &           9.53 \\
          HULL &  302 &      0.06856 &      0.00009 &      0.00083 &      0.93051 &  9.28 &         8.49 &           9.53 \\
       IPSWICH &  104 &      0.08528 &      0.00009 &      0.00000 &      0.91463 &  9.03 &         8.49 &           9.53 \\
         LEEDS &  510 &      0.09935 &      0.00006 &      0.00168 &      0.89891 &  5.92 &         8.49 &           9.53 \\
     LEICESTER &  288 &      0.07103 &      0.00005 &      0.00133 &      0.92759 &  9.00 &         8.49 &           9.53 \\
     LIVERPOOL &  802 &      0.05754 &      0.00004 &      0.00025 &      0.94217 &  5.63 &         8.49 &           9.53 \\
        LONDON & 3389 &      0.04575 &      0.00006 &      0.00021 &      0.95399 &  5.63 &         8.49 &           9.53 \\
    MANCHESTER &  704 &      0.05658 &      0.00003 &      0.00145 &      0.94193 &  7.29 &         8.49 &           9.53 \\
MIDDLESBOROUGH &  146 &      0.06662 &      0.00007 &      0.00067 &      0.93264 & 11.32 &         8.49 &           9.53 \\
     NEWCASTLE &  295 &      0.07129 &      0.00005 &      0.00024 &      0.92843 &  9.19 &         8.49 &           9.53 \\
       NORWICH &  120 &      0.10958 &      0.00005 &      0.00000 &      0.89037 & 12.78 &         8.49 &           9.53 \\
    NOTTINGHAM &  307 &      0.05794 &      0.00004 &      0.00068 &      0.94133 & 11.54 &         8.49 &           9.53 \\
        OLDHAM &  119 &      0.09814 &      0.00007 &      0.00092 &      0.90087 & 11.57 &         8.49 &           9.53 \\
      PLYMOUTH &  209 &      0.08388 &      0.00006 &      0.00077 &      0.91529 & 13.37 &         8.49 &           9.53 \\
    PORTSMOUTH &  240 &      0.07339 &      0.00007 &      0.00295 &      0.92359 &  9.39 &         8.49 &           9.53 \\
       PRESTON &  120 &      0.06501 &      0.00007 &      0.00242 &      0.93251 &  7.52 &         8.49 &           9.53 \\
       READING &  116 &      0.07686 &      0.00005 &      0.00137 &      0.92172 & 12.93 &         8.49 &           9.53 \\
       SALFORD &  178 &      0.08982 &      0.00007 &      0.00109 &      0.90903 &  8.82 &         8.49 &           9.53 \\
     SHEFFIELD &  515 &      0.07818 &      0.00006 &      0.00308 &      0.91869 &  8.10 &         8.49 &           9.53 \\
   SOUTHAMPTON &  181 &      0.11018 &      0.00006 &      0.00391 &      0.88585 & 11.14 &         8.49 &           9.53 \\
      SOUTHEND &  152 &      0.11816 &      0.00008 &      0.00043 &      0.88132 & 16.65 &         8.49 &           9.53 \\
     ST.HELENS &  112 &      0.09871 &      0.00008 &      0.00256 &      0.89864 &  9.89 &         8.49 &           9.53 \\
     STOCKPORT &  142 &      0.09721 &      0.00004 &      0.00231 &      0.90044 &  9.03 &         8.49 &           9.53 \\
         STOKE &  276 &      0.07614 &      0.00006 &      0.00071 &      0.92310 &  8.62 &         8.49 &           9.53 \\
    SUNDERLAND &  178 &      0.06698 &      0.00006 &      0.00088 &      0.93208 & 14.90 &         8.49 &           9.53 \\
       SWANSEA &  162 &      0.08195 &      0.00006 &      0.00052 &      0.91748 & 12.59 &         8.49 &           9.53 \\
       WALSALL &  115 &      0.08378 &      0.00009 &      0.00080 &      0.91533 & 13.75 &         8.49 &           9.53 \\
 WOLVERHAMPTON &  162 &      0.06376 &      0.00006 &      0.00021 &      0.93597 &  8.57 &         8.49 &           9.53 \\
\hline
\caption{\label{measles_C_results} Measles example.  Inferred quantities for model C.}
\end{longtable}
\end{figure}

\end{document}